\documentclass[letterpaper,10pt]{article}
\usepackage[title,toc]{appendix}
\usepackage[margin=1in]{geometry}
\usepackage{setspace}
\usepackage[utf8]{inputenc}
\usepackage[english]{babel}

\usepackage{amssymb,amsmath,amsthm,mathrsfs,mathtools}
\usepackage{graphicx}
\usepackage{subfigure}
\usepackage{enumitem}
\usepackage[names]{xcolor}

\usepackage{bm}

\usepackage[square,numbers,sort&compress]{natbib}

\usepackage{hyperref}
\hypersetup{
    colorlinks=true,
    linkcolor=blue,
    filecolor=magenta,      
    urlcolor=cyan,
}

\newtheorem{theorem}{Theorem}
\newtheorem{corollary}[theorem]{Corollary}
\newtheorem{lemma}[theorem]{Lemma}
\newtheorem{prop}[theorem]{Proposition}

\theoremstyle{definition}\newtheorem{definition}{Definition}
\theoremstyle{definition}\newtheorem{remark}{Remark}
\theoremstyle{definition}\newtheorem{example}{Example}
\theoremstyle{definition}\newtheorem{experiment}{Experiment}

\newcommand{\BR}{\mathbb{R}}
\newcommand{\BE}{\mathbb{E}}
\newcommand{\BN}{\mathbb{N}}
\newcommand{\BZ}{\mathbb{Z}}

\newcommand{\BQ}{\mathbb{Q}}
\newcommand{\SP}{\mathscr{P}}
\newcommand{\SD}{\mathscr{D}}
\newcommand{\SF}{\mathscr{F}}
\newcommand{\pp}{\mathrm{pmmse}}
\newcommand{\mm}{\mathrm{mmse}}
\newcommand{\LL}{\mathrm{lmmse}}
\newcommand{\supp}{\mathrm{supp}}
\newcommand{\wc}{\mathrel{\;\cdot\;}}
\newcommand{\Span}{\mathrm{span}}
\newcommand{\sn}{\mathcal{S}_{[n]}}

\DeclareRobustCommand{\stirling}{\genfrac\{\}{0pt}{}}

\newcommand{\calA}{\mathcal{A}}
\newcommand{\calB}{\mathcal{B}}
\newcommand{\calC}{\mathcal{C}}
\newcommand{\calD}{\mathcal{D}}

\newcommand{\calF}{\mathcal{F}}
\newcommand{\calG}{\mathcal{G}}

\newcommand{\calI}{\mathcal{I}}

\newcommand{\calL}{\mathcal{L}}
\newcommand{\calM}{\mathcal{M}}
\newcommand{\calN}{\mathcal{N}}
\newcommand{\calO}{\mathcal{O}}
\newcommand{\calP}{\mathcal{P}}

\newcommand{\calR}{\mathcal{R}}
\newcommand{\calS}{\mathcal{S}}
\newcommand{\calT}{\mathcal{T}}

\newcommand{\calV}{\mathcal{V}}

\newcommand{\calX}{\mathcal{X}}
\newcommand{\calY}{\mathcal{Y}}

\newcommand{\frakA}{\mathfrak{A}}
\newcommand{\frakB}{\mathfrak{B}}
\newcommand{\frakC}{\mathfrak{C}}
\newcommand{\frakD}{\mathfrak{D}}
\newcommand{\frakE}{\mathfrak{E}}
\newcommand{\frakF}{\mathfrak{F}}

\newcommand{\ba}{\boldsymbol{a}}

\newcommand{\bc}{\boldsymbol{c}}
\newcommand{\bd}{\boldsymbol{d}}
\newcommand{\be}{\boldsymbol{e}}
\newcommand{\bff}{\boldsymbol{f}}
\newcommand{\bg}{\boldsymbol{g}}

\newcommand{\bk}{\boldsymbol{k}}

\newcommand{\bp}{\boldsymbol{p}}
\newcommand{\bq}{\boldsymbol{q}}

\newcommand{\bu}{\boldsymbol{u}}
\newcommand{\bv}{\boldsymbol{v}}
\newcommand{\bw}{\boldsymbol{w}}
\newcommand{\bx}{\boldsymbol{x}}
\newcommand{\by}{\boldsymbol{y}}
\newcommand{\bz}{\boldsymbol{z}}

\newcommand{\bA}{\boldsymbol{A}}
\newcommand{\bB}{\boldsymbol{B}}
\newcommand{\bC}{\boldsymbol{C}}

\newcommand{\bH}{\boldsymbol{H}}
\newcommand{\bI}{\boldsymbol{I}}

\newcommand{\bL}{\boldsymbol{L}}
\newcommand{\bM}{\boldsymbol{M}}
\newcommand{\bN}{\boldsymbol{N}}

\newcommand{\bR}{\boldsymbol{R}}

\newcommand{\bU}{\boldsymbol{U}}
\newcommand{\bV}{\boldsymbol{V}}
\newcommand{\bW}{\boldsymbol{W}}
\newcommand{\bX}{\boldsymbol{X}}
\newcommand{\bY}{\boldsymbol{Y}}
\newcommand{\bZ}{\boldsymbol{Z}}

\newcommand{\bSigma}{\boldsymbol{\Sigma}}

\newcommand{\btheta}{\pmb{\theta}}

\newcommand{\bmu}{\pmb{\mu}}
\newcommand{\bnu}{\pmb{\nu}}
\newcommand{\blambda}{\pmb{\lambda}}
\newcommand{\bpsi}{\pmb{\psi}}
\newcommand{\bxi}{\pmb{\xi}}
\newcommand{\bvarepsilon}{\pmb{\varepsilon}}
\newcommand{\bbeta}{\pmb{\beta}}
\newcommand{\bsigma}{\pmb{\sigma}}
\newcommand{\bcalX}{\pmb{\calX}}

\begin{document}

\title{Measuring Information from Moments}

\date{}

\author{Wael Alghamdi and Flavio P. Calmon\thanks{W. Alghamdi and F. P. Calmon are with the John A. Paulson School of Engineering and Applied Sciences at Harvard University. E-mails: \textsf{alghamdi@g.harvard.edu, flavio@seas.harvard.edu}. 
This work was supported in part by the National Science Foundation under Grants CIF 1900750 and CAREER 1845852. This paper was presented in part at ISIT 2019~\cite{Moments}, and will be presented in part at ISIT 2021~\cite{CondExp}.}}

\maketitle

\begin{abstract}
We investigate the problem of representing information measures in terms of the moments of the underlying random variables. First, we derive polynomial approximations of the conditional expectation operator. We then apply these approximations to bound the best mean-square error achieved by a polynomial estimator---referred to here as the PMMSE. In Gaussian channels, the PMMSE coincides with the minimum mean-square error (MMSE) if and only if the input is either Gaussian or constant, i.e., if and only if the conditional expectation of the input of the channel given the output is a polynomial of degree at most 1. By combining the PMMSE with the I-MMSE relationship, we derive new formulas for information measures (e.g., differential entropy, mutual information) that are given in terms of the moments of the underlying random variables. As an application, we introduce estimators for information measures from data via approximating the moments in our formulas by sample moments. These estimators are shown to be asymptotically consistent and possess desirable properties, e.g., invariance to affine transformations when used to estimate mutual information.
\end{abstract}

\setcounter{tocdepth}{1}
\tableofcontents

\newpage

\section{Introduction}
A fundamental formula in information theory is the I-MMSE relation~\cite{Guo2005}, which shows that in Gaussian channels the mutual information is the integral of the minimum mean-square error (MMSE):
\begin{equation} \label{qh}
    I(X;\sqrt{\gamma}X+N) = \frac12 \int_0^\gamma \mm\left(X \mid \sqrt{t} X + N\right) \, dt.
\end{equation}
Here, $X$ has finite variance and $N$ is standard normal independent of $X.$ In this paper, we build on this relation to express information measures of two random variables $X$ and $Y$ as functions of their moments. For example, whenever $X$ and $Y$ are continuous and subexponential there is a sequence of rational functions $\{ \rho_n \}_{n\in \BN}$---each completely determined by finitely many of the moments of $X$ and $Y$---such that the mutual information is
\begin{equation} \label{qk}
    I(X;Y) = \lim_{n\to \infty} \int_\BR \rho_n(t) \, dt.
\end{equation}

We derive the new expression~\eqref{qk} and a similar formula for differential entropy in three steps. First, we produce polynomial approximations of conditional expectations. Second, we apply these approximations to bound the mean-square error of reconstructing a hidden variable $X$ from an observation $Y$ using an estimator that is a polynomial in $Y$. We call this approximation the PMMSE, in short for Polynomial MMSE. Finally, we use the PMMSE in the I-MMSE relation~\eqref{qh} to approximate mutual information (as in~\eqref{qk}) and differential entropy.

\subsection{Overview of Main Results}

The crux of our work is the study of polynomial approximations of conditional expectations. We produce polynomial approximations for $\BE[X\mid Y]$ for general random variables $X$ and $Y$ in Section~\ref{az}. The  polynomial formulas are instantiated for $\BE[X\mid X+N],$ where $N\sim \calN(0,1)$ is independent of $X,$ and studied further in Section~\ref{jh}. 

A surprising result that motivates the study of polynomial approximations of conditional expectations is a negative answer to the question: Can $\BE[X\mid X+N]$ be a polynomial of degree at least 2? We prove in Theorem~\ref{np} that among all integrable random variables $X$ (i.e., $\BE[|X|]<\infty),$ the only way that $\BE[X\mid X+N]$ can be a polynomial is if $X$ is Gaussian or constant. In other words, $\BE[X\mid X+N]$ is linear or constant if it is a polynomial.\footnote{The fact that $\BE[X\mid X+N]$ is a polynomial if and only if $X$ is Gaussian or constant can be proved in view of the inequality $\BE[X^2 \mid X+N = y] = O(y^2),$ which is derived in~\cite[Proposition 1.2]{Fozunbal} when $X$ has finite variance. Here, we extend the negative conclusion to any integrable $X$.}

Nevertheless, we produce a sequence of polynomials that converges to the conditional expectation. For a finite-variance $X$ and a light-tailed non-finitely-atomic $Y$ (see Theorems \ref{ar} and \ref{hw}), we derive the mean-square polynomial approximation of conditional expectation
\begin{equation} \label{ko}
    \BE[X\mid Y] = \lim_{n\to \infty} \BE\left[ (X,XY,\cdots,XY^n) \right] \bM_{Y,n}^{-1} \left( \begin{array}{c}
    1 \\
    Y \\
    \vdots \\
    Y^n 
    \end{array} \right),
\end{equation}
where the $n$-th order Hankel matrix of moments of $Y$ is denoted by
\begin{equation}
    \bM_{Y,n}:=\left(\BE \left[ Y^{i+j} \right]\right)_{0\le i,j \le n}.
\end{equation}
The light-tail condition on $Y$ is satisfied if $Y$ has a moment-generating function (MGF) or, more generally, if it satisfies Carleman's condition \cite{Lubinsky2007}
\begin{equation} \label{ql}
    \sum_{n\in \BN_{>0}} \BE\left[ Y^{2n} \right]^{-1/(2n)} = \infty.
\end{equation}
Imposing the light-tail condition on $Y$ ensures that polynomials are dense in $L^2(P_Y)$; in this case, the random variable $\BE\left[X \mid Y\right]\in L^2(P_Y)$ will be a limit of polynomials, and~\eqref{ko} gives one such limit. 

Note that the limit~\eqref{ko} holds even when $Y$ is not a Gaussian perturbation of $X.$ Also, in the special case that $Y=X+N$ for $N\sim \calN(0,1)$ independent of $X,$ the MGF of $Y$ exists if that of $X$ exists. In general, when stating our results we do not make an implicit assumption on the relationship between $X$ and $Y$ unless explicitly stated.

The expressions in the right hand side of \eqref{ko} are the orthogonal projections of the conditional expectation $\BE[X\mid Y]$ (or, equivalently, of $X$) onto finite-dimensional subspaces of polynomials in $Y$ of a certain degree. In other words, for each $n\in \BN,$ the polynomial in $Y$ defined by
\begin{equation}
    E_n[X\mid Y] := \BE\left[ (X,XY,\cdots,XY^n) \right] \bM_{Y,n}^{-1} \left( \begin{array}{c}
    1 \\
    Y \\
    \vdots \\
    Y^n 
    \end{array} \right),
\end{equation}
is the orthogonal projection of $\BE[X \mid Y]$ onto the subspace $\SP_n(Y):=\{p(Y) \mid p\in \SP_n\},$ where $\SP_n$ is the set of polynomials in one variable of degree at most $n$ with real coefficients. This projection characterization, in turn, makes $E_n[X \mid Y]$ the best polynomial approximation (in the $L^2(P_Y)$-norm sense) of the conditional expectation $\BE[X \mid Y].$ Specifically, $E_n[X\mid Y]$ uniquely solves the approximation problem
\begin{equation} \label{kp}
    E_n[X\mid Y] = \underset{q(Y)\in \SP_n(Y)}{\mathrm{argmin}} ~  \|q(Y) - \BE[X\mid Y] \|_{L^2(P_Y)}.
\end{equation}
Equation~\eqref{kp} is taken as the definition of $E_n[X\mid Y]$ for random variables $X$ and $Y$ satisfying $\BE[X^2]<\infty$ and $\BE[Y^{2n}]<\infty$ (see Definition~\ref{ho}). The approximation error $\|E_n[X\mid Y] - \BE[X\mid Y]\|_{L^2(P_Y)}$ in~\eqref{kp} can also be quantified when the estimation is done in Gaussian channels, which we briefly overview next.

If $Y=X+N$ for standard normal $N$ independent of $X,$ and if $X$ has a probability density function (PDF) or a probability mass function (PMF) $p_X$ that is compactly-supported, even, and decreasing over $[0,\infty)\cap \supp(p_X),$ then the approximation error $\|E_n[X\mid Y] - \BE[X\mid Y]\|_{L^2(P_Y)}$ decays faster than any polynomial in the degree $n$ (Theorem \ref{kw}). More precisely, for all positive integers $n$ and $k$ satisfying $n\ge \max(k-1,1)$ we have that
\begin{equation} \label{xf}
    \left\|E_n[X\mid Y] - \BE[X\mid Y] \right\|_{L^2(P_Y)} = O_{X,k}\left( \frac{1}{n^{k/2}} \right).
\end{equation}  
The implicit constants in~\eqref{xf} depend only on $X$ and $k.$ 

The bound on the rate of decay in~\eqref{xf} is derived by applying recently developed results on polynomial approximation in weighted Hilbert spaces \cite{Lotsch2009}. The key result (Theorem \ref{og}) is a uniform (in $X$) bound on the derivatives of the conditional expectation of the form
\begin{equation} \label{xg}
    \sup_{\BE[|X|]<\infty} \left\| \frac{d^{k}}{dy^{k}} ~ \mathbb{E}[X \mid Y=y] \right\|_{L^2(P_Y)} \le ~ \eta_k
\end{equation}
for each $k\in \BN,$ where the $\eta_k$ are absolute constants. In~\eqref{xg}, $Y=X+N$ for $N\sim \calN(0,1)$ independent of $X.$

From an estimation-theoretic point of view, the operators $E_n$ are natural generalizations of the linear minimum mean-square error (LMMSE) estimate. By the orthogonality property of conditional expectation, the characterization in equation \eqref{kp} is equivalent to the characterization
\begin{equation}
    E_n[X\mid Y] = \underset{q(Y)\in \SP_n(Y)}{\mathrm{argmin}} ~ \BE\left[ \left( q(Y) - X \right)^2 \right].
\end{equation}
Hence, we call the random variable $E_n[X \mid Y]$ the $n$-th degree polynomial minimum-mean squared error (PMMSE) estimate (of $X$ given $Y$) and denote the estimation error by
\begin{equation} \label{ks}
    \pp_n(X \mid Y) := \BE\left[ \left( E_n[X \mid Y] - X \right)^2 \right].
\end{equation}

We next overview several moments-based representations for distribution functionals proved in this work. Recall that the minimum mean-square error (MMSE) is given by $\mm(X\mid Y) = \BE\left[ \left( \BE[ X \mid Y ] - X \right)^2 \right].$ Using the limit $E_n[X \mid Y]\to \BE[X\mid Y],$ an application of the triangle inequality yields that the PMMSEs converge to the MMSE
\begin{equation} \label{kq}
    \mm(X\mid Y) = \lim_{n\to \infty} \pp_n(X\mid Y)
\end{equation}
when $X$ has finite variance and $Y$ satisfies Carleman's condition~\eqref{ql} (see Theorem~\ref{ar}).

Utilizing the approximation of the MMSE given by \eqref{kq} in the I-MMSE relation, we prove new formulas for differential entropy and mutual information expressing them primarily in terms of moments. For example, a corollary of the I-MMSE relation states that the differential entropy of a finite-variance continuous random variable $X$ can be expressed in terms of the MMSE as \cite{Guo2005}
\begin{equation}
    h(X) = \frac12 \int_0^\infty  \mm(X\mid \sqrt{t}X+N) - \frac{1}{2\pi e + t} \, dt,
\end{equation}
where $N$ is standard normal and is independent of $X.$ In view of the approximability of the MMSE by the PMMSEs \eqref{kq}, we obtain in Proposition \ref{jn} an approximation of differential entropy of an $X$ that has a MGF\footnote{Interestingly, the light-tail condition here, which is required for $\sqrt{t}X+N,$ might necessitate that $X$ satisfy a condition stronger than Carleman's condition in equation \eqref{ql} (though $N$ satisfies Carleman's condition, it might be the case that $X$ satisfies Carleman's condition but the sum $\sqrt{t}X+N$ does not, see~\cite[Proposition 3.1]{Berg1985}); nevertheless, assuming that $X$ has a MGF is sufficient, as then $\sqrt{t}X+N$ would necessarily have a MGF.} as
\begin{equation} \label{kr}
    h(X) = \lim_{n\to \infty} h_n(X)
\end{equation}
where we have introduced the distribution functionals $h_n(X)$ to be
\begin{equation} \label{kt}
    h_n(X) := \frac12 \int_0^\infty   \pp_n(X\mid \sqrt{t}X+N) - \frac{1}{2\pi e + t} \, dt.
\end{equation}

Figure \ref{fig:hn} provides an illustration of how $h_n(X)$ approximates $h(X),$ where $X$ has a chi distribution with two degrees of freedom (commonly denoted by $\chi_2$). It is evident from the figure that $h_n(X)$ approximates the differential entropy of $X$ monotonically more accurately as $n$ grows; indeed, this is true in general, because the monotonicity $\SP_1 \subset \SP_2 \subset \cdots$ implies that the convergence in \eqref{kr} is monotone
\begin{equation}
    h_1(X)\ge h_2(X) \ge \cdots \ge h(X),
\end{equation}
with $h_1(X) = \frac12 \log(2\pi e \sigma_X^2)$ being the differential entropy of a Gaussian with the same variance as that of $X.$

\begin{figure}[!tb]
    \centering
    \includegraphics[width=0.7\textwidth]{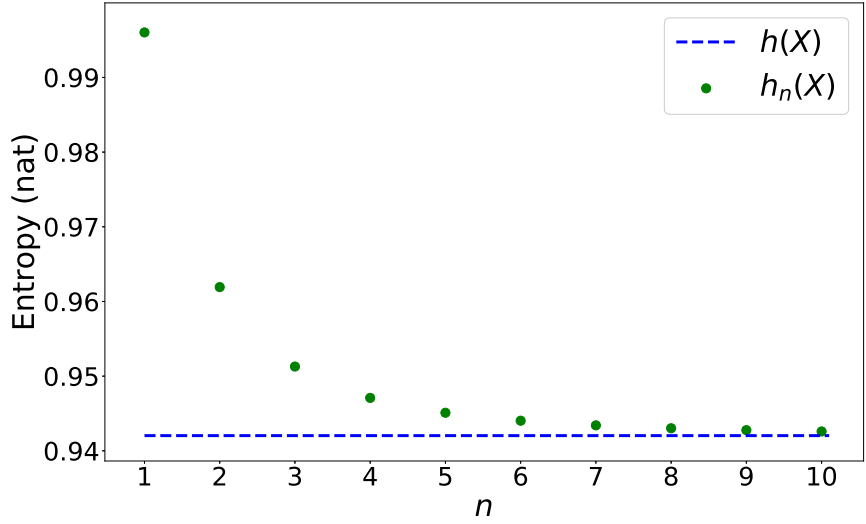}
    \caption{Comparison of the values of $h_n(X)$ (green dots) against the true value $h(X)$ (dashed blue line) for $n\in \{1,\cdots,10\}$ and $X \sim \chi_2.$ We have that $h(X)<h_{10}(X)<h(X)+6\cdot 10^{-4}.$}
    \label{fig:hn}
\end{figure}

Furthermore, closure properties of polynomial subspaces under affine transformations imply that the PMMSE behaves under affine transformations exactly as the MMSE does:
\begin{equation} \label{lw}
    \pp_n(aX+b \mid cY+d) = a^2 ~ \pp_n(X\mid Y)
\end{equation}
for constants $a,b,c,$ and $d$ such that $c \neq 0$ (Proposition~\ref{hn}). Thus, the distribution functionals $h_n$ behave under affine transformations exactly as differential entropy does, namely,
\begin{equation}
    h_n(aX+b) = h_n(X) + \log |a|
\end{equation}
for $a\neq 0$ (Proposition~\ref{js}).

The most noteworthy implication of \eqref{kr} is that it is a formula for differential entropy entirely in terms of moments. This fact follows directly from \eqref{ko}, \eqref{ks}, and \eqref{kt}. A closer analysis of the PMMSE under Gaussian perturbation yields a characterization that is more amenable to numerical computation. More precisely, we show in Theorem~\ref{ir} that the PMMSE when estimating a $2n$-times integrable random variable $X$ given its output under a Gaussian channel
\begin{equation} \label{lx}
    \pp_n(X,t) := \pp_n(X\mid \sqrt{t}X+N)
\end{equation}
is a rational function in the signal-to-noise ratio (SNR) $t$
\begin{equation} \label{lu}
    \pp_n(X,t) = \frac{\sigma_X^2 \prod_{k=1}^n k! + \cdots + (\det  \bM_{X,n}) t^{d_n-1}}{\prod_{k=1}^n k! + \cdots + (\det  \bM_{X,n}) t^{d_n}}
\end{equation}
where the degree of the denominator is $d_n = \binom{n+1}{2}.$ Here, $N$ is standard normal and is independent of $X.$ For example, if $X$ is zero-mean and unit-variance, denoting $\calX_k = \BE[X^k],$ we have the formula
\begin{equation} \label{lv}
    \pp_2(X,t) = \frac{2+4t+(\calX_4-\calX_3^2-1)t^2}{2+6t+(\calX_4+3)t^2+(\calX_4-\calX_3^2-1)t^3}.
\end{equation}
For a general $n\in \BN,$ the coefficients in both the numerator and denominator of the PMMSE in \eqref{lu} are ``homogeneous" polynomials in the moments of $X$ (i.e., for a single coefficient $c(X)$ there is a $k_c \in \BN$ such that $c(\alpha X)=\alpha^{k_c}c(X)$); this and further characterizations of the coefficients are given in Theorem~\ref{iv}.

Pointwise convergence of $\pp_n(X,t)$ to the MMSE
\begin{equation}
    \mm(X,t) := \mm(X \mid \sqrt{t}X+N)
\end{equation}
follows immediately from the general PMMSE-to-MMSE convergence in \eqref{kq} if $X$ has a MGF. In fact, continuity of both the PMMSE and the MMSE in the SNR and the monotonicity of the convergence in \eqref{kq} imply that the convergence is uniform (Theorem~\ref{an})
\begin{equation}
    \lim_{n\to \infty} ~ \sup_{t\ge 0} ~\pp_n(X,t) - \mm(X,t) =  0.
\end{equation}
Figure \ref{fig:pmmses} shows an example of how the PMMSE approximates the MMSE for a random variable $X$ that takes the values $1$ and $-1$ equiprobably. In this case, the MMSE is given by
\begin{equation}
    \mm(X,t) = 1 - \frac{1}{\sqrt{2\pi}} \int_{\BR} \tanh(z\sqrt{t})^2 e^{-(z+\sqrt{t})^2/2} \, dz,
\end{equation}
whereas the functions $\pp_n(X,t)$ are rational in $t,$ e.g., for $n=1$ we have the LMMSE
\begin{equation}
    \pp_1(X,t) = \frac{1}{1+t}
\end{equation}
and for $n=5$ we have the $5$-th degree PMMSE\footnote{In general, $\pp_5(Z,t)$ is a ratio of a degree-$14$ polynomial by a degree-$15$ polynomial as in equation~\eqref{lu}. In the special case of a Rademacher random variable, significant cancellation occurs and we obtain equation~\eqref{ya}.} 
\begin{equation} \label{ya}
    \pp_5(X,t) =  \frac{  45+ 360 t+ 675 t^{2}+300 t^{3}   }{  45+ 405 t+ 1035 t^{2}+ 1005 t^{3}+ 450 t^{4}+ 96 t^{5}+8 t^{6}}.
\end{equation}

\begin{figure}[!tb]
    \centering
    \includegraphics[width=0.7\textwidth]{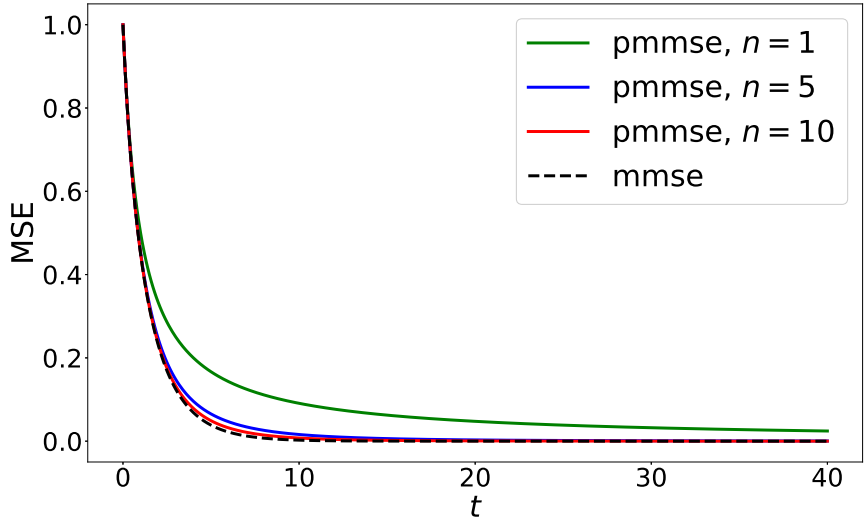}
    \caption{Comparison of the graphs of the functions $t\mapsto \pp_n(X,t)$ (solid lines) against the function $t \mapsto \mm(X,t)$ (dashed black line) for $n\in \{1,5,10\}$ and $X \sim \mathrm{Unif}(\{\pm 1\}).$}
    \label{fig:pmmses}
\end{figure}

The convergence of the distribution functionals $h_n$ to differential entropy gives rise to approximations of mutual information between a discrete random variable and a continuous random variable. Let $X$ and $Y$ be jointly distributed random variables such that $Y$ is finitely-atomic and, for each $y\in \supp(Y),$ the random variable $X_y$ obtained from $X$ conditioned on $Y=y$ is continuous. Then, the sequence $\{ I_n(X;Y) \}_{n\in \BN}$ defined by
\begin{equation} \label{jl}
    I_n(X;Y) := h_n(X) - \sum_{y \in \supp(Y)} \mathrm{Pr}(Y=y) h_n(X_y)
\end{equation}
converges to $I(X;Y)$ when $X$ has a MGF (Theorem \ref{jo}). 

We also obtain a moments-based formula for $I(X;Y)$ when both $X$ and $Y$ are continuous. This expression follows, in part, by extending our results on random variables to random vectors. The multidimensional generalization, in turn, is straightforward in view of the existence of analogous results for the I-MMSE relation \cite{Guo2005} and on the denseness of polynomials in the multidimensional setting. One notable exception is that we derive the MMSE dimension of a continuous random vector, namely, that 
\begin{equation}
    \lim_{t\to \infty} t \cdot \mm(\bX \mid \sqrt{t} \bX+\bN) =  \sum_{i=1}^m \sigma_{N_i}^2
\end{equation}
for continuous square-integrable $m$-dimensional random vectors $\bX$ and $\bN=(N_1,\cdots,N_m)^T$ whenever the density function of $\bN$ is bounded and decaying sufficiently fast (see Theorem~\ref{qm}).

Approximating the PMMSE in \eqref{lu} by plugging in sample moments in place of moments, we derive an estimator of several information measures. If $\{X_j\}_{j=1}^m$ are i.i.d. samples taken from the distribution of $X,$ then a uniform random variable over the samples $U \sim \mathrm{Unif}(\{X_j\}_{j=1}^m)$ provides an estimate $\pp_n(U,t)$ for $\pp_n(X,t).$ The moments of $U$ converge to the moments of $X$ by the law of large numbers. Further, using $\pp_n(U,t)$ to estimate $\pp_n(X,t)$ is a consistent estimator by the continuous mapping theorem, as the PMMSE is a continuous function of the moments. The same can be said of $h_n(U)$ as an estimate of $h_n(X),$ or of $I_n(U;V)$ as an estimate of $I(X;Y)$ when $(U,V) \sim \mathrm{Unif}(\{(X_j,Y_j)\}_{j=1}^m)$ where $\{(X_j,Y_j)\}_{j=1}^m$ are i.i.d. samples drawn according to the distribution of $(X,Y).$ These estimators also satisfy some desirable properties. For example, the behavior of the PMMSE under affine transformations \eqref{lw} implies that the estimate of the PMMSE from data is robust to (injective) affine transformations
\begin{equation}
    \pp_n(aU+b,t) = a^2 \pp_n(U,a^2 t),
\end{equation}
and so are the estimates of differential entropy
\begin{equation}
    h_n(aU+b) = h_n(U) + \log|a|
\end{equation}
and mutual information
\begin{equation}
    I_n(aU+b;V) = I_n(U;V).
\end{equation}

\subsection{Related Literature}

The mutual information between the input and output of the Gaussian channel is known to have an integral relation with the MMSE, referred to in the literature as the I-MMSE relation. This connection was made in the work of Guo, Shamai, and Verd\'{u} in \cite{Guo2005}. Extensions of the I-MMSE relation were investigated in \cite{Zakai2005,Guo2009,Verdu2010,Guo2011,Wu2012,Asnani2014,Han2015,Dytso2017,Dytso2020}, and applications have been established, e.g., in optimal power allocation \cite{Lozano2006} and monotonicity of non-Gaussianness \cite{Tulino2006}. Our work is inscribed within this literature. 

We introduce the PMMSE approximation of the MMSE, derive new representations of distribution functionals in terms of moments, and introduce estimators based on these new representations. We note that utilizing higher-order polynomials as proxies of the MMSE has appeared, e.g., in approaches to denoising \cite{Cha2018}. Also, studying smoothed distributions, e.g., via convolutions with Gaussians, has generated recent interest in the context of information theory~\cite{Calmon2018,Polyanskiy2016a} and learning theory \cite{Goldfeld2019,Goldfeld2020}.

At the heart of our work is the Bernstein approximation problem, on which a vast literature exists within approximation theory. The original Bernstein approximation problem extends Weierstrass approximation to the whole real line by investigating whether polynomials are dense in $L^\infty(\mu)$ for a measure $\mu$ that is absolutely continuous with respect to the Lebesgue measure. Works such as those by Carleson \cite{Carleson1951} and Freud \cite{Freud1977a}, and eventually the more comprehensive solution given by Ditzian and Totik \cite{Lotsch2009}---which introduces moduli of smoothness, a natural extension of the modulus of continuity---show that tools used to solve the Bernstein approximation problem can be useful for the more general question of denseness of polynomials in $L^p(\mu)$ for all $p\ge 1$ (see \cite{Lubinsky2007} for a comprehensive survey). In particular, the case $p=2$ has a close relationship with the Hamburger moment problem, described next.

The Hamburger moment problem asks whether a countably-infinite sequence of real numbers corresponds to a unique positive Borel measure on $\BR.$ A connection between this problem and the Bernstein approximation problem is that if the Hamburger moment problem has a positive answer for the sequence of moments of $\mu$ then polynomials are dense in $L^2(\mu),$ see \cite{Berg1981}. In the context of information theory, the application of the Bernstein approximation problem and the Hamburger moment problem has appeared in \cite{Makur2017}.

The denominator of the PMMSE in Gaussian channels, $\det \bM_{\sqrt{t}X+N,n}$ in \eqref{lu}, as well as the leading coefficient of both the numerator and the denominator, $\det \bM_{X,n},$ can be seen as generalizations of the Selberg integral. Denote
\begin{equation} \label{qi}
    \calI_n(\varphi) = \int_{\BR^{n+1}} \prod_{0\le i <j\le n} (y_i-y_j)^2 \prod_{i=0}^n \varphi(y_i) \, dy_0 \cdots dy_n.
\end{equation}
If $\varphi$ is the PDF of a Beta distribution or a standard normal distribution, then $\calI_n(\varphi)$ is the Selberg integral or the Mehta integral, respectively (both with parameter $\gamma=1$). For a continuous random variable $Y$ whose PDF is $p_Y,$
\begin{equation}
    \det \bM_{Y,n} = \frac{1}{(n+1)!} \calI_n(p_Y).
\end{equation}
The Vandermonde-determinant power $\prod_{i<j} (y_i-y_j)^2$ in the integrand in \eqref{qi} bears a close connection with the quantum hall effect \cite{Scharf1994,King2004}. The connection arises via expanding powers of the Vandermonde determinant and investigating which of the ensuing monomials have nonzero coefficients.

We quantify the rate of convergence of the PMMSE to the MMSE in Theorem \ref{kw}, for which the key ingredient is the bound in Theorem \ref{og} on the derivatives of the conditional expectation. The first-order derivative of the conditional expectation in Gaussian channels has been treated in \cite{DytsoDer}. We note that in parallel to this work the authors were made aware that the higher-order derivative expressions in Proposition~\ref{ars} were also derived in~\cite{Dytso2021}. We also extend the proofs for the MMSE dimension in the continuous case as given in \cite{Wu2011} to higher dimensions.

Distribution functionals, such as mutual information, are popular metrics for quantifying associations between data (e.g.,  \cite{Goodarzi2012,Carro2010,Fleuret2004}), yet reliably estimating distributional functions directly from samples is a non-trivial task. The naive route of first estimating the underlying distribution is generally impractical and imprecise. To address this challenge, a growing number of distribution functionals' estimators have recently been proposed within the information theory and computer science communities (see, e.g., \cite{Kraskov2004,Valiant2011,Jiao2015,Wu2016,Gao2017}). The estimators proposed in this paper satisfy desirable properties, such as shift invariance and scale resiliency, without the need to estimate the underlying distributions.

\subsection{Paper Organization}

In Section~\ref{qj}, we show that the conditional expectation in a Gaussian channel is linear if it is a polynomial. We introduce the PMMSE, prove its convergence to the MMSE, and provide explicit formulas for it in Section~\ref{az}. Basic properties of the PMMSE are given in Section \ref{ca}. A more focused treatment of the Gaussian-channel case occupies Section \ref{jh}. We derive a bound on the derivatives of the conditional expectation in Section~\ref{pc}. This bound, in turn, is used to quantify how close the PMMSE approximates the MMSE as the polynomial degree increases, which is the focus of Section \ref{ji}. We translate the expressions of the PMMSE into new formulas for differential entropy in Section \ref{jj}. We then develop a moments-based formula for the mutual information in Section \ref{ba}. Generalizations to multiple dimensions and to pairs of continuous random variables are given in Section \ref{jc}. The proposed PMMSE-based estimator is introduced, and its performance illustrated with simulations, in Section~\ref{bb}.

\subsection{Notation} 

Throughout, we fix a probability space $(\Omega,\calF,P).$ Let $\calB$ denote the Borel $\sigma$-algebra of $\BR.$ A random variable (RV) is a function $X: \Omega \to \BR$ that is $(\calF,\calB)$-measurable. For any sub-$\sigma$-algebra $\Sigma\subset \calF,$ we denote 
\begin{equation}
    \calM(\Sigma) := \left\{ f: \Omega \to \BR ~ ; ~ f \text{ is } (\Sigma,\calB)\text{-measurable} \, \right\}.
\end{equation}
The $\sigma$-algebra generated by a RV $X$ is denoted by $\sigma(X):=\{ X^{-1}(B) ~ ; ~ B\in \calB\}\subset \calF,$ where $X^{-1}(B)$ denotes the set-theoretic inverse. A function $g:\BR\to \BR$ is called a Borel function if it is $(\calB,\calB)$-measurable. The set of $\sigma(X)$-measurable RVs can be characterized by~\cite[Section II.4.5, Theorem 3]{Shiryaev}
\begin{equation}
    \calM(\sigma(X)) = \{g\circ X ~ ; ~ g:\BR \to \BR  ~ \text{is a Borel function} \, \}.
\end{equation}
For $1\le q < \infty,$ the weighted $L^q$-space $L^q(\Sigma)$ is defined by
\begin{equation}
    L^q(\Sigma) := \left\{ f\in \calM(\Sigma) ~ ; ~ \int_{\Omega} |f|^q \, dP < \infty \right\}. 
\end{equation}
We denote the norm of the Banach space $L^q(\calF)$ (for $1\le q < \infty$) by $\|\wc\|_q,$ i.e., for any $f\in \calM(\calF)$
\begin{equation}
    \| f \|_q := \left( \int_{\Omega} |f|^q \, dP \right)^{1/q}.
\end{equation}
We say that a RV $f$ is integrable if $f\in L^1(\calF),$ and we say that $f$ is $n$-times integrable if $f\in L^n(\calF).$ A shorthand for integration against $P$ is the expectation operator, i.e., if $f\in L^1(\calF)$ we denote $\BE[f] := \int_{\Omega} f \, dP.$ If $f\in L^2(\calF),$ we denote its variance by $\sigma_f^2 := \BE\left[f^2\right] - \BE[f]^2,$ where $\sigma_f\ge 0.$ The inner product in the Hilbert space $L^2(\calF)$ is denoted by $\langle \wc,\wc \rangle,$ i.e., for $f,g\in L^2(\calF)$ we set $\langle f,g \rangle := \BE[fg].$ Note that the notation $L^q(\calF)$ suppresses the dependence on the underlying space $\Omega$ and measure $P.$ The Banach space $L^q(\BR)$ consists of all Lebesgue-measurable functions $f:\BR \to \BR$ such that $\int_{\BR} |f(x)|^q \, dx < \infty,$ and its norm is denoted by $\|\wc\|_{L^q(\BR)}.$

The (Borel) probability measure induced by a RV $X$ is denoted by $P_X,$ i.e., for $B\in \calB$ we set
\begin{equation}
    P_X(B) := P\left(X^{-1}(B)\right).
\end{equation}
We let $\supp(X)$ denote the support of $X$; this is the smallest closed subset $S\subset \BR$ such that $P_X(S)=1,$ i.e., $\supp(X)$ is the complement of the union of all open $P_X$-null sets, or,
\begin{equation} \label{im}
    \supp(X) := \{ x\in \BR ~ ; ~ P_X(U)>0 \text{ for each open } U\ni x\}.
\end{equation}
These are equivalent characterizations of $\supp(X)$ because the standard topology on $\BR$ has a countable basis. We denote the cardinality of a set $S$ by $|S|,$ so by $|\supp(X)|$ we mean the smallest size of a set containing $X$ almost surely
\begin{equation} \label{il}
    |\supp(X)| := \min\{ |S| ~ ; ~ P_X(S)=1 \}.
\end{equation}
If there is no finite-cardinality $S$ for which $P_X(S)=1,$ then we write $|\supp(X)|=\infty$ (and we say $X$ has infinite support), whereas if there is some finite $S$ for which $P_X(S)=1$ then the well-ordering principle yields that the minimum in \eqref{il} is attained. For example, the uniform random variable over the interval $[0,1]$ has infinite support. The Dirac measure is denoted by $\delta_a,$ i.e., for every $a\in \BR$ and $B\in \calB$
\begin{equation}
    \delta_a(B) := \left\{ \begin{array}{cl}
    1 & \text{ if } a \in B, \\
    0 & \text{ otherwise.}
    \end{array} \right.
\end{equation}
For $b\in \BR,$ we write $\delta_a(b):= \delta_a(\{b\}).$ The indicator function of a set $B$ is denoted by $1_B.$

We say that $Y$ is a continuous RV if $P_Y$ is absolutely continuous with respect to the Lebesgue measure $\mu_L$ on $\BR,$ and we call $Y$ discrete if there is a countable set $\{q_y\}_{y\in \calS}\subset [0,1]$ such that for every $B\in \calB$ we have $P_Y(B) = \sum_{y\in \calS} q_y \delta_y(B).$ We let $p_Y$ denote the probability density function (PDF) of a RV $Y$ if $Y$ is continuous, or the probability mass function of $Y$ if $Y$ is a discrete RV. We write $Y\sim f$ to mean that $f$ is a PDF or PMF of $Y$ if it is clear from context that $Y$ is continuous or discrete. The symbol $\sim$ is also used to indicate that a RV follows a named distribution, e.g., $N\sim \calN(0,1)$ means that $N$ is a standard normal RV.\footnote{Additionally, the symbol $\sim$ will be used to indicate asymptotic behavior, namely, for functions $u$ and $v$ defined on the same metric space and having values in $\BR,$ we write $u\sim v$ as $t\to a$ if it holds that $\lim_{t\to a} u(t)/v(t)=1.$} 

If $Y$ is a continuous RV, we denote the differential entropy of $Y$ by
\begin{equation} \label{wq}
    h(Y) := -\int_{\BR} p_Y(y) \log p_Y(y) \, dy
\end{equation}
whenever the integral is well-defined. For example, if $\sigma_Y^2<\infty,$ then $h(Y)$ as given by~\eqref{wq} is well-defined and $- \infty \le h(Y) \le (1/2) \log(2\pi \sigma_Y^2)$ as can be seen by splitting the integrand $p_Y\log p_Y$ into positive and negative parts. The characteristic function of a RV $Y$ is denoted by
\begin{equation}
    \varphi_Y(t):= \BE[e^{itY}].
\end{equation}
We say that the RV $Y$ has a moment-generating function (MGF) if its MGF $\BE[e^{tY}]$ is finite over some nonempty interval $t\in (-\delta,\delta).$

For a nonempty $B\in \calB$ that is finite or of nonzero Lebesgue measure, we write $Y \sim \mathrm{Unif}(B)$ to mean that $Y$ is a uniform RV over $B.$ In other words, if $B$ is finite, then $Y$ is discrete and $Y\sim |B|^{-1} \sum_{y\in B} \delta_y,$ whereas if $\mu_L(B)>0$ then $Y$ is continuous and $Y\sim \mu_L(B)^{-1}1_B.$ We will also consider RVs that are uniform over independent and identically distributed (i.i.d.) samples. In other words, if $X_1,\cdots,X_m$ are i.i.d. RVs, we write $U\sim \mathrm{Unif}(\{X_j\}_{j=1}^m)$ to denote a collection of RVs $\{ U(\omega) \}_{\omega \in \Omega}$ such that for each $\omega \in \Omega$ the RV $U(\omega)$ is uniformly distributed over the multiset $\{X_j(\omega)\}_{j=1}^m.$ For such $U,$ we let $\BE[U]$ be a shorthand for the RV $\omega \mapsto \BE[U(\omega)],$ i.e., $\BE[U] = (X_1+\cdots+X_m)/m.$

If $\BE\left[ Y^{2n} \right]<\infty,$ we denote the Hankel matrix\footnote{Hankel matrices are square matrices with constant skew diagonals.} of moments $\bM_{Y,n}\in \BR^{(n+1)\times (n+1)}$ by
\begin{equation} \label{jd}
\bM_{Y,n} := \left(\BE \left[ Y^{i+j} \right]\right)_{0\le i,j\le n}.
\end{equation}
For a RV $Y,$ we consider random vectors
\begin{equation} \label{jg}
    \bY^{(n)} := \left( 1,Y,\cdots,Y^n \right)^T.
\end{equation}
Note that $\bM_{Y,n}$ is the expectation of the outer product of the random vector $\bY^{(n)},$ 
\begin{equation}
    \bM_{Y,n} = \BE\left[ \bY^{(n)} \left( \bY^{(n)} \right)^T \right].
\end{equation}
Therefore, $\bM_{Y,n}$ is a rank-1 perturbation of the covariance matrix of $\bY^{(n)},$
\begin{equation}
    \bSigma_{\bY^{(n)}}:=\BE\left[ \left(\bY^{(n)}-\BE\left[\bY^{(n)}\right] \right)  \left(\bY^{(n)}-\BE\left[\bY^{(n)}\right] \right)^T \right],
\end{equation}
namely,
\begin{equation}
    \bM_{Y,n} = \bSigma_{\bY^{(n)}} + \BE \left[ \bY^{(n)} \right]\BE \left[ \bY^{(n)} \right]^T. 
\end{equation}
The matrix $\bM_{Y,n}$ is symmetric. If, in addition, $\bM_{Y,n}$ is positive-definite (which occurs if and only if $|\supp(Y)|>n$ as shown in Lemma \ref{kv}), then we use the notation $\bM_{Y,n}^{1/2}$ to denote the lower-triangular matrix in the Cholesky decomposition of $\bM_{Y,n}.$ In other words, $\bM_{Y,n}^{1/2}$ is the unique lower-triangular matrix with positive diagonal entries that satisfies
\begin{equation} \label{ls}
    \bM_{Y,n}^{1/2} \left( \bM_{Y,n}^{1/2} \right)^T = \bM_{Y,n}.
\end{equation}
We also use the shorthand
\begin{equation} \label{lr}
    \bM_{Y,n}^{-1/2} := \left( \bM_{Y,n}^{1/2} \right)^{-1}.
\end{equation}
Note that $(\bM_{Y,n}^{1/2})^{-1}\neq (\bM_{Y,n}^{-1})^{1/2},$ unless $\bM_{Y,n}$ is diagonal; since $\bM_{Y,n}$ is a Hankel matrix, this can only happen if $n=1$ and $\BE[Y]=0.$

We let $\SP_n$ denote the collection of all polynomials of degree at most $n$ with real coefficients, and we define $\SP_n(Y)\subset \calM(\sigma(Y))$ by
\begin{equation}
    \SP_n(Y) := \{ q(Y) ~ ; ~ q \in \SP_n\}.
\end{equation}
If we write $q(Y) \in \SP_n(Y),$ then we implicitly mean that $q\in \SP_n$; this does not necessarily define $q$ uniquely, e.g., if $Y$ is binary then with $q(Y)=Y^2$ we also have $q(Y)=Y,$ so both $t\mapsto t$ and $t\mapsto t^2$ are valid candidates for $q.$ Nevertheless, if $|\supp(Y)|>n,$ then each element in $\SP_n(Y)$ corresponds to a unique element in $\SP_n.$ 

For $n\in \BN,$ we set
\begin{equation}
    [n]:=\{0,1,\cdots,n\}.
\end{equation}
The set of all finite-length tuples of non-negative integers is denoted by $\BN^*.$ Vectors are denoted by boldface letters, in which case subscripted regular letters refer to the entries, e.g., when $\blambda\in \BR^n$ we let $\lambda_j$ denote the $j$-th entry so $\blambda = (\lambda_1,\cdots,\lambda_n)^T.$ Often, we will start indexing at $0$ rather than $1,$ and that will be made clear when it is done (in Section \ref{pc}, indexing of tuples starts at $2$). The $n\times n$ identity matrix is denoted by $\bI_n.$ For a real vector space $V,$ we let $\dim V$ denote its dimension. Note that
\begin{equation}
    \dim \SP_n(Y) = \min(n+1,|\supp(Y)|).
\end{equation}
The closure of a set $S$ will be denoted by $\overline{S}.$ In various parts of the paper, we will use $\calX_k$ as a shorthand for the $k$-th moment of a RV $X$
\begin{equation}
    \calX_k := \BE\left[ X^k \right],
\end{equation}
and the notation $\calY_k$ is used analogously.

For every integer $r\ge 2,$ let $\Pi_r$ be the set of unordered integer partitions $r=r_1+\cdots+r_k$ of $r$ into integers $r_j\ge 2.$ We encode $\Pi_r$ via a list of the multiplicities of the parts as
\begin{equation} \label{so}
    \Pi_r := \left\{ (\lambda_2,\cdots,\lambda_\ell) \in \BN^{*} ~ ; ~ 2\lambda_2+\cdots+\ell \lambda_\ell = r \right\}.
\end{equation}
In \eqref{so}, $\ell\ge 2$ is free, and trailing zeros are ignored (i.e., $\lambda_{\ell}>0$). For a partition $(\lambda_2,\cdots,\lambda_\ell)=\blambda \in \Pi_r$ having $m=\lambda_2+\cdots+\lambda_\ell$ parts, we denote\footnote{The integer $c_{\blambda}$ counts the number of cyclically-invariant ordered set-partitions of an $r$-element set into $m=\lambda_2+\cdots+\lambda_\ell$ subsets where, for each $k\in \{2,\cdots,\ell\},$ exactly $\lambda_k$ parts have size $k.$}
\begin{equation} \label{ss}
    c_{\blambda} := \frac{1}{m} \binom{m}{\lambda_2,\cdots,\lambda_\ell} \binom{r}{\underbrace{2,\cdots,2}_{\textstyle \lambda_2};\cdots;\underbrace{\ell,\cdots,\ell}_{\textstyle  \lambda_\ell}}
\end{equation}
and
\begin{equation} \label{sp}
    e_{\blambda} := (-1)^{m-1} c_{\blambda}.
\end{equation}
We set\footnote{The integer $C_r$ counts the total number of cyclically-invariant ordered set-partitions of an $r$-element set into subsets of sizes at least $2.$} $C_r:=\sum_{\blambda \in \Pi_r} c_{\blambda}.$ Let $\stirling{r}{m}$ denote the Stirling numbers of the second kind (i.e., the number of unordered set-partitions of an $r$-element set into $m$ nonempty subsets). The integer $C_r$ can be expressed as 
\begin{equation} \label{arx}
    C_r = \sum_{k=1}^r (k-1)!  \sum_{j=0}^k (-1)^{j} \binom{r}{j} \stirling{r-j}{k-j}.
\end{equation}
The first few values of $C_r$ (for $2\le r \le 7$) are given by $1, 1, 4, 11, 56, 267,$ and as $r\to \infty$ we have the asymptotic $C_r \sim (r-1)!/\alpha^r$ for some constant $\alpha \approx 1.146$ (see \cite{OEIS}). The crude bound $C_r<r^r$ can also be seen by a combinatorial argument.  For completeness, equation~\eqref{arx} is derived in Appendix~\ref{aaj}.

\subsection{Assumptions}

We assume only that $X$ is integrable and $N\sim \calN(0,1)$ is independent of $X$ to prove that the conditional expectation $\BE[X\mid X+N]$ cannot be a polynomial of degree exceeding~$1$ (Theorem~\ref{np}) and bound the norms of the derivatives of the conditional expectation $y \mapsto \BE[X \mid X+N=y]$ (Theorem~\ref{og}). To prove that the PMMSE estimate converges to the MMSE estimate $\BE[X\mid Y]$ (Theorem~\ref{ar}), we assume that $X$ is square-integrable and $Y$ satisfies Carleman's condition, i.e.,
\begin{equation}
    \sum_{n\ge 1} \BE\left[ Y^{2n} \right]^{-1/(2n)} = \infty.
\end{equation}
For a RV to satisfy Carleman's condition it suffices to have a MGF \cite{Lin2017}. To instantiate the PMMSE formula and its convergence to the MMSE in the Gaussian channel case in Section~\ref{jh}, we either assume that the channel input is sufficiently integrable or that it has a MGF, respectively. The formula for differential entropy of a continuous RV in terms of its moments (Theorem~\ref{jn}) and its generalization to random vectors (Theorem~\ref{xu}) hold whenever the underlying RVs have MGFs. The ensuing formula for mutual information $I(X;Y)$ between a discrete RV $X$ and a continuous RV $Y$ (Theorem~\ref{jo}) holds when $X$ has finite support and $Y$ has a MGF and satisfies $h(Y)>-\infty.$ The formula for $I(X;Y)$ when both $X$ and $Y$ are continuous (Theorem~\ref{xr}) holds when both $X$ and $Y$ have MGFs and $\min(h(X),h(Y))>-\infty.$ Finally, for the Bernstein approximation theorem we prove for $\BE[X\mid X+N]$ (Theorem~\ref{kw}), we impose the assumption that $X$ is either continuous or discrete with a PDF or a PMF belonging to the set we define next.
\begin{definition} \label{kx}
Let $\SD$ denote the set of compactly-supported even PDFs or PMFs $p$ that are non-increasing over $[0,\infty)~\cap~ \mathrm{supp}(p).$
\end{definition}

\section{Polynomial Conditional Expectations in Gaussian Channels are Constant or Linear} \label{qj}

We prove in Theorem~\ref{np} that $\BE[X\mid Y]$ is  a polynomial, for integrable $X$ and $Y=X+N$ a Gaussian perturbation, if and only if $X$ is Gaussian or constant. The proof is carried in two steps. First, we show that a degree-$m$ non-constant polynomial $\BE[X\mid Y]$ requires $p_Y = e^{-h}$ for some polynomial $h$ with $\deg h = m+1.$ The second step is showing that, because $p_Y=e^{-h}$ is a convolution of the Gaussian kernel, $m =1 .$

The key result we use to prove Theorem~\ref{np} is Marcinkiewicz's theorem that a cumulant-generating function has degree at most $2$ if it were a polynomial.

\begin{theorem}[Marcinkiewicz, {\cite[Theorem 2.5.3]{Bryc}}] \label{wr}
If the characteristic function of a nondeterministic RV $R$ satisfies $\varphi_R=e^{g}$ on a neighborhood of $0$ in $\BR$ for some polynomial $g$ with complex coefficients, then $\deg(g) = 2$ and $R$ is Gaussian.
\end{theorem}
\begin{remark}
An analogous statement holds for the MGF. If the MGF satisfies $M_R(t):=\BE[e^{tR}]=e^{f(t)}$ over $t\in (-\delta,\delta)$ for some polynomial $f,$ then the two functions $\BE[e^{zR}]$ and $e^{f(z)}$ are analytic continuations of $M_R$ to the domain $|\Im(z)|<\delta$ (that $\BE[e^{zR}]$ is analytic there can be seen by Morera's theorem). By uniqueness of analytic continuations, we obtain that $\BE[e^{zR}] = e^{f(z)}$ over $|\Im(z)|<\delta.$ In particular, $\varphi_R(t) = e^{g(t)}$ over $t\in (-\delta,\delta),$ where $g$ is the polynomial $g(t):=f(it).$ Therefore, having a polynomial cumulant generating function $\kappa_R(t)=f(t)$ implies by Theorem~\ref{wr} that $\deg(\kappa_R) \in \{0,2\}.$
\end{remark}

The following elementary lemma will be useful for the proof of Theorem~\ref{np}.

\begin{lemma} \label{arv}
For a RV $X$ and a polynomial $p,$ if $p(X)$ is integrable then so is $X^{\deg(p)}.$
\end{lemma}
\begin{proof}
See Appendix~\ref{aae}.
\end{proof}

This lemma will allow us to conclude the finiteness of all moments of $X$ directly from the hypotheses that $\BE[X\mid Y]$ is a polynomial of degree exceeding $1$ and $\|X\|_1<\infty,$ because we have the inequalities $\| \BE[X\mid Y] \|_k \le \|X\|_k$ for every $k\ge 1.$ 

\begin{theorem} \label{np}
For $Y=X+N$ where $X$ is an integrable RV and $N\sim \calN(0,1)$ independent of $X,$ the conditional expectation $\BE[X \mid Y]$ cannot be a  polynomial in $Y$ with degree greater than 1.  Therefore, the MMSE estimator in a Gaussian channel with finite-variance input is a polynomial if and only if the input is Gaussian or constant.
\end{theorem}
\begin{remark}
Note that the fact that the MMSE estimator $\BE[X\mid Y]$ is a polynomial if and only if $X$ is Gaussian or constant is derivable from the fact that $\BE[X^2 \mid Y=y] = O(y^2),$ shown in~\cite[Proposition 1.2]{Fozunbal} under the assumption that the input RV $X$ has finite variance. Theorem~\ref{np} proves this conclusion under the more general setup when $X$ is assumed to be only integrable.
\end{remark}
\begin{proof}
Suppose, for the sake of contradiction, that
\begin{equation} \label{nk}
    \BE[X \mid Y] = q(Y)
\end{equation}
for some polynomial with real coefficients $q$ of degree $m:=\deg q >1.$ The contradiction we derive will be that the probability measure defined by
\begin{equation} \label{nz}
    Q(B) := \frac{1}{a}\int_B e^{-x^2/2} \, dP_X(x)
\end{equation}
for every Borel subset $B\subset \BR,$ where $a= \BE\left[e^{-X^2/2}\right]$ is the normalization constant, would necessarily have a cumulant generating function that is a polynomial of degree $m+1>2.$ Let $R$ be a RV distributed according to $Q.$ We note that the polynomial $q$ is uniquely determined by \eqref{nk} because $Y$ is continuous, for if $q(Y)=g(Y)$ for a polynomial $g$ then the support of $Y$ must be a subset of the roots of $q-g.$ 

The proof strategy is to compute the PDF $p_Y$ in two ways. One way is to compute $p_Y$  as a convolution
\begin{equation} \label{st}
    p_Y(y) = \frac{1}{\sqrt{2\pi}} \BE\left[ e^{-(X-y)^2/2} \right].
\end{equation}
This equation shows by Lebesgue's dominated convergence that $p_Y$ is continuous. The second way to compute $p_Y$ is via the inverse Fourier transform of $\varphi_Y.$ We consider the Fourier transform that takes an integrable function $\varphi$ to $\widehat{\varphi}(y):=\int_{\mathbb{R}} \varphi(t)e^{-iy t}\, dt,$ so the inverse Fourier transform takes an integrable function $p$ to $(2\pi)^{-1}\int_{\BR} p(y)e^{ity}\, dy.$ Now, $\varphi_Y=\varphi_X\varphi_N$ is integrable; indeed, $|\varphi_X|\le 1$ and $\varphi_N(t)=e^{-t^2/2}.$ Also, being a characteristic function, $\varphi_Y$ is continuous too. Therefore, by the Fourier inversion theorem, since $\varphi_Y/(2\pi)$ is the inverse Fourier transform of $p_Y,$ we obtain that $p_Y = \widehat{\varphi_Y}/(2\pi).$ Equating this latter equation with~\eqref{st}, then multiplying both sides by $\sqrt{2\pi}e^{y^2/2}/a,$ that $R\sim Q$ (see~\eqref{nz}) implies
\begin{equation} \label{sz}
    \BE\left[ e^{Ry} \right] = \frac{1}{a\sqrt{2\pi}}e^{y^2/2} \widehat{\varphi_Y}(y).
\end{equation}
Equation \eqref{sz} holds for every $y\in \BR.$ The rest of the proof derives a contradiction by showing that $\widehat{\varphi_Y}=e^G$ for some polynomial $G$ of degree $m+1>2.$

Integrability of $X$ implies integrability of $\BE[X\mid Y],$ so for every $t\in \BR$
\begin{equation} \label{nj}
    \BE\left[ e^{itY} (X-\BE[X\mid Y]) \right] = 0.
\end{equation}
Substituting $X=Y-N$ and $\BE[X \mid Y] = q(Y)$ into \eqref{nj},
\begin{equation}
    \BE\left[ e^{itY} (Y-N-q(Y)) \right] = 0.
\end{equation}
Because the RVs $e^{itY} (Y-q(Y))$ and $e^{itY}N$ are integrable, we may split the expectation to obtain
\begin{equation} \label{nl}
    \BE\left[ e^{itY}(Y-q(Y)) \right] - \BE\left[e^{itY}N \right] = 0.
\end{equation}
We rewrite equation \eqref{nl} in terms of the characteristic functions of $Y$ and $N.$ 

Since $q(Y)$ is integrable, Lemma~\ref{arv} implies that $Y$ is $m$-times integrable. In particular, we have that $\BE\left[ \left| (X+z)^m \right|\right] <\infty$ for some $z\in \BR.$ By Lemma~\ref{arv} again, $X$ is $m$-times integrable. Hence, for each $k\in [m]$ and $Z\in \{X,N,Y\},$ that $\BE\left[ |Z|^k \right]<\infty$ implies that the $k$-th derivative $\varphi_Z^{(k)}$ exists everywhere and
\begin{equation} \label{nm}
    (-i)^k \varphi_Z^{(k)}(t) = \BE\left[ e^{itZ} Z^k \right].
\end{equation}

For the term $\BE\left[ e^{itY}N \right]$ in \eqref{nl}, plugging in $Y=X+N,$ we infer from \eqref{nm} that
\begin{equation} \label{arn}
    \BE\left[ e^{itY}N\right] = \varphi_X(t) \BE\left[ e^{itN}N\right] = -i \varphi_X(t)\varphi_N'(t).
\end{equation}
But $\varphi_N(t) = e^{-t^2/2},$ so $\varphi_N'(t) = -t \varphi_N(t),$ hence \eqref{arn} yields
\begin{equation} \label{nn}
    \BE\left[ e^{itY}N\right] = it \varphi_X(t)\varphi_N(t) = it \varphi_Y(t).
\end{equation}
Let $\alpha_k$ for $k\in [m]$ be real constants such that $q(u) = \sum_{k\in [m]} \alpha_k u^k$ identically over $\BR,$ so $\alpha_m\neq 0.$ For the first term in \eqref{nl}, utilizing \eqref{nm} repeatedly we obtain
\begin{equation} \label{no}
    \BE\left[ e^{itY}(Y-q(Y)) \right] = -i\sum_{k \in [m]} c_k \varphi_Y^{(k)}(t)
\end{equation}
where we define the constants
\begin{equation} \label{nu}
    c_k := (-i)^{k+1}\alpha_k + \delta_{1,k} = \left\{ \begin{array}{cl} 
    (-i)^{k+1}\alpha_k & \text{if } k\in [m]\setminus \{1\}, \\
    1-\alpha_1 & \text{if } k=1.
    \end{array} \right.
\end{equation}
Plugging \eqref{nn} and \eqref{no} in \eqref{nl}, we get the differential equation
\begin{equation} \label{nq}
    t \varphi_Y(t) + \sum_{k \in [m]} c_k \varphi_Y^{(k)}(t) = 0.
\end{equation}
We will transform the differential equation \eqref{nq} into a linear differential equation in the Fourier transform of $\varphi_Y.$ For this, we need first to show that for each $k\in [m]$ the derivative $\varphi_Y^{(k)}$ is integrable so that its Fourier transform is well-defined.

Now, repeated differentiation of $\varphi_Y(t)=\varphi_X(t)e^{-t^2/2}$ shows that for each $k\in [m]$ there is a polynomial $r_k$ in $k+2$ variables such that
\begin{equation} \label{si}
    \varphi_Y^{(k)}(t) = r_k\left( t, \varphi_X(t),\varphi_X'(t), \cdots,\varphi_X^{(k)}(t) \right) e^{-t^2/2}.
\end{equation}
Indeed, we start with $r_0(t,u)=u$ because $\varphi_Y(t)=\varphi_X(t)e^{-t^2/2}.$ Now, suppose \eqref{si} holds for some $k\in [m-1].$ The derivative (with respect to $t$) of the $r_k$ term is
\begin{equation}
    \frac{d}{dt} r_k\left( t, \varphi_X(t), \cdots,\varphi_X^{(k)}(t) \right) = s_k\left( t,\varphi_X(t),\cdots,\varphi_X^{(k+1)}(t) \right)
\end{equation}
for some polynomial $s_k$ in $k+3$ variables. Therefore, differentiating \eqref{si}, we get
\begin{equation}
    \varphi_Y^{(k+1)}(t) = r_{k+1}\left( t, \varphi_X(t),\varphi_X'(t), \cdots,\varphi_X^{(k+1)}(t) \right) e^{-t^2/2}
\end{equation}
where
\begin{align}
    r_{k+1}\left(t,u_0,\cdots,u_{k+1}\right) &:= s_k\left( t,u_0,\cdots,u_{k+1}\right) - t \cdot r_k\left( t,u_0,\cdots,u_k\right)
\end{align}
is a polynomial in $k+3$ variables. Therefore \eqref{si} holds for all $k\in [m].$ Now, for each $j\in [m],$ we have by \eqref{nm} the uniform bound $| \varphi_X^{(j)}(t) | \le \BE\left[ |X|^j \right].$ Therefore, for each $k\in [m],$ letting $v_k$ be the same polynomial as $r_k$ but with the coefficients replaced with their absolute values, the triangle inequality applied to~\eqref{si} yields the bound $| \varphi_Y^{(k)}(t) | \le \eta_k(t) e^{-t^2/2}$ where $\eta_k(t) := v_k\left(|t|,1,\BE[|X|],\cdots,\BE\left[|X|^k\right]\right)$ is a (positive) polynomial in $|t|.$ Since $\int_\BR \eta_k(t)e^{-t^2/2} \, dt <\infty,$ we obtain that $\varphi_Y^{(k)}$ is integrable for each $k\in [m].$

Taking the Fourier transform in the differential equation \eqref{nq} we infer
\begin{equation} \label{nt}
    i\widehat{\varphi_Y}'(y) + \widehat{\varphi_Y}(y) \sum_{k\in [m]} c_k (iy)^k = 0.
\end{equation}
We rewrite this equation in terms of the $\alpha_k$ (see~\eqref{nu}) as
\begin{equation} \label{wf}
    \widehat{\varphi_Y}'(y) - \widehat{\varphi_Y}(y) \sum_{k\in [m]} (\alpha_k-\delta_{1,k}) y^k = 0.
\end{equation}
Equation~\eqref{wf} necessarily implies
\begin{equation}
    \widehat{\varphi_Y}(y) = D \exp\left( \sum_{k\in [m]}  \frac{\alpha_k-\delta_{1,k}}{k+1} y^{k+1}\right)
\end{equation}
for some constant $D.$ Since $p_Y=\widehat{\varphi_Y}/(2\pi),$ we necessarily have $D> 0.$ Therefore, we obtain the desired form for $\widehat{\varphi_Y},$ namely, $\widehat{\varphi_Y}=e^G$ where $G\in \SP_{m+1}\setminus \SP_{m}$ is given by\footnote{It can also be shown that we necessarily have $\alpha_m<0$ and $m$ is odd, but these points are moot since we eventually have a contradiction.}
\begin{equation}
    G(y) := \sum_{k\in [m]} \frac{\alpha_k-\delta_{1,k}}{k+1} y^{k+1} + \log(D).
\end{equation}
Plugging in this formula for $\widehat{\varphi_Y}$ in~\eqref{sz}, we obtain that the cumulant-generating function of the RV $R$ is the degree-$(m+1)$ polynomial $G(y)+y^2/2-\log(a\sqrt{2\pi}),$ contradicting Marcinkiewicz's theorem that a cumulant-generating function has degree at most $2$ if it were a polynomial (see, e.g., \cite[Theorem 2.5.3]{Bryc}). This concludes the proof by contradiction that $\BE[X\mid Y]$ cannot be a polynomial of degree at least $2.$

For the second statement in the theorem, we consider the remaining two cases that $\BE[X\mid Y]$ is a linear expression in $Y$ or is a constant.  If $\BE[X\mid Y]$ is constant, then differentiating and taking the expectation in~\cite{DytsoDer}
\begin{equation} \label{sh}
    \frac{d}{dy} ~ \BE\left[ X \mid Y = y \right] = \mathrm{Var}\left[ X \mid Y=y \right]
\end{equation}
yields that $\|X-\BE[X\mid Y]\|_2=0,$ i.e., $X=\BE[X\mid Y]$ is constant. Finally, under the assumption that $X$ has finite variance, $\BE[X \mid Y]$ is linear if and only if $X$ is Gaussian (see, e.g., \cite{Guo2011}). We note that if one requires only that $X$ be integrable, then one may deduce directly from the differential equation \eqref{nq} that a linear $\BE[X\mid Y]$ implies a Gaussian $X$ in this case too, and, for completeness, we end with a proof of this fact.

Assume that $\BE[X\mid Y]=\alpha_1 Y + \alpha_0$ is linear (so $\alpha_1\neq 0$). The differential equation \eqref{nq} becomes
\begin{equation} \label{sj}
    (t-i\alpha_0)\varphi_Y(t) + (1-\alpha_1)\varphi_Y'(t) = 0.
\end{equation}
From \eqref{sj}, we see that $\alpha_1\neq 1,$ because $\varphi_Y$ is nonzero on an open neighborhood around the origin (since $\varphi_Y(0)=1$ and $\varphi_Y$ is continuous). Therefore,
\begin{equation} \label{sk}
    \varphi_Y(t) = C e^{\frac{1}{\alpha_1-1}\left( \frac12 t^2 - i \alpha_0 t \right)},
\end{equation}
for some constant $C.$ Taking $t=0$ in \eqref{sk}, we see that $C=1.$ Therefore, the characteristic function of $Y$ is equal to the characteristic function of a $\calN\left( \frac{\alpha_0}{1-\alpha_1},\frac{1}{1-\alpha_1}\right)$ random variable (by taking $t\to \infty$ in~\eqref{sk}, we get $\alpha_1<1$). In fact, since $\varphi_Y = \varphi_X \cdot \varphi_N,$ we obtain
\begin{equation}
    \varphi_X(t) = e^{-\frac12\cdot \frac{\alpha_1}{1-\alpha_1} \cdot t^2 + it \cdot \frac{\alpha_0}{1-\alpha_1}}.
\end{equation}
Taking $t\to \infty,$ we see that $\alpha_1/(1-\alpha_1)>0,$ i.e., $\alpha_1\in (0,1)$ (note that $\alpha_1\neq 0$ by the assumption that $\BE[X\mid Y]$ is linear).  Therefore, uniqueness of characteristic functions implies that $X$ is Gaussian too.
\end{proof}

\begin{remark}
The proof of Theorem~\ref{np} is included for completeness, as it is independent of the inequality shown in~\cite[Proposition 1.2]{Fozunbal} that $\BE[X^2\mid Y=y] = O(y^2).$ Alternatively, we may build on this inequality to obtain another proof of Theorem~\ref{np}, as follows. By the Cauchy-Schwarz inequality, we deduce $\BE[X \mid Y=y] = O(|y|)$ as $y\to \infty.$ Therefore, $\BE[X\mid Y]$ cannot be a polynomial of degree exceeding $1$ if $X$ has finite variance. Lemma~\ref{arv} can be used to show that the general case in which $X$ is only integrable can be reduced to the finite-variance case. So, assume that $X$ is integrable, but not necessarily of finite variance. Suppose that $\BE[X\mid Y]$ is a polynomial of degree $m \ge 2.$ Since $X$ is integrable, $\BE[X\mid Y]$ is integrable too. By Lemma~\ref{arv}, we conclude that $Y$ is $m$-times integrable. Thus, $X+y$ is $m$-times integrable for at least one $y\in \BR.$ By Lemma~\ref{arv} again, $X$ is $m$-times integrable. In particular, $X$ has finite variance, and the desired result follows by the finite-variance case.
\end{remark}

\section{Polynomial MMSE: Definition, Convergence to the MMSE, and Explicit Formulas}\label{az}

We begin by defining the PMMSE, which will be an approximation of the MMSE. We  show convergence of the PMMSE to the MMSE, and give an explicit formula for the PMMSE.

\subsection{Definition of the PMMSE}

There are two equivalent ways to define the PMMSE estimate. First, it is an orthogonal projection onto subspaces of polynomials of bounded degree. Second, it is a natural generalization of the LMMSE to higher-degree polynomials. We first expound on both views of PMMSE, and then show equivalence of the two approaches. Finally, we present a formal definition of the PMMSE in Definition \ref{ho}.

If $\SP_n(Y)\subset L^2(\sigma(Y))$ (i.e., if $\BE\left[Y^{2n}\right]<\infty$), then $\SP_n(Y),$ being a finite-dimensional subspace, is closed. It is natural then to consider the orthogonal projection onto $\SP_n(Y),$ which we temporarily denote by 
\begin{equation}
    E_n^A[\wc\mid Y] : L^2(\calF) \to \SP_n(Y).
\end{equation}
A standard result in functional analysis states that such an orthogonal projection exists and is unique (see,~e.g.,~\cite{Stein2019}). In particular, separability of $L^2(\calF)$ yields the formula 
\begin{equation} \label{hl}
    E_n^A[\wc\mid Y] = \sum_{i\in [d]} \langle \wc , p_i(Y) \rangle ~ p_i(Y)
\end{equation}
where $\{p_i(Y)\}_{i\in [d]}$ is any orthonormal basis of $\SP_n(Y)$ (so $d = \mathrm{dim}\, \SP_n(Y)-1$) \cite{Stein2019}. Recall that, restricting attention to only finite-variance RVs, the conditional expectation $\BE[\wc \mid Y]$ is an orthogonal projection operator 
\begin{equation}
    \BE[\wc \mid Y] : L^2(\calF) \to L^2(\sigma(Y)).
\end{equation}
One might hope then that denseness of polynomials,
\begin{equation}
    \overline{\bigcup_{n\in \BN} \SP_n(Y)} = L^2(\sigma(Y)),
\end{equation}
would imply convergence of the orthogonal projection operators $E_n^A[\wc \mid Y]$ to the orthogonal projection operator $\BE[\wc \mid Y].$ This result will be shown to hold in Theorem \ref{ar}.

Another point of view is to introduce higher-degree generalizations of LMMSE estimation, i.e., estimating over $\SP_n(Y)$ rather than only within $\SP_1(Y).$ More precisely, note that the $\mm(X \mid Y)$ is obtained via solving the optimization problem
\begin{equation}
    \label{eq:mmse}
    \mm(X \mid Y) = \min_{f\in  L^2(\sigma(Y))} \BE \left[ \left( X - f \right)^2 \right]
\end{equation}
where the optimization variable ranges over the whole space $L^2(\sigma(Y))$ (and the minimum is uniquely attained at $\BE[X\mid Y]$). On the other extreme, the LMMSE is obtained via
\begin{equation}
    \label{eq:lmmse}
    \LL(X \mid Y) = \min_{a,b \in \BR} \BE \left[ \left( X - (aY+b) \right)^2 \right],
\end{equation}
where the optimization space encompasses only functions that are linear in $Y$ (and the minimum is uniquely attained provided that $Y$ is nondeterministic). We consider the problems that lie in-between \eqref{eq:mmse} and \eqref{eq:lmmse}, namely, the optimization problems where the variable ranges over all  polynomials of degree at most $n$:
\begin{equation} \label{hm}
    E_n^B[X \mid Y] := \underset{p(Y)\in \SP_n(Y)}{\mathrm{argmin}} ~ \BE \left[ \left( X - p(Y) \right)^2 \right].
\end{equation}
This definition implicitly assumes the existence of a unique (in $L^2(\sigma(Y))$) minimizer of the right-hand side of \eqref{hm}. Proposition \ref{iw} states that this indeed is the case and, in fact, $E_n^A[X\mid Y]$ is the unique minimizer.

Indeed, that the two operators $E_n^A$ and $E_n^B$ coincide is a restatement of a standard result in functional analysis, which posits that orthogonal projections in Hilbert spaces are the closest elements to the projected elements \cite[Section 4.4]{Stein2019}. We state this equivalence here and, for completeness, provide a proof in the appendices.

\begin{prop} \label{iw}
Fix $n\in \BN$ and two RVs $X$ and $Y$ such that $\BE\left[X^2\right],\BE\left[Y^{2n}\right]<\infty.$ Equations \eqref{hl} and \eqref{hm} define the same operator, i.e., there exists a unique minimizer $E_n^B[X \mid Y]$ of the right hand side of \eqref{hm} and it satisfies
\begin{equation}
    E_n^A[X \mid Y] = E_n^B[X \mid Y].
\end{equation}
\end{prop}
\begin{proof}
See Appendix \ref{kn}.
\end{proof}

\begin{remark}
    For this equivalence to hold, we require neither density of polynomials nor that monomials be linearly independent. Further, the polynomial $E_n^A[X\mid Y]$ is both the projection of $X$ onto $\SP_n(Y)$ and, from $E_n^A[X\mid Y]=E_n^B[X\mid Y],$ this polynomial also is the unique closest element in $\SP_n(Y)$ to $X.$ However, once we write $E_n^A[X\mid Y] = \sum_{k\in [n]} b_k Y^k,$ the constants $b_k$ might not be unique. For example, if $Y$ is binary and $n=2,$ then $Y^2=Y,$ so for any $b_0,b_1,b_2\in \BR$ for which $E_n^A[X\mid Y] = b_0+b_1Y+b_2Y^2$ we also have
    \begin{align}
        E_n^A[X\mid Y] = b_0+(b_1-1)Y+(b_2+1)Y^2.
    \end{align}
    In particular, there is no unique quadratic $p\in \SP_2$ for which $E_n^A[X\mid Y]= p(Y).$ Nevertheless, in the problems of interest to us, uniqueness is attained (e.g., if $Y$ is continuous); in fact, the coefficients are unique if and only if $|\supp(Y)|>n$ holds (equivalently, $1,Y,\cdots,Y^n$ are linearly independent, i.e., there is no hyperplane in $\BR^{n+1}$ that almost surely contains $\bY^{(n)}$). 
\end{remark}

We choose the estimation-theoretic point of view to define the conditional expectation approximants.

\begin{definition}[Polynomial MMSE] \label{ho}
Fix $n\in \BN$ and two RVs $X$ and $Y$ such that both $\BE\left[X^2\right]$ and $\BE \left[Y^{2n}\right]$ are finite. We define the $n$-th order \textit{polynomial minimum mean-squared error} (PMMSE) for estimating $X$ given $Y$ by
\begin{equation}\label{ak}
\text{pmmse}_n(X\mid Y) := \underset{\bc \in \BR^{n+1}}{\min} \; \BE \left[ \left( X- \bc^T\bY^{(n)} \right)^2 \right] .
\end{equation}
We define the $n$-th order PMMSE \emph{estimate} of $X$ given $Y$ by 
\begin{equation}
    E_n[X\mid Y] := \bc^T \bY^{(n)} \in \SP_n(Y)
\end{equation}
for any minimizer $\bc \in \BR^{n+1}$ in \eqref{ak}.
\end{definition}
\begin{remark} \label{pf}
The minimum in \eqref{ak} is always attained, and the polynomials $\bc^T\bY^{(n)}$ represent the same element of $\SP_n(Y)$ for all minimizers $\bc$ of \eqref{ak}. These two facts are direct consequences of Proposition \ref{iw}.
\end{remark}
\begin{remark} \label{sa}
We may define the pointwise estimator $E_n[X\mid Y=y]$ for $y\in \supp(Y)$ by
\begin{equation} \label{rz}
    E_n[X\mid Y=y] := \sum_{j\in [n]} c_j y^j
\end{equation}
where $\bc = (c_0,\cdots,c_n)^T$ is any minimizer in~\eqref{ak}. To see that~\eqref{rz} uniquely defines $E_n[X\mid Y=y],$ note that there are two distinct polynomials $p,q\in \SP_n$ such that $E_n[X \mid Y]=p(Y)=q(Y)$ if and only if $|\supp(Y)|\le n$ (and then $\supp(Y)$ is a subset of the roots of $p-q$). Therefore, when $|\supp(Y)|>n,$ there is a unique minimizer $\bc$ in~\eqref{ak} (and the definition of $E_n[X\mid Y=y]$ in~\eqref{rz} is extended to all $y\in \BR$). If $\supp(Y)=\{y_1,\cdots,y_s\}$ for $s\le n,$ then two vectors $\bc$ and $\bd$ minimize~\eqref{ak} if and only if there is some polynomial $r\in \SP_{n-s}$ such that $\sum_{j\in [n]} (d_j-c_j)u^j = (u-y_1)\cdots (u-y_s)r(u).$ In such case, using either $\bc$ or $\bd$ in~\eqref{rz} yields the same value of $E_n[X \mid Y=y]$ for every $y\in \supp(Y).$
\end{remark}

Unlike the case of the MMSE, working with the PMMSE is tractable and allows for explicit formulas. The formula for $t \mapsto \pp_n(X \mid \sqrt{t}X+N)$ stated in Theorem \ref{iv} reveals that this mapping is a rational function of $t$ (where $N\sim \calN(0,1)$ is independent of $X$).

The first question we investigate is whether the convergence
\begin{equation}
    \lim_{n\to \infty} E_n[X\mid Y] = \BE[X\mid Y]
\end{equation}
holds. Theorem \ref{ar} (stated below) proves that this convergence holds when polynomials of $Y$ are dense in $L^2(\sigma(Y)).$

\subsection{Convergence of the PMMSE to the MMSE}

The strong convergence (i.e., in the strong operator topology) of operators $E_n[\wc \mid Y]\to \BE [\wc \mid Y]$ immediately follows from linear independence of the monomials $\{Y^k\}_{k\in \BN}$ and density of polynomials in $L^2(\sigma(Y)).$  Indeed, from the complete linearly-independent set $\{Y^k\}_{k\in \BN}$ one may construct an orthonormal basis of $L^2(\sigma(Y))$ consisting of polynomials $\{q_0(Y) \equiv 1 \} \cup \{q_k(Y)\in \SP_k(Y) \setminus \SP_{k-1}(Y)\}_{k\in \BN_{\ge1}}.$  Then, for any $X\in L^2(\calF),$ we have the expansion
\begin{equation} \label{ll}
    \BE[X \mid Y] = \sum_{k\in \BN} \left\langle \BE[X\mid Y],q_k(Y) \right\rangle q_k(Y).
\end{equation}
For each $n\in \BN,$ the term
\begin{equation} \label{lk}
    \sum_{k\in [n]} \left\langle \BE[X\mid Y],q_k(Y) \right\rangle q_k(Y)
\end{equation}
is the orthogonal projection of $\BE[X\mid Y]$ onto $\SP_n(Y),$ because $\{q_k(Y)\}_{k\in [n]}$ is an orthonormal subset of $\SP_n(Y)$ of maximum size. Since $\BE[X\mid Y]$ is the orthogonal projection of $X$ onto $L^2(\sigma(Y)),$ then the function in \eqref{lk} is the orthogonal projection of $X$ onto $\SP_n(Y),$ i.e., it is $E_n[X\mid Y].$ Therefore, \eqref{ll} implies the convergence
\begin{equation}
    \BE[X\mid Y] = \lim_{n\to \infty} E_n[X\mid Y]
\end{equation}
in $L^2(\sigma(Y)).$ Furthermore, the limit $E_n[X \mid Y]\to \BE [X \mid Y]$ yields that $\pp_n(X\mid Y)\to \mm(X \mid Y)$ as $n\to \infty.$ This follows from the orthogonality principle of $\BE[X\mid Y],$ which gives
\begin{equation}
    \pp_n(X\,|\,Y) - \mm(X\,|\, Y) = \|\BE[X \,|\, Y] - E_n[X\,|\, Y] \|_2^2.
\end{equation}

A verifiable condition ensuring denseness of polynomials is, for example, Carleman's condition on $Y.$ Moreover, linear independence of the monomials $\{Y^k\}_{k\in \BN}$ is equivalent to having $|\supp(Y)|=\infty.$ Therefore, we have the following PMMSE-to-MMSE convergence result.

\begin{theorem}\label{ar}
If $X$ and $Y$ are RVs such that $\BE [X^2]<\infty$ and Carleman's condition is satisfied by $Y,$ then
\begin{equation} \label{hp}
    \lim_{n\to \infty} E_n[X \mid Y] = \BE [X \mid Y]
\end{equation}
in the $L^2(\sigma(Y))$-norm, and
\begin{equation} \label{hq}
\lim_{n\rightarrow \infty} \pp_n(X \mid Y) = \mm(X \mid Y).
\end{equation}
\end{theorem}
\begin{proof}
See Appendix \ref{lm}.
\end{proof}

A corollary of the PMMSE-to-MMSE convergence in Theorem \ref{ar} is the pointwise (in the SNR) convergence of the PMMSE to the MMSE in Gaussian channels (if, e.g., the input's MGF exists). Then, for each fixed $t\ge 0$ (see \eqref{lx})
\begin{equation}
    \lim_{n\to \infty} \pp_n(X,t) = \mm(X,t)
\end{equation}
provided that the MGF of $X$ exists. In fact, this convergence is uniform in the SNR $t.$ We prove this uniform convergence in Theorem \ref{an}.

\subsection{PMMSE Formula} \label{lo}

We prove next explicit PMMSE formulas that serve as the foundation for expressing information measures in terms of  moments. Even though these explicit formulas follow from standard results on orthogonal projections, we re-derive them here for the sake of completeness. We build on these formulas in Section \ref{jh} to derive rationality of $t \mapsto \pp_n(X,t)$ along with characterizations of the coefficients of this rational function. Those characterizations, in turn, will simplify the proof of consistency of the estimators for information measures introduced in Section~\ref{bb}.

\begin{theorem} \label{hw}
For two RVs $X$ and $Y$ and $n\in \BN,$ if $\BE\left[ X^2 \right]$ and $\BE\left[ Y^{2n} \right]$ are finite and $|\supp(Y)|>n,$ then the PMMSE estimator is given by
\begin{equation} \label{ip}
    E_n[X \mid Y] = \BE  \left[ X\bY^{(n)} \right]^T \bM_{Y,n}^{-1} \, \bY^{(n)},
\end{equation}
where $\bM_{Y,n} = ( \BE[Y^{i+j}] )_{(i,j)\in [n]^2}$ is the Hankel matrix of moments defined in~\eqref{jd}, and the PMMSE satisfies
\begin{equation} \label{iq}
    \pp_n(X \mid Y) = \BE\left[ X^2\right] -\BE  \left[ X\bY^{(n)} \right]^T  \bM_{Y,n}^{-1} ~ \BE  \left[ X\bY^{(n)} \right].
\end{equation}
\end{theorem}

Equations \eqref{ip} and \eqref{iq} yield the relation
\begin{equation} \label{jv}
    \pp_n(X\mid Y) = \BE\left[X^2\right] - \BE\left[ X E_n[X \mid Y] \right].
\end{equation}
To expound on the expressions in \eqref{ip} and \eqref{iq}, we derive next explicit formulas for the PMMSE and PMMSE estimate for $n=1$ and $n=2.$ By definition of the PMMSE, these expressions recover the LMMSE (and LMMSE estimate) and ``quadratic" MMSE (and ``quadratic" MMSE estimate).

\begin{example}
For $n=1,$ we have from \eqref{ip} that
\begin{equation}
    E_1[X \mid Y] = \left( \BE[X],\BE[XY] \right) \left( \begin{array}{cc}
    1 & \BE[Y] \\
    \BE[Y] & \BE\left[Y^2 \right] 
    \end{array} \right)^{-1}  \left( \begin{array}{c}
    1 \\
    Y
    \end{array} \right).
\end{equation}
Computing the matrix inverse and multiplying out, we obtain
\begin{equation} \label{ju}
    E_1[X\mid Y] = \BE[X] + \frac{\mathrm{cov}(X,Y)}{\sigma_Y^2} \left( Y- \BE[Y] \right),
\end{equation}
where $\mathrm{cov}(X,Y):= \BE[XY]-\BE[X]\BE[Y]$ is the covariance between $X$ and $Y.$ Formula \eqref{ju} indeed gives the LMMSE estimate. Via the relation in \eqref{jv}, we recover
\begin{equation} \label{jw}
    \pp_1(X\mid Y) = \sigma_X^2 - \frac{\mathrm{cov}(X,Y)^2}{\sigma_Y^2} = \sigma_X^2 (1-\rho_{X,Y}^2),
\end{equation}
where $\rho_{X,Y}:= \mathrm{cov}(X,Y)/(\sigma_X\sigma_Y)$ is the Pearson correlation coefficient between $X$ and $Y.$ Formula \eqref{jw} verifies that $\pp_1(X \mid Y)$ is the LMMSE. \hfill \qedsymbol
\end{example}

\begin{example}
We will use the notation
\begin{equation}
    \calY_k := \BE\left[Y^k\right]
\end{equation}
for short. For $n=2,$ \eqref{ip} implies that
\begin{equation}
    E_2[X \mid Y] 
    = \left( \BE[X],\BE[XY],\BE[XY^2] \right) \left( \begin{array}{ccc}
    1 & \calY_1 & \calY_2 \\
    \calY_1 & \calY_2 & \calY_3 \\
    \calY_2 & \calY_3 & \calY_4
    \end{array} \right)^{\hspace{-1mm}-1}  \left( \begin{array}{c} 
    1 \\
    Y \\
    Y^2
    \end{array} \right).
\end{equation}
Multiplying out, we obtain that $E_2[X\mid Y]$ is the quadratic
\begin{equation}
    E_2[X\mid Y] = \frac{\alpha_0}{\delta} + \frac{\alpha_1}{\delta}Y+\frac{\alpha_2}{\delta}Y^2
\end{equation}
where the values of $\alpha_0,\alpha_1,\alpha_2,$ and $\delta$ are
\begin{align}
    \alpha_0 &= (\calY_2\calY_4 - \calY_3^2)\BE[X] +  (\calY_2\calY_3-\calY_1\calY_4)\BE[XY] + (\calY_1\calY_3-\calY_2^2)\BE[XY^2] \\
    \alpha_1 &= (\calY_2\calY_3-\calY_1\calY_4)\BE[X]  + (\calY_4-\calY_2^2)\BE[XY] + (\calY_1 \calY_2 - \calY_3) \BE[XY^2] \\
    \alpha_2 &= (\calY_1\calY_3-\calY_2^2)\BE[X] + (\calY_1\calY_2-\calY_3) \BE[XY] + (\calY_2 - \calY_1^2) \BE[XY^2]
\end{align}
and
\begin{equation} \label{kd}
    \delta =  \calY_2\calY_4 - \calY_1^2 \calY_4 - \calY_2^3 - \calY_3^2 + 2 \calY_1 \calY_2 \calY_3 .
\end{equation}
Note that $\delta=\det \bM_{Y,2}.$ It is not immediately clear that $\delta\neq0,$ but we show in Lemma \ref{kv} that invertibility of $\bM_{Y,n}$ (for any $n\in \BN$) is equivalent to the condition $|\supp(Y)|>n.$ Equation \eqref{jv} then yields the formula
\begin{equation}
\pushQED{\qed} 
    \pp_2(X\mid Y) = \BE\left[X^2\right]-\delta^{-1} \sum_{k=0}^2 \alpha_k \BE\left[ XY^k \right]. \qedhere
\popQED
\end{equation}
\end{example}

We derive Theorem \ref{hw} in two ways according to how the PMMSE estimate is interpreted: as an orthogonal projection onto polynomial subspaces or as a minimizer of $L^2$-distance. For both of these proofs, and for other arguments in the sequel, we need the following straightforward result on invertibility of the Hankel matrix of moments.

\begin{lemma} \label{kv}
For a RV $Y$ and $n\in \BN$ such that $\BE[Y^{2n}]<\infty,$ the Hankel matrix of moments $\bM_{Y,n}$ is invertible if and only if $|\supp(Y)|>n.$
\end{lemma}
\begin{proof}
See Appendix \ref{ln}.
\end{proof}

The next lemma outlines an orthonormal basis for $\SP_n(Y).$ This basis simplifies the proof of Theorem \ref{hw} by allowing the use of standard orthogonal projection results in Hilbert spaces.

\begin{lemma} \label{io}
For a RV $Y$ and $n\in \BN$ such that $\BE[Y^{2n}]<\infty$ and $|\supp(Y)|>n,$ an orthonormal basis for $\SP_n(Y)$ is given by the entries of the random vector $\bM_{Y,n}^{-1/2} \bY^{(n)}$ (see~\eqref{lr} for the definition of $\bM_{Y,n}^{-1/2}$).
\end{lemma}
\begin{proof}
See Appendix \ref{lp}.
\end{proof}
\begin{remark}
This basis is the output of Gram-Schmidt orthonormalization applied to $\{Y^j\}_{j\in [n]}.$
\end{remark}

We are now ready to prove a preliminary formula for the PMMSE in view of the general expansion of orthogonal projections \eqref{hl}. 

\begin{proof}[Proof (of Theorem \ref{hw})]
By assumption on $Y,$ Lemma \ref{io} yields that the entries of the random vector $\bV := \bM_{Y,n}^{-1/2}\bY^{(n)}$ form an orthonormal basis of $\SP_n(Y).$ The expansion \eqref{hl} of orthogonal projections yields the formula
\begin{equation}
    E_n[X \mid Y] = \BE\left[ X \bV^T \right] \bV.
\end{equation}
Substituting $\bV= \bM_{Y,n}^{-1/2}\bY^{(n)}$ we obtain \eqref{ip}, i.e.,
\begin{equation*}
    E_n[X \mid Y] = \BE  \left[ X\bY^{(n)} \right]^T \bM_{Y,n}^{-1} \, \bY^{(n)}.
\end{equation*}
Then, expanding $\pp_n(X\mid Y)=\BE [(X-E_n[X\mid Y])^2],$ we obtain \eqref{iq}, i.e.,
\begin{equation*}
    \pp_n(X \mid Y) = \BE \left[ X^2\right]-\BE  \left[ X\bY^{(n)} \right]^T \hspace{-1mm}\bM_{Y,n}^{-1} \BE  \left[ X\bY^{(n)} \right],
\end{equation*}
and the proof is complete.
\end{proof}
\begin{remark}
The formulas in Theorem \ref{hw} were derived via geometric arguments, but we note that an alternative, analytic derivation directly utilizing Definition \ref{ho} is possible. This proof is via differentiation with respect to the polynomial coefficients in $E_n[X \mid Y]$ in the same way as the LMMSE is usually derived. We provide this analytic proof in Appendix \ref{lq}. 
\end{remark}

The PMMSE formula proved in Theorem \ref{hw} will aid in proving properties of the PMMSE in Section \ref{ca}. Moreover, Theorem \ref{hw} will be instantiated for the Gaussian channel in Section \ref{jh}. This specialization will be used to develop new representation of information measures in terms of moments.

\subsection{Connection to Polynomial Regression}

The goal of single-variable polynomial regression is to model a RV $X$ as a polynomial in a RV $Y$
\begin{equation}
    X = \beta_0+\beta_1Y+\cdots+\beta_nY^n+\varepsilon,
\end{equation}
where $\varepsilon$ is a RV capturing the modeling error. Here, the $\beta_j$ are constants to be determined from data. Given access to samples $\{(x_i,y_i)\}_{i=1}^m,$ this model can be estimated via the equation
\begin{equation}
    \bX = \bY\bbeta + \bvarepsilon,
\end{equation}
where $\bX = (x_1,\cdots,x_m)^T,$ $\bY=(y_i^j)_{i\in \{1,\cdots,m\},j\in [n]},$ $\bvarepsilon = (\varepsilon_1,\cdots,\varepsilon_m)^T$ where the $\varepsilon_j$ are samples from $\varepsilon,$ and $\bbeta=(\beta_0,\cdots,\beta_n)^T.$ It is assumed that the number of distinct $y_i$ is strictly larger than $n.$ A value of $\bbeta$ that minimizes $\|\bvarepsilon\|$ is
\begin{equation} \label{rk}
    \bbeta = \bX^T\bY (\bY^T\bY)^{-1}.
\end{equation}
Indeed, this formula follows from the PMMSE estimate formula in Theorem \ref{hw}, namely,
\begin{equation} \label{rl}
    \bbeta = \BE\left[ U \bV^{(n)} \right] \bM_{V,n}^{-1}
\end{equation}
where we introduce RVs $U$ and $V$ such that $(U,V)$ is uniform over $\{(x_i,y_i)\}_{i=1}^m.$ It immediately follows that
\begin{equation} \label{rm}
    \bX^T\bY = m\BE\left[ U \bV^{(n)} \right]^T
\end{equation}
and
\begin{equation} \label{rn}
    (\bY^T\bY)^{-1} = \frac{1}{m}\bM_{V,n}^{-1}.
\end{equation}
Multiplying \eqref{rm} and \eqref{rn}, we obtain \eqref{rk} in view of~\eqref{rl}. Therefore, the polynomial regression approach solves the restricted problem of finding the PMMSE when both $X$ and $Y$ are discrete with probability mass functions that evaluate to rational numbers, i.e., when the distribution of $(X,Y)$ is uniform over a dataset $\{ (x_i,y_i) \}_{i=1}^m.$

\section{Basic Properties of the PMMSE} \label{ca}

In this section, we investigate the behavior of the PMMSE under affine transformations and exhibit a few additional properties of the PMMSE that parallel those of the MMSE. 

\subsection{PMMSE and affine Transformations}

We discuss next the effect of affine transformations on the PMMSE. Recall that for RVs $X$ and $Y$ such that $\BE[X^2]<\infty,$ and for constants $\alpha,\beta\in \BR,$ we have~\cite{Guo2005}
\begin{equation} \label{jy}
    \mm(X+\alpha \mid Y+\beta) = \mm(X\mid Y),
\end{equation}
i.e., the MMSE is shift-invariant in both entries. Also, if $\beta \neq 0,$ then
\begin{equation} \label{jz}
    \mm(\alpha X \mid \beta Y) = \alpha^2 \mm(X\mid Y).
\end{equation}
These two properties of the MMSE hold, in essence, because $\mm(\wc \mid Y)$ measures the distance to the space of Borel-measurable functions of $Y,$ which is invariant under (injective) affine transformations of $Y,$ i.e., $\calM(\sigma(a Y+b)) = \calM(\sigma(Y))$ for $a,b\in \BR$ with $a\neq 0.$ These properties still hold when the search space is restricted to only the subspace of polynomials in $Y$ of a certain degree, i.e., the $\pp_n(\wc \mid Y)$ measures the distance to $\SP_n(Y)$ and $\SP_n(a Y+b) = \SP_n(Y)$ for $a,b \in \BR$ with $a\neq 0.$ Therefore, the  two properties in \eqref{jy} and \eqref{jz} also hold for the PMMSE. The following proposition follows directly from Definition \ref{ho}, and we note that appealing to formula \eqref{iq} yields a shorter proof.

\begin{prop} \label{hn}
Let $X$ and $Y$ be two RVs and $n\in \BN,$ and assume that both $\BE\left[X^2\right]$ and $\BE\left[Y^{2n}\right]$ are finite. For any $(\alpha,\beta)\in \BR^2,$ 
\begin{equation}\label{ai}
\pp_n(X+\alpha \mid Y+\beta)=\pp_n(X \mid Y)
\end{equation}
and, when $\beta \neq 0,$
\begin{equation}\label{aj}
\pp_n(\alpha X \mid \beta Y)=\alpha^2 \, \pp_n(X \mid Y).
\end{equation} 
\end{prop}
\begin{proof}
    See Appendix \ref{kb}.
\end{proof}

The behavior of the PMMSE under affine transformations shown in Proposition \ref{hn} has desirable implications on the approximations we introduce in Sections \ref{jj} and \ref{ba} for differential entropy and mutual information. For example, recall that differential entropy satisfies
\begin{equation} \label{ka}
    h(aY+b) = h(Y) + \log|a|
\end{equation}
for any constants $a$ and $b$ with $a\neq 0.$ Because of Proposition \ref{hn}, the same property in \eqref{ka} holds for the approximations $h_n$ in \eqref{kt}  (also see Proposition \ref{js}) for differential entropy, i.e.,
\begin{equation}
    h_n(aY+b) = h_n(Y) + \log|a|.
\end{equation}

\subsection{Operator Properties}

The operator $E_n[\wc \mid Y]$ satisfies several properties analogously to the conditional expectation $\BE[\wc \mid Y].$ We note that $E_n[\wc \mid Y]$ is not in general a conditional expectation operator itself, i.e., there are some $n\in \BN$ and $Y\in L^{2n}(\calF)$ such that for every sub-$\sigma$-algebra $\Sigma\subset \calF$ we have $E_n[\wc \mid Y] \neq \BE[\wc \mid \Sigma].$ One way to see this is that $E_n[\wc \mid Y]$ might not preserve positivity. For example, if $X\sim \mathrm{Unif}(0,1)$ and $Y=X+N$ for $N\sim \calN(0,1)$ independent of $X,$ we have that $E_1[X\mid Y] = (Y+6)/13$ (see~\eqref{ju}). Therefore, the probability that $E_1[X \mid Y]<0$ is $P_Y((-\infty,-6))>0.$ In other words, although $X$ is non-negative, $E_1[X \mid Y]$ is not; in contrast, $\BE[X \mid \Sigma]$ is non-negative for every sub-$\sigma$-algebra $\Sigma \subset \calF.$

Since $E_n[\wc \mid Y]$ is an orthogonal projection, it satisfies the following operator properties~\cite[Section 4.4]{Stein2019}.

\begin{prop} \label{ro}
Let $n\in \BN,$ and fix a RV $Y$ such that $\BE [Y^{2n}]<\infty$ and $|\supp(Y)|>n.$ The mapping $E_n[\wc \mid Y]:L^2(\calF)\to \SP_n(Y)$
\begin{enumerate}[label=(\roman*)]
    \item is the unique orthogonal projection onto $\SP_n(Y),$
    \item is a self-adjoint, idempotent, bounded linear operator, and
    \item has operator norm 1.
\end{enumerate}
\end{prop}

In the following proposition, we show properties of the operator $E_n$ that are analogous to those of the conditional expectation operator.

\begin{prop} \label{me}
Fix square-integrable RVs $X,Y,$ and $Z,$ and a natural $n\in \BN$ such that $\BE\left[ Z^{2n} \right] <\infty$ and $|\supp(Z)|>n.$ The following hold:
\begin{enumerate}[label=(\roman*)]

    \item \label{rp} Total expectation:
    \begin{equation} \label{je}
        \BE [E_n[X \mid Z]] = \BE [X].
    \end{equation}
    
    \item \label{rq} Orthogonality: For any polynomial $p\in \SP_n,$ 
    \begin{equation}
        \BE [(X-E_n[X \mid Z])p(Z)]=0.
    \end{equation}
    
    \item \label{rr} Linearity: For constants $a,b\in \BR$
    \begin{equation}
        E_n[aX+bY \mid Z] = aE_n[X \mid Z] + bE_n[Y \mid Z].
    \end{equation}
    
    \item \label{wn} Contractivity: We have the inequality
    \begin{equation}
        \left\| E_n[X \mid Z] \right\|_2 \le \|X\|_2.
    \end{equation}
    
    \item \label{rt} Idempotence: 
    \begin{equation}
        E_n\left[ E_n\left[ X \mid Z \right] \mid Z \right] = E_n[X \mid Z]. 
    \end{equation}
    
    \item \label{wo} Self-Adjointness: $E_n[\wc \mid Z]$ is self-adjoint
    \begin{equation}
        \BE\left[ E_n[X\mid Z] Y \right] = \BE\left[ X E_n[Y\mid Z] \right].
    \end{equation}
    
    \item \label{rs} Independence: If $X$ and $Z$ are independent, then
    \begin{equation} \label{jf}
        E_n[X \mid Z] = \BE[X].
    \end{equation}
    
    \item \label{ru} Markov Chain: If $X$---$Y$---$Z$ forms a Markov chain, then
    \begin{equation}
        E_n\left[ \BE[X \mid Y] \mid Z \right] = E_n[X\mid Z].
    \end{equation}
\end{enumerate}
\end{prop}
\begin{proof}
    See Appendix \ref{mf}.
\end{proof}

\begin{remark}
In view of the properties \ref{rp} and \ref{rs}, one may define the unconditional version of $E_n$ as
\begin{equation}
    E_n[X] := \BE[X]
\end{equation}
for $X\in L^2(\calF).$ With this definition, the total expectation property \ref{rp} becomes
\begin{equation}
    E_n[E_n[X \mid Z]]=E_n[X],
\end{equation}
and the independence property \ref{rs} becomes
\begin{equation}
    E_n[X \mid Z] = E_n[X],
\end{equation}
for independent $X$ and $Z.$ This definition of $E_n[X]$ is consistent with defining it as $E_n[X\mid 1],$ because $\BE[X]$ is the closest constant to $X$ in $L^2(\calF)$-norm.
\end{remark}

The following proposition shows that the polynomial in $X+Z$ closest to $X$ is always of odd degree, provided that $X$ and $Z$ are both symmetric RVs.\footnote{A RV $Y$ is symmetric if $P_{Y-a}=P_{-(Y-a)}$ for some $a\in \BR.$}
\begin{prop} \label{wk}
Fix $k \in \BN_{\ge 1}$ and RVs $X$ and $Z$ satisfying $\BE\left[Z^2\right],\BE\left[ X^{4k} \right] <\infty,$ and $|\supp(X+Z)|>2k.$ If $X$ and $Z$ are both symmetric, then we have that
\begin{equation} \label{wi}
    E_{2k}[X \mid X+Z] = E_{2k-1}[X \mid X+Z].
\end{equation}
In such case, we also have that
\begin{equation} \label{wj}
    \pp_{2k}(X\mid X+Z) = \pp_{2k-1}(X\mid X+Z).
\end{equation}
\end{prop}
\begin{proof}
See Appendix~\ref{wl}.
\end{proof}

\subsection{Convergence Theorems} 

Analogous to conditional expectation, dominated convergence, monotone convergence, and Fatou's lemma all hold for the PMMSE estimate. The notation $E_n[X \mid Y=y]$ is used here similarly to how the notation $\BE[X \mid Y=y]$ is customarily used (see Remark~\ref{sa}). 

\begin{prop}[Convergence Theorems] \label{xk}
Fix a sequence of square-integrable RVs  $\{X_k\}_{k\in \BN},$ and let $n\in \BN$ and the RV $Y$ be such that $\BE\left[ Y^{2n} \right]<\infty$ and $|\supp(Y)|>n.$ The following hold for every $y\in \BR$
\begin{enumerate}[label=(\roman*)]
    
    \item \label{xi} Monotone Convergence: If $\{X_k\}_{k\in \BN}$ is monotone, and either $Y\ge 0$ or $Y\le 0$ holds almost surely, then the pointwise limit $X=\lim_{k\to \infty} X_k$ satisfies     \begin{equation}
        E_n[X \mid Y=y] = \lim_{k\to \infty} E_n \left[ X_k \mid Y=y \right].
    \end{equation}
    
    \item \label{xj} Dominated Convergence: If there is a square-integrable RV $M$ such that $\sup_{k\in \BN} |X_k|\le M,$ and if the pointwise limit $X:=\lim_{k\to \infty} X_k$ exists, then
    \begin{equation}
        E_n[X \mid Y=y] = \lim_{k\to \infty} E_n \left[ X_k \mid Y=y \right].
    \end{equation}
\end{enumerate}
\end{prop}
\begin{proof}
Note that in~\ref{xi} the sequences $\{X_kY^j\}_{k\in \BN},$ for each fixed $j\in [n],$ are monotone almost surely. Also, $X_0$ is integrable, as we are assuming that $X_0\in L^2(\calF).$ Note also that in~\ref{xj} each sequence $\{X_kY^j\}_{k\in \BN},$ for $j\in [n],$ is dominated by $M|Y|^j,$ which is integrable since both $M$ and $Y^j$ are square-integrable. Thus, monotone convergence and dominated convergence both hold in $L^1(\calF)$ for each of the sequences $\{X_kY^j\}_{k\in \BN},$ where $j\in [n]$ is fixed. In addition, the formula
\begin{equation}
    E_n\left[ X_k \mid Y=y \right] = \BE\left[ X_k \bY^{(n)} \right]^T \bM_{Y,n}^{-1} \left( \begin{array}{c} 1 \\ y \\ \vdots \\ y^n \end{array} \right) = \sum_{j=0}^n c_j \BE\left[X_kY^j\right]
\end{equation}
expresses $E_n\left[ X_k \mid Y=y \right]$ as an $\BR$-linear combination of $\{X_kY^j\}_{j\in [n]}$ (where the $c_j$ do not depend on $k$). Thus, the convergence theorems in~\ref{xi} and~\ref{xj} also hold.
\end{proof}
\begin{remark}
A version of Fatou's lemma that holds for a subset of values of $y$ is also derivable. Namely, suppose that there is a RV $M\in L^1(\calF)$ such that $X_kY^j \ge -M$ for every $(k,j)\in \BN\times [n].$ Then, the same argument in the proof of Proposition~\ref{xk} shows that 
\begin{equation}
    \liminf_{k\to \infty} E_n[X_k \mid Y=y] \ge E_n\left[ \liminf_{k\to \infty} X_k \;\middle|\; Y=y \right]
\end{equation}
for every $y\in \BR$ such that $\bM_{Y,n}^{-1}(1,y,\cdots,y^n)^T$ consists of non-negative entries. For example, when $n=1,$ Fatou's lemma holds for $y \ge \BE[Y]$ if $\BE[Y] \le 0,$ and it holds for $y\in [\BE[Y],\BE[Y^2]/\BE[Y]]$ if $\BE[Y]>0.$
\end{remark}

\section{PMMSE for Gaussian Channels} \label{jh}

We now take a closer look at the special case of Gaussian channels, where
\begin{equation}
    Y = \sqrt{t}X+N
\end{equation}
for $t\ge 0$ and $N \sim \calN(0,1)$ independent of $X.$ We analyze both the PMMSE estimate $E_n[X \mid \sqrt{t}X+N]$ and the PMMSE
\begin{equation}
    \pp_n(X,t) := \pp_n(X \mid \sqrt{t}X+N).
\end{equation}
We also use the shorthand
\begin{align}
    \mm(X,t) &:= \mm(X\mid \sqrt{t}X+N) \\
    \LL(X,t) &:= \LL(X\mid \sqrt{t}X+N).
\end{align}
The Gaussian channel allows us to extrapolate---via the I-MMSE relation---new formulas for differential entropy and mutual information primarily in terms of moments (see Sections \ref{jj} and \ref{ba}), which then pave the way for new estimators for these quantities (see Section \ref{bb}).

Approximating the MMSE with the PMMSE in Gaussian channels is valid whenever the MGF of the input exists. In other words, the pointwise (in $t$) limit
\begin{equation} 
    \mm(X,t) = \lim_{n\to \infty} \pp_n(X,t)
\end{equation}
is a direct consequence of Theorem \ref{ar}. Furthermore, as will be shown in Theorem \ref{an}, the convergence of the PMMSE to the MMSE is in fact uniform in $t.$ Uniform convergence follows from rationality of the PMMSE as a function of $t.$ This rationality result, stated in Theorem \ref{iv}, will be the focus of this section. 

The mapping over the positive half-line defined by $t\mapsto \pp_n(X,t)$ will be shown to be a rational function that starts at $\sigma_X^2$ when $t=0$ and satisfies $\pp_n(X,t) < 1/t$ for $t>0.$ If $|\supp(X)|>n,$ then we also have the asymptotic $\pp_n(X,t) \sim 1/t$ as $t \to \infty.$ A simplified statement of the main theorem of this section (Theorem \ref{iv}) is as follows.

\begin{theorem} \label{ir}
Fix a natural $n \ge 1,$ and let $X$ be a RV satisfying $\mathbb{E}\left[X^{2n}\right]<\infty.$ The mapping $t\mapsto \mathrm{pmmse}_n(X,t)$ over $[0,\infty)$ is given by a rational function
\begin{equation} \label{in}
    \pp_n(X,t) = \frac{\sigma_X^2 G(n+2) + \cdots + (\det  \bM_{X,n}) t^{d_n-1}}{G(n+2) + \cdots + (\det  \bM_{X,n}) t^{d_n}}
\end{equation}
where $d_n:=\binom{n+1}{2}$ and $G(k):=\prod_{j=1}^{k-2} j!$ (for integers $k\ge 1$) is the Barnes $G$-function~\cite{Adamchik2001}.
\end{theorem}

\begin{remark}
We note that the dots in \eqref{in} are not to imply a specific pattern; rather, the statement of the theorem emphasizes only the rationality of the function along with the leading and constant coefficients. These will be enough to conclude the results we present about convergence and asymptotic behavior. The middle terms can be computed via the formulas presented in Theorem \ref{iv}. For example, the denominator in~\eqref{in} equals $\det \bM_{\sqrt{t}X+N,n}$ for $N\sim \calN(0,1)$ independent of $X.$ We also note that, for each $n\in \BN,$ the constant term $G(n+2)$ satisfies
\begin{equation}
    G(n+2) = \det \bM_{N,n} = \prod_{k=1}^n k!.
\end{equation}
\end{remark}

\begin{example}
The results of Theorem \ref{ir}, when applied to $n=1,$ recover the LMMSE formula
\begin{equation} \label{jx}
    \pp_1(X,t) = \LL(X,t) = \frac{\sigma_X^2}{1+\sigma_X^2 t},
\end{equation}
because $d_1=1,$ $G(3)=1,$ and $\det  \bM_{X,1}=\sigma_X^2.$ Note that \eqref{jx} directly follows from the general formula~\eqref{jw} upon setting $Y=\sqrt{t}X+N.$ \hfill \qedsymbol
\end{example}

We discuss three direct corollaries of Theorem \ref{ir} next. The remainder of the section is then  devoted to proving Theorem \ref{ir}. First, upon taking $t=0$ or $t\to \infty$ in \eqref{in}, we immediately obtain that the first-order asymptotic of the PMMSE (for any $n\in \BN$) is equivalent to that of the LMMSE, which is also the asymptotic of the MMSE for continuous RVs \cite{Wu2011}.

\begin{corollary} \label{is}
For $n\in \mathbb{N}$ and a RV $X$ such that $\mathbb{E}\left[ X^{2n} \right]<\infty,$ we have that 
\begin{equation} \label{qt}
    \pp_n(X,0) = \sigma_X^2,
\end{equation}
and, for every $t>0,$
\begin{equation} \label{qu}
    \pp_n(X,t) \le \frac{\sigma_X^2}{1+\sigma_X^2 t} < \frac{1}{t}.
\end{equation}
If in addition we have $|\supp(X)|>n,$ then
\begin{equation} \label{qv}
    \mathrm{pmmse}_n(X,t) = \frac{1}{t} + O(t^{-2})
\end{equation}
as $t\to \infty.$
\end{corollary}
\begin{proof}
Equation \eqref{qt} follows from \eqref{in} by setting $t=0.$ Inequality \eqref{qu} follows since $\pp_n(X,t)\le \pp_1(X,t).$ Also, if $|\supp(X)|>n$ then $\det \bM_{X,n} \neq 0,$ so \eqref{qv} follows from \eqref{in}.
\end{proof}

Second, by rationality of the PMMSE, and since the denominator $\det \bM_{\sqrt{t}X+N,n}$ is a polynomial that is strictly positive over $t\in [0,\infty),$ we obtain analyticity of the PMMSE.
\begin{corollary} \label{hs}
For $n\in \mathbb{N}$ and a RV $X$ such that $\mathbb{E}\left[ X^{2n} \right]<\infty,$ the map $t\mapsto \mathrm{pmmse}_n(X,t)$ is real analytic at each $t\in [0,\infty).$ 
\end{corollary}
\begin{proof}
A rational function is analytic at each point in its domain. For each $t\ge 0,$ $|\supp(\sqrt{t}X+N)|=\infty$ where $N\sim \calN(0,1)$ independent of $X.$ Therefore, $ \bM_{\sqrt{t}X+N}$ is invertible for every $t\ge 0,$ i.e., the denominator in \eqref{in} is never zero for $t\ge 0.$
\end{proof}

Our final by-product of Theorem \ref{ir} builds upon Corollaries \ref{is} and \ref{hs} to obtain the uniform convergence (in the SNR) of the PMMSE to the MMSE.
\begin{theorem}\label{an}
If the MGF of a RV $X$ exists, then we have the uniform convergence
\begin{equation} \label{ef}
    \lim_{n\to \infty} ~ \sup_{t \ge 0} ~ \mathrm{pmmse}_n(X,t) - \mathrm{mmse}(X,t)  = 0.
\end{equation}
\end{theorem}
\begin{proof}
See Appendix \ref{br}.
\end{proof}
\begin{remark}
The assumption that the MGF of $X$ exists is imposed so that $\sqrt{t}X+N$ satisfies Carleman's condition (for $N\sim \calN(0,1)$ independent of $X,$ and $t\ge 0$ fixed), which holds because $\sqrt{t}X+N$ will then have a MGF. It is not true in general that Carleman's condition is satisfied by the sum of two independent RVs each satisfying Carleman's condition, see~\cite[Proposition 3.1]{Berg1985}.
\end{remark}

In the remainder of this section, we prove Theorem \ref{ir}.

\subsection{Setup for the Proof of Theorem \ref{ir}}

We prove a slight strengthening of Theorem \ref{ir} by characterizing  the coefficients of the rational function $\pp_n(X,t),$ for an arbitrary fixed $2n$-times integrable RV $X.$ For convenience, we start with some additional notation.

We denote the moments of $X$ by $\calX_k$ so that for each $k\in \BN$
\begin{equation}
    \calX_k := \BE\left[ X^k \right].
\end{equation}
Note that $\calX_0=1$ holds regardless of what the RV $X$ is. It is convenient to look at the following notion of weighted-degree polynomial expression in the moments $\{ \calX_k \}_{k\in \BN}.$ 

Recall that (see \eqref{kd})
\begin{equation}
    \det \bM_{X,2} =  \calX_4\calX_2 - \calX_4 \calX_1^2   - \calX_3^2  + 2 \calX_3 \calX_2 \calX_1 - \calX_2^3.
\end{equation}
Thus, $\det \bm{M}_{X,2}$ is a polynomial in the moments of $X$ with integer coefficients. Further, each monomial appearing in the expression for $\det \bm{M}_{X,2}$ is of the form $\calX_a \calX_b \calX_c$ for $a+b+c=6.$ We formalize this observation next.

\begin{definition}
For $(\ell,m,k)\in \BN^3,$ let $\Pi_{\ell,m,k}$ denote the set of unordered partitions of $\ell$ into at most $m$ parts each of which is at most $k$
\begin{equation}
    \Pi_{\ell,m,k} := \left\{ \blambda \in \BN^m ~;~ k\ge \lambda_1 \ge \cdots \ge \lambda_m, ~ \blambda^T \mathbf{1} = \ell \right\}.
\end{equation}
\end{definition}
\begin{example}
The only unordered partition of $5$ into at most $2$ parts each of which at most $3$ is $5 = 3+2.$ Thus, we have $\Pi_{5,2,3} = \{ (3,2) \}.$ Another example is that the partitions of $6$ into at most $3$ parts each of which at most $4$ are
\begin{equation} \label{ke}
    \pushQED{\qed} 
    \Pi_{6,3,4} = \{(4,2,0),(4,1,1),(3,3,0),(3,2,1),(2,2,2)\}.
    \qedhere
    \popQED
\end{equation}
\end{example}

Note the resemblance between the partitions in \eqref{ke} comprising $\Pi_{6,3,4}$ and the terms appearing in the expression for $\det \bm{M}_{X,2},$ 
\begin{equation} \label{kf}
    \det \bM_{X,2} =  \calX_4\calX_2 - \calX_4 \calX_1^2   - \calX_3^2  + 2 \calX_3 \calX_2 \calX_1 - \calX_2^3.
\end{equation}
Namely, a term $\prod_{i=1}^3 \calX_{\lambda_i}$ with $\lambda_1\ge \lambda_2 \ge \lambda_3$ appears in $\det \bM_{X,2}$ if and only if $\blambda=(\lambda_1,\lambda_2,\lambda_3)$ is in $\Pi_{6,3,4}.$

Leibniz's formula for the determinant can be used to show that, in general, $\det \bM_{X,n}$ is an integer linear combination of terms $\prod_{i=1}^{n+1} \calX_{\lambda_i}$ where $\blambda \in \Pi_{n(n+1),n+1,2n},$ i.e., we may write
\begin{equation} \label{qn}
    \det \bM_{X,n} = \sum_{\blambda \in \Pi_{n(n+1),n+1,2n}}  d_{\blambda} \hspace{2mm} \prod_{i=1}^{n+1} \calX_{\lambda_i}
\end{equation}
for some integers $d_{\blambda}.$ Each term $\prod_{i=1}^{n+1} \calX_{\lambda_i}$ in \eqref{qn} shares the property that 
\begin{equation}
    \sum_{i=1}^{n+1}\lambda_i=n(n+1)
\end{equation}
is constant. Looking at $\calX_k$ as an indeterminate of ``degree" $k,$ we may view $\det \bM_{X,n}$ as a ``homogeneous" polynomial in the moments of $X$ (of ``degree" $n(n+1)$). In other words, we may write $\det \bm{M}_{X,n}$ as an integer linear combination of terms of the form $\prod_{i=1}^{2n} \BE\left[ X^i \right]^{\alpha_i}$ for integers $\alpha_i$ such that $\alpha_1+\cdots+2n\alpha_{2n}$ is constant (and equal to $n(n+1)$). Then, for any constant $c,$ $\det \bm{M}_{cX,n} = c^{n(n+1)} \det \bm{M}_{X,n}$; in fact, this homogeneity holds for each term in the sum, $\prod_{i=1}^{2n} \BE\left[ (cX)^i \right]^{\alpha_i} = c^{n(n+1)} \prod_{i=1}^{2n} \BE\left[ X^i \right]^{\alpha_i}.$

\begin{definition} \label{nd}
For $(\ell,m,k)\in \BN^3,$ we define the set of homogeneous integer-coefficient polynomials of weighted-degree $\ell$ of width at most $m$ in the first $k$ moments $\calX_1,\cdots,\calX_k$ of $X$ as
\begin{equation}
    H_{\ell,m,k}(X) := \left\{ \sum_{\blambda \in \Pi_{\ell,m,k}} d_{\blambda} \prod_{i=1}^m \calX_{\lambda_i} ~;~ d_{\blambda} \in \BZ \right\}. 
\end{equation}
If $\Pi_{\ell,m,k}=\emptyset,$ we set $H_{\ell,m,k}(X)=\BZ.$
\end{definition}
\begin{remark}
An element in $H_{\ell,m,k}(X)$ will be an integer linear combination of terms $\prod_{i=1}^m \calX_{\lambda_i}.$ Each of these terms is a product of at most $m$ of the moments of $X$ (hence the terminology \emph{width}). The highest moment that can appear is $\calX_k,$ because $\blambda \in \Pi_{\ell,m,k}.$ Each summand shares the property that $\sum_{i=1}^m \lambda_i = \ell.$
\end{remark}

\begin{example}
Because $\Pi_{4,2,3} = \{ (3,1),(2,2) \},$ we have
\begin{equation}
    H_{4,2,3}(X) = \{a \calX_3 \calX_1 + b \calX_2^2  ~ ; ~ a,b\in \BZ\}.
\end{equation}
Also, since $\Pi_{2,2,2} = \{ (2,0),(1,1) \},$ 
\begin{equation}
    \sigma_X^2 = \calX_2 - \calX_1^2 \in H_{2,2,2}(X).
\end{equation}
In general, Leibniz's formula yields (see \eqref{qn})
\begin{equation}
    \pushQED{\qed}
    \det \bM_{X,n}  \in H_{n(n+1),n+1,2n}(X).
    \qedhere
    \popQED
\end{equation}
\end{example}

For brevity, we introduce the following functions. Let $N\sim \calN(0,1)$ be independent of $X.$ For $k \in [n],$ if $\|X\|_{k+1} < \infty,$ we define the function $  v_{X,k}:[0,\infty)\to \BR$ at each $t\ge 0$ by
\begin{align} \label{ix}
    v_{X,k}(t) &:= \BE \left[ X\left(\sqrt{t}X+N\right)^k \right].
\end{align}
For example, $v_{X,0}(t) = \calX_1,$ $v_{X,1}(t)=\sqrt{t} \calX_2,$ and $v_{X,2}(t) = t\calX_3 + \calX_1$ for $X\in L^3(\calF).$ We also define vector-valued functions $\bv_{X,n}:[0,\infty)\to \BR^{n+1}$ for $n\in \BN$ and $X\in L^{n+1}(\calF)$ via
\begin{equation} \label{my}
    \bv_{X,n} := (v_{X,0},\cdots,v_{X,n})^T.
\end{equation}
In view of Theorem \ref{hw} and this definition of $\bv_{X,n},$ we may represent the PMMSE as 
\begin{equation}
    \pp_n(X,t) = \BE \left[ X^2 \right] - \bv_{X,n}(t)^T  \bM_{\sqrt{t}X+N,n}^{-1} \bv_{X,n}(t).
\end{equation}
Therefore, defining $F_{X,n}:[0,\infty)\to [0,\infty)$ by
\begin{align}
F_{X,n}(t) &:= \bv_{X,n}(t)^T  \bM_{\sqrt{t}X+N,n}^{-1} \bv_{X,n}(t), \label{mx}
\end{align}
we have the equation
\begin{equation} \label{iy}
    \pp_n(X,t) = \BE \left[ X^2 \right] - F_{X,n}(t).
\end{equation}
The functions $F_{X,n}$ are non-negative because the matrices $ \bM_{\sqrt{t}X+N,n}$ are positive-definite (see Lemma \ref{io}). In view of \eqref{iy},  PMMSE is fully characterized by $F_{X,n},$ and we focus on this function in the next subsections. 

\subsubsection{An Exact Characterization of PMMSE in Gaussian Channels}

We utilize Cramer's rule along with the Leibniz formula for determinants to prove the following characterization of the PMMSE in Gaussian channels. This characterization is a generalization of Theorem~\ref{ir}.
\begin{theorem} \label{iv}
Fix a natural $n\ge 1$ and a RV $X$ satisfying $\BE[X^{2n}]<\infty,$ and set $d_n:= \binom{n+1}{2}.$ The mapping $t\mapsto \pp_n(X,t)$ over $[0,\infty)$ is a rational function
\begin{equation} \label{iu}
    \pp_n(X,t) = \frac{\sum_{j\in [d_n-1]} ~ a_{X}^{n,j} ~ t^j}{\sum_{j \in [d_n]} ~ b_{X}^{n,j} ~ t^j},
\end{equation}
with the constants $a_X^{n,j}$ and $b_X^{n,j}$ satisfying
\begin{equation} \label{mp}
    a_X^{n,j} \in H_{2j+2,\, \min(n,2j)+2,\, 2\min(n,j+1)\, }(X)
\end{equation}
and
\begin{equation} \label{mq}
    b_X^{n,j} \in H_{2j,\, \min(n+1,2j),\, 2\min(n,j)\,}(X),
\end{equation}
where $H_{\ell,m,k}(X)$ is as given in Definition \ref{nd}. Also, we have the formulas
\begin{align}
    b_X^{n,0} &= G(n+2) \label{lz} \\
    b_X^{n,1} &= G(n+2)d_n \sigma_X^2  \label{ma} \\
    b_X^{n,d_n} &= \det  \bM_{X,n} \\
    a_X^{n,0} &=  G(n+2)  \sigma_X^2\label{ly} \\
    a_X^{n,d_n-1} &= \det  \bM_{X,n}, \label{mb}
\end{align}
where $G(k+2)=\prod_{j=1}^{k-2} j!$ is the Barnes $G$-function. Furthermore, all of the constants $a_X^{n,j}$ and $b_X^{n,j}$ are shift-invariant, i.e., for any $s \in \BR$ we have that
\begin{equation} \label{sb}
    a_{X+s}^{n,j} = a_X^{n,j} \quad \text{and} \quad b_{X+s}^{n,j} = b_X^{n,j}.
\end{equation}
\end{theorem}
\begin{remark} \label{qw}
We list next additional details that are omitted from the statement of the theorem but are evident in its proof. First, if $N\sim \calN(0,1)$ is independent of $X,$ then the denominator and numerator in \eqref{iu} are $\det  \bM_{\sqrt{t}X+N,n}$ and
\begin{equation}
   \left( \BE\left[ X^2 \right] - F_{X,n}(t) \right) \det  \bM_{\sqrt{t}X+N,n},
\end{equation}
respectively (see~\eqref{mx} for the definition of $F_{X,n}$). Second, we show a stricter relation than the one in \eqref{mp}, namely,
\begin{equation} \label{ne}
    a_X^{n,j} \in H_{2j+2,\min(n,2j)+2,\tau_n(j)}(X)
\end{equation}
where $\tau_n(j) \le 2\min(n,j+1)$ is defined by
\begin{equation}
    \tau_n(j) = \left\lbrace \begin{array}{cl}
    2 & \text{if } j=0, \\
    2j+1 & \text{if } 1\le j \le \frac{n}{2}, \\
    2j & \text{if } \frac{n+1}{2} \le j \le n, \\
    2n & \text{if } j>n. 
    \end{array} \right.
\end{equation}
For example, \eqref{ne} says that $a_X^{2,1} \in H_{4,4,3}(X),$ whereas \eqref{mp} only gives the weaker relation $a_X^{2,1} \in H_{4,4,4}(X).$ Note that $H_{4,4,3}(X) \subsetneq H_{4,4,4}(X),$ and, in fact, $a_X^{2,1} = \sigma_X^4 = \calX_2^2-2\calX_1^2\calX_2+\calX_1^4 \in H_{4,4,2}(X).$ Third, we give formulas for all of the coefficients $a_X^{n,j}$ and $b_X^{n,j}$ in expanded polynomial form. Consider tuples $\bk=(k_0,\cdots,k_n)\in [2n]^{n+1},$ let $\sn$ be the symmetric group of permutations on $[n],$ and denote the sign of a permutation $\pi \in \sn$ by $\mathrm{sgn}(\pi).$ We show that, for each $j\in [d_n],$
\begin{equation} \label{tr}
    b_X^{n,j} \hspace{2mm} = \hspace{1cm} \sum_{\mathclap{\substack{
    \pi \in  \sn \\ 
    k_r \in [r+\pi(r)], ~ \forall r\in [n] \\
    k_0+\cdots+k_n = 2j }}} \hspace{6mm}  \beta_{\pi,\bk} ~ \calX_{k_0}\cdots \calX_{k_n},
\end{equation}
where the $\beta_{\pi,\bk}$ are integers given by
\begin{equation}
    \beta_{\pi,\bk} = \mathrm{sgn}(\pi) \prod_{r\in [n]} \binom{r+\pi(r)}{k_r} \BE[N^{r+\pi(r)-k_r}].
\end{equation}
Also, for each $j\in [d_n-1],$ denoting the restricted sums
\begin{equation}
    s_i(\bk) = \sum_{r\in [n]\setminus \{i\}} k_r,
\end{equation}
we derive the formula
\begin{align} \label{tq}
    a_X^{n,j} \hspace{3mm} = \hspace{1cm} \sum_{\mathclap{\substack{
    \pi \in  \sn \\ 
    k_r \in [r+\pi(r)], ~ \forall r\in [n] \\
    k_0+\cdots+k_n = 2j }}} \hspace{8mm}  \beta_{\pi,\bk} \calX_2 \calX_{k_0}\cdots \calX_{k_n} \hspace{3mm} - \hspace{1cm} \sum_{\mathclap{\substack{
    (i,\pi) \in  [n]\times \sn \\
    (w,z) \in [i]\times [\pi(i)] \\
    k_r \in [r+\pi(r)], ~ \forall r\in [n] \setminus\{i\} \\
    w+z+s_i(\bk) = 2j}}} \hspace{1cm} \gamma_{i,\pi,\bk,w,z} \calX_{w+1}\calX_{z+1}  \prod_{r\in [n]\setminus\{i\}}  \calX_{k_r}.
\end{align}
where the integers $\gamma_{i,\pi,\bk,w,z}$ are given by
\begin{equation}
     \gamma_{i,\pi,\bk,w,z} =  (-1)^{i+\pi(i)}\mathrm{sgn}(\pi)\binom{i}{w} \binom{\pi(i)}{z}\BE[N^{i-w}]\BE[N^{\pi(i)-z}]  \prod_{r\in [n]\setminus\{i\}}  \binom{r+\pi(r)}{k_r} \BE[N^{r+\pi(r)-k_r}].
\end{equation}
Finally, the Barnes $G$-function is also the determinant of the Hankel matrix of Gaussian moments
\begin{equation}
    G(n+2) = \det \bM_{N,n} = \prod_{k=1}^n k!.
\end{equation}
\end{remark}

\subsubsection{Proof Steps} 

Note that Theorem~\ref{iv} subsumes Theorem~\ref{ir}, because Theorem~\ref{iv} asserts the PMMSE-rationality claim in Theorem~\ref{ir} and gives an additional characterization of the coefficients of the numerator and denominator. We present the proof of Theorem~\ref{iv} in a series of intermediate results: 

\begin{enumerate}[label=(\roman*)]
    \item We show that both functions $t \mapsto \det  \bM_{\sqrt{t}X+N,n}$ and $t \mapsto F_{X,n}(t)\det  \bM_{\sqrt{t}X+N,n}$ are polynomials in $t$ of degree at most $d_n := \binom{n+1}{2}$ (Lemmas~\ref{kk} and \ref{kl}). In view of \eqref{iy}, this implies  that the PMMSE is a rational function
\begin{equation}
    \pp_n(X,t) = \frac{\sum_{j=0}^{d_n} a_{X}^{n,j} t^j}{\sum_{j=0}^{d_n} b_{X}^{n,j} t^j}
\end{equation}
for some constants $a_{X}^{n,j}$ and $b_{X}^{n,j},$ where the denominator is 
\begin{equation}
    \sum_{j=0}^{d_n} b_{X}^{n,j} t^j = \det  \bM_{\sqrt{t}X+N,n}
\end{equation}
and the numerator is 
\begin{equation}
    \sum_{j=0}^{d_n} a_{X}^{n,j} t^j = \left( \calX_2 - F_{X,n}(t) \right) \det  \bM_{\sqrt{t}X+N,n}.
\end{equation}
Shift-invariance~\eqref{sb} and formulas \eqref{lz}--\eqref{ly} follow immediately (Lemma~\ref{qy}). Also, $a_X^{n,d_n}=0$ follows because $\pp_n(X,t) \le \LL(X,t)$ and $\LL(X,t)\to 0$ as $t\to \infty.$

\item We show that $a_X^{n,d_n-1}=\det \bM_{X,n}$ (equation \eqref{mb}) holds under the assumption that $X$ is continuous. This is done by leveraging results on the MMSE dimension of continuous RVs~\cite{Wu2011}.

\item We derive the polynomial formulas for the constants $a_X^{n,j}$ and $b_X^{n,j}$ stated in Remark \ref{qw}.

\item We finish the proof for a general RV $X$ by leveraging results on the truncated Hamburger moment problem~\cite{Curto91}. In particular, we prove that if the restriction of a multivariate polynomial to moments of arbitrary continuous RVs vanishes then the polynomial must be identically zero (Proposition \ref{qx}).
\end{enumerate}

Proofs of the component lemmas in the following subsection are deferred to Appendices~\ref{ki}--\ref{sg}.

\subsection{Proof of Theorem \ref{iv}} \label{kg}

Throughout the proof, $N\sim \calN(0,1)$ is independent of $X.$ Let $\sn$ denote the symmetric group of permutations on the $n+1$ elements $[n].$ We utilize the following auxiliary result on the parity of $i+\pi(i)$ for a permutation $\pi \in \sn.$
 
\begin{lemma}\label{ac}
For any permutation $\pi \in \sn,$ there is an even number of elements $i\in [n]$ such that $i+\pi(i)$ is odd, i.e., the integer
\begin{equation} \label{kj}
    \delta(\pi) := \left| \{ i\in [n] \; ; \; i+\pi(i) ~~ \mathrm{is}~\mathrm{odd}\}\right|
\end{equation}
is even. 
\end{lemma}
\begin{proof}
 See Appendix \ref{ki}.
\end{proof}

\subsubsection{Rationality of the PMMSE} \label{yi}

We introduce the following auxiliary polynomials, where $R$ is a RV independent of $N\sim \calN(0,1).$ For $\ell$ even, we set
\begin{equation}\label{cj}
e_{R,X,\ell}(t):= \mathbb{E}\left[ R\left( \sqrt{t}X+N \right)^\ell \right],
\end{equation}
and for $\ell$ odd we set (for $t>0$)
\begin{equation}\label{ck}
o_{R,X,\ell}(t):= t^{-1/2}\mathbb{E}\left[ R\left( \sqrt{t}X+N \right)^\ell \right] .
\end{equation}
That $e_{R,X,\ell}$ and $o_{R,X,\ell}$ are polynomials in $t$ can be seen as follows. Recall that $\mathbb{E}[N^r]=0$ for odd $r\in \BN.$ If $\ell$ is even then expanding the right hand side of \eqref{cj} yields
\begin{equation}\label{ch}
    e_{R,X,\ell}(t) = \sum_{k \; \text{even}} \binom{\ell}{k} t^{k/2}   \BE\left[ RX^k \right] \BE \left[ N^{\ell-k} \right],
\end{equation}
whereas if $\ell$ is odd then expanding the right hand side of \eqref{ck} yields
\begin{equation}\label{ci}
    o_{R,X,\ell}  =  \sum_{k \; \text{odd}} \binom{\ell}{k} t^{(k-1)/2}  \BE \left[ RX^k \right] \BE \left[N^{\ell-k}\right].
\end{equation}
Both expressions on the right hand side of \eqref{ch} and \eqref{ci} are polynomials of degree at most $\lfloor \ell/2 \rfloor.$ Further, the coefficient of $t^{\lfloor \ell/2 \rfloor}$ in either polynomial is $\mathbb{E}\left[ R X^{\ell} \right].$

We show first that the function $t \mapsto \det  \bM_{\sqrt{t}X+N,n}$ is a polynomial in $t,$ and find the leading coefficient. For the proof, we utilize the polynomials $e_{1,X,\ell}$ and $o_{1,X,\ell}$ (i.e., $R=1$) as defined in \eqref{cj} and \eqref{ck}.

\begin{lemma} \label{kk}
For a RV $X$ and $n\in \BN_{>0}$ such that $\BE\left[ X^{2n} \right]<\infty,$ and for $N\sim \calN(0,1)$ independent of $X,$ the function $t\mapsto  \det \bM_{\sqrt{t}X+N,n}$ over $t\in [0,\infty)$ is a polynomial of degree at most $d_n = \binom{n+1}{2}.$ Further, the coefficient of $t^{d_n}$ is $\det \bM_{X,n},$ which is nonzero if and only if $|\supp(X)|>n.$
\end{lemma}
\begin{proof}
By Leibniz's formula,
\begin{equation}\label{cl}
    \det  \bM_{\sqrt{t}X+N,n} = \sum_{\pi \in \sn}  \mathrm{sgn}(\pi) \prod_{r\in [n]} \BE\left[ \left( \sqrt{t}X+N \right)^{r+\pi(r)} \right].
\end{equation}
With the auxiliary polynomials $e_{1,X,\ell}$ and $o_{1,X,\ell}$ as defined in (\ref{cj}) and (\ref{ck}), and $\delta$ as defined in \eqref{kj}, we may write
\begin{equation}
\det  \bM_{\sqrt{t}X+N,n} = \sum_{\pi \in \sn} \mathrm{sgn}(\pi) t^{\delta(\pi)/2}  \quad \quad
\prod_{\mathclap{\substack{i \; : \; i+\pi(i) \; \text{odd} }}} ~~ o_{1,X,i+\pi(i)}(t) \quad \quad \prod_{\mathclap{\substack{j \; : \; j+\pi(j) \; \text{even} }}}  ~~ e_{1,X,j+\pi(j)}(t),
\end{equation}
thereby showing that $\det  \bM_{\sqrt{t}X+N,n}$ is a polynomial in $t$ by evenness of each $\delta(\pi)$ (Lemma \ref{ac}). Furthermore, for each permutation $\pi \in \sn,$
\begin{align}
\deg \left( t^{\delta(\pi)/2} \hspace{3mm} \prod_{\mathclap{\substack{ i+\pi(i) \; \text{odd} }}} o_{1,X,i+\pi(i)}(t) \hspace{3mm} \prod_{\mathclap{\substack{ j+\pi(j) \; \text{even} }}}  e_{1,X,j+\pi(j)}(t) \right) &\le \frac{\delta(\pi)}{2} + \hspace{-2mm}\sum_{ i+\pi(i) \; \text{odd}} \hspace{-2mm} \frac{i+\pi(i)-1}{2} + \hspace{-2mm} \sum_{ j+\pi(j) \; \text{even}} \hspace{-2mm} \frac{j+\pi(j)}{2} \\
&= \frac{1}{2} \sum_{k=0}^n k + \pi(k) = \frac{n(n+1)}{2} = d_n.
\end{align}
Therefore, we also have
\begin{equation}
    \deg\left( \det \bM_{\sqrt{t}X+N,n} \right) \le d_n.
\end{equation}
Finally, taking the terms of highest degrees (in $\sqrt{t}$) in (\ref{cl}), we obtain that the coefficient of $t^{d_n}$ in $\det  \bM_{\sqrt{t}X+N,n}$ is 
\begin{equation}
    \sum_{\pi \in \sn} \mathrm{sgn}(\pi) \prod_{r\in [n]} \calX_{r+\pi(r)},
\end{equation}
which is equal to $\det  \bM_{X,n}$ by the Leibniz determinant formula. This coefficient is non-negative because $\bM_{X,n}$ is positive-semidefinite, and it is nonzero if and only if $|\mathrm{supp}(X)|>n$ by Lemma \ref{io}.
\end{proof}

The same method of proof in Lemma \ref{kk} can be used to show that $F_{X,n}(t)\det  \bM_{\sqrt{t}X+N,n}$ is a polynomial in $t$ and to characterize its leading coefficient. In this case, we utilize $e_{X,X,\ell}$ and $o_{X,X,\ell}$ (i.e., $R=X$).

\begin{lemma} \label{kl}
For $n\in \BN,$ $N\sim \calN(0,1),$ and a RV $X$ independent of $N$ and satisfying $\BE\left[ X^{2n} \right] < \infty,$ the function $t \mapsto F_{X,n}(t) \det \bM_{\sqrt{t}X+
N,n}$ is a polynomial of degree at most $d_n = \binom{n+1}{2}$ and the coefficient of $t^{d_n}$ is
\begin{equation} \label{sd}
    \sum_{i \in [n]}  \sum_{\pi \in \sn} \mathrm{sgn}(\pi) \calX_{i+1}\calX_{\pi(i)+1} \prod_{r\in [n]\setminus \{i\}} \calX_{r+\pi(r)}.
\end{equation}
\end{lemma}
\begin{proof}
See Appendix \ref{km}.
\end{proof}
\begin{remark}
The coefficient of $t^{d_n}$ given by~\eqref{sd} is shown in the remainder of this section to simplify to $\calX_2 \det \bM_{X,n}.$ See also Appendix~\ref{aw} for another proof.
\end{remark}

We next combine Lemmas \ref{kk} and \ref{kl} to build the rational polynomial structure of $\pp_n(X,t).$ Recall from \eqref{iy} that 
\begin{equation} \label{mh}
    \pp_n(X,t) = \calX_2 - F_{X,n}(t).
\end{equation}
Multiplying and dividing by $\det \bM_{\sqrt{t}X+N,n},$ 
\begin{equation} \label{mr}
    \pp_n(X,t) = \frac{\left( \calX_2 - F_{X,n}(t) \right)\det \bM_{\sqrt{t}X+N,n}}{\det \bM_{\sqrt{t}X+N,n}}.
\end{equation}
From Lemmas \ref{kk} and \ref{kl}, we know that both expressions in the numerator and the denominator in \eqref{mr} are polynomials in $t,$ each of degree at most $d_n.$ Therefore, we deduce that $\pp_n(X,t)$ is a rational function
\begin{equation} \label{mg}
    \pp_n(X,t) = \frac{\sum_{j\in [d_n]} ~ a_{X}^{n,j} ~ t^j}{\sum_{j \in [d_n]} ~  b_{X}^{n,j} ~ t^j}
\end{equation}
where the constants $a_X^{n,j}$ and $b_X^{n,j}$ are defined by
\begin{align}
    \left( \calX_2 - F_{X,n}(t) \right)\det \bM_{\sqrt{t}X+N,n} &= \sum_{j\in[d_n]} a_{X}^{n,j} ~ t^j \label{mv} \\
    \det \bM_{\sqrt{t}X+N,n} &= \sum_{j\in[d_n]} b_{X}^{n,j} ~ t^j. \label{mu}
\end{align}

In \eqref{mg}, the numerator has terms of degree up to $d_n.$ However, as stated in Theorem \ref{iv} (see equation \eqref{iu}), we only need a numerator of degree $d_n-1,$ i.e., we have that
\begin{equation} \label{mi}
    a_X^{n,d_n} = 0.
\end{equation}
To see that \eqref{mi} holds, note that
\begin{equation}
    \pp_n(X,t) \le \LL(X,t) \le \frac{1}{t}.
\end{equation}
Hence, $\pp_n(X,t) \to 0$ as $t\to \infty,$ thereby yielding~\eqref{mi} in view of~\eqref{mg}. 

The following lemma derives equations \eqref{lz}--\eqref{ly}, as stated in Theorem~\ref{iv}.

\begin{lemma} \label{qy}
Consider a RV $X$ and a positive integer $n$ such that $\BE[X^{2n}]<\infty,$ and set $d_n = \binom{n+1}{2}.$ Equations~\eqref{lz}--\eqref{ly} hold, i.e.,
\begin{align}
    b_X^{n,0} &= G(n+2) \\
    b_X^{n,1} &= G(n+2) d_n \sigma_X^2  \label{ra} \\
    b_X^{n,d_n} &= \det  \bM_{X,n} \label{qz} \\
    a_X^{n,0} &= G(n+2) \sigma_X^2 .
\end{align}
Further, for any $j\in [d_n]$ and $s\in \BR,$ we have the shift-invariance \eqref{sb}
\begin{equation} \label{rb}
    a_{X+s}^{n,j} = a_X^{n,j} \quad \text{and} \quad b_{X+s}^{n,j} = b_X^{n,j}
\end{equation}
\end{lemma}
\begin{proof}
See Appendix \ref{rd}.
\end{proof}

We have yet to prove the claim in the theorem (see equation \eqref{mb}) that $a_X^{n,d_n-1} = \det \bM_{X,n}.$ We prove this equation next for continuous $X,$ then generalize for every RV $X.$

\subsubsection{Proof of~\eqref{mb} for Continuous RVs} \label{rf}

Assume for now that $X$ is continuous. In particular, $|\supp(X)|=\infty,$ so $\det \bM_{X,n}\neq 0.$ Now, note that the PMMSE is bounded by the LMMSE and the MMSE,
\begin{equation} \label{mj}
    \mm(X,t) \le \pp_n(X,t) \le \LL(X,t).
\end{equation}
We have that $\LL(X,t) \sim 1/t,$ and $\mm(X,t) \sim 1/t$ as $X$ is assumed to be continuous~\cite{Wu2011}. Thus, by \eqref{mj}, $\pp_n(X,t)\sim 1/t$ as $t\to \infty.$ Moreover, in the denominator of $\pp_n(X,t)$ in \eqref{mg} we have that, by Lemma~\ref{qy},
\begin{equation}
    b_X^{n,d_n} = \det \bM_{X,n} \neq 0,
\end{equation}
i.e., the denominator is of degree exactly $d_n.$ Therefore, from \eqref{mg} and \eqref{mi}, we deduce
\begin{equation}
    a_X^{n,d_n-1} = b_X^{n,d_n} = \det \bM_{X,n}.
\end{equation}
We have thus shown the desired equation~\eqref{mb} when $X$ is continuous.

Next, we return to the general case (i.e., not necessarily continuous $X$). Our approach is first to realize that the equation $a_X^{n,d_n-1}= \det \bM_{X,n}$ asserts the vanishing of a multivariate polynomial in the moments of every continuous $2n$-times integrable RV $X.$ Then, we show that such locus is too large for any nonzero polynomial, i.e., that such a polynomial must vanish identically. This result would imply, in particular, that $a_X^{n,d_n-1}= \det \bM_{X,n}$ holds even when $X$ is not continuous.

\subsubsection{Coefficient Formulas} \label{rg}

We develop next multivariate-polynomial expressions for the coefficients in the PMMSE as given by \eqref{mg}, which show that those coefficients are homogeneous polynomials in the moments; in particular, we prove relations \eqref{mp} and \eqref{mq}. First, we investigate $\det \bM_{\sqrt{t}X+N,n}$ when seen as a polynomial in $t.$ 

We have the expansion (see \eqref{cl})
\begin{equation} \label{amz}
    \det  \bM_{\sqrt{t}X+N,n} = \sum_{\pi \in \sn} \mathrm{sgn}(\pi) \prod_{r\in [n]} \BE\left[ \left( \sqrt{t}X+N \right)^{r+\pi(r)} \right]
\end{equation}
by the Leibniz formula. In the expressions that follow, we denote the tuple $\bk=(k_0,\cdots,k_n).$ Expanding the powers inside the expectation and computing the expectation, we get a formula of the form
\begin{equation} \label{na}
    \det \bM_{\sqrt{t}X+N,n} \quad = \quad \quad \sum_{\mathclap{\substack{
    \pi \in  \sn \\ 
    k_r \in [r+\pi(r)], ~ \forall r\in [n] }}} \quad \quad t^{(k_0+\cdots+k_n)/2} \calX_{k_0}\cdots \calX_{k_n} \beta_{\pi;\bk},
\end{equation}
where the $\beta_{\pi;\bk}$ are integers given by\footnote{From this formula, one may deduce an alternative proof of Lemma \ref{kk}. The term $\beta_{\pi;\bk}$ is nonzero if and only if all the differences $r+\pi(r)-k_r$ are even. Suppose, for the sake of contradiction, that this is true for some fixed permutation $\pi\in \sn$ and naturals $k_0,\cdots,k_n$ for which $k_0+\cdots+k_n$ is odd. Then, there is an odd number of odd numbers $k_r.$ But, by Lemma \ref{ac}, there is an even number of odd numbers $r+\pi(r).$ Therefore, there is an $r\in [n]$ for which $r+\pi(r)$ and $k_r$ have different parities, contradicting evenness of $r+\pi(r)-k_r.$}
\begin{equation}
    \beta_{\pi,\bk} :=  \mathrm{sgn}(\pi) \prod_{r\in [n]} \binom{r+\pi(r)}{k_r} \BE[N^{r+\pi(r)-k_r}].
\end{equation}
The summation may be restricted  further but, for the purpose of proving Theorem~\ref{iv}, it suffices to show that the coefficients are homogeneous polynomials in the moments. By Lemma \ref{kk}, only the summands for which the integer $k_0+\cdots+k_n$ is even can be non-trivial because $\det \bM_{\sqrt{t}X+N,n}$ is a polynomial in $t.$ Thus, we have
\begin{equation}
    \det \bM_{\sqrt{t}X+N,n} = \sum_{j\in [d_n]} t^j \hspace{10mm} \sum_{\mathclap{\substack{
    \pi \in  \sn \\ 
    k_r \in [r+\pi(r)], ~ \forall r\in [n] \\
    k_0+\cdots+k_n = 2j }}} \hspace{4mm}  \beta_{\pi;\bk} \calX_{k_0}\cdots \calX_{k_n}.
\end{equation}
Because the coefficients $b_X^{n,j}$ were defined by (see \eqref{mu})
\begin{equation} \label{mw}
    \det \bM_{\sqrt{t}X+N,n} = \sum_{j\in [d_n]} b_X^{n,j} t^j,
\end{equation}
then, we obtain that for each $j\in [d_n]$
\begin{equation}
    b_X^{n,j} \hspace{2mm} = \hspace{1cm} \sum_{\mathclap{\substack{
    \pi \in  \sn \\ 
    k_r \in [r+\pi(r)], ~ \forall r\in [n] \\
    k_0+\cdots+k_n = 2j }}} \hspace{6mm}  \beta_{\pi;\bk} \calX_{k_0}\cdots \calX_{k_n}.
\end{equation}
In particular, the relation in \eqref{mq} follows, namely,
\begin{equation}
    b_X^{n,j} \in H_{2j,\min(n+1,2j),2\min(n,j)}(X).
\end{equation}

The coefficients $a_X^{n,j}$ in the numerator of \eqref{mg} may be treated similarly, so we defer the details to Appendix~\ref{re}.

\subsubsection{Proof of \eqref{mb} for a General Random Variable}

We have shown in Section \ref{rf} that $a_X^{n,d_n-1}= \det \bM_{X,n}$ when $X$ is continuous. We generalize this fact here for any RV $X.$ In view of Section \ref{rg}, this equation takes the form $q(\BE[X],\cdots,\BE[X^{2n}])=0$ for some fixed polynomial $q.$ We demonstrate here that  $a_X^{n,d_n-1}= \det \bM_{X,n}$ generalizes to every $X$ by showing that the polynomial $q$ must vanish identically. 

The solution to the truncated Hamburger moment problem implies that for $\mu_0=1$ and any tuple $(\mu_1,\cdots,\mu_{2h+1})\in \BR^{2h+1},$ if the Hankel matrix $\left( \mu_{i+j} \right)_{(i,j)\in [h]^2}$ is positive definite, then there is a finitely-support discrete RV $W$ such that $\BE\left[W^k\right] = \mu_k$ for each $k\in [2h+1]$ (see~\cite[Theorem 3.1]{Curto91}). A consequence of this result for continuous RVs, which we use in the sequel, is the following.

\begin{lemma} \label{sf}
Fix $m\in \BN_{>0},$ set $\ell = \lfloor m/2 \rfloor$ and $\mu_0=1,$ and let $(\mu_1,\cdots,\mu_m)\in \BR^m$ be such that the matrix $(\mu_{i+j})_{(i,j)\in \left[ \ell \right]^2}$ is positive definite. For every $\varepsilon > 0,$ there exists a continuous RV $Z$ such that $\left| \BE\left[ Z^k \right] - \mu_k \right| < \varepsilon $ for every $k\in [m].$
\end{lemma}
\begin{proof}
See Appendix~\ref{se}.
\end{proof}

In the other direction, if $\mu_0=1$ and $(\mu_1,\cdots,\mu_{2h}) \in \BR^{2h}$ come from a continuous RV, i.e., $\BE\left[ Z^{k} \right] = \mu_k$ for each $k\in [2h],$ then it must be that the Hankel matrix $\bH = \left( \mu_{i+j} \right)_{(i,j)\in [h]^2}$ is positive definite; indeed, since $|\supp(Z)|=\infty,$ we have that for every nonzero real vector $\bv = (v_0,\cdots,v_h)^T$
\begin{equation}
    \bv^T \bH \bv = \left\|  \sum_{k\in [h]} v_k Z^k \right\|_2^2 > 0 .
\end{equation}
For each integer $m \ge 2,$ let $\calR_m \subset L^m(\calF)$ be the set of all continuous RVs $X$ such that $\BE[|X|^m]<\infty.$ Consider the set $\calC_m\subset \BR^m$ defined by
\begin{equation}
    \calC_m = \{ (\BE[X],\cdots,\BE[X^{m}])  ~ ; ~ X \in \calR_m\}.
\end{equation}

\begin{prop} \label{qx}
Let $p$ be a polynomial in $m$ variables with real coefficients. If 
\begin{equation}
    p\left( \BE[X],\cdots,\BE[X^m]\right) = 0
\end{equation}
for every continuous RV $X$ satisfying $\BE[|X|^m]<\infty,$ then $p$ is the zero polynomial.
\end{prop}
\begin{proof}
See Appendix~\ref{sg}.
\end{proof}

Since $a_{X}^{n,d_n-1}-\det \bM_{X,n}=p\left( \BE[X],\cdots,\BE[X^{2n}]\right)$ for some polynomial $p,$ and since we have shown that $p$ vanishes over $\calC_m,$ we conclude that $p$ vanishes identically.

\begin{corollary} \label{sl}
For any RV $X$ and $n\in \BN$ such that $\BE[X^{2n}]<\infty,$ we have that $a_X^{n,d_n-1}=\det \bM_{X,n}.$
\end{corollary}

With Corollary~\ref{sl}, the proof of Theorem~\ref{iv} is complete.

\begin{remark}
We provide in Appendix~\ref{aw} an alternative proof for the formulas $a_X^{n,d_n-1}=\det \bM_{X,n}$ and $a_X^{n,d_n}=0$ via a self-contained algebraic argument.
\end{remark}

\section{Conditional Expectation Derivatives} \label{pc}

We develop formulas for the higher-order derivatives of the conditional expectation, and establish upper bounds. The bounds in Theorem \ref{og} on the norm of the derivatives of the conditional expectation will be crucial in Section \ref{ji} for establishing a Bernstein approximation theorem that shows how well polynomials can approximate the conditional expectation in the mean-square sense.

\begin{theorem} \label{og}
Fix an integrable RV $X$ and an independent $N\sim \calN(0,1),$ and set $Y=X+N.$ Let $r\ge 2$ be an integer, let $C_r$ be as defined in \eqref{rx}, and denote $q_r:=\lfloor (\sqrt{8r+9}-3)/2 \rfloor$ and $\gamma_r:=(2rq_r)!^{1/(4q_r)}.$ We have the bound
\begin{equation}
    \left\| \frac{d^{r-1}}{dy^{r-1}}~ \BE[X\mid Y=y] \right\|_{2}   \le ~ 2^r C_r \min\left( \gamma_r, \|X\|_{2rq_r}^r \right).
\end{equation}
\end{theorem}

For $2\le r \le 7,$ we obtain the first few values of $q_r$ as $1,1,1,2,2,2,$ and we have $q_r \sim \sqrt{2r}$ as $r\to \infty$ (see Remark~\ref{sy} for a way to reduce $q_r$). To prove Theorem~\ref{og}, we first express the derivatives of $y\mapsto \BE[X\mid Y=y]$ as polynomials in the moments of the RV $X_y-\BE[X_y],$ where $X_y$ denotes the RV obtained from $X$ by conditioning on $\{Y=y\}.$

\begin{prop} \label{ars}
Fix an integrable RV $X$ and an independent $N\sim \calN(0,1),$ and let $Y=X+N.$ For each $(y,k)\in \BR\times \BN,$ denote $f(y) := \BE[X \mid Y=y]$ and
\begin{equation}
    g_k(y) := \mathbb{E}\left[ \left( X - \mathbb{E}[X \mid Y] \right)^k \mid Y = y \right].
\end{equation}
For $(\lambda_2,\cdots,\lambda_\ell)=\blambda\in \BN^*,$ denote $\bg^{\blambda} := \prod_{i=2}^\ell g_i^{\lambda_i},$ with the understanding that $g_i^{0}=1.$ Then, for every integer $r\ge 2,$ we have that
\begin{equation} \label{sw}
    f^{(r-1)} = \sum_{\blambda \in \Pi_r} e_{\blambda} \bg^{\blambda},
\end{equation}
where the integers $e_{\blambda}$ are as defined in~\eqref{ss}--\eqref{sp}.
\end{prop}
\begin{proof}
See Appendix~\ref{aad}.
\end{proof}
\begin{remark}
We note that in parallel to this work, Dytso et. al. also derived in~\cite{Dytso2021} the formula in Proposition~\ref{ars} via a shorter and more elegant proof. Further, it is shown in~\cite{Dytso2021} that formula~\eqref{sw} is the expansion of the $r$-th cumulant of $X_y$ in terms of the moments of $X_y$ via Bell polynomials.
\end{remark}

Now we are ready to prove Theorem~\ref{og}. 

\begin{proof}[Proof of Theorem~\ref{og}]
We use the notation of Proposition~\ref{ars}. Fix $(\lambda_2,\cdots,\lambda_\ell)=\blambda \in \Pi_r.$ By the generalization of H\"{o}lder's inequality stating $\|\psi_1\cdots \psi_k\|_1 \le \prod_{i=1}^k \|\psi_i\|_k,$ we have that
\begin{equation}  \label{sv}
    \left\| \bg^{\blambda}(Y) \right\|_{2}^{2} = \left\| \prod_{\lambda_i \neq 0} g_i^{2\lambda_i}(Y) \right\|_{1} \le \prod_{\lambda_i\neq 0} \left\| g_i^{2\lambda_i}(Y) \right\|_{s}
\end{equation}
where $s$ is the number of nonzero entries in $\blambda.$ By Jensen's inequality for conditional expectation, for each $i$ such that $\lambda_i\neq 0$ we have that
\begin{equation} \label{su}
    \left\| g_i^{2\lambda_i}(Y) \right\|_{s} \le \| X - \BE[X\mid Y] \|_{2i \lambda_i s}^{2i\lambda_i}.
\end{equation}
Now, $r = \sum_{i=2}^\ell i\lambda_i \ge \sum_{i=2}^{s+1} i = \frac{(s+1)(s+2)}{2}-1,$ so we have that $s^2+3s-2r \le 0,$ i.e., $s\le q_r.$ Further, $i\lambda_i \le r$ for each $i.$ Hence, monotonicity of norms and inequalities \eqref{sv} and \eqref{su} imply the uniform (in $\blambda$) bound
\begin{equation}
    \left\| \bg^{\blambda}(Y) \right\|_{2} \le  \| X - \BE[X\mid Y] \|_{2rq_r}^{r}.
\end{equation}
Observe that $\| X - \BE[X\mid Y] \|_{k}\le 2 \min\left( (k!)^{1/(2k)}, \|X\|_k\right)$ (see \cite{Guo2011}), so applying the triangle inequality in~\eqref{sw} we obtain
\begin{align}
    \left\| f^{(r-1)}(Y) \right\|_2 &\le \sum_{\blambda \in \Pi_r} c_{\blambda} \left\| \bg^{\blambda}(Y) \right\|_2 \le 2^rC_r ~ \min\left( \gamma_r, \|X\|_{2rq_r}^r \right),
\end{align}
where $\gamma_r=(2rq_r)!^{1/(4q_r)},$ as desired.
\end{proof}
\begin{remark} \label{sy}
A closer analysis reveals that $i\lambda_i s$ in~\eqref{su} cannot exceed $\beta_r:=t_r^2(t_r+1/2)$ where $t_r:=(\sqrt{6r+7}-1)/3.$ For $r\to \infty,$ we have $rq_r/\beta_r \sim 3^{3/2}/2 \approx 2.6.$ The reduction  when, e.g., $r=7,$ is from $rq_r=14$ to $\beta_r=10.$
\end{remark}

\section{A Bernstein Approximation Theorem for Conditional Expectation} \label{ji}

If $p \in \SD$ (see Definition~\ref{kx}) and $X\sim p,$ then we show that the error in approximating $\BE[X\mid Y]$ by the best polynomial $E_n[X\mid Y],$ where $Y$ is a Gaussian perturbation of $X,$ decays faster than any polynomial in $n.$

\begin{theorem} \label{kw}
Fix $p \in \SD,$ let $X\sim p,$ suppose $N\sim\calN(0,1)$ is independent of $X,$ and set $Y=X+N.$ There exists a sequence $\{ D(p,k) \}_{k\in \mathbb{N}}$ of constants such that for all integers  $n\ge \max(k-1,1)$ we have
\begin{equation}
    \left\|E_n[X\mid Y] - \mathbb{E}[X \mid Y] \right\|_2 \le \frac{D(p,k)}{n^{k/2}}.
\end{equation}
\end{theorem}

The proof relies on results on the Bernstein approximation problem in weighted $L^p$ spaces. In particular, we consider the Freud case~\cite[Definition 3.3]{Lubinsky2007}, where the weight is of the form $e^{-Q}$ for $Q$ of polynomial growth, e.g., a Gaussian weight.

\begin{definition}[Freud Weights,~\cite{Lubinsky2007}] \label{ig}
A function $W:\mathbb{R} \to (0,\infty)$ is called a \emph{Freud Weight}, and we write $W\in \SF,$ if it is of the form $W = e^{-Q}$ for $Q: \mathbb{R}\to \mathbb{R}$ satisfying:
\begin{enumerate}[label=(\arabic*)]
    \item $Q$ is even, \label{aak}
    \item $Q$ is differentiable, and $Q'(y)>0$ for $y>0,$ \label{aal}
    \item $y\mapsto yQ'(y)$ is strictly increasing over $(0,\infty),$ \label{aam}
    \item $yQ'(y)\to 0$ as $y\to 0^+,$ and \label{aan}
    \item \label{aao} there exist $\lambda,a,b,c > 1$ such that for every $y>c$ 
    \begin{equation}
        a \le \frac{Q'(\lambda y)}{Q'(y)} \le b.
    \end{equation} 
\end{enumerate}
\end{definition}

The convolution of a weight in $\SD$ with the Gaussian weight $\varphi(x) := e^{-x^2/2}/\sqrt{2\pi}$ is a Freud weight. This can be shown by noting that with $p_Y=e^{-Q}$ we have $Q'(y)=\BE[N \mid Y=y].$

\begin{theorem} \label{if}
If $p\in \SD$ and $X\sim p,$ then the probability density function of $X+N,$ for $N\sim \mathcal{N}(0,1)$ independent of $X,$ is a Freud weight.
\end{theorem}
\begin{proof}
See Appendix~\ref{aaf}.
\end{proof}

To be able to state the theorem we borrow from the Bernstein approximation literature, we need first to define the Mhaskar–Rakhmanov–Saff number.

\begin{definition} \label{lb}
If $Q:\mathbb{R}\to \mathbb{R}$ satisfies conditions~\ref{aal}--\ref{aan} in Definition~\ref{ig}, and if $yQ'(y)\to \infty$ as $y\to \infty,$ then the \emph{$n$-th Mhaskar–Rakhmanov–Saff number $a_n(Q)$ of $Q$} is defined as the unique positive root $a_n$ of the equation
\begin{equation}
    n = \frac{2}{\pi} \int_0^1 \frac{a_ntQ'(a_nt)}{\sqrt{1-t^2}} \, dt.
\end{equation}
\end{definition}
\begin{remark}
The condition $yQ'(y) \to \infty$ as $y\to \infty$ in Definition \ref{lb} is satisfied if $e^{-Q}$ is a Freud weight. Indeed, in view of properties~\ref{aal}--\ref{aam} in Definition \ref{ig}, the quantity $\ell := \lim_{y\to \infty} yQ'(y)$ is well-defined and it belongs to $(0,\infty].$ If $\ell\neq \infty,$ then because $\lim_{y\to \infty} \lambda y Q'(\lambda y) = \ell$ too, property~\ref{aao} would imply that $a \le 1/\lambda \le b$ contradicting that $\lambda,a>1.$ Therefore, $\ell=\infty.$
\end{remark}

For example, the weight $W(y)=e^{-y^2},$ for which $Q(y)=y^2,$ has $a_n(Q)=\sqrt{n}$ because $\int_0^1 t^2/\sqrt{1-t^2} \, dt = \frac{\pi}{4}.$ If $X\sim p \in \SD,$ say $\supp(p)\subset [-M,M],$ and $p_{Y}=e^{-Q}$ (where $N\sim \calN(0,1)$ is independent of $X,$ and $Y=X+N$), then (see Appendix~\ref{aag})
\begin{equation} \label{sx}
    a_n(Q) \le \left( 2 M +\sqrt{2} \right)\sqrt{n}.
\end{equation}

We apply the following Bernstein approximation theorem \cite[Corollary 3.6]{Lubinsky2007} to prove Theorem~\ref{kw}. 

\begin{theorem} \label{arq}
Fix $W\in \SF,$ and let $u$ be an $r$-times continuously differentiable function such that $u^{(r)}$ is absolutely continuous. Let $a_n=a_n(Q)$ where $W=e^{-Q},$ and fix $1\le s \le \infty.$ Then, for some constant $D(W,r,s)$ and every $n\ge \max(r-1,1)$ 
\begin{equation}
    \inf_{q\in \mathscr{P}_n} \|(q-u)W\|_{L^s(\mathbb{R})} \le D(W,r,s) \left( \frac{a_n}{n} \right)^r \|u^{(r)}W\|_{L^s(\mathbb{R})}.
\end{equation}
\end{theorem}

\begin{proof}[Proof of Theorem~\ref{kw}] 
Fix $k\in \BN$ and $n\ge \max(k-1,1).$ We apply Theorem~\ref{arq} for the function $u(y)=\BE[X\mid Y=y],$ the weight $W=\sqrt{p_Y},$ and for $s=2.$ By our choice of weight, $\|hW\|_{L^2(\BR)} = \|h(Y)\|_2$ for any Borel $h:\BR\to \BR.$  Recall from~\eqref{kp} that $E_n[X\mid Y]$ minimizes $\| q(Y) - \BE[X\mid Y\|_2$ over $q(Y)\in \SP_n(Y).$ By~\eqref{sx}, we have the bound $a_n = O_p(\sqrt{n}).$ Furthermore, by Theorem~\ref{og}, $\|(d^k/dy^k)\BE[X\mid Y]\|_2 = O_k(1).$ Note that $W\in \SF,$ because $W^2=p_Y\in \SF$ by Theorem~\ref{if}. Therefore, by Theorem~\ref{arq}, we obtain a constant $D(p,k)$ such that
\begin{equation}
    \|E_n[X\mid Y] - \mathbb{E}[X \mid Y]\|_2 \le \frac{D(p,k)}{n^{k/2}},
\end{equation}
as desired.
\end{proof}

\section{Entropy from Moments} \label{jj}

In this section, we leverage the I-MMSE relation to express the differential entropy in terms of moments. We prove in Theorem~\ref{jn} that, for any continuous RV $Y$ that has a MGF, there is a sequence of rational functions $\rho_n:[0,\infty) \to \mathbb{R},$ $n\in \mathbb{N},$ whose coefficients are multivariate polynomials in the first $2n$ moments of $Y,$ such that
\begin{equation} \label{ie}
    h(Y) = \lim_{n\to \infty} \int_0^{\infty} \rho_n(t) \, dt.
\end{equation}
Further, the convergence in~\eqref{ie} is monotone from above. The starting point in deriving this formula is the I-MMSE relation, which we briefly review next.

\subsection{The I-MMSE Relation}

The I-MMSE relationship can be stated as follows. 

\begin{theorem}[I-MMSE relation, \cite{Guo2005}] \label{bt}
For any square-integrable RV $X,$ an independent $N\sim \calN(0,1),$ and $\gamma>0,$ we have that
\begin{equation}
I(X~;~\sqrt{\gamma}X+N) = \frac12 \int_0^\gamma \mathrm{mmse}(X,t) \, dt.
\end{equation}
\end{theorem}

The I-MMSE relation directly yields the following formula for differential entropy: for a square-integrable continuous RV $Y$ we have that~\cite{Guo2005}
\begin{equation} \label{hz}
    h(Y) = \frac12 \log\left( 2\pi e \sigma_Y^2 \right) - \frac12 \int_0^\infty \frac{\sigma_Y^2}{1+\sigma_Y^2 t} - \mathrm{mmse}(Y,t) \, dt.
\end{equation}
Noting that for $a,b>0$
\begin{equation} \label{hx}
    \int_0^{\infty} \frac{a}{1+at} - \frac{b}{1+bt} \, dt = \lim_{\gamma \to \infty} \log \frac{1+a\gamma}{1+b\gamma} = \log \frac{a}{b},
\end{equation}
we obtain
\begin{equation} \label{hy}
    h(Y) = \frac12 \int_0^\infty  \mm(Y,t) - \frac{1}{2\pi e + t} \, dt.
\end{equation}
The MMSE term in the above expression can be approximated by the PMMSE, resulting in an expression for differential entropy as a function of moments of $Y$. We explore this connection next.

\subsection{A New Formula for Differential Entropy}

From~\eqref{hz} and~\eqref{hy}, and since $\mm(Y,t)\le \LL(Y,t),$ replacing the MMSE with the LMMSE gives the upper bound on differential entropy $h(Y)$
\begin{equation}
    h(Y) \le h_1(Y) := \frac12 \int_0^\infty  \mathrm{lmmse}(Y ,t) - \frac{1}{2\pi e + t} \, dt  = \frac12 \log\left( 2\pi e \sigma_Y^2 \right) = h(\mathcal{N}(0,\sigma_Y^2)),
\end{equation}
which is the maximum possible differential entropy for a continuous RV with a prescribed variance of $\sigma_Y^2.$ We take this a step further and introduce for each $n\in \mathbb{N}$ (assuming $\mathbb{E}[Y^{2n}]<\infty$)
\begin{equation} \label{jp}
    h_n(Y) := \frac12 \int_0^\infty  \pp_n(Y,t) - \frac{1}{2\pi e + t} \, dt.
\end{equation}
From the definition of the PMMSE, we have the monotonicity 
\begin{equation}
    \pp_1(Y,t) \ge \pp_2(Y,t) \ge \cdots \ge \mm(Y,t).
\end{equation}
Hence, we obtain the monotonicity 
\begin{equation}
    h_1(Y) \ge h_2(Y) \ge \cdots \ge h(Y)
\end{equation}
for a RV $Y$ having moments of all orders. In fact, if $Y$ has a MGF, then $h_n(Y) \to h(Y)$ by the monotone convergence theorem in view of the limit $\pp_n(Y,t) \to \mm(Y,t)$ shown in Theorem~\ref{an}. Hence, we have the following formula for differential entropy in terms of the moments.

\begin{theorem} \label{jn}
Let $Y$ be a continuous RV whose MGF exists. Then, we have a decreasing sequence
\begin{equation}
    \frac{1}{2} \log \left( 2\pi e \sigma_Y^2 \right) = h_1(Y)\ge h_2(Y) \ge \cdots \ge h(Y)
\end{equation}
converging to the differential entropy
\begin{equation} \label{ia}
    \lim_{n\to \infty} h_n(Y) = h(Y).
\end{equation}
\end{theorem}
\begin{proof}
The functions $g_n(t):= \LL(Y,t)-\pp_n(Y,t)$ are nonnegative and nondecreasing. By Theorem~\ref{an}, $g_n \to g$ pointwise, where $g(t):= \LL(Y,t)-\mm(t).$ Therefore, by the monotone convergence theorem, $\int_{\BR} g_n(t) \, dt \to \int_{\BR} g(t) \, dt.$ Adding and subtracting $1/(2\pi e + t)$ to each integrand, and noting that $t \mapsto \LL(Y,t)-1/(2\pi e + t)$ is absolutely integrable, we conclude that $h_n(Y) \to h(Y).$
\end{proof}

\subsection{Behavior of $h_n$ Under Affine Transformations} 

Each approximation $h_n$ behaves under (injective) affine transformations exactly as differential entropy does. Specifically, by the behavior of the PMMSE under affine transformations shown in Proposition~\ref{hn}, we have that
\begin{equation} \label{li}
    \pp_n(\alpha Y+\beta,t) = \alpha^2 \pp_n(Y,\alpha^2 t).
\end{equation}
From \eqref{li}, and after a change of variables in 
\begin{equation}
    h_n(Y) := \frac12 \int_0^\infty  \pp_n(Y,t) - \frac{1}{2\pi e + t} \, dt,
\end{equation}
one can show that $h_n(\alpha Y+\beta)=h_n(Y) + \log|\alpha|.$

\begin{prop} \label{js}
For $Y \in L^{2n}(\calF)$ and $(\alpha,\beta)\in \BR^2$ with $\alpha \neq 0,$ we have that
\begin{equation}
    h_n(\alpha Y+\beta) = h_n(Y) + \log |\alpha|.
\end{equation}
\end{prop}
\begin{proof}
See Appendix \ref{lj}.
\end{proof}

\section{A New Formula for Mutual Information} \label{ba}

From the formula we developed for differential entropy in Proposition \ref{jn}, we immediately obtain a new formula for mutual information between a discrete RV $X$ and a continuous RV $Y.$ We use the shorthand $Y^{(x)}$ for the RV $Y$ conditioned on $\{ X=x\}$ (as opposed to the subscript notation used in Section~\ref{pc}), as we will later consider i.i.d. samples which we will indicate with subscripts. 

First, note that for each $x$ such that $P_X(x)>0,$ the RV $Y^{(x)}$ is continuous. Indeed, if $B\in \calB$ has Lebesgue measure zero, we have that
\begin{equation} \label{wm}
    0 = P_Y(B) = \sum_{a\in \supp(X)} P_{X,Y}(\{a\}\times B) = \sum_{a\in \supp(X)} P_{Y^{(a)}}(B) P_X(a) \ge  P_{Y^{(x)}}(B) P_X(x).
\end{equation}
Since $P_X(x)>0,$ we infer from~\eqref{wm} that $P_{Y^{(x)}}(B)=0.$ Further, the joint measure $P_{X,Y}$ is absolutely continuous with respect to the product measure $P_X\otimes P_Y.$ Indeed, if $B$ is a Borel subset of $\BR^2,$ and $B_x := \{y\in \BR ~ ; ~ (x,y)\in B\},$ then $(P_X\otimes P_Y)(B) = \sum_{x\in \supp(X)} P_Y(B_x)P_X(x).$ Suppose that $(P_X\otimes P_Y)(B)=0,$ so for each $x\in \supp(X),$ either $P_X(x)=0$ or $P_Y(B_x)=0.$ In addition, we have that $P_{X,Y}(B) = \sum_{x\in \supp(X)} P_{Y^{(x)}}(B_x)P_X(x).$ For every $x\in \supp(X)$ such that $P_X(x)\neq 0,$ we have that 
\begin{equation}
    0 = P_{Y}(B_x) = \sum_{a\in \supp(X)} P_{X,Y}(\{a\}\times B_x) = \sum_{a\in \supp(X)} P_{Y^{(x)}}(B_x) P_X(a) \ge P_{Y^{(x)}}(B_x) P_X(x),
\end{equation}
hence $P_{Y^{(x)}}(B_x)=0.$ Therefore, $P_{X,Y}(B)=0,$ and we conclude that $P_{X,Y}$ is absolutely continuous with respect to $P_{X}\otimes P_Y.$ We have the Radon-Nikodym derivative
\begin{equation}
    \frac{d P_{X,Y}}{d P_X\times P_Y} (x,y) = \frac{p_{Y^{(x)}}(y)}{p_Y(y)}.
\end{equation}
Now, by the disintegration theorem, we have that
\begin{equation} \label{wp}
    I(X;Y) = \sum_{x\in \supp(X)} P_X(x) \int_{\BR} p_{Y^{(x)}}(y) \log \frac{p_{Y^{(x)}}(y)}{p_Y(y)} ~ dy.
\end{equation}
Suppose that $\sigma_Y^2<\infty,$ so we also have that $\sigma_{Y^{(x)}}^2 < \infty$ for each $x\in \supp(X).$ If we also have that $h(Y)>-\infty,$ then we may split the integral in~\eqref{wp} to obtain
\begin{equation}
    I(X;Y) = -\sum_{x\in \supp(X)} P_X(x) \int_{\BR} p_{Y^{(x)}}(y) \log p_Y(y) ~ dy +  \sum_{x\in \supp(X)} P_X(x) \int_{\BR} p_{Y^{(x)}}(y) \log p_{Y^{(x)}}(y) ~ dy.
\end{equation}
In other words,
\begin{equation} \label{jm}
    I(X;Y) = h(Y) - \BE_X\left[ h\left(Y^{(X)}\right) \right],
\end{equation}
where the subscript in $\BE_X$ is to emphasize that the expectation is taken with respect to $X$ only.

Next, we discuss implications of equation \eqref{jm} in view of the formula for differential entropy proved in Proposition \ref{jn} and the behavior of the PMMSE under affine transformations proved in Proposition \ref{hn}. In particular, we show that the ensuing approximants of mutual information are affine-transformation invariant, which is in agreement with how mutual information behaves. In Section~\ref{jc}, we obtain another formula for the mutual information $I(X;Y)$ when both RVs $X$ and $Y$ are continuous.

\subsection{Mutual Information in Terms of Moments}

Under the assumptions used in Theorem~\ref{jn}, and utilizing equation \eqref{jm}, we obtain a formula for mutual information primarily in terms of moments.
\begin{theorem} \label{jo}
For a discrete RV $X$ with finite support and a continuous RV $Y$ with a MGF such that $h(Y)>-\infty,$ the mutual information is given by
\begin{equation}
I(X;Y)  = \frac12  \lim_{n\to \infty}  \int_0^\infty  \pp_n(Y,t) - \BE_X \left[\pp_n(Y^{(X)},t)\right] \, dt.
\end{equation}
\end{theorem}
\begin{proof}
By assumption on the MGF of $Y,$ we have that $\BE[e^{Yt}]<\infty$ whenever $|t|<\delta$ for some fixed $\delta>0.$ Then, whenever $|t|<\delta,$ we also have that for each $x\in \supp(X)$
\begin{equation}
    P_X(x)\BE[e^{Y^{(x)}t}] \le \sum_{a\in \supp(X)} P_X(a) \BE[e^{Y^{(a)} t}] = \BE[e^{Yt}] < \infty.
\end{equation}
Therefore, the MGF of $Y^{(x)}$ for each $x\in \supp(X)$ is finite over the nonempty neighborhood $(-\delta,\delta)$ of $0.$ Therefore, Theorem~\ref{jn} and equation \eqref{jm} yield the desired equation for $I(X;Y).$
\end{proof}

Equipped with the relationship between the moments and $I(X;Y)$ given in Theorem \ref{jo}, we will introduce a moments-based estimator of mutual information in Section~\ref{bb}. Specifically, we approximate the mutual information by fixing $n,$ then further approximate the ensuing expression by replacing moments with sample moments. Therefore, the estimator makes use of the approximants
\begin{equation} \label{xl}
    I_n(X;Y) := \frac12   \int_0^\infty  \pp_n(Y,t) - \BE_X \left[ \pp_n(Y^{(X)},t)  \right]\, dt
\end{equation}
for every $n \in \BN,$ where $X$ is a discrete RV with finite support and $Y$ is a continuous RV that is $2n$-times integrable. Under the premises of Theorem~\ref{jo}, we have the limit
\begin{equation}
    I(X;Y) = \lim_{n\to \infty} I_n(X;Y).
\end{equation}
Also, in view of the definition of $h_n$ in equation \eqref{jp},
\begin{equation} \label{jt}
    I_n(X;Y) = h_n(Y) - \BE_X \left[ h_n(Y^{(X)}) \right].
\end{equation}

\subsection{Properties of $I_n$: Affine Transformations and Independence}

We prove properties of the approximants $I_n$ that resemble those for the mutual information. First, the behavior of the PMMSE under affine transformations exhibited in Proposition \ref{hn} implies that $I_n(X;Y)$ is invariant under injective affine transformations of $Y.$ Indeed, this can be seen immediately from the behavior of $h_n$ in Proposition \ref{js} in view of equation \eqref{jt}.

\begin{corollary}
Suppose $X$ and $Y$ are RVs satisfying the premises of Theorem \ref{jo}. For any constants $(\alpha,\beta) \in \mathbb{R}^2$ with $\alpha \neq 0,$ and for any $n\in \BN$
\begin{equation}
I_n(X; \alpha Y+\beta)=I_n(X;Y).
\end{equation}
\end{corollary}

Also, the approximations $I_n(X;Y)$ detect independence exactly. 

\begin{prop}
Suppose $X$ and $Y$ are RVs satisfying the premises of Theorem \ref{jo}. If $X$ and $Y$ are independent, then for any $n\in \BN$
\begin{equation}
    I_n(X;Y)=0.
\end{equation}
\end{prop}
\begin{proof}
By independence, each $Y^{(x)}$ is identically distributed to $Y.$ Therefore, $h_n(Y^{(x)})=h_n(Y)$ for each $x\in \supp(X).$ The desired result then follows from equation \eqref{jt}.
\end{proof}

We give full expressions for the first two approximants of mutual information that are generated by the LMMSE and quadratic MMSE. 

\begin{example}
When $n=1,$ we obtain
\begin{equation}
I_1(X;Y) = \log \sigma_Y - \mathbb{E}_X \left[ \log \sigma_{Y^{(X)}} \right],
\end{equation}
which is the exact formula for $I(X;Y)$ when both $Y$ is Gaussian and each $Y^{(x)}$ (for $x\in \supp(X)$) is Gaussian; indeed, in such a case, the MMSE is just the LMMSE. \hfill $\qed$
\end{example}

\begin{example}
For $n=2,$ we obtain the formula
\begin{align*}
I_2(X;Y) = \frac{1}{6} \log \frac{b_Y^{2,3}}{\prod_{x\in \supp(X)}\left( b_{Y^{(x)}}^{2,3} \right) ^{P_X(x)}} +\frac12 \int_{0}^\infty & \frac{a_Y^{2,1}t}{2+b_Y^{2,1}t+b_Y^{2,2}t^2+b_Y^{2,3}t^3} \nonumber \\
&\quad - \mathbb{E}_X ~ \frac{a_{Y^{(X)}}^{2,1}t}{2+b_{Y^{(X)}}^{2,1}t+b_{Y^{(X)}}^{2,2}t^2+b_{Y^{(X)}}^{2,3}t^3} \, dt
\end{align*}
where we may compute for any $R\in L^4(\calF)$
\begin{equation}
b_R^{2,3} := \left| \begin{array}{ccc}
1 & \mathbb{E}R & \mathbb{E}R^2 \\
\mathbb{E}R & \mathbb{E}R^2 & \mathbb{E}R^3 \\
\mathbb{E}R^2 & \mathbb{E}R^3 & \mathbb{E}R^4 \end{array} \right| = \sigma_R^2\mathbb{E}R^4 + 2(\mathbb{E}R)(\mathbb{E}R^2)\mathbb{E}R^3 - (\mathbb{E}R^2)^3  - (\mathbb{E}R^3)^2,
\end{equation}
which is strictly positive when $|\mathrm{supp}(R)|>2,$ and
\begin{align}
b_R^{2,2}&=-4(\mathbb{E}R)\mathbb{E}R^3 + 3(\mathbb{E}R^2)^2 + \mathbb{E}R^4 \\
b_R^{2,1}&= 6\sigma_R^2 \\
a_R^{2,1}&=4(\mathbb{E}R)^4 - 8(\mathbb{E}R)^2\mathbb{E}R^2+ \frac83(\mathbb{E}R)\mathbb{E}R^3 + 2(\mathbb{E}R^2)^2  - \frac23 \mathbb{E}R^4.
\end{align}
\hfill $\qed$
\end{example}

\section{Generalizations to Arbitrary Bases and Multiple Dimensions} \label{jc}

We extend our approximation results for the conditional expectation from the polynomial-basis setting to arbitrary bases, and from conditioning on random variables to conditioning on arbitrary $\sigma$-algebras. An extension to the multidimensional case is also presented, which straightforwardly yields an approximation theorem for differential entropy of random vectors. Another byproducts of the multidimensional generalization is an expression for mutual information between two continuous random variables completely in terms of moments, which is presented at the end of this section. 

\subsection{Arbitrary Bases and $\sigma$-Algebras}

Up to here, our exposition dealt with the polynomial basis of $L^2(\sigma(Y)).$ However, our  results can be extended to a more general setup. Recall that we have defined
\begin{equation}
    \bM_{Y,n} = \BE\left[ \bY^{(n)} \left( \bY^{(n)} \right)^T \right],
\end{equation}
and derived
\begin{equation} \label{ph}
    \BE[X \mid Y] = \lim_{n\to \infty} \BE\left[ X \bY^{(n)} \right] \bM_{Y,n}^{-1} \bY^{(n)}
\end{equation}
in Theorem~\ref{ar} under two requirements: $Y$ satisfies Carleman's condition, and $|\supp(Y)|=\infty.$ Along similar lines, we derive a generalization where the set of polynomials of $Y$ is replaced with any linearly-independent subset of $L^2(\Sigma)$ having a dense span, where $\Sigma\subset \calF$ is any $\sigma$-algebra. Denseness replaces Carleman's condition, while linear independence replaces the infinite-support requirement.

\begin{theorem} \label{pi}
Fix a $\sigma$-algebra $\Sigma \subset \calF$ and a set $\{\psi_k \}_{k \in \BN} = \calS \subset L^2(\Sigma).$ For each $n\in \BN,$ define the random vector $\bpsi^{(n)} = (\psi_0,\cdots,\psi_n)^T$ and the matrix of inner products
\begin{equation}
    \bM_{\calS,n} := \BE\left[ \bpsi^{(n)} \left( \bpsi^{(n)} \right)^T \right].
\end{equation}
If $\calS$ is linearly independent and $\Span(\calS)$ is dense in $L^2(\Sigma),$ then
\begin{equation}
    \BE[X \mid \Sigma] = \lim_{n\to \infty} \BE\left[ X \bpsi^{(n)} \right]^T \bM_{\calS,n}^{-1} ~ \bpsi^{(n)}
\end{equation}
in $L^2(\Sigma)$ for any RV $X\in L^2(\calF).$
\end{theorem}

For the proof of Theorem \ref{pi}, we will need the following formula for the closest element in a finite-dimensional subspace of $L^2(\calF)$ to a RV $X\in L^2(\calF),$ which will also be used for the extension of our results to random vectors later in this section.

\begin{lemma} \label{pk}
For any fixed finite-dimensional subspace $\calV\subset L^2(\calF)$ having a basis $\{V_0,V_1,\cdots,V_n\},$ we have that for every $X\in L^2(\calF)$ 
\begin{equation} \label{pj}
    \BE\left[ X \bV \right]^T \BE\left[ \bV \bV^T \right]^{-1} \bV = \underset{V \in \calV}{\mathrm{argmin}} \left\| X - V \right\|_2,
\end{equation}
where $\bV = (V_0,V_1,\cdots,V_n)^T.$
\end{lemma}
\begin{proof}
See Appendix~\ref{wx}.
\end{proof}

In view of Lemma \ref{pk}, we introduce the following notation.
\begin{definition}
Fix a RV $X\in L^2(\calF),$ a $\sigma$-algebra $\Sigma \subset \calF,$ and a linearly-independent set $\{\theta_j\}_{j\in \BN} = \Theta \subset L^2(\Sigma).$ Write $\btheta^{(n)}= (\theta_0,\cdots,\theta_n)^T$ for each $n\in \BN.$ We define the $n$-th approximation of $\BE[X \mid \Sigma]$ with respect to $\Theta$ by
\begin{equation}
    E_{n,\Theta}\left[ X \mid \Sigma \right] := \BE\left[ X \btheta^{(n)} \right] \BE\left[ \btheta^{(n)} \left( \btheta^{(n)} \right)^T \right]^{-1} \btheta^{(n)}.
\end{equation}
\end{definition}

Note that $E_{n,\Theta}[X \mid \Sigma]$ belongs to $\Span(\{\theta_j\}_{j\in [n]}).$ Further, according to Lemma \ref{pk}, $E_{n,\Theta}[X \mid \Sigma]$ is the unique closest element in $\Span(\{\theta_j\}_{j\in [n]})$ to $X$
\begin{equation}
    E_{n,\Theta}[X \mid \Sigma] = \underset{V \, \in ~ \Span(\{\theta_j\}_{j\in[n]})}{\mathrm{argmin}} \left\| X - V \right\|_2.
\end{equation}
If $Y\in L^{2n}(\calF),$  $\Theta=\{Y^j \}_{j\in \BN},$ and $\Sigma = \sigma(Y),$ then
\begin{equation}
    E_{n,\Theta}[X \mid \Sigma] = E_n[X \mid Y].
\end{equation}

The central claim in Theorem \ref{pi} is that if $\Span(\Theta)$ is dense in $L^2(\Sigma)$ then 
\begin{equation}
    \BE[X \mid \Sigma] = \lim_{n\to \infty} E_{n,\Theta}[X \mid \Sigma].
\end{equation}

\begin{proof}[Proof of Theorem \ref{pi}]
For each $n\in \BN,$ define the subspace
\begin{equation}
    \calV^{(n)} = \Span(\{\psi_j\}_{j\in [n]}).
\end{equation}
As in the proof of Lemma \ref{pk}, the entries of the vector
\begin{equation}
    \bxi^{(n)} = \bM_{\calS,n}^{-1/2} \bpsi^{(n)}
\end{equation}
form an orthonormal basis for $\calV^{(n)}.$ Note that $\bxi^{(n)}$ is the Gram-Schmidt orthonormalization of $\bpsi^{(n)}.$ Hence, $\bxi^{(n)}$ is a prefix of $\bxi^{(n+1)}.$ Let $\Theta=\{\xi_j \}_{j\in \BN}$ be the ensuing orthonormal basis for $L^2(\Sigma),$ i.e., $\bxi^{(n)}=(\xi_0,\cdots,\xi_n)^T.$ Then, $\Span(\Theta)$ is dense in $L^2(\Sigma).$ Therefore,
\begin{equation}
    \BE[X \mid \Sigma] = \sum_{j\in \BN} \BE[X \xi_j] \xi_j = \lim_{n\to \infty} \BE[ X \bxi^{(n)}]^T \bxi^{(n)} = \lim_{n\to \infty} E_{n,\Theta}[X \mid \Sigma],
\end{equation}
as desired.
\end{proof}

\subsection{The Multidimensional Case}

We extend our results on the PMMSE of random variables to random vectors. We will modify our notation for this subsection only.

Denote the Borel $\sigma$-algebra of $\BR^m$ by $\calB_m.$
We are still working with the probability space $(\Omega,\calF,P).$ By an $m$-dimensional random vector ($m$-RV) we mean a measurable function from $\left( \Omega,\calF \right)$ to $(\BR^m,\calB_m).$ For a sub-$\sigma$-algebra $\Sigma\subset \calF,$ we denote the set of $(\Sigma,\calB_m)$-measurable functions by $\calM_m(\Sigma)$; so, the set of $m$-RVs is denoted by $\calM_m(\calF).$ Additionally, for each $q\ge 1,$ we define
\begin{equation} \label{tb}
    L^q(\BR^m,\Sigma) := \left\{ \bff \in \calM_m\left( \Sigma \right) ~ ; ~ \int_{\Omega} \|\bff(\omega)\|_q^q \, dP(\omega) < \infty \right\}.
\end{equation}
In~\eqref{tb}, $\|\bff(\omega)\|_q$ refers to the $q$-norm of the vector $\bff(\omega)\in \BR^m.$ In other words, if $\bff = (f_1,\cdots,f_m)^T,$ then $\|\bff(\omega)\|_q^q = \sum_{i=1}^m |f_i(\omega)|^q.$ We will also use $\|\wc\|_q$ for the norm of $L^q(\BR^m,\Sigma),$ i.e., for $\bff\in L^q(\BR^m,\Sigma)$ we have 
\begin{equation}
    \|\bff\|_q^q = \int_{\Omega} \|\bff(\omega)\|_q^q \, dP(\omega).
\end{equation}
The distinction is that we consider the $q$-norm of $\BR^m$ when we specify the input of $\bff$ (i.e., $\bff(\omega)$), and when no input is given to $\bff$ then $\|\bff\|_q$ refers to the norm of $L^q(\BR^m,\Sigma).$ We keep the notation $L^q(\BR,\Sigma)=L^q(\Sigma).$

A function $\bff:\BR^m \to \BR^k$ is called Borel if it is $(\calB_m,\calB_k)$-measurable. For any $\bY \in \calM_m(\calF)$ and integer $k\ge 1$
\begin{equation}
    \calM_k(\sigma(\bY)) = \{ \bff(\bY) ~ ; ~ \bff:\BR^m \to \BR^k \text{ Borel } \}.
\end{equation}
By a generalization of H\"{o}lder's inequality, for any $m$-RV $\bY=(Y_1,\cdots,Y_m)^T \in L^{\beta}(\BR^m,\calF),$ we also have that $Y_1^{\alpha_1}\cdots Y_m^{\alpha_m}\in L^1(\calF)$ for any constants $\alpha_1,\cdots,\alpha_m\ge 0$ such that $\alpha_1+\cdots+\alpha_m\le \beta.$

We extend the notation $\bY^{(n)}$ in \eqref{jg} to random vectors as follows. For an $m$-RV $\bY=(Y_1,\cdots,Y_m)^T,$ we let $\bY^{(n,m)}$ denote the random vector whose entries are monomials in the $Y_j$ of total degree at most $n,$ ordered first by total degree then reverse-lexicographically in the exponents. For example, if $m=3$ so $\bY=(Y_1,Y_2,Y_3)^T,$ then for $n=2$
\begin{equation}
    \bY^{(2,3)} = (1,Y_1,Y_2,Y_3,Y_1^2,Y_1Y_2,Y_1Y_3,Y_2^2,Y_2Y_3,Y_3^2)^T
\end{equation}
because the totally ordered set of exponents $(~\{ \bv \in \BN^3 \mid \mathbf{1}^T\bv \le 2 \}~,~<~)$ 
has the order\footnote{Note that this ordering is not the same as the degree reverse lexicographical order nor its reverse.}
\begin{align*}
    (0,0,0)<(1,0,0)<(0,1,0)<(0,0,1)<(2,0,0)<(1,1,0)<(1,0,1)<(0,2,0)<(0,1,1)<(0,0,2).
\end{align*}
A straightforward stars-and-bars counting argument reveals that the length of $\bY^{(n,m)}$ is $\binom{n+m}{m}.$

Let $\SP_{n,m}$ denote the set of polynomials in $m$ variables with real coefficients of total degree at most $n.$ For an $m$-RV $\bY,$ denote
\begin{equation}
    \SP_{n,m}(\bY) := \left\{ p(\bY) ~ ; ~ p \in \SP_{n,m} \right\}.
\end{equation}
Note that $\SP_{n,1}=\SP_{n}.$ Also, the notation $\bY^{(n,1)},$ while avoided, is disambiguated by interpreting it as $\bY^{(n)},$ i.e., $\bY^{(n,1)}=(1,Y,\cdots,Y^n)^T$ where the subscript on $Y_1$ is dropped. We denote the product sets of $\SP_{n,m}(\bY)$ by $\SP_{n,m}^{\ell}(\bY),$ and consider their elements as vectors rather than tuples. In other words, we denote the set of length-$\ell$ vectors whose coordinates are multivariate polynomial expressions of an $m$-RV $\bY$ with total degree at most $n$ by
\begin{equation}
    \SP_{n,m}^{\ell}(\bY) = \left\{ (p_1(\bY),\cdots,p_\ell(\bY))^T ~ ; ~ p_1,\cdots,p_\ell \in \SP_{n,m} \right\}.
\end{equation}

The multivariate generalization of the PMMSE is defined as follows.

\begin{definition}[Multivariate Polynomial MMSE]
Fix positive integer $\ell,m,$ and $n$. Fix an $\ell$-RV $\bX\in L^2(\BR^{\ell},\calF)$ and an $m$-RV $\bY\in L^{2n}(\BR^{m},\calF),$ and set $k = \binom{n+m}{m}.$ We define the $n$-th order PMMSE for estimating $\bX$ given $\bY$ by
\begin{equation} \label{pg}
\text{pmmse}_n(\bX \mid \bY) := \underset{\bC \in \BR^{\ell \times k}}{\min} \; \left\| \bX- \bC \bY^{(n,m)} \right\|_2,
\end{equation}
and the $n$-th order PMMSE estimate of $\bX$ given $\bY$ by 
\begin{equation} \label{xd}
    E_n[\bX \mid \bY] := \bC \bY^{(n,m)} \in \SP_{n,m}^{\ell}(\bY)
\end{equation}
for any minimizing matrix $\bC \in \BR^{\ell \times k}$ in \eqref{pg}.
\end{definition}
\begin{remark}
For any minimizer $\bC$ in~\eqref{pg}, the $\ell$-RV $\bC \bY^{(n,m)}$ is the unique orthogonal projection of $\bX$ onto $\SP_{n,m}^{\ell}(\bY)$; in particular, $E_n[\bX \mid \bY]$ is well-defined by~\eqref{xd}.
\end{remark}

Denote, for $\bY \in L^{2n}(\BR^m,\calF),$
\begin{equation}
    \bM_{\bY,n} := \BE\left[  \bY^{(n,m)} \left( \bY^{(n,m)} \right)^T  \right].
\end{equation}
For $n\in \BN$ and an $\ell$-RV $(X_1,\cdots,X_\ell)^T=\bX \in L^2(\BR^\ell,\calF),$ if $\bM_{\bY,n}$ is invertible, Lemma~\ref{pk} yields that
\begin{equation} \label{wz}
    E_n[\bX \mid \bY] = \left( \begin{array}{c} 
    E_n[X_1 \mid \bY] \\
    \vdots \\
    E_n[X_\ell \mid \bY] 
    \end{array} \right) = 
    \left( \begin{array}{c}
    \vspace{2mm}
    \BE\left[X_1 \bY^{(n,m)}\right]^T \bM_{\bY,n}^{-1} \bY^{(n,m)} \\
    \vspace{2mm}
    \vdots \\
    \BE\left[X_\ell \bY^{(n,m)}\right]^T \bM_{\bY,n}^{-1} \bY^{(n,m)}
    \end{array} \right).
\end{equation}
We say that the $Y_j$ do not satisfy a polynomial relation if the monomials $\prod_{j=1}^m Y_j^{\alpha_j},$ for $\alpha_1,\cdots,\alpha_m \in \BN,$ are linearly independent, i.e., if the mapping
\begin{equation}
    \varphi : \bigcup_{n\in \BN} \SP_{n,m} \to \bigcup_{n\in \BN} \SP_{n,m}(\bY), \quad\quad \varphi(p) = p(\bY)
\end{equation}
is an isomorphism of vector spaces.

Generalizing our results on RVs to $m$-RVs can be done in view of the following polynomial denseness result.

\begin{theorem}[\cite{Petersen82}] \label{pe}
For any $m$-RV $\bY=(Y_1,\cdots,Y_m)^T$ and $q > 1,$ if 
\begin{equation}
    \overline{\bigcup_{n\in \BN} \SP_n(Y_j)} = L^q(\sigma(Y_j))
\end{equation}
for each $j\in \{1,\cdots,m\},$ then
\begin{equation}
    \overline{\bigcup_{n\in \BN} \SP_{n,m}(\bY)} = L^r(\sigma(\bY))
\end{equation}
for every $r\in [1,q).$
\end{theorem}

An immediate corollary that we use in this section is as follows.

\begin{corollary} \label{pd}
Fix an integer $m\ge 1$ and an $m$-RV $\bY=(Y_1,\cdots,Y_m)^T.$ If each of the RVs $Y_1,\cdots,Y_m$ satisfies Carleman's condition, then the set of vectors of polynomials $\bigcup_{n\in \BN} \SP_{n,m}^m(\bY)$ is dense in $L^q(\BR^m,\sigma(\bY))$ for any $q\ge 1.$
\end{corollary}
\begin{proof}
See Appendix~\ref{wy}.
\end{proof}

We deduce the following result on the convergence of the multivariate PMMSE to the MMSE.

\begin{theorem} \label{wu}
Fix an $m$-RV $\bY=(Y_1,\cdots,Y_m)^T$ and an $\ell$-RV $\bX\in L^2(\BR^\ell,\calF).$ If each $Y_j$ satisfies Carleman's condition, and if the $Y_j$ do not satisfy a polynomial relation, then we have the $L^2(\BR^\ell,\sigma(\bY))$-limit
\begin{equation}
    \BE[ \bX \mid \bY] = \lim_{n\to \infty} E_n[\bX \mid \bY].
\end{equation}
\end{theorem}
\begin{proof}
See Appendix~\ref{xa}.
\end{proof}

With the definition of the multivariate PMMSE at hand, we show that the PMMSE estimate satisfies a tower property similar to the conditional expectation.

\begin{prop}[Tower Property] \label{xb}
Fix $n\in \BN$ and three RVs $X\in L^2(\calF)$ and $Y_1,Y_2 \in L^{2n}(\calF).$ Suppose that $|\supp(Y_1)|,|\supp(Y_2)|>n.$ Then
\begin{equation} \label{pm}
    E_n\left[ E_n[X \mid Y_1] \mid Y_1,Y_2 \right] = E_n[X \mid Y_1],
\end{equation}
and
\begin{equation} \label{pn}
    E_n\left[ E_n[X \mid Y_1,Y_2] \mid Y_2 \right] = E_n[X \mid Y_2].
\end{equation}
\end{prop}
\begin{proof}
See Appendix~\ref{xc}.
\end{proof}

Now, to generalize our results on estimation in Gaussian channels, we first note a straightforward generalization of the MMSE dimension to the multidimensional case.

\begin{theorem} \label{qm}
Fix two square-integrable continuous $m$-RVs $\bX$ and $\bN$ that are independent. Suppose that $p_{\bN}$ is bounded and that\footnote{The exponent $m+2$ in~\eqref{tj} may be replaced with $m+1+\varepsilon$ for any $\varepsilon>0,$ see~\cite[Section 3.2]{Stein2019}} 
\begin{equation} \label{tj}
    p_{\bN}(\bz)=O\left( \|\bz\|^{-(m+2)} \right)
\end{equation}
as $\|\bz\|\to \infty.$ Then, we have that
\begin{equation} \label{ps}
    \lim_{t\to \infty} t \cdot \mm\left(\bX \mid \sqrt{t} \bX+\bN\right) = \mathrm{tr}~\Sigma_{\bN}.
\end{equation}
\end{theorem}
\begin{proof}
See Appendix~\ref{qo}.
\end{proof}

The approach for showing the rationality of $t\mapsto \pp_n(X,t)$ for a RV $X\in L^{2n}(\calF)$ in Section~\ref{jh} may be generalized to deduce rationality of $t \mapsto \pp_n(\bX,t)$ for an $m$-RV $\bX\in L^{2n}(\BR^m,\calF).$ Here, we are denoting $\pp_n(\bX,t):= \pp_n(\bX \mid \sqrt{t}\bX+\bN),$ where $\bN\sim \calN(\bm{0},\bI_m)$ is independent of $\bX.$ For brevity, we give a blueprint of how this generalization of rationality can be obtained.

First, Lemma~\ref{kk} may be generalized to yield that $\det \bM_{\sqrt{t}\bX+\bN}$ is a polynomial in $t$ of degree at most $d_{n,m}$ which is given by
\begin{align}
    d_{n,m} &:= \sum_{k\in [n]} k \cdot |\{ (\lambda_1,\cdots,\lambda_m)\in \BN^m ~ ; ~ \lambda_1+\cdots+\lambda_m = k \}| \\
    &= \sum_{k\in [n]} k \binom{k+m-1}{m-1} = \sum_{k\in [n]} m \binom{k+m-1}{m} = m \binom{n+m}{m+1}.
\end{align}
Further, the coefficient of $t^{d_{n,m}}$ in $\det \bM_{\sqrt{t}\bX+\bN}$ is $\det \bM_{\bX}.$ Note that $d_{n,1}=d_n.$ Then, generalizing Lemma~\ref{kl} we obtain an analogous expression to the scalar case given in Theorem~\ref{ir}, namely,
\begin{equation} \label{xe}
    \pp_n(\bX,t) = \frac{\left( \mathrm{tr}~\Sigma_{\bX} \right) \det \bM_{\bN,n}+\cdots+\left( \mathrm{tr}~\Sigma_{\bN} \right) \left(\det \bM_{\bX,n}\right) ~ t^{d_{n,m}-1}}{ \det \bM_{\bN,n}+\cdots+ \left( \det \bM_{\bX,n} \right) ~ t^{d_{n,m}}}.
\end{equation}
To deduce~\eqref{xe}, the multidimensional MMSE dimension result in Theorem~\ref{qm} is used, as follows. Note that $\mathrm{tr}~\Sigma_{\bN} = m$ for $\bN\sim \calN(\mathbf{0},\bI_{m}).$ By Theorem~\ref{qm}, we have that $\mm(\bX,t) \sim m/t.$ It is also true that $\LL(\bX,t) \sim m/t.$ Therefore, $\pp_n(\bX,t)\sim m/t$ for every integer $n\ge 1.$ Note that $\pp_n(\bX,0) = \mathrm{tr}~\Sigma_{\bX}.$ Expression~\eqref{xe} follows.

\subsection{Mutual Information for Continuous Random Variables}

We prove an analogous approximation result to that in Theorem~\ref{jo} for the mutual information between continuous RVs.

\begin{theorem} \label{xr}
For two continuous RVs $X$ and $Y$ whose MGFs exist, if $\min(h(X),h(Y))>-\infty,$ then the mutual information is given by
\begin{equation}
I(X;Y)  = \frac12  \lim_{n\to \infty} \int_0^\infty \pp_n(X,t) + \pp_n(Y,t) - \pp_n(\bW,t) ~ dt,
\end{equation}
where $\bW := (X,Y)^T.$
\end{theorem}
\begin{remark}
This formula expresses the mutual information between two continuous RVs entirely in terms of moments, because $\pp_n(\wc,t)$ is determined completely by the moments. This is in contrast to the formula in Theorem~\ref{jo}, which expresses the mutual information between a discrete RV and a continuous RV in terms of moments along with the expectation operator of the discrete RV.
\end{remark}

To prove Theorem~\ref{xr}, we use the following generalization of Theorem~\ref{jn} to higher dimensions. For an $m$-RV $\bV\in L^{2n}(\BR^m,\calF),$ we define
\begin{equation}
    h_n(\bV) := \frac12 \int_0^\infty \pp_n(\bV,t) - \frac{m}{2\pi e + t} \, dt.
\end{equation}

\begin{theorem} \label{xu}
Let $\bV$ be a continuous $m$-RV whose MGF exists. Then, we have a decreasing sequence
\begin{equation}
    \frac{1}{2} \log \left( (2\pi e)^m \det \bSigma_{\bV} \right) = h_1(\bV)\ge h_2(\bV) \ge \cdots \ge h(\bV)
\end{equation}
converging to the differential entropy
\begin{equation}
    \lim_{n\to \infty} h_n(\bV) = h(\bV).
\end{equation}
\end{theorem}
\begin{proof}
In view of monotonicity of $\pp_n(\bV,t)$ in $n,$ and since $h_1(\bV)$ is finite, it suffices by the monotone convergence theorem and the equation
\begin{equation}
    h(\bV) = \frac12 \int_0^\infty \mm(\bV,t) - \frac{m}{2\pi e + t} \, dt
\end{equation}
to show that $\pp_n(\bV,t) \to \mm(\bV,t)$ as $n\to \infty.$ Let $\bN\sim \calN(0,\bI_m)$ be independent of $\bV.$ A simple application of the triangle inequality yields that it suffices to prove the convergence
\begin{equation} \label{xs}
    E_n\left[ \bV \mid \sqrt{t}\bV + \bN \right] \to \BE\left[ \bV \mid \sqrt{t}\bV + \bN \right].
\end{equation}
We deduce~\eqref{xs} from Theorem~\ref{wu}, as follows.

Denote $\bZ^{(t)} := \sqrt{t}\bV+\bN,$ and let $Z_j^{(t)}$ be the $j$-th entry of $\bZ^{(t)}.$ Fix $t\ge 0.$ To apply Theorem~\ref{wu}, we only need to show that the $Z_j^{(t)}$ do not satisfy a nontrivial polynomial relation. We show this by induction on $m.$ The case $m=1$ follows since $Z_1^{(t)}$ is continuous. Assume that we have shown that $Z_1^{(t)},\cdots,Z_{m-1}^{(t)}$ do not satisfy a nontrivial polynomial relation, and that $m\ge 2.$ Suppose, for the sake of contradiction, that $q$ is a polynomial in $m$ variables such that $q(\bZ^{(t)})=0.$ Write $q(u_1,\cdots,u_m) = \sum_{k\in [d]} q_k(u_1,\cdots,u_{m-1})u_m^k$ for some polynomials $q_k$ in $m-1$ variables such that $q_d\neq 0.$ Squaring $q(\bZ^{(t)})=0$ and taking the conditional expectation with respect to $N_m$ we obtain
\begin{equation} \label{xt}
    0 = \BE\left[ q\left( \bZ^{(t)} \right) \right] = \sum_{k\in [2d]} \beta_k N_m^{k}
\end{equation}
for some constants $\beta_k \in \BR$ with $\beta_{2d}:= \| q_d(Z_1^{(t)},\cdots,Z_{m-1}^{(t)}) \|_2^2.$ Since $N_m$ is continuous, equation~\eqref{xt} cannot be a nontrivial polynomial relation for $N_m.$ Thus, we must have $\beta_{2d}=0,$ i.e., $q_d(Z_1^{(t)},\cdots,Z_{m-1}^{(t)})=0.$ By the induction hypothesis, $q_d=0$ identically, a contradiction. Therefore, no nontrivial polynomial relation $q(\bZ^{(t)})=0$ can hold, and the inductive proof is complete. Finally, applying Theorem~\ref{wu}, we deduce the limit in~\eqref{xs}, thereby completing the proof of the theorem.
\end{proof}

Now, we deduce Theorem~\ref{xr} from Theorem~\ref{xu}.

\begin{proof}[Proof of Theorem~\ref{xr}.]
Since $X$ and $Y$ have finite variance, the mutual information $I(X;Y)$ is shown in~\cite{Verdu2010} to satisfy
\begin{equation} \label{ws}
I(X;Y) = \frac12 \hspace{-1pt} \int_0^\infty \hspace{-7pt} \mathbb{E}\left[ \left( \mathbb{E}\left[ X \mid Z_1^{(t)} \right]-\mathbb{E}\left[X \mid Z_1^{(t)},Z_2^{(t)} \right] \right)^2 \right] dt  + \frac12 \hspace{-1.5pt} \int_0^\infty \hspace{-7pt} \mathbb{E}\left[ \left( \mathbb{E}\left[ Y \mid Z_2^{(t)} \right]-\mathbb{E} \left[ Y \mid Z_1^{(t)},Z_2^{(t)} \right] \right)^2 \right]  dt
\end{equation}
where $Z_1^{(t)} := \sqrt{t}X+N_1$ and $Z_2^{(t)}:=\sqrt{t}Y+N_2.$ We may rewrite~\eqref{ws} as
\begin{equation} \label{xo}
    I(X;Y) = \frac12 \int_0^\infty \mm(X,t) + \mm(Y,t) - \mm(\bW,t) ~ dt,
\end{equation}
where $\bW:=(X,Y)^T.$ Adding and subtracting $2/(2\pi e + t)$ to the integrand in~\eqref{xo}, and noting that the assumption $\min(h(X),h(Y))>-\infty$ allows us to split the integral, we obtain
\begin{align}
    I(X;Y) &= \frac12 \int_0^\infty \hspace{-4pt} \mm(X,t) - \frac{1}{2\pi e+t} \, dt + \frac12 \int_0^\infty \hspace{-4pt} \mm(Y,t) - \frac{1}{2\pi e+t} \, dt  - \frac12 \int_0^\infty \hspace{-4pt} \mm(\bW,t) - \frac{2}{2\pi e+t} \, dt \label{xp} \\
    & = h(X)+h(Y)-h(\bW).
\end{align}
Finally, note that the MGF of $\bW$ exists by the assumption that the MGFs of $X$ and $Y$ exist. Thus, by Theorem~\ref{xu}, we have that $h_n(A) \to h(A)$ for $A\in \{X,Y,\bW\}.$ Hence, we obtain the desired formula
\begin{equation}
    I(X;Y) = \lim_{n\to \infty} I_n(X;Y)
\end{equation}
where we define
\begin{align}
    I_n(X;Y) &:= h_n(X) + h_n(Y) - h_n(\bW) \\
    &= \frac12 \int_0^\infty \pp_n(X,t) - \frac{1}{2\pi e+t} \, dt + \frac12 \int_0^\infty \pp_n(Y,t) - \frac{1}{2\pi e+t} \, dt  \nonumber \\
    & \hspace{5.6cm} - \frac12 \int_0^\infty \pp_n(\bW,t) - \frac{2}{2\pi e+t} \, dt \\
    &= \frac12 \int_0^\infty \pp_n(X,t) + \pp_n(Y,t) - \pp_n(\bW,t) ~ dt,
\end{align}
and the proof is complete.
\end{proof}

\section{Application: Estimation of Information Measures from Data} \label{bb}

The approximations introduced in the previous sections naturally motivate estimators for information measures. These estimators are based on (i) approximating moments with sample moments, then (ii) plugging the sample moments into the formulas we have developed for information measures.  Since the formulas for information measures depend continuously on the underlying moments, the resulting estimators  are asymptotically consistent. Moreover, the estimators also behave as the target information measure under affine transformations, being inherently robust to, for example, rescaling of the samples.

We estimate $h(X)$ from i.i.d. samples $X_1,\cdots,X_m$ as $h_n(U)$ for $U\sim \mathrm{Unif}(\{X_1,\cdots,X_m\}).$ More precisely, we introduce the following estimator of differential entropy.

\begin{definition} \label{to}
Let $X,X_1,\cdots,X_m$ be i.i.d. continuous RVs, denote $\calS=\{X_j\}_{j=1}^m,$ and consider the uniform RV $U \sim \mathrm{Unif}(\calS).$ We define the $n$-th estimate $\widehat{h}_n(\calS)$ of the differential entropy $h(X)$ by $\widehat{h}_n(\calS) := h_n(U).$
\end{definition}

The estimator of mutual information $I(X;Y)$ between a discrete $X$ and a continuous $Y$ is defined next. We utilize Theorem~\ref{jo}. We will need to invert the Hankel matrices of moments $(\BE[V^{i+j} \mid U=u])_{i,j\in [n]}$ for each $u\in \supp(U),$ where $(U,V)$ is uniformly distributed over the samples $\calS = \{(X_j,Y_j)\}_{j=1}^m.$ These Hankel matrices are invertible if and only if for each $u\in \{X_j\}_{j=1}^m$ there are more than $n$ distinct samples $(X_j,Y_j)$ for which $X_j=u$; equivalently, the size of the support set of the RV $V$ conditioned on $U=u$ exceeds $n.$ Thus, we remove all values $u$ that appear at most $n$ times in the samples $\calS.$ In other words, we replace $\calS$ with the subset
\begin{equation}
    \mathcal{S}^{(n)} := \left\{ (X',Y')\in\calS \; ; \; |\{ 1\le i \le m \; ; \; X_i = X' \}|>n \right\}.
\end{equation}

\begin{definition} \label{bo}
Let $(X,Y),(X_1,Y_1),\cdots,(X_m,Y_m)$ be i.i.d. $2$-RVs such that $X$ is discrete with finite support and $Y$ is continuous, and denote $\mathcal{S}=\{(X_j,Y_j) \}_{j=1}^m.$ Define $\mathcal{S}^{(1)}\supseteq \mathcal{S}^{(2)} \supseteq \cdots$ by
\begin{equation} \label{eo}
    \mathcal{S}^{(n)} := \left\{ (X',Y')\in\calS \; ; \; |\{ 1\le i \le m \; ; \; X_i = X' \}|>n \right\}.
\end{equation}
For each $n\ge 1$ such that $\mathcal{S}^{(n)}$ is nonempty, let $(U^{(n)},V^{(n)}) \sim \mathrm{Unif}(\mathcal{S}^{(n)}).$ We define the $n$-th estimate $\widehat{I}_n(\calS)$ of the mutual information $I(X;Y)$ by $\widehat{I}_n(\calS) := I_n(U^{(n)};V^{(n)}).$
\end{definition}

We show in this section how to implement these estimators numerically, prove that they are consistent, and discuss their sample complexity. We end the section by empirically comparing their performance with other estimators from the literature. For convenience, define the function $\delta_{X,n}:(0,\infty) \to [0,\infty)$ by
\begin{equation}
    \delta_{X,n}(t) := \det \bm{M}_{\sqrt{t}X+N,n}
\end{equation}
for a $2n$-times integrable RV $X.$ Recall that $\delta_{X,n}$ is the denominator of $\pp_n(X,\wc).$

\subsection{Simplification of the Differential Entropy Formula from Moments for Numerical Stability}

We develop the expressions of our approximations of differential entropy further to avoid possible issues that could arise from numerically computing the improper integral over $[0,\infty).$ To illustrate this issue, consider the expression for $h_2(X).$ Recall from \eqref{lv} that a zero-mean unit-variance RV $X$ satisfies
\begin{equation}
    \pp_2(X,t) = \frac{2+4t+(\calX_4-\calX_3^2-1)t^2}{2+6t+(\calX_4+3)t^2+(\calX_4-\calX_3^2-1)t^3}.
\end{equation}
For example, when $X \sim \mathrm{Unif}([-\sqrt{3},\sqrt{3}]),$ so 
\begin{equation}
    (\calX_1,\calX_2,\calX_3,\calX_4) = \left( 0 , 1 , 0 , \frac{9}{5} \right),
\end{equation}
we obtain
\begin{equation}
    \pp_2(X,t) = \frac{5+10t+2t^2}{5+15t+12t^2+2t^3}.
\end{equation}
Now, consider the expression for $h_2(X)$ in \eqref{jp}, namely,
\begin{equation} \label{qp}
    h_2(X) = \frac12 \int_0^\infty  \frac{5+10t+2t^2}{5+15t+12t^2+2t^3} - \frac{1}{2\pi e + t} \, dt.
\end{equation}
The integral in \eqref{qp} converges, but a numerical computation might not be able to capture this convergence as the expression for the integrand is a difference of non-integrable functions that both decay as $1/t.$ To avoid this possible issue, we subtract a $1/t$ term from both of these non-integrable functions. More precisely, denoting differentiation with respect to $t$ by a prime, we write
\begin{align}
    \pp_2(X,t) &= \frac{5+10t+2t^2 - \frac13 \delta_{X,2}'(t) + \frac13 \delta_{X,2}'(t)}{\delta_{X,2}(t)} \\
    &= \frac{2t}{5+15t+12t^2+2t^3} + \frac13 \frac{d}{dt} \log \delta_{X,2}(t)
\end{align}
and
\begin{equation}
    \frac{1}{2\pi e + t} = \frac{d}{dt} \log (2\pi e + t).
\end{equation}
The integrand $\pp_2(X,t)-1/(2\pi e+t)$ now becomes
\begin{equation}
    \frac{2t}{5+15t+12t^2+2t^3} + \frac{d}{dt} \log \frac{\delta_{X,2}(t)^{1/3}}{2\pi e + t}.
\end{equation}
The advantage in having the integrand in this form is that the first term is well-behaved (it decays as $1/t^2$), and the second term's integral can be given in closed form
\begin{equation}
    \int_0^\infty \left( \log \frac{\delta_{X,2}(t)^{1/3}}{2\pi e + t} \right)' \, dt = \log \left( 2 \pi e \left( \frac25 \right)^{1/3} \right).
\end{equation}
Therefore, equation \eqref{qp} becomes
\begin{equation} \label{qq}
    h_2(X) = \frac12 \log  \frac{2 \pi e}{(5/2)^{1/3}} +  \int_0^\infty \frac{t}{5+15t+12t^2+2t^3} \, dt.
\end{equation}
We use equation \eqref{qq} instead of \eqref{qp} for numerical computation. Note that this resolves the same numerical instability issue when estimating from data: if $\calS = \{X_j\}_{j=1}^m$ is a multiset of i.i.d. samples distributed according to $P_X,$ and if $U \sim \mathrm{Unif}(\calS),$ we compute the estimate $\widehat{h}_2(\calS)=h_2(U)$ of $h_2(X)$ via an expression analogous to that in \eqref{qq} where $X$ is replaced with $U.$ 

The procedure of obtaining expression \eqref{qq} from \eqref{qp} can be carried out for a general $X$ and $n$ such that $\BE[X^{2n}]<\infty$ and $|\supp(X)|>n,$ as follows. Let $\theta_{X,n}:[0,\infty)\to [0,\infty)$ be the polynomial that is the numerator of $\pp_n(X,t),$ i.e., $\theta_{X,n}(t) := \delta_{X,n}(t) \cdot  \pp_n(X,t).$ Thus, we have that
\begin{equation} \label{ur}
    \pp_n(X,t) = \frac{\theta_{X,n}(t)}{\delta_{X,n}(t)}.
\end{equation}
We define the function $\rho_{X,n}:[0,\infty) \to \BR$ by
\begin{equation} \label{tm}
    \rho_{X,n}(t) := \frac{\theta_{X,n}(t) - d_n^{-1} \delta_{X,n}'(t)}{2\delta_{X,n}(t)},
\end{equation}
where $d_n = \binom{n+1}{2}.$ By the analysis of the coefficients in $\pp_n(X,t)$ proved in Theorem~\ref{iv}, we have that $\rho_{X,n}(0)=0$ and \begin{equation}
    \rho_{X,n}(t) = O\left( t^{-2} \right)
\end{equation}
as $t\to \infty.$ In particular, $\rho_{X,n}$ is integrable over $[0,\infty).$ The following formula for differential entropy directly follows from the definition of $h_n$ in~\eqref{jp}.

\begin{lemma} \label{tn}
For any RV $X$ satisfying $\BE[X^{2n}]<\infty$ and $|\supp(X)|>n,$ we have the formula
\begin{equation} \label{xw}
    h_n(X) = \frac12 \log \left( 2\pi e \left( \frac{\det \bM_{X,n}}{\det \bM_{N,n}} \right)^{1/d_n} \right) + \int_0^\infty \rho_{X,n}(t) \, dt,
\end{equation}
where $d_n = \binom{n+1}{2},$ $N\sim \calN(0,1),$ and $\rho_{X,n}$ is as defined in~\eqref{tm}.
\end{lemma}

A similar conclusion holds for mutual information in view of equation~\eqref{jt} that expresses $I_n$ in terms of $h_n.$

\begin{lemma} \label{xv}
Fix a discrete RV $X$ with finite support, and a $2n$-times integrable continuous RV $Y.$ We have that
\begin{align}
    I_n(X;Y) = \frac{1}{n(n+1)} \log \frac{\det \bM_{Y,n}}{\prod_{x\in \supp(X)} \left( \det \bM_{Y^{(x)},n} \right)^{P_X(x)}} + \int_0^\infty \rho_{Y,n}(t) - \mathbb{E}_X \left[ \rho_{Y^{(X)},n}(t) \right] \, dt, \label{xx}
\end{align}
where for each $x\in \supp(X)$ we denote by $Y^{(x)}$ the RV $Y$ conditioned on $\{ X=x \}.$
\end{lemma}

Note that in Lemmas~\ref{tn} and~\ref{xv}, the determinants $\det \bM_{A,n}$ and the rational functions $\rho_n(A;t),$ for $A\in \{X,Y\}$ or $A\in \{Y^{(x)} ~ ; ~ x\in \supp(X)\},$ are completely determined by the first $2n$ moments of $A.$ To obtain the estimates $\widehat{h}_n$ and $\widehat{I}_n$ given samples, the moments of $A$ are replaced with their respective sample moments in formulas~\eqref{xw} and~\eqref{xx}.

\subsection{Consistency} \label{ub}

As sample moments converge almost surely to the moments, and as our expressions for differential entropy and mutual information depend continuously on the moments, the continuous mapping theorem yields that the estimators of differential entropy and mutual information introduced in the beginning of this section are consistent.

\begin{theorem} \label{bp}
Let $X$ be a continuous RV that has a MGF. Let $\{ X_j \}_{j=1}^{\infty}$ be i.i.d. samples drawn according to $P_X.$ Then, for every $n \in \BN,$ we have the almost-sure convergence
\begin{equation}
    \lim_{m\to \infty} \widehat{h}_n\left( \{ X_j \}_{j=1}^m \right) = h_n(X).
\end{equation}
Furthermore, we have that
\begin{equation} \label{tz}
    h(X) = \lim_{n\to \infty} \lim_{m\to \infty} \widehat{h}_n\left( \{ X_j \}_{j=1}^m \right)
\end{equation}
where the convergence in $m$ is almost-sure convergence.
\end{theorem}
\begin{proof}
See Appendix~\ref{ue}.
\end{proof}

\begin{corollary} \label{uc}
Let $X$ be discrete RV with finite support, and $Y$ be a continuous RV with a MGF. Let $\{(X_j,Y_j)\}_{j=1}^\infty$ be i.i.d. samples drawn according to $P_{X,Y}.$ For every $n\in \BN,$ we have the almost-sure convergence
\begin{equation} \label{ug}
    \lim_{m\to \infty} \widehat{I}_n\left(\{(X_j,Y_j)\}_{j=1}^m \right) = I_n(X;Y).
\end{equation}
Furthermore, 
\begin{equation} \label{uh}
I(X;Y) = \lim_{n\rightarrow \infty} \lim_{m\rightarrow \infty} \widehat{I}_n\left(\{(X_j,Y_j)\}_{j=1}^m\right)
\end{equation}
where the convergence in $m$ is almost-sure convergence.
\end{corollary}
\begin{proof}
See Appendix~\ref{ud}.
\end{proof}

\subsection{Sample Complexity} \label{wd}

When $X$ is a continuous RV of bounded support, we may derive the following sample complexity of the estimator of differential entropy in Definition~\ref{to} from Hoeffding's inequality.

\begin{prop} \label{uj}
Fix a bounded-support continuous RV $X \in L^{2n}(\calF).$ There is a constant $C=C(X,n)$ such that, for all small enough $\varepsilon,\delta>0,$ any collection $\calS$ of i.i.d. samples drawn according to $P_X$ of size
\begin{equation}
    |\calS| > \frac{C}{\varepsilon^2} \log \frac{1}{\delta}
\end{equation}
must satisfy
\begin{equation}
    \mathrm{Pr}\left\lbrace \left| \widehat{h}_n(\calS) - h_n(X) \right| < \varepsilon \right\rbrace \ge 1- \delta.
\end{equation}
\end{prop}
\begin{proof}
See Appendix~\ref{bz}.
\end{proof}

From Proposition~\ref{uj}, we may also obtain a sample complexity result for the estimate $\widehat{I}_n$ in Definition~\ref{bo}.

\begin{prop} \label{xy}
Fix a finitely-supported discrete RV $X$ and a bounded-support continuous RV $Y \in L^{2n}(\calF).$ There is a constant $C=C(X,Y,n)$ such that, for all small enough $\varepsilon,\delta>0,$ any collection $\calS$ of i.i.d. samples drawn according to $P_{X,Y}$ of size
\begin{equation}
    |\calS| > \frac{C}{\varepsilon^2} \log \frac{1}{\delta}
\end{equation}
must satisfy
\begin{equation}
    \mathrm{Pr}\left\lbrace \left| \widehat{I}_n(\calS) - I_n(X;Y) \right| < \varepsilon \right\rbrace \ge 1- \delta.
\end{equation}
\end{prop}
\begin{proof}
See Appendix~\ref{xz}.
\end{proof}

\subsection{Numerical Results} \label{bc}

We compare via synthetic experiments the performance of our estimators\footnote{A Python code can be found at~\cite{githubMIE}.} against some of the estimators in the literature. 

Our proposed estimator for differential entropy is $\widehat{h}_{10},$ i.e., given samples $\calS$ of $X$ we estimate $h(X)$ by $\widehat{h}_{10}(\calS)$ as given by Definition~\ref{to}. We compare this estimator with two estimation methods: $k$-Nearest-Neighbors ($k$-NN), and Kernel Density Estimation (KDE). The $k$-NN-based method we compare against is as provided by the Python package \textsf{`entropy\_estimators'} \cite{Steeg2014}, which we will refer to in this section as KSG. The kernel used for the KDE method is Gaussian, and it is obtained by computing from a set of samples $\{X_j\}_{j=1}^m$ a kernel $\Phi$ via the Python function \textsf{`scipy.stats.gaussian\_kde'}~\cite{scipy}; then, the estimate for differential entropy will be $\frac{-1}{m}\sum_{j=1}^m \log \Phi(X_j).$ The parameters for the KSG and the KDE estimators are the default parameters, namely, $k=3$ for the KSG estimator, and the bandwidth for the KDE estimator is chosen according to Scott's rule (i.e., $m^{-1/(d+4)}$ for a set of $m$ samples of a $d$-RV).

The mutual information is estimated using $\widehat{I_5},$ i.e., given samples $\mathcal{S}$ of $(X,Y)$ our estimate for $I(X;Y)$ will be $\widehat{I}_5(\mathcal{S})$ as given by Definition~\ref{bo}. This estimator is compared against the partitioning estimator and the Mixed KSG estimator~\cite{Gao2017} (which is a $k$-NN-based estimator); we utilize the implementation in \cite{Gao2017} for both estimators. In particular, the parameters are fixed throughout, namely, we utilize the parameters used in~\cite{Gao2017} ($k=5$ for the Mixed KSG, and $8$ bins per dimension for the partitioning estimator).

We perform $250$ independent trials for each experiment and each fixed sample size, then plot the absolute error as a percentage of the true value (except for the last experiment, where the ground truth is $0,$ so we plot the absolute error) against the sample size.

We note that we also performed the mutual information experiments for the Noisy KSG estimator based on the estimator in~\cite{Kraskov2004} (with noise strength $\sigma=0.01$ as in~\cite{Gao2017}), but its performance was much worse than the other estimators, so we do not include it in the plots.

\begin{experiment} \label{yj}
We estimate the differential entropy of a RV $X$ distributed according to Wigner's semicircle distribution, i.e.,
\begin{equation}
    p_X(x) := \frac{2}{\pi} \sqrt{1-x^2} \cdot 1_{[-1,1]}(x).
\end{equation}
The ground truth is $h(X)\approx 0.64473$ nats. We generate a set $\calS$ of i.i.d. samples distributed according to $P_X.$ The size of $\calS$ ranges from $800$ to $4000$ in increments of $800,$ and for each fixed sample size we independently generate $250$ such sets $\calS$ (so we generate a total of $1250$ sets of samples). The differential entropy $h(X)$ is estimated by three methods: the moments-based estimator that we propose $\widehat{h}_{10},$ the $k$-NN-based estimator implemented in~\cite{Steeg2014}  (which we refer to as the KSG estimator), and the Gaussian KDE estimator. For the proposed estimator, we use $\widehat{h}_{10}(\calS)$ as an estimate for $h(X).$ For the KSG estimator, we use the default setting, for which $k=3.$ We also use the default setting for the Gaussian KDE estimator; in particular, the bandwidth is chosen according to Scott's Rule as $m^{-1/(d+4)}$ where $m=|\calS|$ and $d=1$ is the dimensionality of $X.$ The percentage relative absolute error in the estimation (e.g., $100\cdot |\widehat{h}_{10}(\calS)/h(X)-1|$) is plotted against the sample size for the three estimators in Figure~\ref{DE_W1}. The solid lines in Figure~\ref{DE_W1} are the means of the errors, i.e., the mean in the $250$ independent trials of the percentage relative absolute error for each fixed sample size in $\{800,1600,2400,3200,4000\}.$ Via bootstrapping, we infer confidence intervals, which are indicated by the shaded areas around the solid lines in Figure~\ref{DE_W1}. We see that the proposed estimator outperforms the KSG estimator and the KDE estimator for this experiment.
\end{experiment}

\begin{figure}[!tb]
    \centering
    \includegraphics[width=0.7\textwidth]{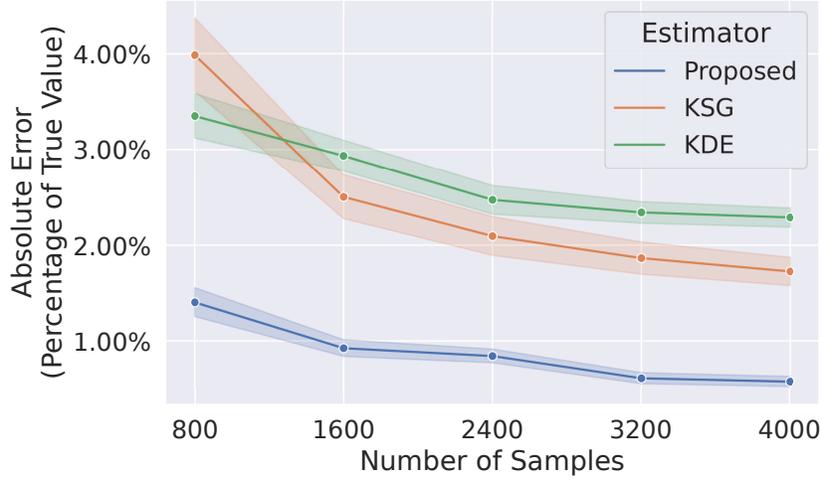}
    \caption{Estimation of differential entropy for a semicircle distribution as in Experiment~\ref{yj}. The vertical axis shows the percentage relative absolute error in the estimation, e.g., for the proposed estimator it is $100\cdot |\widehat{h}_{10}(\calS)/h(X) - 1|$ where $\calS$ is the set of samples and $h(X)\approx 0.64473$ nats is the ground truth. The horizontal axis shows $|\calS|,$ the sample size. The proposed estimator $\widehat{h}_{10}$ outperforms the $k$-NN-based estimator (denoted KSG) and the Gaussian KDE estimator for this experiment.}
    \label{DE_W1}
\end{figure}

\begin{experiment} \label{yh}
We estimate the differential entropy $h(\bX)$ of a random vector $\bX=(X_1,X_2)^T$ where $X_1$ and $X_2$ are i.i.d. distributed according to Wigner's semicircle distribution, namely, $\bX$ has the PDF
\begin{equation}
    p_{\bX}(x,y) = \frac{4}{\pi^2} \sqrt{(1-x^2)(1-y^2)} \cdot 1_{[-1,1]\times [-1,1]}(x,y).
\end{equation}
The ground truth is $h(\bX) \approx 1.28946$ nats. The same numerical setup as in Experiment~\ref{yj} is performed here. The results are plotted in Figure~\ref{DE_W2}, where we see a similar behavior to the comparison in the 1-dimensional case; in particular, the proposed estimator outperforms the KSG estimator and the KDE estimator for this experiment.
\end{experiment}

\begin{figure}[!tb] 
    \centering
    \includegraphics[width=0.7\textwidth]{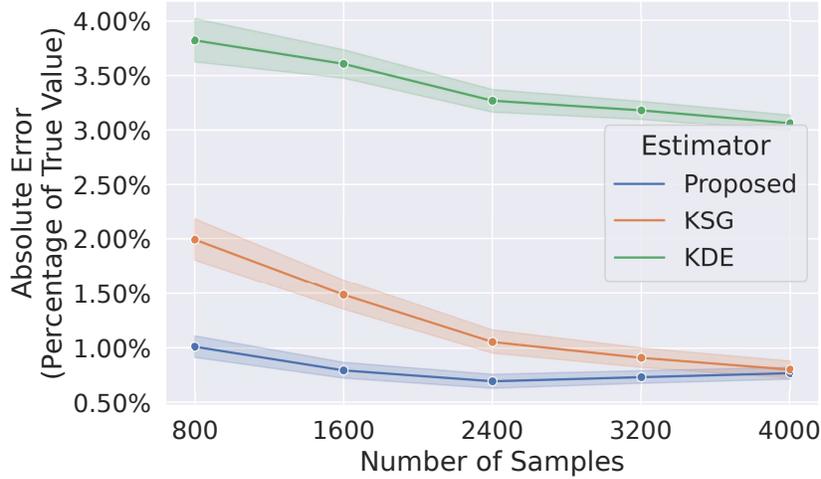} 
    \caption{Estimation of differential entropy for a 2-dimensional semicircle distribution as in Experiment~\ref{yh}. The proposed estimator $\widehat{h}_{10}$ outperforms both the KSG and the KDE estimators for this experiment.} 
    \label{DE_W2}
\end{figure}

\begin{experiment} \label{qs}
We estimate the differential entropy $h(X)$ of a Gaussian mixture $X$ whose PDF is given by
\begin{equation}
    p_X(x) = \sum_{i=1}^4 \frac{p_i}{\sqrt{2\pi \sigma_i^2}} e^{-(x-\mu_i)^2/(2\sigma_i^2)},
\end{equation}
where 
\begin{align}
    \bp &= (0.1,0.2,0.3,0.4) \\
    \bmu &= (-2,0,1,5) \\
    \bsigma &= (1.5,1,2,1).
\end{align}
The ground truth is $h(X) \approx 2.34249$ nats. The same numerical setup in Experiments~\ref{yj} and~\ref{yh} is used here. The results are plotted in Figure~\ref{DE1}. For this experiment, the proposed estimator outperforms the KSG estimator, and it is essentially indistinguishable from the KDE estimator. Note that it is expected that the KDE estimator performs well in this Gaussian mixture experiment, since it is designed specifically to approximate densities by Gaussian mixtures.
\end{experiment}

\begin{figure}[!tb]
    \centering
    \includegraphics[width=0.7\textwidth]{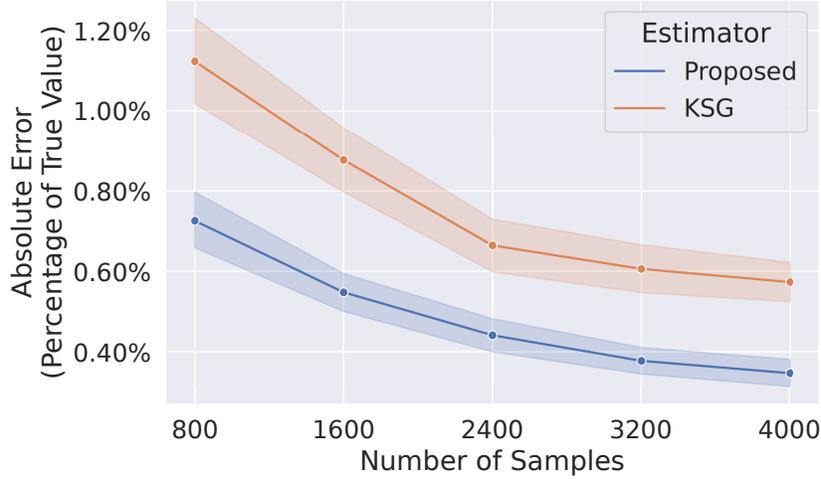}
    \caption{Estimation of differential entropy for a Gaussian mixture as in Experiment~\ref{qs}. The proposed estimator $\widehat{h}_{10}$ outperforms both the KSG and KDE estimators for this experiment. The plot of the KDE estimator's performance is omitted to avoid cluttering, as it lies just above the line for the proposed estimator but overlaps significantly with its uncertainty region.}
    \label{DE1}
\end{figure}

\begin{experiment} \label{yf}
We estimate the differential entropy $h(\bX)$ of a random vector $\bX$ that is a mixture of two Gaussians, namely, $\bX$ has the PDF
\begin{equation}
    p_{\bX}(\bx) = \frac{1}{4\pi \sqrt{\det(\bA)}} e^{-(\bx-\bmu)^T\bA^{-1}(\bx-\bmu)/2} + \frac{1}{4\pi \sqrt{\det(\bB)}} e^{-(\bx-\bnu)^T\bB^{-1}(\bx-\bnu)/2},
\end{equation}
where we have the means $\bmu=(-1,-1)^T$ and $\bnu=(1,1)^T,$ and the covariance matrices
\begin{equation}
    \bA = \left( \begin{array}{cc}
    1 & 1/2 \\
    1/2 & 1 
    \end{array} \right) 
\end{equation}
and $\bB = \bI_2.$ The ground truth is $h(\bX) \approx 3.22406$ nats. The same numerical setup as in Experiments~\ref{yj}-\ref{qs} is performed here. The results are plotted in Figure~\ref{DE2}. As in the 1-dimensional case in Experiment~\ref{qs}, the proposed estimator outperforms the KSG estimator for this experiment. Further, the proposed estimator also outperforms the KDE estimator in this 2-dimensional setting.
\end{experiment}

\begin{figure}[!tb]
    \centering
    \includegraphics[width=0.7\textwidth]{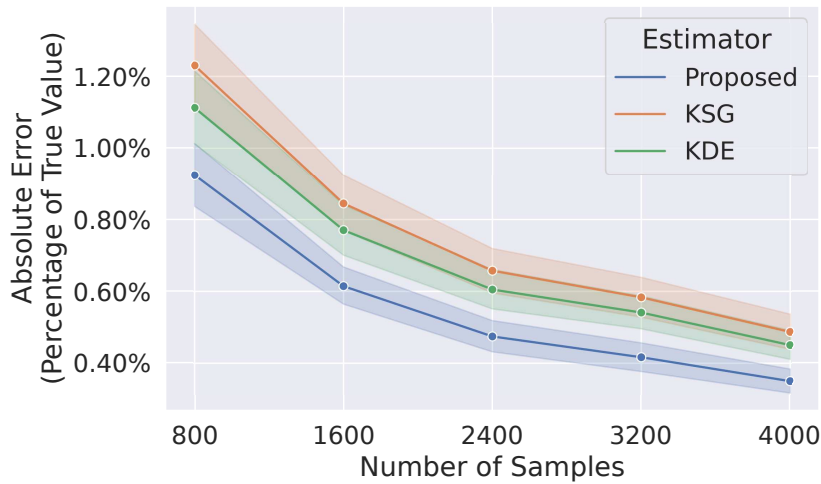}
    \caption{Estimation of differential entropy for a vector Gaussian mixture as in Experiment~\ref{yf}. The proposed estimator $\widehat{h}_{10}$ outperforms both the KSG and KDE estimators for this experiment.}
    \label{DE2}
\end{figure}

\begin{experiment} \label{qr}
We replicate the mixture-distribution part of the zero-inflated Poissonization experiment of~\cite{Gao2017}. In detail, we let $Y\sim \mathrm{Exp}(1),$ and let $X=0$ with probability $0.15$ and $X\sim \mathrm{Pois}(y)$ given that $Y=y$ with probability $0.85.$ The quantity to be estimated is the mutual information $I(X;Y),$ and the ground truth is $I(X;Y) \approx 0.25606$ nats. We generate a set of i.i.d. samples $\calS$ according to the distribution $P_{X,Y},$ where $\calS$ has size in $\{800,1600,2400,3200\}.$ We estimate $I(X;Y)$ via the proposed estimator by $\widehat{I}_5(\calS),$ and we also consider the estimates given by the Mixed KSG estimator and the partitioning estimator, both as implemented in~\cite{Gao2017} (including the parameters used therein). This estimation process is repeated independently $250$ times. The comparison of estimators' performance is plotted in Figure~\ref{MI_Unscaled}. The solid lines indicate the mean percentage relative absolute error, and the shaded areas indicate confidence intervals obtained via bootstrapping. We see in Figure~\ref{MI_Unscaled} that the proposed estimator outperforms the other considered estimators for this experiment. We also test the affine-transformation invariance property of the proposed estimator. In particular, we consider estimating the mutual information from the scaled samples $\calS'$ obtained from $\calS$ via scaling the $Y$ samples by $10^4,$ i.e.,
\begin{equation}
    \calS' := \{ (A,10^4B) ~;~ (A,B) \in \calS \}. 
\end{equation}
Plotted in Figure~\ref{MI_Scaled} is a comparison of the same estimators using the same samples as those used to generate Figure~\ref{MI_Unscaled}, but where $Y$ is processed through this affine transformation. The ground truth stays unchanged, and so do our estimator and the partitioning estimator, but the Mixed KSG estimates change. This experiment illustrates the resiliency of the proposed estimator to affine transformations. In fact, the computed numerical values in the modified setting by the proposed estimator differ by no more than $10^{-15}$ nats from those numerically computed in the original setting for each of the $1000$ different sets of samples $\calS$; in theory, these pairs of values are identical, and the less than $10^{-15}$ discrepancy is an artifact of the computer implementation. Finally, we note that although the setup is more general than the assumptions we prove our results under in this paper (as $X$ here is not finitely supported), the proposed estimator outperformed the other estimators.
\end{experiment}

\begin{figure}[!tb]
    \centering
    \includegraphics[width=0.7\textwidth]{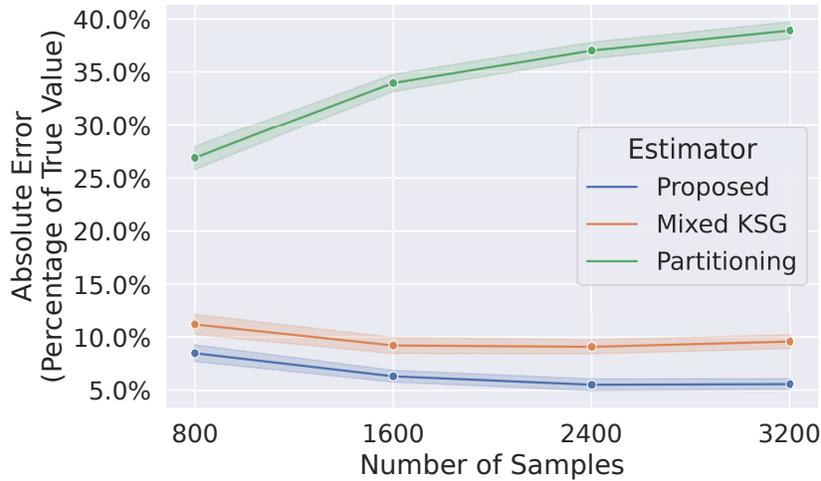}
    \caption{Percentage relative absolute error vs. sample size for unscaled zero-inflated poissonization in Experiment~\ref{qr}. The proposed estimator $\widehat{I}_5$ outperforms both the $k$-NN-based estimator (denoted Mixed KSG) and the partitioning estimator.}
    \label{MI_Unscaled}
\end{figure}

\begin{figure}[!tb]
    \centering
    \includegraphics[width=0.7\textwidth]{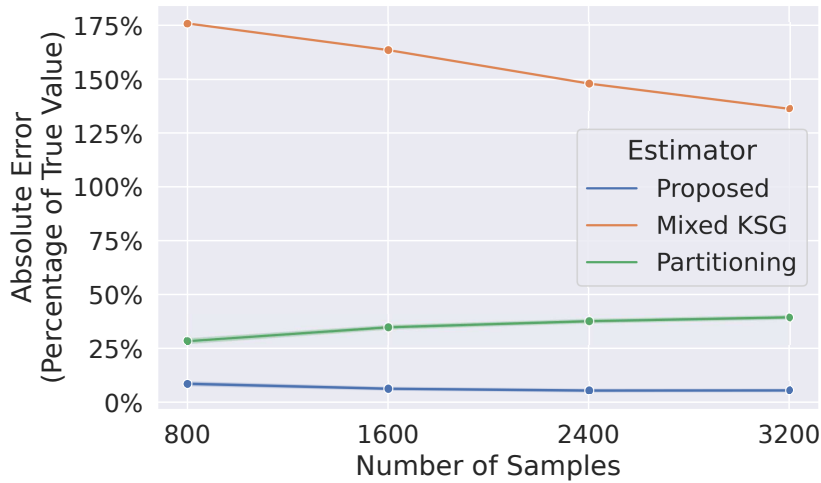}
    \caption{Percentage relative absolute error vs. sample size for the scaled zero-inflated poissonization in Experiment~\ref{qr}. To generate these plots, we use the same samples that yield the plots in Figure~\ref{MI_Unscaled}, but we process them through an affine transformation. Specifically, each sample $(A,B)$ is replaced with $(A,10^4 B).$ Then the samples are passed to the three estimators. We see that the proposed estimator $\widehat{I}_5$ is resilient to scaling, i.e., the same performance line in Figure~\ref{MI_Unscaled} is observed here too. This is in contrast to the performance of the Mixed KSG estimator. The partitioning estimator is resilient to scaling, but its performance is not favorable in this experiment (with above $\% 25$ relative absolute error).}
    \label{MI_Scaled}
\end{figure}

\begin{experiment} \label{yg}
We test for independence under the following settings. We consider independent $X\sim \mathrm{Bernoulli}(0.5)$ and $Y\sim \mathrm{Unif}([0,2]).$ We estimate $I(X;Y),$ whose true value is $I(X;Y)=0.$ We employ the same estimation procedure as in Experiment~\ref{qr}. The results are plotted in Figure~\ref{MI_Indep}, which shows that the proposed estimator predicted independence more accurately than the other estimators for the same sample size. Note that in this case the plot shows the absolute error (in nats) rather than the relative absolute error, as the ground truth is zero.
\end{experiment}

\begin{figure}[!tb]
    \centering
    \includegraphics[width=0.7\textwidth]{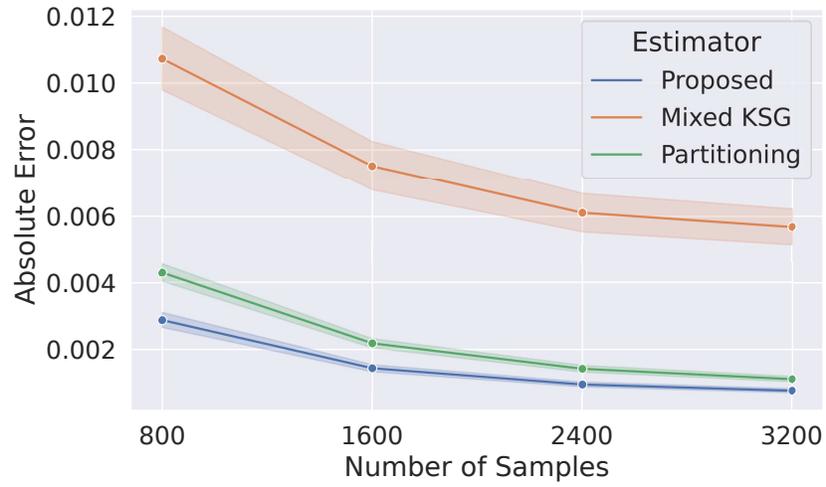}
    \caption{Absolute error (in nats) vs. sample size for the independence testing in Experiment~\ref{yg}. The proposed estimator $\widehat{I}_5$ outperforms the Mixed KSG and the partitioning estimators in this experiment.}
    \label{MI_Indep}
\end{figure}

\newpage

\section*{Acknowledgements}

The authors would like to thank Dr. Shahab Asoodeh (Harvard University) for pointing out that our approach for expressing the mutual information in terms of moments is applicable also to differential entropy. The authors are also grateful to Prof. Alex Dytso (NJIT) for noting that the higher-order derivatives of the conditional expectation in Gaussian channels are expressible in terms of the conditional cumulants.
\addcontentsline{toc}{section}{Acknowledgements}
\clearpage 
\newpage 

\begin{appendices}
\addtocontents{toc}{\protect\setcounter{tocdepth}{0}}

\section{A Derivation of Equation \texorpdfstring{\eqref{arx}}{(58)}} \label{aaj}

Using the notation of \cite{partition}, we have that
\begin{equation} \label{aaa}
    C_r = \sum_{k=1}^r (k-1)! \stirling{r}{k}_{\ge 2}
\end{equation}
where $\stirling{r}{k}_{\ge 2}$ denotes the number of partitions of an $r$-element set into $k$ subsets each of which contains at least $2$ elements (note that there are $(k-1)!$ cyclically-invariant arrangements of $k$ parts). The exponential generating function of the sequence $r \mapsto \stirling{r}{k}_{\ge 2}$ is $(e^x-1-x)^k/k!.$ Now, we may write
\begin{equation}
    (e^x-1-x)^k = \sum_{a+b\le k} \binom{k}{a,b} (-1)^{k-a}x^b \sum_{t\in \BN} \frac{(ax)^t}{t!}.
\end{equation}
Therefore, the coefficient of $x^r$ in $(e^x-1-x)^k/k!$ is
\begin{align}
    \frac{1}{r!}\stirling{r}{k}_{\ge 2} &= \sum_{a+b \le k} \frac{(-1)^{k-a}a^{r-b}}{a!b!(k-a-b)!(r-b)!} \\
    &= \frac{1}{r!} \sum_{b=0}^k \binom{r}{b} \sum_{a=0}^{k-b} (-1)^{k-a} \frac{a^{r-b}}{a!(k-a-b)!} \\
    &= \frac{1}{r!} \sum_{b=0}^k \binom{r}{b} \stirling{r-b}{k-b} (-1)^b,
\end{align}
which when combined with~\eqref{aaa} gives~\eqref{arx}.

\section{Proof of Lemma~\ref{arv}} \label{aae}

Assume that $\BE\left[|X|^{\deg(p)}\right]=\infty$ (so $\deg(p)\ge 1$), and we will show that $\BE\left[|p(X)|\right] = \infty$ too. Let $k\in [\deg(p)-1]$ be the largest integer for which $\BE\left[ |X|^k \right] < \infty,$ and write $p(u)=u^{k+1}q(u)+r(u)$ for a nonzero polynomial $q$ and a remainder $r\in \SP_k.$ By monotonicity of norms, $\BE\left[ |X|^j \right]<\infty$ for every $j\in [k].$ Hence, $r(X)$ is integrable. Therefore, it suffices to prove that $X^{k+1}q(X)$ is non-integrable, which we show next.

Consider the set $\calD = \left\{ u\in \BR ~ ; ~ \left| q(u) \right| < |a| \right\}$ where $a\neq 0$ is the leading coefficient of $q.$ If $q$ is constant, then $\calD$ is empty, whereas if $\deg q \ge 1$ then $|q(u)|\to \infty$ as $|u|\to \infty$ implies that $\calD$ is bounded; in either case, there is an $M\in \BR$ such that $\calD\subset [-M,M].$ Now, writing $1=1_{\calD}+1_{\calD^c},$ we obtain
\begin{equation} \label{aru}
    \BE\left[ |X|^{k+1} |q(X)| \right] \ge |a| ~ \BE\left[ |X|^{k+1} 1_{\calD^c}(X) \right].
\end{equation}
But we also have that
\begin{equation}
    \infty = \BE\left[ |X|^{k+1}\right] \le M^{k+1} +  \BE\left[ |X|^{k+1} 1_{\calD^c}(X) \right],
\end{equation}
so $\BE\left[ |X|^{k+1} 1_{\calD^c}(X) \right]=\infty.$ Therefore, inequality \eqref{aru} yields that $\BE\left[ \left|X^{k+1}q(X)\right|\right]=\infty,$ concluding the proof.

\section{Proofs of Section~\ref{az}: PMMSE Preliminaries} 

\subsection{Proof of Proposition \ref{iw}: Equivalence of PMMSE Definitions} \label{kn}

In this proof, we drop the subscript on the norm, so $\|\wc\|$ is the norm of $L^2(\calF).$ Let $\{p_i(Y)\}_{i\in [d]}$ be an orthonormal basis for $\SP_n(Y).$ Then, the orthogonal projection, $E_n^A[X \mid Y],$ of $X$ onto $\SP_n(Y)$ is given by
\begin{equation}
    E_n^A[X\mid Y] = \sum_{i\in [d]} \langle X,p_i(Y) \rangle p_i(Y),
\end{equation}
with $d= \dim \SP_n(Y)-1.$ 

We next show that $E_n^A[X\mid Y]$ is the unique closest element in $\SP_n(Y)$ to $X.$ Note that $L^2(\calF)$ is separable. Extend the orthonormal basis $\{p_i(Y)\}_{i\in [d]}$ to an orthonormal basis $\{p_i(Y)\}_{i\in [d]} \cup \{f_j \}_{j\in \BN}$ for $L^2(\calF).$ 

Fix an arbitrary $q(Y)\in \SP_n(Y)$ and let $\ba \in \BR^{d+1}$ be such that $q(Y)=\sum_{i\in [d]} a_i p_i(Y).$ We may expand $X$ as
\begin{equation}
    X = \sum_{i\in [d]} \langle X,p_i(Y)\rangle p_i(Y) + \sum_{j\in \BN} \langle X,f_j\rangle f_j.
\end{equation}
Then, the projection $E_n^A[X \mid Y]$ is at least as close to $X$ as $q(Y)$ is, because
\begin{align}
    \left\|X-E_n^A[X \mid Y]\right\|^2 &=\left\|X - \sum_{i\in [d]} \langle X,p_i(Y) \rangle p_i(Y) \right\|^2 = \left\|\sum_{j\in \BN} \langle X,f_j \rangle f_j \right\|^2 = \sum_{j\in \BN} |\langle X,f_j \rangle|^2 \label{ic} \\
    &\le \sum_{i\in [d]} |\langle X,p_i(Y) \rangle - a_i|^2 + \sum_{j\in \BN} |\langle X,f_j \rangle|^2 \label{ib} \\
    &= \left\|\sum_{i\in [d]} (\langle X,p_i(Y)\rangle - a_i)p_i(Y) + \sum_{j\in \BN} \langle X,f_j \rangle f_j \right\|^2 = \| X - q(Y)\|^2. \label{id}
\end{align}

Finally, we show uniqueness. So, assume that $r\in \SP_n$ is such that
\begin{equation}
    \|X-E_n^A[X\mid Y]\| = \|X- r(Y) \|,
\end{equation}
so this value lower bounds $\|X-p(Y)\|$ for any $p(Y)\in \SP_n(Y).$ Note that $(E_n^A[X\mid Y]+r(Y))/2\in \SP_n(Y).$ Hence, $\|X-(E_n^A[X\mid Y]+r(Y))/2\|$ is lower bounded by both $\|X-E_n^A[X\mid Y]\|$ and $\|X- r(Y) \|.$ We utilize the parallelogram law, $\|a+b\|^2+\|a-b\|^2=2\|a\|^2+2\|b\|^2$ for $a,b \in L^2(\calF).$ Setting $a=X-E_n^A[X\mid Y]$ and $b=X-r(Y)$ in the parallelogram law, we have that
\begin{align}
    0 &\le \| E_n^A[X\mid Y] - r(Y) \|^2 = \|(X-r(Y))-(X-E_n^A[X\mid Y])\|^2 \\
    &= 2\|X-r(Y)\|^2+ 2\|X-E_n^A[X\mid Y]\|^2  - 4\|X-(E_n^A[X\mid Y]+r(Y))/2\|^2 \le 0.
\end{align}
Therefore, $\|E_n^A[X\mid Y]-r(Y)\|=0,$ i.e., $E_n^A[X\mid Y]=r(Y).$ Hence, $E_n^A[X\mid Y]$ uniquely minimizes $\|X-p(Y)\|$ over $p(Y)\in \SP_n(Y).$ Therefore, by definition of $E_n^B[X\mid Y],$ we get that $E_n^B[X\mid Y]$ is well defined and that
\begin{equation}
    E_n^A[X\mid Y] = E_n^B[X\mid Y],
\end{equation}
and the proof is complete.

\subsection{Proof of Theorem \ref{ar}: PMMSE Converges to MMSE} \label{lm}

First, we show that the limits hold if $Y$ has finite support. Suppose $\supp(Y) = \{y_j\}_{j\in [m]}.$ Let $q\in \SP_m$ be the polynomial interpolant of $\{\left(y_j,\BE[X\mid Y=y_j]\right)\}_{j\in [m]},$ so $q(Y)=\BE[X\mid Y].$ Then, for every $n\ge m,$ we have that $E_n[X\mid Y] = q(Y) = \BE[X\mid Y].$ Hence, the limits~\eqref{hp} and~\eqref{hq} hold in this case. Thus, we may assume that $|\supp(Y)|=\infty.$

By Carleman's condition, polynomials are dense in $L^2(\sigma(Y))$~\cite{Lubinsky2007}, i.e.,
\begin{equation}
    \overline{\bigcup_{n\in \BN} \SP_n(Y)} = L^2(\sigma(Y)).
\end{equation}
In particular, the collection of monomials $\{Y^k\}_{k\in \BN}$ is a complete set, i.e., if $Z\in L^2(\sigma(Y))$ satisfies $\langle Z,Y^k \rangle = 0$ for every $k\in \BN$ then it must be that $Z=0.$ Further, since $|\supp(Y)|=\infty,$ the monomials are linearly independent. Hence, applying Gram-Schmidt, one obtains an orthonormal basis for the Hilbert space $L^2(\sigma(Y))$ consisting of polynomials $\{q_k(Y) \}_{k\in \BN}$ such that $q_k \in \SP_k$ for each $k\in \BN.$ 

Therefore, we may write
\begin{equation}
    \BE [X \mid Y] = \sum_{k\in \BN} \langle \BE [X \mid Y] , q_k(Y)\rangle  q_k(Y)
\end{equation}
where the series converges in the $L^2(\sigma(Y))$-norm sense. Further, by the orthogonality principle for $\BE [X \mid Y]$
\begin{equation}
    \langle \BE [X \mid Y],q_k(Y)\rangle = \langle X,q_k(Y) \rangle
\end{equation}
for every $k\in \BN.$ Now, by the characterization in \eqref{hl} that $E_n[X\mid Y]$ is the orthogonal projection of $X$ onto $\SP_n(Y),$ we know that
\begin{equation}
    E_n[X \mid Y] = \sum_{k=0}^n \langle X,q_k(Y) \rangle q_k(Y)
\end{equation}
for each $n\in \BN$ (note that $\{q_k(Y)\}_{k\in [n]}$ is an orthonormal basis for $\SP_n(Y)$). Therefore, 
\begin{equation}
    \BE [X \mid Y] = \lim_{n\to \infty} \sum_{k=0}^n \langle X,q_k(Y) \rangle q_k(Y) = \lim_{n\to \infty} E_n[X \mid Y],
\end{equation}
i.e., \eqref{hp} holds.

Finally, for \eqref{hq}, note that
\begin{align}
    \pp_n(X\mid Y) &= \|X-E_n[X \mid Y]\|_2^2 \\
    \mm(X\mid Y) &= \|X-\BE[X\mid Y]\|_2^2.
\end{align}
Since both RVs $E_n[X \mid Y]$ and $\BE[X\mid Y]$ are elements in $L^2(\sigma(Y)),$ the orthogonality principle of $\BE[X\mid Y]$ implies that $X-\BE[X\mid Y]$ and $E_n[X \mid Y]-\BE[X\mid Y]$ are orthogonal. Then, the Pythagorean theorem implies that
\begin{equation}
    \pp_n(X\mid Y) - \mm(X \mid  Y) = \|\BE[X \mid  Y] - E_n[X \mid Y] \|_2^2,
\end{equation}
and the limit \eqref{hq} follows from the limit \eqref{hp}.

\subsection{Proof of Lemma~\ref{kv}: Invertibility of $\bM_{Y,n}$} \label{ln}

The matrix $\bM_{Y,n}$ is symmetric. We show that it is positive-semidefinite, and that it is positive-definite if and only if $|\supp(Y)|>n.$ For any $\bd\in \BR^{n+1},$ we have the inequality
\begin{align}
\bd^T \bM_{Y,n} \bd & = \bd^T \BE  \left[ \bY^{(n)} \left(\bY^{(n)}\right)^T \right] \bd  =  \BE  \left[ \bd^T\bY^{(n)} \left(\bY^{(n)}\right)^T \bd \right]  = \BE \left[ \left| \bd^T\bY^{(n)} \right|^2 \right] \ge 0,
\end{align}
so $\bM_{Y,n}$ is positive-semidefinite. Furthermore, the equality case 
\begin{equation}
    \BE \left[ \left| \bd^T\bY^{(n)} \right|^2 \right]=0
\end{equation}
holds if and only if $  \left| \bd^T\bY^{(n)} \right|^2=0,$ and this latter relation holds if and only if $\bd^T\bY^{(n)}=0.$ Therefore, $\bM_{Y,n}$ is positive-definite if and only if $\bd^T\bY^{(n)}=0$ implies $\bd = \mathbf{0},$ i.e., $\bY^{(n)}$ does not lie almost surely in a hyperplane in $\BR^{n+1}.$ Finally, $\bY^{(n)}$ lies almost surely in a hyperplane in $\BR^{n+1}$ if and only if $|\supp(Y)|\le n.$ Therefore, the desired result that $\bM_{Y,n}$ is invertible if and only if $|\supp(Y)|>n$ follows.

\subsection{Proof of Lemma~\ref{io}: An Orthonormal Basis for $\SP_n(Y)$} \label{lp}

First, note that Lemma \ref{kv} implies that $\bM_{Y,n}$ is invertible. To show orthonormality, we show that the expectation of the outer product of the random vector 
\begin{equation}
\bV = \bM_{Y,n}^{-1/2} \bY^{(n)}    
\end{equation}
is the identity matrix. We have that
\begin{align}
    \BE\left[ \bV \bV^T \right] &= \BE \left[ \bM_{Y,n}^{-1/2} \bY^{(n)} \left( \bY^{(n)} \right)^T \left( \bM_{Y,n}^{-1/2} \right)^T \right] = \bM_{Y,n}^{-1/2} \bM_{Y,n} \left( \bM_{Y,n}^{-1/2} \right)^T  = \bI_{n+1}, \label{lt}
\end{align}
where we have used that (see~\eqref{ls})
\begin{equation}
    \bM_{Y,n}^{1/2} \left( \bM_{Y,n}^{1/2} \right)^T = \bM_{Y,n}
\end{equation}
and (see~\eqref{lr})
\begin{equation}
    \bM_{Y,n}^{-1/2} = \left( \bM_{Y,n}^{1/2} \right)^{-1}.
\end{equation}
Hence, the entries of the vector $\bM_{Y,n}^{-1/2}\bY^{(n)}$ form an orthonormal subset of $\SP_n(Y).$ Since we have that $\mathrm{span}(\{1,\cdots,Y^n\})=\SP_n(Y),$ and $\bM_{Y,n}^{-1/2}$ is invertible, we conclude that the entries of $\bM_{Y,n}^{-1/2} \bY^{(n)}$ also span $\SP_n(Y),$ which completes the proof.

\subsection{A Second Proof for Theorem \ref{hw}} \label{lq}

By assumption on $Y,$ Lemma \ref{io} yields invertibility of $\bM_{Y,n}.$ Denote 
\begin{equation} \label{cc}
\bc_{X,Y,n} := \bM_{Y,n}^{-1} \BE  \left[ X\bY^{(n)} \right],
\end{equation}
and we will show that strict convexity of the function $g:\BR^{n+1}\to [0,\infty)$ defined by
$$
g(\bd) =  \BE \left[ \left( X- \bd^T\bY^{(n)} \right)^2 \right]
$$
implies that $\bc_{X,Y,n}$ is its unique minimizer. For any $\bd \in \BR^{n+1},$ linearity of expectation implies that the gradient of $g$ is
$$
\nabla g (\bd) = \left( 
\BE \left[ 2Y^k\left( \bd^T\bY^{(n)}-X \right) \right]  \right)_{0\le k \le n},
$$
so the Hessian of $g$ is the constant $2\bM_{Y,n}.$ As $\bM_{Y,n}$ is positive-definite, $g$ is strictly convex. As $\nabla g(\bd)={\bf 0}$ is equivalent to $\bM_{Y,n}\bd = \BE  \left[ X\bY^{(n)} \right],$ i.e., to $\bd = \bc_{X,Y,n},$ the desired result follows.

\section{Proofs of Section \ref{ca}: Basic PMMSE Properties}

\subsection{Proof of Proposition \ref{hn}: PMMSE and Affine Transformations} \label{kb}

Set $U=Y+\beta.$ For any $\bc\in \BR^{n+1},$ 
\begin{equation}
X+\alpha-\bc^T\bU^{(n)} = X-(\bM\bc-\alpha \be_1)^T\bY^{(n)},
\end{equation}
where we define the matrix
\begin{equation}
\bM := \left( \beta^{i-j}\binom{i}{j} \right)_{(i,j)\in [n]^2},
\end{equation}
and we set $\beta^{i-j}\binom{i}{j} =0$ when $j>i$ and $\beta^0\binom{i}{i}=1$ when $\beta=0.$ Then $\bM$ is lower-triangular with an all-$1$ diagonal, so the inverse $\bM^{-1}$ exists. Thus, the mapping $\BR^{n+1}\to \BR^{n+1}$ defined by $\bc\mapsto \bM\bc-\alpha \be_1$ is invertible (where $\bd \mapsto \bM^{-1}(\bd+\alpha \be_1)$ is the inverse mapping). Therefore, the following two subsets of $L^2(\calF)$ are equal
\begin{equation}
    \left\{ X+\alpha-\bc^T\bU^{(n)} ~ ; ~ \bc \in \BR^{n+1} \right\} = \left\{ X-\bd^T\bY^{(n)} ~ ; ~ \bd \in \BR^{n+1} \right\}.
\end{equation}
Then, by the definition of the PMMSE, equality~(\ref{ai}) holds.

Equation (\ref{aj}) may be treated similarly. First, note that the case $\alpha = 0$ is immediate, because then $\alpha X = 0 \in \SP_n(\beta Y)$ implies $\pp_n(\alpha X \mid \beta Y)= 0.$ Assume $\alpha \neq 0.$ Setting $V=\beta Y,$ one has that for any $\bc \in \BR^{n+1}$
\begin{equation}
\alpha X - \bc^T \bV^{(n)} = \alpha \left(X-(\bL \bc)^T\bY^{(n)}\right)
\end{equation}
where we define the invertible matrix
\begin{equation}
\bL:= \mathrm{diag}\left( \left(\beta^k/\alpha\right)_{k\in [n]} \right).
\end{equation}
As $\bc\mapsto \bL \bc$ is a bijection of $\BR^{n+1},$ the definition of the PMMSE yields equation (\ref{aj}).

\subsection{Proof of Lemma \ref{me}: PMMSE Operator Properties} \label{mf}

Let $\be_0,\be_1,\cdots,\be_n$ denote the standard basis vectors for $\BR^{n+1}.$ For~\ref{rp}, we first note that
    \begin{equation} \label{md}
        \bM_{Z,n}^{-1} \BE  \left[ \bZ^{(n)} \right] = \be_0,
    \end{equation}
    because $\BE[\bZ^{(n)}] = \bM_{Z,n} \be_0.$ By the formula for $E_n[X\mid Z]$ in Theorem \ref{hw}
    \begin{equation} \label{ja}
    E_n[X \mid Z] = \BE  \left[ X\bZ^{(n)} \right]^T \bM_{Z,n}^{-1} \, \bZ^{(n)}.
    \end{equation}
    Taking the expectation and using \eqref{md},
    \begin{equation}
    \BE \left[E_n[X\mid Z]\right] = \BE [X].
    \end{equation}
    
    For~\ref{rq}, we generalize this approach and note that
    \begin{equation} \label{rv}
        \bM_{Z,n}^{-1}\BE \left[ Z^j\bZ^{(n)} \right]=\be_j.
    \end{equation}
    for each $j \in [n].$ Then, for any monomial $Z^j$ with $j\in [n]$ we have from~\eqref{ja} and~\eqref{rv} that
    \begin{equation}
        \BE\left[ E_n[X\mid Z] Z^j \right]=\BE  \left[ X\bZ^{(n)} \right]^T \be_j = \BE[XZ^j].
    \end{equation}
    In other words, for all $j\in [n],$
    \begin{equation}
        \BE  \left[ (X-E_n[X\mid Z]) Z^j \right] = 0.
    \end{equation}
    By linearity of expectation, we conclude that for any polynomial $p\in \SP_n$
    \begin{equation}
    \BE [(X-E_n[X \mid Z])p(Z)]=0.
    \end{equation}
    
    Properties~\ref{rr}-\ref{wo} follow immediately by Proposition~\ref{ro}. Alternatively, the linearity in~\ref{rr} follows from linearity of expectation in view of the formula for $E_n[\wc \mid Z]$ in~\eqref{ja}, the contractivity in~\ref{wn} follows since $0\le \pp_n(X\mid Y) = \|X\|_2^2 - \|E_n[X\mid Y]\|_2^2,$ the idempotence in~\ref{rt} follows directly from the fact that $E_n[X\mid Z]\in\SP_n(Z),$ and the self-adjointness in~\ref{wo} can be verified via formula~\eqref{ja}.

    If $X$ and $Z$ are independent, then
    \begin{equation}
        \BE\left[ X\bZ^{(n)} \right] = \BE [X] ~ \BE \left[ \bZ^{(n)} \right].
    \end{equation}
    Therefore, from $\bM_{Z,n}^{-1}\BE\left[\bZ^{(n)}\right]=\be_0$ in \eqref{md}, and in view of formula \eqref{ja} for $E_n[X\mid Z],$ we have that
    \begin{equation}
        E_n[X\mid Z] = \BE[X] \be_0^T \bZ^{(n)} = \BE[X],
    \end{equation}
    so \ref{rs} follows.
    
    Finally, for \ref{ru}, assume that $X$---$Y$---$Z$ is a Markov chain. Then, for every $j\in [n]$
    \begin{equation} \label{rw}
        \BE\left[ XZ^j \right] = \BE\left[ \BE\left[XZ^j \mid Y \right] \right] = \BE\left[ \BE\left[X \mid Y \right] \BE\left[Z^j \mid Y \right] \right].
    \end{equation}
    Further, we have that
    \begin{equation} \label{rx}
        \BE\left[ \BE\left[X \mid Y \right] \BE\left[Z^j \mid Y \right] \right] = \BE\left[ \BE\left[X \mid Y \right] Z^j \right]
    \end{equation}
    by the orthogonality property for $\BE[\wc \mid Y]$:
    \begin{equation}
         \BE\left[ \BE\left[X \mid Y \right] \left( Z^j - \BE\left[Z^j \mid Y \right] \right) \right] = 0.
    \end{equation}
    Combining \eqref{rw} and \eqref{rx}, we obtain that for every $j\in [n]$
    \begin{equation} \label{ry}
        \BE\left[ XZ^j \right] =  \BE\left[ \BE\left[X \mid Y \right] Z^j \right].
    \end{equation}
    Multiplying \eqref{ry} on the right by $\bM_{Z,n}^{-1}\bZ^{(n)},$ the formula for $E_n[\wc \mid Z]$ (see~\eqref{ja}) implies that 
    \begin{equation}
        E_n\left[ X \mid Z \right] = E_n\left[ \BE[X\mid Y] \mid Z \right],
    \end{equation}
    as desired.

\subsection{Proof of Proposition~\ref{wk}: PMMSE for Symmetric RVs} \label{wl}

We may assume that $X$ and $Z$ are symmetric around $0,$ since $E_{m}[X+a\mid X+Z+b] = a + E_m[X \mid X+Z]$ for every $m\in \BN$ and $a,b\in \BR.$ Then, $\BE[X^j]=\BE[Z^j]=0$ for every odd $j \in \BN.$ Set $Y=X+Z$ and $n=2k.$ Then, $\BE[Y^j]=0$ for every odd $j\in \BN,$ and $\BE[XY^\ell]=0$ for every even $\ell \in \BN.$ Then, the coefficient of $Y^n$ in $E_n[X \mid Y]$ is 
\begin{equation}
     \frac{1}{\det \bM_{Y,n}}  \sum_{\substack{\ell \in [n] \\ \ell ~ \text{odd}}} \BE\left[ XY^{\ell} \right] \left[ \bM_{Y,n}^{-1} \right]_{\ell,n},
\end{equation}
where $\left[ \bM_{Y,n}^{-1} \right]_{\ell,n}$ denotes the $(\ell,n)$-th entry of $\bM_{Y,n}^{-1}.$ Fix an odd $\ell \in [n].$ Let $T_n^{(\ell,n)}$ denote the set of permutations of $[n]$ that send $\ell$ to $n.$ We have that
\begin{equation}
    \left[ \bM_{Y,n}^{-1} \right]_{\ell,n} = -\sum_{\pi \in T_n^{(\ell,n)}} \mathrm{sgn}(\pi) \prod_{r \in [n]\setminus \{\ell\}} \BE\left[ Y^{r+\pi(r)} \right].
\end{equation}
We have that, for every $\pi \in T_n^{(\ell,n)},$ $\sum_{r\in [n] \setminus \{\ell\}} r+\pi(r) = n(n+1)-\ell-n,$ which is odd. Therefore, for at least one $r\in [n]\setminus \{\ell\},$ the integer $r+\pi(r)$ is odd. Hence, $\BE[Y^{r+\pi(r)}]=0,$ implying that $\left[ \bM_{Y,n}^{-1} \right]_{\ell,n}=0.$ As this is true for every odd $\ell\in [n],$ we conclude that the coefficient of $Y^n$ in $E_n[X\mid Y]$ is $0,$ i.e., \eqref{wi} holds. Equation~\eqref{wj} follows from~\eqref{wi} as $\pp_m(X\mid Y) = \|X-E_m[X\mid Y]\|_2^2$ for every $m\in \BN.$

\section{Proofs of Section~\ref{jh}}

\subsection{Proof of Theorem \ref{an}: Uniform Convergence of the PMMSE} \label{br}

We start the proof by obtaining from Proposition \ref{ar} pointwise convergence. Let $N\sim \calN(0,1)$ be independent of $X.$ We shall verify the assumptions in Proposition \ref{ar} on $\sqrt{t}X+N$ for fixed $t\ge 0$: \textbf{i)} As $N$ is continuous and $X$ is independent of $N$ we must have $|\mathrm{supp}(\sqrt{t}X+N)|=\infty,$ and \textbf{ii)} The MGF of $\sqrt{t}X+N$ exists (it is the product of the MGFs of $\sqrt{t}X$ and $N$) and this implies that $\sqrt{t}X+N$ satisfies Carleman's condition~\cite{Lubinsky2007}. Hence, by Proposition \ref{ar}, we get that for every $t\ge 0$
\begin{equation}\label{bg}
    \lim_{n\to \infty} \pp_n(X,t) = \mm(X,t).
\end{equation}
Now, we show that the convergence is uniform.

Set, for each $n\in \mathbb{N}$ and $t\in [0,\infty),$
\begin{equation}
    g_n(t):=\pp_n(X,t)-\mm(X,t)
\end{equation}
for short. We will show that 
\begin{equation} \label{eg}
    \lim_{n\rightarrow \infty} \sup_{t \in [0,\infty)} g_n(t) = 0,
\end{equation}
which is the uniform convergence in (\ref{ef}). 

The limit \eqref{bg} says that
\begin{equation} \label{ht}
    \lim_{n \to \infty} g_n(t) = 0
\end{equation}
for every fixed $t\ge 0.$ In addition, the asymptotics given in Corollary \ref{is} imply that for each fixed $n\in \mathbb{N}$
\begin{equation} \label{hu}
    \lim_{t \to \infty} g_n(t) = 0.
\end{equation}

By definition of the PMMSE as the minimum over sets of increasing size (in $n$), for each fixed $t\ge 0$ the sequence $\{\pp_n(X,t)\}_{n\in \mathbb{N}}$ is decreasing. So $\{g_n\}_{n\in \mathbb{N}}$ is a pointwise decreasing sequence of functions (i.e., $g_1(t)\ge g_2(t) \ge \cdots$ for each fixed $t\ge 0$). Note that the $g_n$ are nonnegative. We finish the proof via Cantor's intersection theorem.

Fix $\varepsilon>0.$ For each $n\in \mathbb{N},$ let $C_{\varepsilon,n}=g^{-1}_n([\varepsilon,\infty)),$ where $g^{-1}$ denotes the set-theoretic inverse. As $\{g_n\}_{n\in \mathbb{N}}$ is decreasing, $C_{\varepsilon,1}\supseteq C_{\varepsilon,2} \supseteq \cdots$ is decreasing too. As each $g_n$ is continuous, each $C_{\varepsilon,n}$ is closed. Further, $\lim_{t\to \infty} g_1(t) = 0$ implies that $C_{\varepsilon,1}$ is bounded, hence each $C_{\varepsilon,n}$ is bounded. Thus, each $C_{\varepsilon,n}$ is compact. But, the intersection $\bigcap_{n\in \mathbb{N}} C_{\varepsilon,n}$ is empty, for if $t_0$ were in the intersection then $\liminf_{n\to \infty} g_n(t_0)\ge \varepsilon$ violating that $\lim_{n\to \infty} g_n(t_0)=0.$ Hence, by Cantor's intersection theorem, it must be that the $C_{\varepsilon,n}$ are eventually empty, i.e., there is an $m\in \mathbb{N}$ such that $\sup_{t\in [0,\infty)} g_n(t)\le \varepsilon$ for every $n>m.$ This is precisely the uniform convergence in (\ref{eg}), and the proof is complete.

\subsection{Proof of Lemma \ref{ac}} \label{ki}

The integer $i+\pi(i)$ is odd if and only if $i$ and $\pi(i)$ have opposite parities. Thus, the desired result follows from the following more general characterization. For any partition $[n]=A\cup B,$ the cardinality of the set
 \begin{equation}
     I := \left\{ i \in \{1,\cdots,n\} \; ; \; (i,\pi(i))\in (A\times B)\cup (B\times A) \right\}
 \end{equation}
 is even. The desired result follows by letting $A$ and $B$ be even and odd integers, respectively, in $[n].$ Now, we show that the general characterization holds.
 
 Let $A_\pi\subset A$ denote the subset of elements of $A$ that get mapped by $\pi$ into $B,$ i.e.,
\begin{equation}
    A_\pi := \{ i \in A \; ; \; \pi(i)\in B \},
\end{equation}
and define $B_\pi$ similarly. Then, $I= A_\pi \cup B_\pi$ is a partition. As $|A_\pi|=|B_\pi|,$ the desired result follows.

\subsection{Proof of Lemma \ref{kl}} \label{km}

For each $(i,j)\in [n]^2$ let the subset $T_n^{(i,j)}\subset \sn$ denote the collection of permutations sending $i$ to $j,$ i.e.,
\begin{equation}
    T_n^{(i,j)} := \left\{ \pi \in \sn ~ ; ~ \pi(i) =j \right\}.
\end{equation}
We define, for each $(i,j)\in [n]^2,$ the cofactor functions $c_{X,n}^{(i,j)}:[0,\infty)\to \BR$ and the products $d_{X,n}^{(i,j)}:[0,\infty)\to \BR$ by
\begin{align}
    c_{X,n}^{(i,j)}(t) &:= \sum_{\pi \in T_n^{(i,j)}} \mathrm{sgn}(\pi) \prod_{k\neq i} \left(  \bM_{\sqrt{t}X+N,n} \right)_{k,\pi(k)}, \\
    d_{X,n}^{(i,j)}(t)&:= v_{X,i}(t) ~  c_{X,n}^{(i,j)}(t) ~ v_{X,j}(t).
\end{align}
Here, $\left(  \bM_{\sqrt{t}X+N,n} \right)_{a,b}$ is the $(a,b)$-th entry of $ \bM_{\sqrt{t}X+N,n},$ i.e.,
\begin{equation}
    \left(  \bM_{\sqrt{t}X+N,n} \right)_{a,b} = \BE\left[ \left(\sqrt{t}X+N \right)^{a+b} \right].
\end{equation}
Note that $c_{X,n}^{(i,j)}(t)$ is the $(i,j)$-th cofactor of $ \bM_{\sqrt{t}X+N,n}.$ The cofactor matrix $\bC_{X,n}:[0,\infty) \to \BR^{(n+1)\times (n+1)}$ of $t \mapsto \bM_{\sqrt{t}X+N,n}$ is given by
\begin{align}
    \bC_{X,n} &:=\left( c_{X,n}^{(i,j)} \right)_{(i,j)\in [n]^2}. \label{it}
\end{align}
We define the function $D_{X,t}:[0,\infty)\to \BR$ by
\begin{equation}
    D_{X,n} := \bv_{X,n}^T \bC_{X,n} \bv_{X,n}.
\end{equation}

We have the following two relations. First, $D_{X,n}$ is the sum of the $d_{X,n}^{(i,j)}$
\begin{equation}
    D_{X,n}(t) = \sum_{(i,j)\in [n]^2} d_{X,n}^{(i,j)}(t).
\end{equation}
Second, by Cramer's rule, and because symmetry of the matrix $\bM_{\sqrt{t}X+N,n}$ implies that its cofactor is equal to its adjugate, we have the formula 
\begin{equation}
    \bM_{\sqrt{t}X+N,n}^{-1}= \frac{1}{\det  \bM_{\sqrt{t}X+N,n}} ~ \bC_{X,n}.
\end{equation}
Therefore, we obtain
\begin{equation}\label{av}
    F_{X,n}(t) = \frac{D_{X,n}(t)}{\det  \bM_{\sqrt{t}X+N,n}} = \frac{\sum_{(i,j)\in [n]^2} d_{X,n}^{(i,j)}(t)}{\det  \bM_{\sqrt{t}X+N,n}}.
\end{equation}
Hence, it suffices to study the  $d_{X,n}^{(i,j)}.$

We start with a characterization of the cofactors $c_{X,n}^{(i,j)}.$ Namely, we show that if $i+j$ is even then $c_{X,n}^{(i,j)}(t)$ is a polynomial in $t,$ and if $i+j$ is odd then $\sqrt{t}c_{X,n}^{(i,j)}(t)$ is a polynomial in $t.$ If $i+j$ is even, then
\begin{equation}
c_{X,n}^{(i,j)}(t) = \sum_{\pi \in T_n^{(i,j)}} \mathrm{sgn}(\pi) ~ t^{\delta(\pi)/2} \quad \quad \prod_{\mathclap{\substack{k\in [n] \\ k+\pi(k) \; \text{odd} }}} \quad o_{1,X,k+\pi(k)}(t) \quad \quad \prod_{\mathclap{\substack{r \in [n], ~ r\neq i  \\  r+\pi(r) \; \text{even} \; }}} \quad  e_{1,X,r+\pi(r)}(t),
\end{equation}
whereas if $i+j$ is odd then
\begin{equation}
c_{X,n}^{(i,j)}(t) = \sum_{\pi \in T_n^{(i,j)}} \mathrm{sgn}(\pi) ~ t^{\frac{\delta(\pi)-1}{2}} \quad \quad \prod_{\mathclap{\substack{k\in [n], ~ k\neq i \\ k+\pi(k) \; \text{odd} }}} \quad  o_{1,X,k+\pi(k)}(t) \quad \quad \prod_{\mathclap{\substack{r \in [n] \\ r+\pi(r) \; \text{even} }}} \quad e_{1,X,r+\pi(r)}(t).
\end{equation}
Thus, evenness of $\delta(\pi)$ for each $\pi \in \sn$ implies that each $c_{X,n}^{(i,j)}(t)$ is a polynomial when $i+j$ is even and that each $\sqrt{t}c_{X,n}^{(i,j)}(t)$ is a polynomial when $i+j$ is odd. Further, the degree of $c_{X,n}^{(i,j)}$ for even $i+j$ is upper bounded by
\begin{equation}
\frac{\delta(\pi)}{2} + \sum_{k+\pi(k) \; \text{odd}} \frac{k+\pi(k)-1}{2} + \sum_{r+\pi(r) \; \text{even} \, ; \, r\neq i} \frac{r+\pi(r)}{2} = \frac{n(n+1)}{2} - \frac{i+j}{2},
\end{equation}
whereas the degree of $\sqrt{t}c_{X,n}^{(i,j)}$ and for odd $i+j$ is upper bounded by
\begin{equation}
\frac{\delta(\pi)}{2} + \sum_{k+\pi(k) \; \text{odd}\, ; \, k\neq i} \frac{k+\pi(k)-1}{2} + \sum_{r+\pi(r) \; \text{even}} \frac{r+\pi(r)}{2} = \frac{n(n+1)}{2} - \frac{i+j-1}{2}.
\end{equation}
We note that both upper bounds are equal to
\begin{equation}
    \frac{n(n+1)}{2} - \left\lfloor \frac{i+j}{2} \right\rfloor.
\end{equation}
Finally, considering the terms of highest order, we see that the term
\begin{equation}
    \sum_{\pi \in T_n^{(i,j)}} \mathrm{sgn}(\pi) \prod_{k\in [n]\setminus \{i\}} \calX_{k+\pi(k)}
\end{equation}
is the coefficient of $t^{ \frac{n(n+1)}{2} - \left\lfloor \frac{i+j}{2} \right\rfloor}$ in $c_{X,n}^{(i,j)}$ when $i+j$ is even and in $\sqrt{t}c_{X,n}^{(i,j)}$ when $i+j$ is odd.

Now, to show that $D_{X,n}$ is a polynomial, it suffices to check that each $d_{X,n}^{(i,j)}$ is. We consider separately the parity of $i+j$ and build upon the characterization of $c_{X,n}^{(i,j)}.$ If $i+j$ is even, so $i$ and $j$ have the same parity, then
$$
\mathbb{E}\left[ X  \left(\sqrt{t}X+N\right)^i  \right] \mathbb{E}\left[ X \left(\sqrt{t}X+N\right)^j  \right]
$$
is a polynomial in $t$ of degree at most $(i+j)/2$ with the coefficient of $t^{(i+j)/2}$ being $\calX_{i+1}\calX_{j+1}.$ If $i+j$ is odd, so $i$ and $j$ have different parities, then 
$$
t^{-1/2}\mathbb{E}\left[ X  \left(\sqrt{t}X+N\right)^i  \right] \mathbb{E}\left[ X  \left(\sqrt{t}X+N\right)^j   \right]
$$
is a polynomial in $t$ of degree at most $(i+j-1)/2$ with the coefficient of $t^{(i+j-1)/2}$ being $\calX_{i+1}\calX_{j+1}.$ 

Thus, from the characterization of $c_{X,n}^{(i,j)},$ regardless of the parity of $i+j$ we obtain that $d_{X,n}^{(i,j)}(t)$ is a polynomial in $t$ of degree at most $n(n+1)/2$ with the coefficient of $t^{n(n+1)/2}$ being 
\begin{equation}
    \calX_{i+1}\calX_{j+1} \sum_{\pi \in T_n^{(i,j)}} \mathrm{sgn}(\pi) \prod_{k\in [n]\setminus \{i\}} \calX_{k+\pi(k)}.
\end{equation}
Summing over all $(i,j)\in [n]^2,$ we obtain that the coefficient of $t^{n(n+1)/2}$ in $D_{X,n}(t)$ is
\begin{equation}
    \sum_{(i,j)\in [n]^2}  \sum_{\pi \in T_n^{(i,j)}} \mathrm{sgn}(\pi) \calX_{i+1}\calX_{j+1} \prod_{k\in [n]\setminus \{i\}} \calX_{k+\pi(k)}.
\end{equation}
The proof is completed by noting that, for each $i \in [n],$ we have a partition
\begin{equation}
\calS_{[n]} = \bigcup_{j\in [n]} T_n^{(i,j)}.
\end{equation}

\subsection{Proof of Lemma \ref{qy}} \label{rd}

The formulas for $a_X^{n,0}$ and $b_X^{n,0}$ follow by setting $t=0$ in \eqref{mv} and \eqref{mu}. Indeed, if $N\sim \calN(0,1)$ is independent of $X,$ then
\begin{equation}
    F_{X,n}(0) = \calX_1^2 ~ \BE\left[ \bN^{(n)} \right]^T \bM_{N,n}^{-1} ~ \BE\left[ \bN^{(n)} \right] = \calX_1^2
\end{equation}
because $\BE\left[ \bN^{(n)} \right]$ is the leftmost column of $\bM_{N,n}.$ Therefore, 
\begin{equation}
    a_X^{n,0} = \sigma_X^2 \det \bM_{N,n} = \sigma_X^2 b_X^{n,0}.
\end{equation}
Further, by direct computation or using the connection between Hankel matrices and orthogonal polynomials~\cite[Appendix A]{Simon1998} along with the fact that the probabilist's Hermite polynomials $q_k$ satisfy the recurrence $xq_k(x) = q_{k+1}(x)+kq_{k-1}(x),$ it follows that $\det \bM_{N,n} = \prod_{k=1}^{n} k! = G(n+2)$ where $G$ is the Barnes $G$-function. Equation~\eqref{qz} for $b_X^{n,d_n}$ is proved in Lemma~\ref{kk}. To simplify the proof of formula~\eqref{ra} for $b_X^{n,1},$ we first show the shift-invariance stated in~\eqref{rb}.

Fix $s\in \BR.$ For any i.i.d. RVs $Z,Z_0,\cdots,Z_n,$ we have that (see, e.g.,~\cite[Appendix A]{Simon1998})
\begin{equation} \label{dy}
    \det \bM_{Z,n} = \frac{1}{(n+1)!} ~ \mathbb{E} \left[ \prod_{0\le i < j \le n} (Z_i-Z_j)^2 \right].
\end{equation}
From equation (\ref{dy}), since $(Z_i+s)-(Z_j+s)= Z_i-Z_j,$ we obtain that
\begin{equation} \label{rc}
    \det \bM_{Z+s,n}=\det \bM_{Z,n}.
\end{equation}
Let $N\sim \calN(0,1)$ be independent of $X.$ Then, for every $t\in [0,\infty),$ considering $Z=\sqrt{t}X+N$ in \eqref{rc}, we obtain
\begin{equation} \label{dz}
    \det \bM_{\sqrt{t}(X+s)+N,n} = \det \bM_{\sqrt{t}X+N,n}.
\end{equation}
As both sides of (\ref{dz}) are polynomials in $t,$ we obtain that $b_{X+s}^{n,j} = b_{X}^{n,j}$ for every $j\in [d_n].$ Since we also have $\pp_n(X+s,t)=\pp_n(X,t),$ it follows that
\begin{equation}
    t \mapsto \sum_{j\in [d_n]} a_X^{n,j} t^j = \pp_n(X,t) \sum_{j\in [d_n]} b_X^{n,j} t^j
\end{equation}
is also invariant under shifting $X,$ so we also obtain $a_{X+s}^{n,j} = a_X^{n,j}.$

By the shift-invariance of $b_X^{n,1},$ we may assume that $\calX_1=0$ (so $\calX_2= \sigma_X^2$). Now, as each entry in $\bM_{\sqrt{t}X+N,n}$ is a polynomial in $\sqrt{t},$ we see that we may drop any term of order $(\sqrt{t})^3$ or above for the sake of finding $b_X^{n,1}$ (which is the coefficient of $t$ in $\det \bM_{\sqrt{t}X+N,n}$). In other words, 
\begin{equation}
    b_X^{n,1} = \det \left( \binom{i+j}{2} \sigma_X^2 \BE\left[ N^{i+j-2} \right] t + \BE\left[ N^{i+j} \right]   \right)_{(i,j)\in [n]^2} \hspace{-1pt}.
\end{equation}
By Leibniz's formula, we conclude
\begin{equation}
    b_X^{n,1} = \sigma_X^2  \sum_{\substack{\pi \in \sn \\ k\in [n]}}  \mathrm{sgn}(\pi) \binom{k+\pi(k)}{2}  \BE\left[ N^{k+\pi(k)-2}\right]  \prod_{i\in [n]\setminus \{k\}}  \BE\left[N^{i+\pi(i)}\right]. \label{eb}
\end{equation}
But, for any non-negative integer $m$
\begin{equation}
    \binom{m}{2} \BE\left[ N^{m-2} \right] = \frac{m}{2} \BE\left[ N^m \right].
\end{equation}
Therefore, (\ref{eb}) simplifies to
\begin{equation}
    b_X^{n,1} = \frac{\sigma_X^2}{2} \hspace{-5pt} \sum_{\substack{\pi \in \sn \\ k\in [n]}}  \mathrm{sgn}(\pi) (k+\pi(k))  \prod_{i\in [n]}  \BE\left[N^{i+\pi(i)}\right].
\end{equation}
Evaluating the summation over $k$ for each fixed $\pi,$ we obtain that
\begin{equation}
    b_X^{n,1} = \binom{n+1}{2} \sigma_X^2 \sum_{\pi \in \sn} \mathrm{sgn}(\pi) \prod_{i\in [n]}  \BE\left[N^{i+\pi(i)}\right].
\end{equation}
Finally, by Leibniz's formula for $\det \bM_{N,n},$ we obtain that
\begin{equation}
    b_X^{n,1} =  \binom{n+1}{2}\sigma_X^2  \det \bM_{N,n},
\end{equation}
as desired.

\subsection{Formulas for $a_X^{n,j}$} \label{re}

We have (see \eqref{mr}) the polynomial in $t$
\begin{equation}
    \left( \calX_2 - F_{X,n}(t) \right)\det \bM_{\sqrt{t}X+N,n} = \sum_{j\in [d_n]} a_X^{n,j} t^j.
\end{equation}
We obtain from \eqref{mw}
\begin{equation}
    \calX_2 \det \bM_{\sqrt{t}X+N,n} = \sum_{j\in [d_n]} \calX_2 b_{X}^{n,j} t^j.
\end{equation}
The function $t \mapsto F_{X,n}(t) \det \bM_{\sqrt{t}X+N,n}$ was also shown in Lemma \ref{kl} to be a polynomial of degree at most $d_n.$ By definition of $F_{X,n}$ in \eqref{mx}, we obtain that
\begin{equation} \label{ta}
    F_{X,n}(t) \det \bM_{\sqrt{t}X+N,n} = \bv_{X,n}(t)^T \bC_{\sqrt{t}X+N,n} \bv_{X,n}(t),
\end{equation}
where $\bC_{\sqrt{t}X+N,n}$ is the cofactor matrix of $\bM_{\sqrt{t}X+N,n}$ and $\bv_{X,n}$ is the vector-valued function defined in \eqref{ix} and \eqref{my}. As the entries in the cofactor matrix are $\pm 1$ times determinants of minors of $\bM_{\sqrt{t}X+N,n},$ we may use Leibniz's formula here too. Explicitly, if $\bu=(u_0,\cdots,u_n)^T\in \BR^{n+1}$ is any vector and $\bH=(h_{i+j})_{(i,j)\in [n]^2}\in \BR^{(n+1)\times (n+1)}$ is any Hankel matrix, and if $\bR$ is the cofactor matrix of $\bH,$ then
\begin{equation} \label{mz}
    \bu^T \bR \bu \hspace{2mm} =  \sum_{(i,\pi)\in [n] \times \sn}  (-1)^{i+\pi(i)} \hspace{1mm} \mathrm{sgn}(\pi) u_iu_{\pi(i)}  \prod_{r\in [n] \setminus \{i\}}  h_{r+\pi(r)}.
\end{equation}
Applying \eqref{mz} to the triplet
\begin{equation} 
    (\bu,\bH,\bR) = (\bv_{X,n}(t),\bM_{\sqrt{t}X+N,n},\bC_{\sqrt{t}X+N,n}),
\end{equation}
we obtain a formula for $F_{X,n}(t) \det \bM_{\sqrt{t}X+N,n}$ similar to \eqref{mz}. Then, analogously to how we obtained \eqref{na} from \eqref{amz}, expanding the powers and expectations we obtain from~\eqref{ta} and~\eqref{mz} that
\begin{align}
    F_{X,n}(t) \det \bM_{\sqrt{t}X+N,n} \hspace{3mm}   = \hspace{10mm} \sum_{\mathclap{\substack{
    (i,\pi) \in  [n]\times \sn \\
    (w,z) \in [i]\times [\pi(i)] \\
    k_r \in [r+\pi(r)], ~ \forall r\in [n]\setminus\{i\}}}} \hspace{5mm} t^{(w+z+s_i(\bk))/2} \gamma_{i,\pi,\bk,w,z} \calX_{w+1}\calX_{z+1}  \prod_{r\in [n]\setminus\{i\}}  \calX_{k_r}, \label{nb}
\end{align}
where $\gamma_{i,\pi,\bk,w,z}$ are the integers\footnote{An alternative proof of Lemma \ref{kl} is derivable via analysing when these coefficients are nonzero.}
\begin{align}
     \gamma_{i,\pi,\bk,w,z}=(-1)^{i+\pi(i)}\mathrm{sgn}(\pi)\binom{i}{w}\binom{\pi(i)}{z}\BE[N^{i-w}]\BE[N^{\pi(i)-z}]\prod_{r\in [n]\setminus\{i\}}  \binom{r+\pi(r)}{k_r} \BE[N^{r+\pi(r)-k_r}],
\end{align}
and where we define the restricted sums
\begin{equation}
    s_i(\bk)=\sum_{r\in [n]\setminus\{i\}} k_r.
\end{equation}
Since $F_{X,n}(t) \det \bM_{\sqrt{t}X+N,n}$ is shown in Lemma \ref{kl} to be a polynomial in $t,$ the summation in \eqref{nb} may be restricted to run through only those parameters for which $w+z+s_i(\bk)$ is even. Therefore,
\begin{align}
    F_{X,n}(t) \det \bM_{\sqrt{t}X+N,n} \hspace{3mm}  = \hspace{2mm} \sum_{j\in [d_n]} \hspace{1mm} t^j \hspace{13mm} \sum_{\mathclap{\substack{
    (i,\pi) \in  [n]\times \sn \\
    (w,z) \in [i]\times [\pi(i)] \\
    k_r \in [r+\pi(r)], ~ \forall r\in [n] \setminus\{i\}  \\
    w+z+s_i(\bk) = 2j}}} \hspace{8mm} \gamma_{i,\pi,\bk,w,z} \calX_{w+1}\calX_{z+1}  \prod_{r\in [n]\setminus\{i\}}  \calX_{k_r}.
\end{align}
Now, returning to the definition of the $a_X^{n,j}$ in \eqref{mv}
\begin{equation}
    \left( \calX_2 - F_{X,n}(t) \right)\det \bM_{\sqrt{t}X+N,n} = \sum_{j\in [d_n]} a_{X}^{n,j} t^j,
\end{equation}
we obtain that
\begin{align}
    a_X^{n,j} \hspace{3mm} = \hspace{1cm} \sum_{\mathclap{\substack{
    \pi \in  \sn \\ 
    k_r \in [r+\pi(r)], ~ \forall r\in [n] \\
    k_0+\cdots+k_n = 2j }}} \hspace{8mm}  \beta_{\pi;k_0,\cdots,k_n} \calX_2 \calX_{k_0}\cdots \calX_{k_n} \hspace{3mm} - \hspace{1cm} \sum_{\mathclap{\substack{
    (i,\pi) \in  [n]\times \sn \\
    (w,z) \in [i]\times [\pi(i)] \\
    k_r \in [r+\pi(r)], ~ \forall r\in [n] \setminus\{i\} \\
    w+z+s_i(\bk) = 2j}}} \hspace{1cm} \gamma_{i,\pi,\bk,w,z} \calX_{w+1}\calX_{z+1}  \prod_{r\in [n]\setminus\{i\}}  \calX_{k_r}. \label{nc}
\end{align}
From \eqref{nc}, we obtain
\begin{equation}
    a_X^{n,j} \in H_{2j+2,\min(n,2j)+2,\tau_n(j)}(X)
\end{equation}
where
\begin{equation}
    \tau_n(j) = \left\lbrace \begin{array}{cl}
    2 & \text{if } j=0, \\
    2j+1 & \text{if } 1\le j \le \frac{n}{2}, \\
    2j & \text{if } \frac{n+1}{2} \le j \le n, \\
    2n & \text{if } j>n. 
    \end{array} \right.
\end{equation}

\subsection{Proof of Lemma~\ref{sf}} \label{se}

Fix $\varepsilon>0.$ Let $W$ be a discrete RV with finite support such that $\BE\left[ W^k \right] = \mu_k$ for each $k\in [m],$ which exists by the solution to the truncated Hamburger moment problem~\cite[Theorem 3.1]{Curto91}. Set $A=\supp(W),$ and let $\{ q_a \}_{a\in A}\subset [0,1]$ be such that $p_W = \sum_{a\in A} q_a \delta_a$ is the PMF of $W.$ Then, for each $k\in [m],$ we have that
\begin{equation} \label{rh}
    \sum_{a\in A} q_a a^k = \mu_k.
\end{equation}
For each $s>0,$  $a \in A,$ and $k\in [m],$ we have that 
\begin{equation} \label{ri}
    \int_{a-s}^{a+s} \frac{z^k}{2s} \, dz = \frac{(a+s)^{k+1}-(a-s)^{k+1}}{2s (k+1)}.
\end{equation}
Further, we have the limit
\begin{equation}
    \lim_{s\to 0^+} \frac{(a+s)^{k+1}-(a-s)^{k+1}}{2s (k+1)} = a^k.
\end{equation}
Therefore, there exist constants $\eta_{a,k}>0,$ for $(a,k)\in A\times [m],$ such that $0 < s \le \eta_{a,k}$ implies
\begin{equation} \label{rj}
\left| \frac{(a+s)^{k+1}-(a-s)^{k+1}}{2s (k+1)} -a^k\right| < \varepsilon.
\end{equation}
Let $\eta = \min_{(a,k)\in A \times [m]} \eta_{a,k}$ and consider the continuous RV $Z$ whose PDF is given by
\begin{equation}
    p_Z(z) = \sum_{a\in A} \frac{q_a}{2\eta} 1_{[a-\eta,a+\eta]}(z).
\end{equation}
Thus, for each $k \in [m],$ 
\begin{equation}
    \BE\left[ Z^k \right] = \sum_{a\in A} q_a \cdot \frac{(a+\eta)^{k+1}-(a-\eta)^{k+1}}{2\eta (k+1)} \in \left( \sum_{a\in A} q_a (a^k-\varepsilon), \sum_{a\in A} q_a(a^k+\varepsilon)
    \right) = \left( \mu_k - \varepsilon, \mu_k+\varepsilon \right),
\end{equation}
which follow by \eqref{ri}, \eqref{rj}, and \eqref{rh}, respectively.

\subsection{Proof of Proposition~\ref{qx}} \label{sg}

We proceed by induction on $m.$ The case $m=1$ follows because then by assumption on $p$ we have that $p(k)=0$ for every positive integer $k$ as can be seen by taking $X\sim \calN(k,1),$ but the only polynomial with infinitely many zeros is the zero polynomial. Now, assume that the statement of the proposition holds for every polynomial in $m-1$ variables, where $m\ge 2.$

Fix a polynomial $p$ in $m$ variables, and assume that $p|_{\calC^m}=0.$ Regarding $p$ as a polynomial in one of the variables with coefficients being polynomials in the remaining $m-1$ variables, we may write
\begin{equation}
    p(u_1,\cdots,u_m) = \sum_{j\in [d]} p_j(u_1,\cdots,u_{m-1})u_m^j,
\end{equation}
for some polynomials $p_0,\cdots,p_d$ in $m-1$ variables, where $d$ is the total degree of $p.$ We show that $p=0$ identically by showing that each $p_j$ vanishes on $\calC_{m-1}$ and using the induction hypothesis to obtain $p_j=0$ identically.

Fix $\bmu = (\mu_1,\cdots,\mu_{m-1})\in \calC_{m-1}.$ Let $\mu_m$ be a variable, and set $\ell= \lfloor m/2 \rfloor.$ We have that $\ell=(m-1)/2$ if $m$ is odd, and $\ell=m/2$ if $m$ is even. Set $\bH=(\mu_{i+j})_{(i,j)\in [\ell]^2}.$ If $m$ is even, then $\det \bH = \alpha \mu_{m}+\beta$ for some constants $\alpha,\beta \in \BR$ determined by $\bmu,$ with $\alpha = \det (\mu_{i+j})_{(i,j)\in [\ell-1]^2} > 0.$ In the case $m$ is even, we set $t=-\beta/\alpha,$ and in the case $m$ is odd, we set $t=0.$ Then, $\bH$ is positive definite whenever $\mu_m>t.$

For each integer $k\ge 1$ and real $\varepsilon>0,$ Lemma~\ref{sf} yields a RV $X_{k,\varepsilon} \in \calR_m$ satisfying
\begin{equation}
    \delta_{k,j}(\varepsilon) := \BE[X_{k,\varepsilon}^j] - \mu_j \in (-\varepsilon,\varepsilon)
\end{equation}
for each $j\in \{1,\cdots,m-1\}$ and
\begin{equation}
    \delta_{k,m}(\varepsilon) := \BE[X_{k,\varepsilon}^m] - (t+k) \in (-\varepsilon,\varepsilon).
\end{equation}
Then, by assumption on $p,$ for every $\varepsilon>0$ and $k \in \BN_{\ge 1},$
\begin{equation}
    \sum_{j\in [d]} p_j(\mu_1+\delta_{k,1}(\varepsilon),\cdots,\mu_{m-1}+\delta_{k,m-1}(\varepsilon))(t+k+\delta_{k,m}(\varepsilon))^j = 0.
\end{equation}
Taking the limit $\varepsilon \to 0^+,$ we deduce that
\begin{equation} \label{wv}
    \sum_{j\in [d]} p_j(\mu_1,\cdots,\mu_{m-1})(t+k)^j = 0.
\end{equation}
Considering~\eqref{wv} as a univariate polynomial in $k,$ we see that its vanishing at infinitely many values of $k$ implies that
\begin{equation}
    p_j(\mu_1,\cdots,\mu_{m-1}) = 0
\end{equation}
for every $j\in [d].$ This holds for every $(\mu_1,\cdots,\mu_{m-1})\in \calC_{m-1},$ i.e., the premise of the proposition applies to each $p_j$ (namely, for every $X\in \calR_{m-1}$ we have $p_j(\BE[X],\cdots,\BE[X^{m-1}])=0$). By the induction hypothesis, we obtain $p_j = 0,$ as polynomials, for every $j\in [d].$ Therefore, $p=0,$ and the proof is complete.

\section{Proofs of Section~\ref{pc}}

\subsection{Proof of Proposition~\ref{ars}: Conditional Expectation Derivatives} \label{aad}

Recall that the conditional expectation can be expressed as
\begin{equation} \label{aab}
    \BE[Z \mid Y=y] = \frac{\BE\left[ Z e^{-(X-y)^2/2} \right]}{\BE\left[ e^{-(X-y)^2/2} \right]}
\end{equation}
for any RV $Z$ for which $Z e^{-(X-y)^2/2}$ is integrable. This formula applies for both $Z=X$ and $Z=(X-\BE[X\mid Y=y])^k,$ where $(y,k)\in \BR\times \BN,$ because they are polynomials in $X$ and the map $x\mapsto q(x)e^{-(x-y)^2/2}$ is bounded for any polynomial $q.$

Differentiating~\eqref{aab} for $Z=X$ and rearranging terms, we obtain
\begin{equation} \label{ob}
    \frac{d}{dy} ~ \BE[X \mid Y=y] = \frac{\BE\left[ (X-\BE[X \mid Y=y])^2 e^{-(X-y)^2/2} \right]}{\BE\left[ e^{-(X-y)^2/2} \right]},
\end{equation}
i.e., $f'=g_2.$ Note that $g_0 \equiv 1$ and $g_1 \equiv 0.$ Differentiating $g_r$ for $r\ge 1,$ we obtain that
\begin{equation} \label{oc}
    g_r' = g_{r+1} - r g_2 g_{r-1}.
\end{equation} 
We apply successive differentiation to $f'=g_2$ and recover patterns by utilizing~\eqref{oc} at each step. 

From $f'=g_2$ and \eqref{oc}, we infer the first few derivatives
\begin{equation} \label{oe}
    f^{(2)} = g_3, ~ f^{(3)} = g_4-3g_2^2, ~ f^{(4)} = g_5 - 10g_2g_3.
\end{equation}
We see a homogeneity in \eqref{oe}, namely, $f^{(r-1)}$ is an integer linear combination of terms of the form $g_{i_1}^{\alpha_1}\cdots g_{i_\ell}^{\alpha_\ell}$ with $i_1\alpha_1 + \cdots + i_\ell \alpha_\ell = r.$ This homogeneity can be shown to hold for a general $r$ by induction, which we show next. For most of the remainder of the proof, we forget the numerical values of the $f^{(k)}$ and the $g_r^{(k)}$ and only treat them as symbols satisfying $f'=g_2$ and $g_r'=g_{r+1}-rg_2g_{r-1}$ that respect rules of differentiation and which commute.

We call $\sum_{j=1}^\ell i_j \alpha_j$ the \textit{weighted degree} of any nonzero integer multiple of $g_{i_1}^{\alpha_1}\cdots g_{i_\ell}^{\alpha_\ell}.$ This is a well-defined degree because it is invariant to the way the product is arranged. We also say that a sum is of weighted degree $r$ if each summand is of weighted degree $r.$ To prove the claim of homogeneity, i.e., that $f^{(r-1)}$ is of weighted degree $r,$ we differentiate and apply the relation in \eqref{oc} to a generic term $g_{i_1}^{\alpha_1}\cdots g_{i_\ell}^{\alpha_\ell}$ whose weighted degree is $r.$ We have the derivative
\begin{equation} \label{of}
    \left( g_{i_1}^{\alpha_1}\cdots g_{i_\ell}^{\alpha_\ell} \right)' =  \left(g_{i_1}^{\alpha_1}\right)'\cdots g_{i_\ell}^{\alpha_\ell} + \cdots +  g_{i_1}^{\alpha_1}\cdots \left(g_{i_\ell}^{\alpha_\ell} \right)'.
\end{equation}
From \eqref{oc}, for integers $i,\alpha\ge 1,$ 
\begin{equation} \label{oh}
    \left(g_{i}^{\alpha}\right)' = \alpha g_i^{\alpha-1} g_{i+1} - \alpha i g_2 g_{i-1} g_i^{\alpha-1}.
\end{equation}
Therefore, the derivative of $g_i^{\alpha}$ has weighted degree $i\alpha +1.$ In other words, differentiation increased the weighted degree of $g_i^{\alpha}$ by $1.$ From \eqref{of}, then, we see that the weighted degree of $\left( g_{i_1}^{\alpha_1}\cdots g_{i_\ell}^{\alpha_\ell} \right)'$ is $r+1.$ Since $f'=g_2$ is of weighted degree $2,$ induction and linearity of differentiation yield that $f^{(r-1)}$ is of weighted degree $r$ for each $r\ge 2.$

Now, we fix the way we are writing products of the $g_i.$ We ignore explicitly writing $g_0$ and $g_1,$ collect identical terms into an exponent, and write lower indices first. One way to keep this notation is via integer partitions. Consider the ``homogeneous" sets
\begin{equation}
    G_r := \left\{ \sum_{\blambda \in \Pi_r} \beta_{\blambda} \bg^{\blambda} ~ ; ~ \beta_{\blambda} \in \BZ ~ \text{for each} ~ \blambda \in \Pi_r \right\}.
\end{equation}
The homogeneity property for the derivatives of $f$ can be written as $f^{(r-1)} \in G_r$ for each $r\ge 2.$ 

Next, we investigate the exact integer coefficients $h_{\blambda}$ in the expression of the derivatives of $f$ in terms of the $\bg^{\blambda}.$ Homogeneity of the derivatives of $f$ says that we may write each $f^{(r-1)},$ $r\ge 2,$ as an integer linear combination of $\{ \bg^{\blambda}\}_{\blambda \in \Pi_r}.$ One way to obtain such a representation is via repeated differentiation of $f'=g_2,$ applying the relation \eqref{oh}, and discarding any term that is a multiple of $g_1.$ Applying these steps, we arrive at representations
\begin{equation}
    f^{(r-1)} = \sum_{\blambda \in \Pi_r} h_{\blambda} \bg^{\blambda}, ~~~ c_{\blambda} \in \BZ.
\end{equation}

The terms $\bg^{\bnu}$ that appear upon differentiating a term $\bg^{\blambda}$ can be described as follows. For $(\lambda_2,\cdots,\lambda_\ell)=\blambda \in \Pi_r,$ we call $\lambda_2$ the leading term of $\blambda.$ Consider for a tuple $\blambda \in \Pi_r$ the following two sets of tuples $\tau_+(\blambda),\tau_-(\blambda) \subset \Pi_{r+1}$:
\begin{itemize}
    \item The set $\tau_+(\blambda)$ consists of all tuples obtainable from $\blambda$ via replacing a pair $(\lambda_i,\lambda_{i+1})$ with $(\lambda_i-1,\lambda_{i+1}+1)$ (so, necessarily $\lambda_i\ge 1$) while keeping all other entries unchanged;
    
    \item The set $\tau_-(\blambda)$ consists of all tuples obtainable from $\blambda$ via replacing a pair $(\lambda_{i-1},\lambda_i),$ for which $i\ge 3,$ with the pair $(\lambda_{i-1}+1,\lambda_i-1)$ (so, necessarily $\lambda_i \ge 1$) and additionally increasing the leading term by $1$ while keeping all other terms unchanged.
\end{itemize} 
For example, if $\blambda = (0,5,0,1) \in \Pi_{20}$ then
\begin{equation}
    \tau_+(\blambda) = \{ (0,4,1,1),(0,5,0,0,1) \} \subset \Pi_{21}
\end{equation}
and
\begin{equation}
    \tau_-(\blambda) = \{(2,4,0,1),(1,5,1) \} \subset \Pi_{21}.
\end{equation}
The relation \eqref{oh} yields, in view of
\begin{equation} \label{ok}
    \left( g_{2}^{\lambda_2}\cdots g_{\ell}^{\lambda_\ell} \right)' =  \left(g_{2}^{\lambda_2}\right)'\cdots g_{\ell}^{\lambda_\ell} + \cdots +  g_{2}^{\lambda_2}\cdots \left(g_{\ell}^{\lambda_\ell} \right)',
\end{equation}
that
\begin{equation} \label{ol}
    \left( \bg^{\blambda} \right)' = \sum_{\bnu \in \tau_+(\blambda)} a_{\blambda,\bnu}  \bg^{\bnu} - \sum_{\bnu \in \tau_-(\blambda)} b_{\blambda,\bnu} \bg^{\bnu}
\end{equation}
for some positive integers $a_{\blambda,\bnu}$ and $b_{\blambda,\bnu},$ which we describe next. Finding $a_{\blambda,\bnu}$ and $b_{\blambda,\bnu}$ can be straightforwardly done from \eqref{oh} in view of \eqref{ok}. If $\bnu \in \tau_+(\blambda),$ say
\begin{equation}
    (\nu_i,\nu_{i+1}) = (\lambda_i-1,\lambda_{i+1}+1),
\end{equation}
then $a_{\blambda,\bnu} = \lambda_i.$ If $\bnu \in \tau_-(\blambda),$ say
\begin{equation}
    (\nu_{i-1},\nu_{i}) = (\lambda_{i-1}+1,\lambda_{i}-1),
\end{equation}
then $b_{\blambda,\bnu} = i \lambda_i.$ In our example of $\blambda = (0,5,0,1),$ we get
\begin{align}
    a_{(0,5,0,1),(0,4,1,1)} &= 5 \\
    a_{(0,5,0,1),(0,5,0,0,1)} &= 1,
\end{align}
whereas
\begin{align}
    b_{(0,5,0,1),(2,4,0,1)} &= 15 \\
    b_{(0,5,0,1),(1,5,1)} &= 5.
\end{align}
Note that the two sets $\tau_+(\blambda)$ and $\tau_-(\blambda)$ are disjoint because, e.g., the sum of entries of a tuple in $\tau_+(\blambda)$ is the same as that for $\blambda,$ whereas the sum of entries of a tuple in $\tau_-(\blambda)$ is one more than that for $\blambda.$

We next describe how to use what we have shown thus far to deduce a recurrence relation for the  $h_{\blambda}.$ Let $\theta$ be a process inverting $\tau,$ i.e., define for $\bnu \in \Pi_{r+1}$ the two sets
\begin{equation}
    \theta_+(\bnu) := \left\{ \blambda \in \Pi_{r} ~ ; ~ \bnu \in \tau_+(\blambda) \right\}
\end{equation}
and 
\begin{equation}
    \theta_-(\bnu) := \left\{ \blambda \in \Pi_r ~ ; ~ \bnu \in \tau_-(\blambda) \right\}
\end{equation}
The two sets $\theta_+(\bnu)$ and $\theta_-(\bnu)$ are disjoint because the two sets $\tau_+(\blambda)$ and $\tau_-(\blambda)$ are disjoint for each fixed $\blambda.$ Recall our process for defining $h_{\blambda}$: we start with $f'=g_2,$ so $h_{(1)}=1$; we successively differentiate $f'=g_2$; after each differentiation, we use \eqref{oh} and \eqref{ok} (recall that we have the understanding $g_i^0=1$); we discard any ensuing multiple of $g_1$; after $r-2$ differentiations, we get an equation $f^{(r-1)}=\sum_{\blambda \in \Pi_r} h_{\blambda} \bg^{\blambda},$ which we take to be the definition of the $h_{\blambda}.$ The point here is that it could be that $f^{(r-1)}$ is representable as an integer linear combination of the $\bg^{\blambda}$ in more than one way, which can only be verified after the numerical values for the $g_i$ are taken into account, but we are not doing that: our approach treats the $g_i$ as symbols following the laid out rules. Now, we look at one of the steps of this procedure, starting at differentiating $f^{(r-1)} = \sum_{\blambda \in \Pi_r} h_{\blambda} \bg^{\blambda},$ so $f^{(r)} = \sum_{\blambda \in \Pi_r} h_{\blambda} \left( \bg^{\blambda} \right)'.$ Replacing $(\bg^{\blambda})'$ via \eqref{ol},
\begin{equation}
    f^{(r)} = \sum_{\blambda \in \Pi_r} \left( \sum_{\bnu \in \tau_+(\blambda)} a_{\blambda,\bnu}  \bg^{\bnu} - \sum_{\bnu \in \tau_-(\blambda)} b_{\blambda,\bnu} \bg^{\bnu} \right).
\end{equation}
Exchanging the order of summations (for which we use $\theta$),
\begin{equation}
    f^{(r)} = \sum_{\bnu \in \Pi_{r+1}} \left( \sum_{\blambda \in \theta_+(\bnu)} h_{\blambda} a_{\blambda,\bnu} - \sum_{\blambda \in \theta_-(\bnu)} h_{\blambda} b_{\blambda,\bnu} \right) \bg^{\bnu}.
\end{equation}
Therefore, by definition of the $h_{\blambda},$ we have the recurrence: for each $\bnu \in \Pi_{r+1}$ 
\begin{equation} \label{om}
    h_{\bnu} = \sum_{\blambda \in \theta_+(\bnu)} h_{\blambda} a_{\blambda,\bnu} - \sum_{\blambda \in \theta_-(\bnu)} h_{\blambda} b_{\blambda,\bnu}, \quad h_{(1)}=1.
\end{equation}
One instance of this recurrence is, e.g.,
\begin{equation}
    h_{(2,1)} = 3h_{(3)} - 4h_{(1,0,1)}-6h_{(0,2)}.
\end{equation}

Now, we show that the recurrence in~\eqref{om} also generates $e_{\blambda}$ as defined in~\eqref{sp}. For $(\lambda_2,\cdots,\lambda_\ell)=\blambda\in\Pi_r,$ denote $\sigma(\blambda) = \lambda_2+\cdots+\lambda_\ell.$ If $\bnu \in \tau_+(\blambda)$ then $\sigma(\bnu)=\sigma(\blambda),$ and if $\bnu\in \tau_-(\blambda)$ then $\sigma(\bnu)=\sigma(\blambda)+1.$ Therefore, $\blambda \in \theta_+(\bnu)$ implies $\sigma(\bnu)=\sigma(\blambda),$ and $\blambda\in \theta_-(\bnu)$ implies $\sigma(\bnu)=\sigma(\blambda)+1.$ Multiplying~\eqref{om} by $(-1)^{\sigma(\bnu)-1}$ yields the equivalent recurrence
\begin{equation} \label{aac}
    t_{\bnu} = \sum_{\blambda \in \theta_+(\bnu)} t_{\blambda} a_{\blambda,\bnu} + \sum_{\blambda \in \theta_-(\bnu)} t_{\blambda} b_{\blambda,\bnu}, \quad t_{(1)}=1,
\end{equation}
where $t_{\blambda}:= (-1)^{\sigma(\blambda)-1}h_{\blambda}.$ We show that $c_{\blambda}=(-1)^{\sigma(\blambda)-1}e_{\blambda}$ (see~\eqref{sp}) satisfies this recurrence, which is equivalent to $e_{\blambda}$ satisfying the recurrence~\eqref{om}. Clearly, $c_{(1)}=1,$ so consider $c_{\bnu}$ for $\bnu \in \Pi_r$ with $r\ge 3.$

Consider labelled elements $s_1,s_2,\cdots,$ and let $S_k=\{s_1,\cdots,s_k\}$ for each $k\ge 2.$ For any $\blambda\in \Pi_k,$ let $\calC_{\blambda}$ be the set of arrangements of cyclically-invariant set-partitions of $S_k$ according to $\blambda,$ so $|\calC_{\blambda}|=c_{\blambda}.$ Now, fix $\bnu \in \Pi_{r+1},$ and we will build $\calC_{\bnu}$ from the $\calC_{\blambda}$ where $\blambda$ ranges over $\theta_+(\bnu)\cup \theta_-(\bnu).$ Consider first $\blambda\in\theta_+(\bnu),$ where a partition in $\calC_{\bnu}$ is constructed from a partition in $\calC_{\blambda}$ by appending $s_{r+1}$ to one of the parts of the latter partition. Note that adding $s_{r+1}$ to two distinct partitions of $S_r$ cannot produce the same partition of $S_{r+1}$; indeed, just removing $s_{r+1}$ shows that that is impossible. Now, let $i$ be the unique index such that $(\nu_i,\nu_{i+1}) = (\lambda_i-1,\lambda_{i+1}+1).$ Then, a partition $\calP\in \calC_{\bnu}$ of $S_{r+1}$ is induced by a partition $\calP'\in\calC_{\blambda}$ of $S_r$ if and only if $s_{r+1}$ is added to a part in $\calP'$ of size $i,$ of which there are exactly $\lambda_i = a_{\blambda,\bnu}.$ Therefore, we get a contribution of $\sum_{\blambda \in \theta_+(\bnu)} c_{\blambda} a_{\blambda,\bnu}$ towards $c_{\bnu},$ which is the first part in~\eqref{aac}. 

For the second part, $\sum_{\blambda \in \theta_-(\bnu)} c_{\blambda} b_{\blambda,\bnu},$ we consider the remaining ways of generating a partition in $\calC_{\bnu}$ from a partition according to some $\blambda\in \theta_-(\bnu).$ In this case, $s_{r+1}$ is not appended to an existing part, but it is used to create a new part of size $2.$ Thus, we need to also move an element $s_j,$ $1\le j \le r,$ from a part of size at least $3$ to be combined with $s_{r+1}$ to create a new part of size $2.$ It is also clear in this case that such a procedure applied to two distinct partitions in $\calC_{\blambda}$ cannot produce the same partition in $\calC_{\bnu}.$ Let $i$ be the unique index for which $(\nu_{i-1},\nu_{i}) = (\lambda_{i-1}+1,\lambda_{i}-1).$ There are $\lambda_i$ parts to choose from, and $i$ elements to choose from once a part is chosen, so there are a total of $i\lambda_i = b_{\blambda,\bnu}$ ways to generate a partition in $\calC_{\bnu}$ from a partition in $\calC_{\blambda}.$ This gives the second sum in~\eqref{aac}, and we conclude that
\begin{equation}
    c_{\bnu} = \sum_{\blambda \in \theta_+(\bnu)} c_{\blambda} a_{\blambda,\bnu} + \sum_{\blambda \in \theta_-(\bnu)} c_{\blambda} b_{\blambda,\bnu}.
\end{equation}

Therefore, the $c_{\blambda}$ and the $t_{\blambda}$ satisfy the same recurrence, which takes the form: for $\bnu\in \Pi_{r+1}$ there are integers $\{ d_{\blambda,\bnu} \}_{\blambda\in \Pi_r}$ such that
\begin{equation} \label{aah}
    u_{\bnu} = \sum_{\blambda\in \Pi_r} d_{\blambda,\bnu} u_{\blambda}
\end{equation}
with the initial condition $u_{(1)}=1.$ Then, we can induct on $r$ to conclude that the $c_{\blambda}$ and the $t_{\blambda}$ are the same sequence. Since $\Pi_2= \{(1)\},$ we see that $c_{\blambda}=t_{\blambda}$ for every $\blambda \in \Pi_2.$ Suppose $r\ge 2$ is such that $c_{\blambda}=t_{\blambda}$ for every $\blambda \in \Pi_r.$ Hence, for every $\bnu\in \Pi_{r+1},$ we have that 
\begin{equation} \label{aai}
    \sum_{\blambda\in \Pi_r} d_{\blambda,\bnu} c_{\blambda} = \sum_{\blambda\in \Pi_r} d_{\blambda,\bnu} t_{\blambda}.
\end{equation}
Since both sequences $c_{\blambda}$ and $t_{\blambda}$ satisfy the recurrence \eqref{aah}, we obtain from~\eqref{aai} that $c_{\bnu}=t_{\bnu}$ for every $\bnu \in \Pi_{r+1}.$ Therefore, we obtain by induction that $c_{\blambda} = t_{\blambda}$ for every $\blambda\in \Pi_r$ for every $r,$ as desired.

\subsection{Proof of Theorem~\ref{if}} \label{aaf}

Fix $p\in \SD,$ suppose $X\sim p,$ and write $Y=X+N$ and $p_Y=e^{-Q}.$ First, we note that $Q'(y)$ is equal to $\mathbb{E}[N \mid Y=y].$

\begin{lemma} \label{ii}
Fix a random variable $X$ and let $Y=X+N$ where $N\sim \mathcal{N}(0,1)$ is independent of $X.$ Writing $p_Y=e^{-Q},$ we have that $ Q'(y) = \mathbb{E}[N \mid Y=y].$
\end{lemma}
\begin{proof}
We have that $p_Y(y) = \mathbb{E}[e^{-(y-X)^2/2}]/\sqrt{2\pi}.$ Differentiating this equation, we obtain that $p_Y'(y) = \mathbb{E}[(X-y)e^{-(y-X)^2/2}]/\sqrt{2\pi},$ where the exchange of differentiation and integration is warranted since $t\mapsto te^{-t^2/2}$ is bounded. Now, $Q=-\log p_Y,$ so $Q' = -p_Y'/p_Y,$ i.e.,
\begin{equation} \label{ij}
    Q'(y) = y - \frac{\mathbb{E}[Xe^{-(y-X)^2/2}]}{\mathbb{E}[e^{-(y-X)^2/2}]} = y - \BE\left[ X \mid Y = y \right].
\end{equation}
The proof is completed by substituting $X= Y-N.$
\end{proof}

In view of Lemma~\ref{ii}, that $p$ is even and non-increasing over $[0,\infty)\cap \mathrm{supp}(p)$ imply that $Q$ satisfies conditions \ref{aak}--\ref{aan} of Definition~\ref{ig}. It remains to show that property~\ref{aao} holds. To this end, we show that if $\supp(p) \subset [-M,M]$ and $\lambda=M+2,$ then for every $y>M+4$ we have that  
\begin{equation}
    1<\frac{M^2+5M+8}{2(M+2)} \le \frac{Q'(\lambda y)}{Q'(y)} \le \frac{M^2+7M+8}{4}.
\end{equation}
First, since $Q'(y)=y-\BE[X\mid Y=y]$ (see \eqref{ij}), we have the bounds $y-M \le Q'(y) \le y+M$ for every $y\in \BR.$ Therefore, $y>M$ and $\lambda>1$ imply that
\begin{equation}
    \frac{\lambda y - M}{y+M} \le \frac{Q'(\lambda y)}{Q'(y)} \le \frac{\lambda y +M}{y-M}.
\end{equation}
Further, since $y>M+4$ and $\lambda=M+2,$ we have
\begin{equation}
    \frac{M^2+5M+8}{2(M+2)} < \lambda - \frac{M(M+3)}{y+M} = \frac{\lambda y - M}{y+M}
\end{equation}
and
\begin{equation}
    \frac{\lambda y +M}{y-M} = \lambda + \frac{M(M+3)}{y-M} \le \frac{M^2+7M+8}{4}.
\end{equation}
The fact that $1<\frac{M^2+5M+8}{2(M+2)}$ follows since the discriminant of $M^2+3M+4$ is $-7<0.$ Therefore, $p_Y$ is a Freud weight.

\subsection{Proof of Inequality \texorpdfstring{\eqref{sx}}{(245)}} \label{aag}

By Lemma~\ref{ii}, 
\begin{equation}
    Q'(y)= \BE[N\mid Y=y] = y - \BE[X \mid Y=y].
\end{equation}
Therefore $X\le M$ implies that, for any constant $z \ge 0,$ we have
\begin{align}
    \int_0^1 \frac{ztQ'(zt)}{\sqrt{1-t^2}} \, dt &= \frac{\pi}{4} z^2 - z \int_0^1 \hspace{-4pt} \frac{t}{\sqrt{1-t^2}}  \frac{\mathbb{E}\left[ X e^{-(X-zt)^2/2} \right]}{ \mathbb{E}\left[e^{-(X-zt)^2/2}\right]} dt \\
    &\ge \frac{\pi}{4} z^2 - Mz.
\end{align}
We have $\pi z^2/4-Mz>n$ for $z=(2M+\sqrt{2})\sqrt{n}.$ Since $y \mapsto yQ'(y)$ is strictly increasing over $(0,\infty)$ (condition~\ref{aam} of Definition~\ref{ig}), we conclude that $a_n(Q)\le (2M+\sqrt{2})\sqrt{n}.$

\section{Proof of Proposition \ref{js}: $h_n$ Under Affine Transformations} \label{lj}

Fix $n \in \BN_{\ge 1}$ and let $N\sim \calN(0,1)$ be independent of $Y.$ From equation \eqref{ai} in Proposition \ref{hn}, for any fixed $t\ge 0$
\begin{align}
    \pp_n(\alpha Y+\beta,t) 
    &= \pp_n(\alpha Y+\beta \mid \sqrt{t} (\alpha Y + \beta) + N)  \\
    &= \pp_n(\alpha Y+\beta \mid \sqrt{\alpha^2 t}  Y + \beta\sqrt{t} + N) \\
    &= \pp_n(\alpha Y \mid \sqrt{\alpha^2 t}  Y  + N).
\end{align}
Also, from equation \eqref{aj} in Proposition \ref{hn},
\begin{equation}
    \pp_n(\alpha Y \mid \sqrt{\alpha^2 t}  Y  + N) = \alpha^2 \pp_n( Y \mid \sqrt{\alpha^2 t}  Y  + N) = \alpha^2 \pp_n(Y,\alpha^2 t).
\end{equation}
Combining these two relations, we obtain
\begin{equation}
    \pp_n(\alpha Y+\beta,t) = \alpha^2 \pp_n(Y,\alpha^2 t).
\end{equation}
Therefore, by definition of $h_n$ (see equation \eqref{jp})
\begin{align}
    h_n(\alpha Y + \beta) &= \frac12 \int_0^\infty  \pp_n(\alpha Y+\beta , t) - \frac{1}{2\pi e + t} \, dt \\
    &= \frac12 \int_0^\infty  \alpha^2 \pp_n(Y,\alpha^2 t) - \frac{1}{2\pi e + t} \, dt \\
    &= \frac12 \int_0^\infty   \left( \pp_n(Y,\alpha^2 t) - \frac{1}{2\pi e \alpha^2 + \alpha^2 t} \right) \alpha^2 dt.
\end{align}
Performing the substitution $u=\alpha^2 t,$
\begin{equation} \label{jq}
    h_n(\alpha Y+ \beta) = \frac12 \int_0^\infty \pp_n(Y,u) - \frac{1}{2\pi e \alpha^2 + u} \, du.
\end{equation}
Finally, from \eqref{hx} (with $a=1/(2\pi e \alpha^2)$ and $b=1/(2\pi e)$), we obtain
\begin{equation} \label{jr}
    \frac12 \int_0^\infty \frac{1}{2\pi e \alpha^2 + u} - \frac{1}{2\pi e + u} \, du = \frac12 \log \frac{1}{\alpha^2} = -\log |\alpha|.
\end{equation}
Note that the integrand in the right hand side of \eqref{jq} and the integrand in the left hand side of \eqref{jr} are both absolutely integrable over $(0,\infty).$ Therefore, adding \eqref{jq} and \eqref{jr}, and utilizing the definition of $h_n$ in \eqref{jp} again, we obtain
\begin{equation}
    h_n(\alpha Y+ \beta)-\log |\alpha| = \frac12 \int_0^\infty  \pp_n(Y,u) - \frac{1}{2\pi e + u} \, du = h_n(Y),
\end{equation}
and the proof is complete.

\section{Proofs of Section~\ref{jc}}

\subsection{Proof of Lemma~\ref{pk}} \label{wx}

Fix $X,$ and let $W$ be its orthogonal projection onto the closed subspace $\calV.$ Consider the symmetric matrix $\bM=\BE[\bV \bV^T].$ Since $\bc^T \bM \bc = \BE[|\bc^T \bV|^2]$ for each $\bc\in \BR^{n+1},$ linear independence of the $V_j$ implies that $\bM$ is positive-definite. Let $\bM^{1/2}$ be the unique lower-triangular matrix with positive diagonal such that $\bM^{1/2} \left(\bM^{1/2}\right)^{T} = \bM,$ and denote its inverse by $\bM^{-1/2}.$ Let $(U_0,\cdots,U_n)^T = \bU = \bM^{-1/2}\bV.$ Then, $\{U_j\}_{j\in [n]} \subset \calV$ and $\BE[\bU\bU^T] = \bI_{n+1}.$ Therefore, $\{U_j\}_{j\in [n]}$ is an orthonormal basis of $\calV$; indeed, it is the output of Gram-Schmidt orthonormalization on $\{V_j\}_{j\in [n]}.$ Then, because $W$ is the orthogonal projection of $X$ onto $\calV,$ we can express $W$ as
\begin{equation} \label{pl}
    W = \sum_{j\in [n]} \BE[XU_j] ~ U_j = \BE[X\bU^T] \bU.
\end{equation}
Plugging the defining formulas of $\bU$ and $\bM$ into \eqref{pl},
\begin{equation}
    W = \BE\left[ X \bV \right]^T \BE\left[ \bV \bV^T \right]^{-1} \bV.
\end{equation}
Finally, being the orthogonal projection of $X$ onto $\calV,$ $W$ is the unique closest element in $\calV$ to $X$; hence, equation \eqref{pj} follows.

\subsection{Proof of Corollary~\ref{pd}} \label{wy}

Fix $q\ge 1.$ For each $j,$ Carleman's condition on $Y_j$ yields that the set of polynomials in $Y_j,$ i.e., $\bigcup_{n\in \BN} \SP_n(Y_j),$ is dense in $L^{2q}(\sigma(Y_j)).$ Therefore, by Theorem \ref{pe},
\begin{equation}
    \overline{\bigcup_{n\in \BN} \SP_{n,m}(\bY)} = L^q(\sigma(\bY)).
\end{equation}
Now, fix $(f_1,\cdots,f_m)^T=\bff \in L^{q}(\BR^m,\sigma(\bY)).$ For each $j,$ $f_j \in L^{q}(\sigma(\bY)).$ Hence, there is a sequence $\{g_{j,n} \}_{n\in \BN} \subset \SP_{n,m}(\bY)$ such that $f_j = \lim_{n\to \infty} g_{j,n}$ in $L^{q}(\sigma(\bY))$-norm. Set $\bg_n = (g_{1,n},\cdots,g_{m,n})^T,$ and note that $\bg_n \in \SP_{n,m}^m(\bY).$ By definition of the norm in $L^q(\BR^m,\sigma(\bY)),$ we deduce
\begin{equation}
    \lim_{n\to \infty} \|\bff - \bg_n \|_q^q = \lim_{n\to \infty} \sum_{j=1}^m \|f_j - g_{j,n}\|_q^q = 0,
\end{equation}
and the desired denseness result follows.

\subsection{Proof of Theorem~\ref{wu}} \label{xa}

Since the $Y_j$ do not satisfy a polynomial relation, the matrix $\bM_{\bY,n}$ is invertible for each $n\in \BN.$ Further, the entries of $\bY^{(n,m)}$ are linearly independent for each $n.$ Then, by Lemma~\ref{pk}, equation~\eqref{wz} follows, i.e., $E_n[\bX\mid \bY]$ is the $\ell$-RV whose $k$-th entry is $\BE\left[X_k \bY^{(n,m)}\right]^T \bM_{\bY,n}^{-1} \bY^{(n,m)}.$ By Corollary~\ref{pd}, since each $Y_j$ satisfies Carleman's condition, the set of vectors of polynomials $\bigcup_{n\in \BN} \SP_{n,m}^m(\bY)$ is dense in $L^2(\BR^m,\sigma(\bY)).$ In particular, $\bigcup_{n\in \BN} \SP_{n,m}(\bY)$ is dense in $L^2(\sigma(\bY)).$ By Theorem~\ref{pi}, we have the $L^2(\sigma(\bY))$ limits
\begin{equation}
        \BE[X_k\mid \bY] = \lim_{n\to \infty} \BE\left[X_k \bY^{(n,m)}\right]^T \bM_{\bY,n}^{-1} \bY^{(n,m)}
\end{equation}
for each $k\in \{1,\cdots,\ell\}.$ We conclude that $E_n[\bX\mid \bY] \to \BE[\bX \mid \bY]$ in $L^2(\BR^\ell,\sigma(\bY)),$ as desired.

\subsection{Proof of Proposition~\ref{xb}} \label{xc}

Set $\bY = (Y_1,Y_2)^T.$ Equation \eqref{pm} is straightforward: since $E_n[X \mid Y_1] \in \SP_n(Y_1) \subset \SP_{n,2}(\bY),$ the projection of $E_n[X \mid Y_1]$ onto $\SP_{n,2}(\bY)$ is $E_n[X \mid Y_1]$ again. Equation \eqref{pn} also follows by an orthogonal projection argument. There is a unique representation $X=p_{1,2}+p_{1,2}^{\perp}$ for $(p_{1,2},p_{1,2}^{\perp}) \in \SP_{n,2}(\bY)\times \SP_{n,2}(\bY)^{\perp}.$ There is also a unique representation $p_{1,2}=q_2 + q_2^{\perp}$ for $(q_2,q_2^{\perp})\in \SP_n(Y_2)\times \SP_n(Y_2)^{\perp}.$ The projection of $X$ onto $\SP_{n,2}(\bY)$ is $p_{1,2},$ whose projection onto $\SP_n(Y_2)$ is $q_2,$ i.e.,
\begin{equation} \label{po}
    E_n\left[ E_n[X \mid Y_1,Y_2] \mid Y_2 \right] = q_2.
\end{equation}
Furthermore, we have the representation $X=q_2+(q_2^{\perp}+p_{1,2}^{\perp}),$ for which $(q_2,q_2^{\perp}+p_{1,2}^{\perp}) \in \SP_n(Y_2)\times \SP_n(Y_2)^{\perp}.$ Hence, the projection of $X$ onto $\SP_n(Y_2)$ is $q_2$ too, i.e.,
\begin{equation} \label{pp}
    E_n[X\mid Y_2] = q_2.
\end{equation}
From \eqref{po} and \eqref{pp} we get \eqref{pn}. Equation \eqref{pn} can also be deduced from the formula of $W:=\BE[X \mid \bY].$ Denote $\bY_2^{(n)} = (1,Y_2,\cdots,Y_2^n)^T.$ We have that
\begin{equation} \label{pq}
    W = \BE\left[ X \bY^{(n,2)} \right]^T \bM_{\bY,n}^{-1} \bY^{(n,2)}
\end{equation}
and
\begin{equation} \label{pr}
    E_n[W \mid Y_2] = \BE\left[ W \bY_2^{(n)} \right]^T \bM_{Y_2,n}^{-1} \bY_2^{(n)}.
\end{equation}
For $k\in [n],$ let $\delta(k)\in \left[ \binom{n+2}{2} - 1 \right]$ be the index of the entry in $\bY^{(n,2)}$ that equals $Y_2^k.$ Then,
\begin{equation}
    \BE\left[ Y_2^k \bY^{(n,2)} \right] = \bM_{\bY,n} \be_{\delta(k)},
\end{equation}
where $\be_0,\cdots,\be_{\binom{n+2}{2}-1}$ are the standard basis vectors of $\BR^{\binom{n+2}{2}}.$ Therefore, plugging \eqref{pq} into \eqref{pr}, we obtain
\begin{equation}
    E_n[W \mid Y_2] = \BE\left[ X \bY_2^{(n)} \right]^T \bM_{Y_2,n}^{-1} \bY_2^{(n)},
\end{equation}
which is just $E_n[X \mid Y_2],$ as desired.

\section{Multidimensional MMSE Dimension (Theorem~\ref{qm})} \label{qo}

In this proof, we will denote the Euclidean norm in $\BR^m$ by the unsubscripted $\|\wc\|,$ so if $\bN=(N_1,\cdots,N_m)^T$ then $\|\bN\| = \sqrt{N_1^2+\cdots+N_m^m}$ is a RV. 

Our proof is a straightforward extension of the proof for the one-dimensional case given in~\cite{Wu2011}. In particular, we use the concept of approximations to the identity.

\begin{definition} \label{qa}
The set of functions $\{K_\delta : \BR^m \to \BR \}_{\delta>0}$ is called \emph{an approximation to the identity} if
\begin{enumerate}[label = (\roman*)]

    \item \label{tf} $\int_{\BR^m} K_{\delta}(\by) \, d\by = A$ for each $\delta>0,$
    
    \item \label{tg} $|K_\delta(\by)| \le B \delta^{-m}$ for each $\delta>0,$ and $\by \in \BR^m,$
    
    \item \label{th} $|K_\delta(\by)| \le B \delta/\|\by\|^{m+1}$ for each $\delta >0$ and $\by \in \BR^m \setminus \{0\},$
    
\end{enumerate}
where $A$ and $B$ are constants that are independent of $\delta$ and $\by.$
\end{definition}

We apply the following theorem \cite[Theorem 2.1]{Stein2019} stating that approximations to the identity closely approximate an identity operation under convolution. Recall that a point $\bx\in \BR^m$ is called a Lebesgue point of a function $f:\BR^m \to \BR$ if
\begin{equation}
    \lim_{r\to 0^+} \frac{1}{\mu_L(B_r(\bx))} \int_{B_r(\bx)} |f(\by)-f(\bx)| \, d\by = 0,
\end{equation}
where $B_r(\bx) \subset \BR^m$ is the open ball around $\bx$ of radius $r,$ and $\mu_L$ is the Lebesgue measure in $\BR^m.$

\begin{theorem} \label{qb}
If $\{ K_\delta : \BR^m \to \BR\}_{\delta>0}$ is an approximation to the identity with $\int_{\BR^m} K_{\delta}(\by) \, d\by = A$ for every $\delta>0$ (see Definition~\ref{qa}), then for every $f\in L^1(\BR^m)$ we have that
\begin{equation}
    \lim_{\delta \to 0^+} (K_\delta \ast f)(\bx) = A\cdot f(\bx)
\end{equation}
at each Lebesgue point $\bx$ of $f.$
\end{theorem}

We use the following auxiliary result showing that the Lebesgue set is Borel.

\begin{lemma} \label{yb}
The set of Lebesgue points of an integrable Borel function is Borel.
\end{lemma}
\begin{proof}
Fix a Borel function $f\in L^1(\BR^m),$ denote the set of its Lebesgue points by $\calL,$ and define the function $F: \BR^m \times (0,\infty) \to [0,\infty)$ by
\begin{equation}
    F(\bx,r) := \frac{1}{r^m} \int_{B_r(\bx)} |f(\by)-f(\bx)| \, d\by.
\end{equation}
Then $\calL=\{\bx \in \BR^m ~ ; ~ \lim_{r\to 0^+} F(\bx,r) = 0 \}.$ By definition of limits, the set $\calL$ might be rewritten as
\begin{equation}
    \calL = \bigcap_{k \in \BN} \bigcup_{\delta \in \BQ_{>0}} \bigcap_{r\in (0,\delta)} \left\{ \bx \in \BR^m ~ ; ~ F(\bx,r) \le 2^{-k} \right\}.
\end{equation}
Now, we show that the uncountable intersection over $r\in (0,\delta)$ may be replaced with a countable intersection over only the rational points $r \in  \BQ \cap (0,\delta).$ For this, it will suffice to show that, for each fixed $\bx\in \BR^m,$ the function $r\mapsto F(\bx,r)$ is continuous over $r\in (0,\infty).$ Clearly, the function $r\mapsto r^{-m}$ is continuous. In addition, by the monotone convergence theorem, we obtain the continuity of the function
\begin{equation}
    r \mapsto \int_{B_r(\bx)} |f(\by)-f(\bx)| \, d\by.
\end{equation}
Therefore, $r\mapsto F(\bx,r)$ is continuous too. Thus, if $F(\bx,r) \le 2^{-k}$ for every $r\in \BQ \cap (0,\delta),$ the density of $\BQ$ in $\BR$ implies that $F(\bx,r) \le 2^{-k}$ for every $r\in (0,\delta).$ Hence, we may write
\begin{equation}
    \calL = \bigcap_{k \in \BN} \bigcup_{\delta \in \BQ_{>0}} \bigcap_{r\in \BQ \cap (0,\delta)} \left\{ \bx \in \BR^m ~ ; ~ F(\bx,r) \le 2^{-k} \right\}.
\end{equation}
Finally, we show that each set $\{\bx \in \BR^m ~ ; ~ F(\bx,r) \le 2^{-k}\}$ is a Borel set (for fixed $r>0$ and $k\in \BN$) by showing that $\bx \mapsto F(\bx,r)$ is a Borel function. Indeed, we may write $F(\bx,r) = \int_{\BR^m} g(\bx,\by) \, d\nu(\by)$ where $\nu$ is the restriction of the Lebesgue measure of $\BR^{m}$ to the Borel subsets, and
\begin{equation}
    g(\bx,\by) := \frac{1}{r^m} |f(\by)-f(\bx)| \cdot 1_{B_r(\bx)}(\by).
\end{equation}
Note that $1_{B_r(\bx)}(\by) = 1_A(\bx,\by),$ where $A=\{(\bx,\by) \in \BR^{2m} ~ ; ~ \|\bx-\by\|_2 < r \}$ is a Borel set. Thus, $g$ is a Borel function. Therefore, by the Fubini-Tonelli theorem~\cite[Section 2.5, Theorem 14]{Adams} $\bx \mapsto \int_{\BR^m} g(\bx,\by) \, d\nu (\by)$ is a Borel function, and the proof is complete.
\end{proof}

\subsection{Proof of Theorem~\ref{qm}}

Denote $\delta = 1/\sqrt{t},$ and consider the $m$-RVs $\bY_{\delta} := \bX+\delta \bN.$ We have that
\begin{align}
    \mm(\bX \mid \sqrt{t}\bX+\bN) &=  \mm(\bX \mid \bX+ \delta \bN) = \left\| \bX - \BE\left[ \bX \mid \bY_\delta \right] \right\|_2^2 = \left\| \bX - \BE\left[ \bY_\delta - \delta \bN \mid \bY_{\delta} \right] \right\|_2^2 \\
    &= \delta^2 \left\|\bN - \BE\left[ \bN \mid \bY_{\delta} \right] \right\|_2^2 = \delta^2 ~ \mm\left( \bN \mid \bY_{\delta} \right).
\end{align}
Therefore, the statement of the theorem is equivalent to
\begin{equation} \label{tc}
    \lim_{\delta \to 0^+} \mm(\bN \mid \bY_\delta ) = \mathrm{tr}~\Sigma_{\bN}.
\end{equation}
To show that~\eqref{tc} holds, it suffices to show that
\begin{equation} \label{pt}
    \liminf_{\delta \to 0^+} ~ \mm(\bN \mid \bY_\delta ) \ge \mathrm{tr}~\Sigma_{\bN}.
\end{equation}
Indeed, the constant estimator $\psi(\by) = \BE[\bN]$ attains an error of $\mathrm{tr}~\Sigma_{\bN}$ when estimating $\bN$ given $\bY_{\delta},$ hence
\begin{equation}
    \limsup_{\delta \to 0^+} ~ \mm(\bN \mid \bY_\delta ) \le \mathrm{tr}~\Sigma_{\bN}.
\end{equation}
Fix $\varepsilon>0.$ By the square-integrability assumption on $\bN,$ there is an $M \ge 0$ such that 
\begin{equation} \label{td}
    \left\| \bN 1_{(M,\infty)}(\|\bN\|) \right\|_2 \le \varepsilon.
\end{equation}
By the triangle inequality in $L^2(\BR^m,\calF),$
\begin{align}
    \mm\left(\bN \mid \bY_{\delta}\right)^{1/2} &= \left\| \bN - \BE\left[ \bN \mid \bY_{\delta} \right] \right\|_2 \\
    &\ge \left\| \bN - \BE\left[ \bN 1_{[0,M]}(\|\bN\|) \mid \bY_{\delta} \right] \right\|_2 - \left\| \BE\left[ \bN 1_{(M,\infty)}(\|\bN\|) \mid \bY_{\delta} \right]\right\|_2. \label{qe}
\end{align}
Now, since $\left\| \BE\left[ \bN 1_{(M,\infty)}(\|\bN\|) \mid \bY_{\delta} \right]\right\|_2 \le  \left\| \bN 1_{(M,\infty)}(\|\bN\|) \right\|_2,$ we obtain from~\eqref{td} that
\begin{equation}
    \mm\left(\bN \mid \bY_{\delta}\right)^{1/2} \ge \left\| \bN - \BE\left[ \bN 1_{[0,M]}(\|\bN\|) \mid \bY_{\delta} \right] \right\|_2 - \varepsilon.
\end{equation}
Therefore, to obtain~\eqref{pt}, it suffices to show that
\begin{equation} \label{te}
    \liminf_{\delta \to 0^+} \left\| \bN - \BE\left[ \bN 1_{[0,M]}(\|\bN\|) \mid \bY_{\delta} \right] \right\|_2 \ge \sqrt{\mathrm{tr}~\Sigma_{\bN}}
\end{equation}
since $\varepsilon$ is arbitrary. Now, consider the function $\bff:\BR^m \times \BR^m \times (0,\infty) \to \BR^m$ given by
\begin{equation}
    \bff(\bx,\bz;\delta) := \BE\left[ \bN ~ 1_{[0,M]}(\|\bN\|) \mid \bY_{\delta} = \bx + \delta \bz \right],
\end{equation}
and note that we have
\begin{equation}
    \bff(\bX,\bN;\delta) = \BE\left[ \bN ~ 1_{[0,M]}(\|\bN\|) \mid \bY_{\delta}  \right].
\end{equation}
We will show that
\begin{equation} \label{pu}
    \lim_{\delta \to 0^+} \bff(\bX,\bN;\delta) = \BE\left[ \bN ~ 1_{[0,M]}(\|\bN\|) \right]
\end{equation}
almost surely. The limit in~\eqref{pu} suffices to conclude the proof of the theorem, because it implies inequality~\eqref{te} via Fatou's lemma. Indeed, from~\eqref{pu} we obtain
\begin{equation}
    \liminf_{\delta \to 0^+} \left\| \bN - \bff(\bX,\bN;\delta) \right\|_2 \ge \left\| \bN - \BE\left[ \bN ~ 1_{[0,M]}(\|\bN\|) \right] \right\|_2 \ge \left\| \bN - \BE\left[ \bN  \right] \right\|_2 = \sqrt{\mathrm{tr}~\Sigma_{\bN}},
\end{equation}
where the penultimate step follows because $\BE[\bN]$ is the best constant estimator of $\bN$ in the mean-square sense. Thus, the rest of the proof is devoted to showing that the limit in~\eqref{pu} holds.

Note that
\begin{equation}
    p_{\bY_{\delta}}(\by) = \frac{1}{\delta^m} \BE\left[ p_{\bN}\left( \frac{1}{\delta} \left( \by-\bX \right) \right) \right] = \BE\left[ p_{\bX}\left( \by - \delta \bN  \right) \right].
\end{equation}
We have the conditional expectation formulas
\begin{align}
    \BE\left[ \bX \mid \bY_\delta = \by \right] &= \BE\left[ \bX ~ p_{\bN}\left( \frac{1}{\delta} \left( \by-\bX \right) \right)\right]  \cdot \frac{1}{\BE\left[ p_{\bN}\left( \frac{1}{\delta} \left( \by-\bX \right) \right) \right]}, \\
    \BE\left[ \bN \mid \bY_\delta = \by \right] &= \BE\left[ \bN ~ p_{\bX}\left(\by - \delta \bN \right) \right] \cdot \frac{1}{\BE\left[ p_{\bX}\left( \by - \delta \bN  \right) \right]}, \\
    \bff(\bx,\bz;\delta) &= \BE\left[ \bN ~ 1_{[0,M]}(\|\bN\|) ~ p_{\bX}(\bx+\delta\bz-\delta \bN) \right] \cdot \frac{1}{p_{\bY_{\delta}}(\bx+\delta\bz)}.
\end{align}
Define the function $\bg:\BR^m \times \BR^m \times (0,\infty) \to \BR^m$
\begin{align}
    \bg(\bx,\bz;\delta) &:=  \BE\left[ \bN ~ 1_{[0,M]}(\|\bN\|) ~ p_{\bX}(\bx+\delta\bz-\delta \bN) \right],
\end{align}
and note that
\begin{equation} \label{tl}
    \bff(\bx,\bz;\delta) = \bg(\bx,\bz;\delta) \cdot  \frac{1}{p_{\bY_\delta}(\bx+\delta\bz)}.
\end{equation}
Let $\calL \subset \BR^m$ be the Lebesgue set of $p_{\bX},$ $\calS := \{ \bx\in \BR^m ~ ; ~ p_{\bX}(\bx)>0 \},$ and $\calT:= \calL \cap \calS.$ We will show that the limit in \eqref{pu} holds almost surely by showing that the following two limits hold for every $(\bx,\bz) \in \calT \times \BR^m$
\begin{align}
    \lim_{\delta \to 0^+} \bg(\bx, \bz;\delta) &= \BE\left[ \bN ~ 1_{[0,M]}(\|\bN\|) \right] p_{\bX}(\bx) \label{pv}, \\
    \lim_{\delta \to 0^+} p_{\bY_{\delta}}(\bx + \delta \bz) &= p_{\bX}(\bx) \label{pw}.
\end{align}
We describe first how \eqref{pv} and \eqref{pw} imply \eqref{pu}. We have that $P_{\bX}(\calS)=1$ because $\calS$ is a Borel set and
\begin{equation}
    P_{\bX}(\calS^c) = \int_{\calS^c} p_{\bX}(\bx) \, d\bx = \int_{\calS^c} 0 \, d\bx = 0.
\end{equation}
By the Lebesgue differentiation theorem, since $p_{\bX}\in L^1(\BR^m),$ the complement $\calL^c$ has zero Lebesgue measure. Further, $\calL$ is a Borel set by Lemma~\ref{yb}, since $p_{\bX}$ is a Borel function. Thus, as $P_{\bX}$ is absolutely continuous with respect to the Lebesgue measure, we deduce that $P_{\bX}(\calL^c)=0,$ i.e., $P_{\bX}(\calL)=1.$ Therefore, $P_{\bX}(\calT)=1,$ from which we obtain $P_{\bX,\bN}(\calT\times \BR^m) = P_{\bX}(\calT)=1.$ In other words, if~\eqref{pv} and~\eqref{pw} hold for every $(\bx,\bz)\in \calT \times \BR^m$ then~\eqref{pu} holds almost surely. We now focus on showing~\eqref{pv} and~\eqref{pw}.

Fix $\bz \in \BR^m,$ and denote $\bff=(f_1,\cdots,f_m)^T$ and $\bg=(g_1,\cdots,g_m)^T.$ We apply Theorem~\ref{qb}. The following convolution relations hold for every $j \in \{1,\cdots,\ell\}$
\begin{align}
    g_j(\bx,\bz;\delta) &= (p_{\bX} \ast G_{\delta,j})(\bx) \label{py} \\
    p_{\bY_{\delta}}(\bx+\delta \bz) &= (p_{\bX} \ast K_\delta)(\bx), \label{pz}
\end{align}
where the functions $G_{\delta,j}$ and $K_{\delta}$ are defined by
\begin{align}
    G_{\delta,j}(\by) &:= (\delta^{-1}y_j+z_j)1_{[0,M]}  \left( \left\| \delta^{-1} \by + \bz \right\| \right)  p_{\bN} \left( \delta^{-1} \by + \bz \right) \delta^{-m} \label{qc} \\
    K_{\delta}(\by) &:= \delta^{-m} p_{\bN} \left( \delta^{-1} \by + \bz \right). \label{qd}
\end{align}
Indeed, denoting $q_j(\bu)= u_j 1_{[0,M]}  \left( \left\| \bu \right\| \right),$ and using the substitution $\by = \delta(\bu-\bz),$ we may compute
\begin{align}
    (p_{\bX} \ast G_{\delta,j})(\bx) & \hspace{-0.5mm} = \int_{\BR^m}  p_{\bX}(\bx - \by) q_j(\delta^{-1} \by + \bz) p_{\bN} \left( \delta^{-1} \by + \bz \right) \delta^{-m}  d\by \\
    &= \int_{\BR^m} p_{\bX}(x-\delta(\bu-\bz)) q_j(\bu) p_{\bN}(\bu)  \, d\bw = \BE\left[ p_{\bX}(\bx - \delta(\bN-\bz)) q_j(\bN)  \right] = g_j(\bx,\bz;\delta). \label{px} 
\end{align}
Also, noting that
\begin{equation} \label{ti}
    K_\delta(\by) = \delta^{-m} p_{\bN}(\delta^{-1} \by + \bz) = p_{\delta(\bN-\bz)}(\by),
\end{equation}
and by independence of $\bX$ and $\delta(\bN-\bz),$ we obtain the convolution
\begin{align}
    (p_{\bX} \ast K_\delta)(\bx) = (p_{\bX} \ast p_{\delta (\bN-\bz)} )(\bx)  = p_{\bX + \delta(\bN-\bz)}(\bx) = p_{\bY_{\delta}}(\bx+\delta \bz).
\end{align}
Therefore, equations \eqref{py} and \eqref{pz} hold. Fix $j\in \{1,\cdots,m\}.$ We show next that $\{G_{\delta,j}\}_{\delta>0}$ and $\{K_\delta\}_{\delta>0}$ are approximations to the identity. For condition~\ref{tf} of Definition~\ref{qa}, we note that the substitution $\by= \delta(\bu-\bz)$ implies that
\begin{equation}
    \int_{\BR^m} G_{\delta,j}(\by) \, d\by = \BE\left[ N_j 1_{[0,M]}(\|\bN\|) \right],
\end{equation}
which is a constant independent of $\delta.$ Also, since $K_\delta$ is the PDF of $\delta(\bN-\bz)$ (see~\eqref{ti}), we have
\begin{equation}
    \int_{\BR^m} K_\delta(\by) \, d\by = 1.
\end{equation}
Therefore, condition~\ref{tf} in Definition \ref{qa} is satisfied by both the $G_{\delta,j}$ and the $K_\delta.$ We now show that the bounds in conditions~\ref{tg} and~\ref{th} of Definition~\ref{qa} hold with the constant $B$ chosen as $B := (1+\|\bz\|)^{m+1}C$ where $C>0$ is any constant such that 
\begin{equation} \label{tk}
    p_{\bN}(\bu) \le \frac{C}{1+\|\bu\|^{m+2}}
\end{equation}
holds for every $\bu\in \BR^m.$ Note that such a constant $C$ exists by assumption~\eqref{tj} on the decay of $p_{\bN}.$

Since $v^\beta \le 1+v^{m+2}$ for every $v\ge 0$ and $\beta \in [0,m+2],$ inequality~\eqref{tk} implies the uniform bounds
\begin{equation}
    \|\bu\|^{\beta}p_{\bN}(\bu) \le C 
\end{equation}
for every $\bu\in \BR^m$ and $\beta \in [0,m+2].$ Denote
\begin{equation}
    S_\beta := \sup_{\bu \in \BR^m} \|\bu\|^{\beta}p_{\bN}(\bu),
\end{equation}
so $S_\beta \le C \le B$ for every $0\le \beta \le m+2.$ For condition~\ref{tg}, we have that for every $\delta>0$ and $\by\in \BR^m,$ 
\begin{equation}
    \delta^m |G_{\delta,j}(\by)| \le \left| \delta^{-1} y_j + z_j \right| p_{\bN}(\delta^{-1}\by+\bz) \le \left\| \delta^{-1} \by + \bz \right\| p_{\bN}(\delta^{-1}\by+\bz) \le S_1 \le C \le B
\end{equation}
and
\begin{equation}
    \delta^m K_\delta(\by) = p_{\bN}(\delta^{-1}\by+\bz) \le S_0 \le C \le B.
\end{equation}
Therefore, condition~\ref{tg} of Definition \ref{qa} is satisfied by both the $G_{\delta,j}$ and the $K_\delta.$ Finally, for condition~\ref{th}, noting that
\begin{equation}
    \left\| \delta^{-1} \by \right\| \le \left\| \delta^{-1}\by + \bz \right\| + \left\| \bz \right\|,
\end{equation}
we have the inequalities
\begin{align}
    \frac{\|\by\|^{m+1}}{\delta} |G_\delta(\by)| &\le \left\| \delta^{-1} \by \right\|^{m+1} \left| \delta^{-1} y_j + z_j \right| p_{\bN}(\delta^{-1}\by+\bz)  \\
    &\le \left( \left\| \delta^{-1}\by + \bz \right\| + \left\| \bz \right\| \right)^{m+1} \left\| \delta^{-1} \by+\bz \right\| p_{\bN}(\delta^{-1}\by+\bz)  \\
    &\le \left( S_{m+2}^{1/(m+1)} +\|\bz\|S_1^{1/(m+1)} \right)^{m+1} \le (1+\|\bz\|)^{m+1}C = B
\end{align}
and
\begin{align}
    \frac{\|\by\|^{m+1}}{\delta} K_\delta(\by) = \left\| \delta^{-1} \by \right\|^{m+1} p_{\bN}(\delta^{-1} \by + \bz) \le \left( S_{m+1}^{1/(m+1)} + \|\bz\| S_0^{1/(m+1)} \right)^{m+1} \le (1+\|\bz\|)^{m+1}C = B
\end{align}
for every $\delta > 0$ and $\by \in \BR^m.$ Therefore, condition~\ref{th} in Definition \ref{qa} is satisfied by both the $G_{\delta,j}$ and the $K_\delta.$ In other words, each of $\{G_{\delta,j}\}_{\delta>0}$ and $\{K_\delta\}_{\delta>0}$ is an approximation to the identity.

Therefore, by Theorem \ref{qb}, for each $j\in \{1,\cdots,m\}$ and $(\bx,\bz)\in \calL\times \BR^m,$ we have that
\begin{equation}
    \lim_{\delta \to 0^+} g_j(\bx,\bz;\delta) = \BE\left[ N_j 1_{[0,M]}(\|\bN\|) \right] p_{\bX}(\bx)
\end{equation}
and
\begin{equation}
    \lim_{\delta \to 0^+} p_{\bY_{\delta}}(\bx+\delta \bz) = p_{\bX}(\bx).
\end{equation}
Hence, by~\eqref{tl},  we obtain that for every $(\bx,\bz)\in \calT \times \BR^m$
\begin{equation}
    \lim_{\delta \to 0^+} \bff(\bX,\bN;\delta) = \BE\left[ \bN ~ 1_{[0,M]}(\|\bN\|) \right],
\end{equation}
which is the desired limit~\eqref{pu}, completing the proof.

\section{Proofs of Subsection~\ref{ub}} \label{bx}

\subsection{Proof of Theorem~\ref{bp}: Consistency of the Differential Entropy Estimator} \label{ue}

We use the formula for $h_n$ given in Lemma~\ref{tn},
\begin{equation} \label{tp}
    h_n(X) = \frac12 \log \left( 2\pi e \left( \frac{\det \bM_{X,n}}{\det \bM_{N,n}} \right)^{1/d_n} \right) + \int_0^\infty \rho_{X,n}(t) \, dt,
\end{equation}
where $d_n = \binom{n+1}{2}$ and $N\sim \calN(0,1).$ We may assume that $N$ is independent of $X$ and the $X_j.$ For each $m\in \BN,$ let $\calS_m := \{X_j\}_{j\in [m]},$ and consider the sequence $\{ U_m \sim \mathrm{Unif}(\calS_m) \}_{m\in \mathbb{N}}.$ For each $m\in \BN,$ let $\frakE_m$ be the event that $X_0,\cdots,X_m$ are distinct, and let $\frakE$ be the event that the $X_j,$ for $j\in \BN,$ are all distinct. Whenever $m\ge n$ and $\frakE_m$ occurs, we have by Definition~\ref{to} of $\widehat{h}_n$ and formula~\eqref{tp} for $h_n$ the following estimate
\begin{equation} \label{uk}
    \widehat{h}_n\left( \calS_m \right) = \frac12 \log \left( 2\pi e \left( \frac{\det \bM_{U_m,n}}{\det \bM_{N,n}} \right)^{1/d_n} \right) + \int_0^\infty \rho_{U_m,n}(t) \, dt.
\end{equation}
Since $X$ is continuous, we have that $P(\frakE_m)=1$ for every $m\in \BN.$ Further, $\frakE_0 \supset \frakE_1 \supset \cdots$ and $\frakE = \bigcap_{m\in \BN} \frakE_m,$ hence $P(\frakE)=1.$ Therefore, for the purpose of proving the almost-sure limit $\widehat{h}_n\left( \calS_m \right)\to h_n(X),$ we may assume that $\frakE$ occurs. We first treat convergence of the integral part. We show that the integral part is a continuous function of the moments, then the continuous mapping theorem yields that
\begin{equation}\label{eh}
    \int_0^\infty \rho_{U_m,n}(t) \, dt \to  \int_0^\infty \rho_{X,n}(t) \, dt
\end{equation}
almost surely as $m\to \infty$ because sample moments converge almost surely to the moments. A similar method is then applied to the convergence of the $\log \det \bM_{X,n}$ part. 

We fix $n\in \BN_{\ge 1},$ and assume $m\ge n$ throughout the proof. We use the following notation. The $2n$-RV $\bmu^{(m)}$ consists of the first $2n$ moments of $U_m$
\begin{equation}
    \bmu^{(m)} := \left( \frac{\sum_{j=0}^m X_j}{m+1}, \cdots, \frac{\sum_{j=0}^m X_j^{2n}}{m+1} \right)^T.
\end{equation}
Let $\mu^{(m)}_k$ be the $k$-th coordinate of $\bmu^{(m)},$ so $\bmu^{(m)}=\left( \mu_1^{(m)},\cdots,\mu_{2n}^{(m)}\right)^T.$ We write $\calX_k := \BE\left[ X^k \right]$ for $k\in \BN,$ and consider the constant vector
\begin{equation}
    \bcalX := \left( \calX_k \right)_{1\le k \le 2n}.
\end{equation}
By the strong law of large numbers, we have the almost-sure convergence $\mu_k^{(m)} \to \calX_k$ for each $1\le k \le 2n.$ Then, $\bmu^{(m)} \to \bcalX$ almost surely as $m\to \infty.$ We show next that the function $\bcalX \mapsto \int_0^\infty \rho_{X,n}(t) \, dt$ is continuous.

By definition of $\rho_{X,n}$ (see~\eqref{tm}), there are polynomials $A_1,\cdots,A_{d_n-2}$ and $B_1,\cdots,B_{d_n}$ in $2n$ variables such that 
\begin{equation} \label{ul}
    \rho_{X,n}(t) = \frac{\sum_{j=1}^{d_n-2} A_j(\bcalX) ~ t^j}{c_n + \sum_{j=1}^{d_n} B_j(\bcalX) ~ t^j}
\end{equation}
where $c_n:= \prod_{k=1}^n k!$ (we are subsuming the $1/2$ factor in~\eqref{tm} in the numerator, so we have the equality $\delta_{X,n}(t) = c_n+ \sum_{j=1}^{d_n} B_j(\bcalX)t^j$). Being polynomials, each of the $A_j$ and the $B_\ell$ is continuous over $\mathbb{R}^{2n}.$ Then, by the continuous mapping theorem, we have the almost-sure convergences
\begin{equation}
    A_j\left({\bm \mu}^{(m)}\right) \to  A_j(\bcalX) \quad  \text{and} \quad B_\ell\left({\bm \mu}^{(m)}\right) \to  B_\ell(\bcalX) 
\end{equation}
as $m\to \infty$ for each $1 \le j \le d_n-2$ and $1\le \ell \le d_n.$ Denote
\begin{align}
    \bA(\bcalX) &:= \left( A_j (\bcalX) \right)_{1\le j \le d_n-2}, \\
    \bB(\bcalX) &:= \left( B_j (\bcalX) \right)_{1\le j\le d_n}.
\end{align}
We show next that the there is an open set $\calO\subset \BR^{d_n}$ containing the point $\bB(\bcalX)$ such that the mapping $f:\mathbb{R}^{d_n-2} \times \calO \to \mathbb{R}$ defined by
\begin{equation} \label{ek}
    f(p_1,\cdots,p_{d_n-2},q_1,\cdots,q_{d_n}) := \int_0^\infty \frac{\sum_{j=1}^{d_n-2} p_jt^j}{c_n + \sum_{j=1}^{d_n} q_jt^j} \, dt
\end{equation}
is continuous at the point $(\bA(\bcalX), \, \bB(\bcalX)).$ To this end, we shall show first that the mapping in (\ref{ek}) is well-defined on an open neighborhood of $(\bA(\bcalX), \, \bB(\bcalX)).$ In other words, the denominator of the integrand $t\mapsto c_n+\sum_{j=1}^{d_n} q_j t^j$ cannot have a root $t\in [0,\infty)$ for any $\bq \in \calO,$ and the rational function integrand has to be integrable. For integrability, we will restrict the set $\calO$ to contain only points having $q_{d_n}>0,$ so showing that the integrand's denominator is strictly positive over $t\in [0,\infty)$ will be enough to deduce integrability in~\eqref{ek}.

We consider the subset $\calG \subset \mathbb{R}^{d_n}$ defined by
\begin{equation} \label{tu}
    \calG := \left\{ \bg \in \mathbb{R}^{d_n} ~ ; ~ g_{d_n}>0, \text{ and } \sum_{\ell=1}^{d_n}  g_{j} t^j > - c_n \; \text{ for every } t\ge 0    \right\}
\end{equation}
where in this definition and the subsequent argument we set ${\bm g} = \left(g_1,\cdots,g_{d_n}\right)^T.$ Note that $\bB(\bcalX)\in \calG.$ Indeed, since $X$ is continuous, $B_{d_n}(\bcalX)=\det \bM_{X,n}>0$; similarly, for every $t \in [0,\infty),$ continuity of $\sqrt{t}X+N$ implies that $\det \bM_{\sqrt{t}X+N}>0$ (recall that $c_n+\sum_{j=1}^{d_n} B_j(\bcalX) t^j = \det \bM_{\sqrt{t}X+N}$). We show that $\calG$ is an open set. Fix $\bg \in \calG$ and $\varepsilon_1 \in \left( 0 , g_{d_n} \right).$ We have that the polynomial $\sum_{j=1}^{d_n} (g_j- \varepsilon_1)t^j$ is eventually increasing and approaches infinity as $t\to \infty.$ Let $t_0>1$ be such that for every $t>t_0$ we have
\begin{equation} \label{el}
    \sum_{\ell=1}^{d_n} (g_j- \varepsilon_1)t^j > -c_n.
\end{equation}
Being continuous, the polynomial $\sum_{j=1}^{d_n} g_j t^j$ attains its minimum over the compact set $[0,t_0].$ Let $s$ denote this minimum, and note that $s>-c_n.$ Let $\varepsilon \in (0,1)$ be defined by
\begin{equation}
    \varepsilon := \frac12 \min\left( \varepsilon_1, \frac{(s+c_n)(t_0-1)}{t_0(t_0^{d_n}-1)} \right).
\end{equation}
As $\varepsilon<\varepsilon_1,$ inequality (\ref{el}) yields that for every $t>t_0$
\begin{equation} \label{ts}
    \sum_{j=1}^{d_n} (g_j-\varepsilon)t^j > -c_n.
\end{equation}
In addition, for any $t\in [0,t_0],$ 
\begin{align}
    \sum_{j=1}^{d_n} (g_j-\varepsilon)t^j =  \sum_{j=1}^{d_n} g_j t^j - \varepsilon \sum_{j=1}^{d_n}  t^j \ge s -\varepsilon \sum_{j=1}^{d_n}  t_0^j > s-  \frac{(s+c_n)(t_0-1)}{t_0(t_0^{d_n}-1)}\sum_{j=1}^{d_n}  t_0^j = s- (s+c_n) = -c_n. \label{tt}
\end{align}
Thus, combining~\eqref{ts} and~\eqref{tt} we obtain
\begin{equation}
     \sum_{j=1}^{d_n} (g_j-\varepsilon)t^j  > -c_n
\end{equation}
for every $t\in [0,\infty).$ Hence, for any $(\delta_j)_{1\le j \le d_n} =: {\bm \delta} \in \mathbb{R}^{d_n}$ such that $\|{\bm \delta}\|_2<\varepsilon,$ we have that for all $t \in [0,\infty)$
\begin{equation}
    \sum_{j=1}^{d_n} (g_j - \delta_j) t^j \ge \sum_{j=1}^{d_n} (g_j - \|{\bm \delta}\|_2) t^j \ge \sum_{j=1}^{d_n} (g_j - \varepsilon) t^j > -c_n.
\end{equation}
In other words, the open ball $\{ \bq \in \BR^{2n} ~ ; ~ \|\bq-\bg\| < \varepsilon \}$ lies within $\calG.$ This completes the proof that $\calG$ is open. Then, the function $f$ given by~\eqref{ek} is well-defined on the open set $\BR^{d_n-2}\times \calG.$ We will replace $\calG$ with an open box $\calO \subset \calG$ to simplify the notation for the proof of continuity of $f.$ 

By openness of $\calG,$ there is an $\eta_1 \in (0,B_{d_n}(\bcalX))$ such that the open box
\begin{equation}
    \calO_1 := \prod_{j =1}^{d_n} \left( B_j(\bcalX)- \eta_1, B_j(\bcalX) + \eta_1 \right) \subset \calG
\end{equation}
contains $\bB(\bcalX).$ Since $\calO_1 \subset \calG,$ we have by the definition of $\calG$ in~\eqref{tu} that for any $\bg \in \mathcal{O}_1$ the lower bound
\begin{equation}
     c_n+ \sum_{\ell=1}^{d_n}  g_\ell t^\ell>0
\end{equation}
holds for every $t\ge 0.$ In particular, with $\eta:= \eta_1/2,$ the set 
\begin{equation}
    \calO := \prod_{j =1}^{d_n} \left( B_j(\bcalX)- \eta, B_j(\bcalX) + \eta \right) \subset \calO_1 \subset \calG
\end{equation}
is an open set containing $\bB(\bcalX),$ and the point $(B_j(\bcalX)-\eta)_{1\le j \le 2n}$ lies inside $\calG.$ Then, the function $f: \mathbb{R}^{d_n-2}\times \mathcal{O} \to \mathbb{R}$ given by~\eqref{ek} is well-defined, and for any $\bg\in \calO$ we have the lower bound (over $t\in [0,\infty)$)
\begin{equation} \label{tw}
    c_n+ \sum_{\ell=1}^{d_n}  g_\ell t^\ell \ge c_n+ \sum_{\ell=1}^{d_n}  (B_\ell(\bcalX)-\eta) t^\ell > 0.
\end{equation}
From~\eqref{tw}, Lebesgue's dominated convergence shows continuity of $f$ at $(\bA(\bcalX), \bB(\bcalX)),$ as follows.

Let ${\bm w} := ({\bm u},{\bm v}) \in  \mathbb{R}^{d_n-2} \times \mathcal{O}$ be such that $\|{\bm w}\|_2<\eta.$ The integrand in $f$ at $(\bA(\bcalX), \bB(\bcalX)) - ({\bm u},{\bm v})$ may be bounded as
\begin{align}
    \left| \frac{\sum_{j=1}^{d_n-2} (A_j(\bcalX)-u_j) t^j}{c_n+\sum_{\ell=1}^{d_n} (B_\ell(\bcalX)-v_\ell) t^\ell} \right| =  \frac{\left|\sum_{j=1}^{d_n-2} (A_j(\bcalX)-u_j) t^j \right|}{c_n+\sum_{\ell=1}^{d_n} (B_\ell(\bcalX)-v_\ell) t^\ell} \le \frac{\sum_{j=1}^{d_n-2} (|A_j(\bcalX)|+\eta) t^j}{c_n+\sum_{\ell=1}^{d_n} (B_\ell(\bcalX)-\eta) t^\ell}. \label{tv}
\end{align}
The bound in~\eqref{tv} is uniform in $\bw,$ and the upper bound is integrable over $[0,\infty)$ as the denominator's degree exceeds that of the numerator by at least $2$ and the denominator is strictly positive by~\eqref{tw}. Hence, by Lebesgue's dominated convergence
\begin{equation}
    \lim_{\|{\bm w}\| \to 0} f\left( (\bA(\bcalX), \bB(\bcalX)) -{\bm w} \right) = f\left( \bA(\bcalX), \bB(\bcalX) \right),
\end{equation}
i.e., $f$ is continuous at $(\bA(\bcalX), \bB(\bcalX)),$ as desired. Denote
\begin{align}
    \bA^{(m)} &:= \left( A_j({\bm \mu}^{(m)}) \right)_{1\le j \le d_n-2}, \\
    \bB^{(m)} &:= \left( B_\ell({\bm \mu}^{(m)}) \right)_{1\le \ell \le d_n}.
\end{align}
We have the formulas
\begin{equation}
     f(\bA^{(m)},\bB^{(m)}) = \int_0^\infty \rho_{U_m,n}(t) \, dt
\end{equation}
and
\begin{equation}
     f(\bA(\bcalX), \bB(\bcalX)) = \int_0^\infty \rho_{X,n}(t) \, dt.
\end{equation}
Since $(\bA^{(m)},\bB^{(m)}) \to (\bA(\bcalX), \bB(\bcalX))$ almost surely, continuity of $f$ at $(\bA(\bcalX), \bB(\bcalX))$ implies by the continuous mapping theorem that
\begin{equation} \label{ty}
    f(\bA^{(m)},\bB^{(m)}) \to f(\bA({\bm \nu}), \bB({\bm \nu}))
\end{equation}
almost surely as $m\to \infty,$ i.e.,~\eqref{eh} holds.

Now, for the convergence of the logarithmic part, recall that we have the almost sure convergence 
\begin{equation}
    \det \bM_{U_m,n} = B_{d_n}({\bm \mu}^{(m)}) \to B_{d_n}(\bcalX) = \det \bM_{X,n}
\end{equation}
as $m\to \infty.$ As the mapping $\mathbb{R}_{>0}\to \mathbb{R}$ defined by $q \mapsto \log q$ is continuous, the continuous mapping theorem yields that 
\begin{equation} \label{tx}
    \log \det \bM_{U_m,n} \to \log\det \bM_{X,n}
\end{equation}
almost surely as $m \to \infty.$ Combining~\eqref{ty} and~\eqref{tx}, we obtain that
\begin{equation} \label{ua}
    \widehat{h}_n\left( \calS_m \right) \to h_n(X)
\end{equation}
almost surely as $m\to \infty.$ Finally,~\eqref{tz} follows from~\eqref{ua} by Theorem~\ref{jn}.

\subsection{Proof of Corollary~\ref{uc}: Consistency of the Mutual Information Estimator} \label{ud}

Denote $\calS_m = \{(X_j,Y_j)\}_{j\in [m]},$ and consider the empirical measure 
\begin{equation}
    \widehat{P}_m(x) := \sum_{j\in [m]} \frac{\delta_x(X_j)}{m+1}.
\end{equation}
Let $\frakD_m$ be the event that for each $x\in \supp(X)$ there is a subset of indices $J_x\subset [m]$ of size at least $n+1$ such that: \textbf{i)} $X_j=x$ for each $j\in J_x,$ and \textbf{ii)} the $Y_j,$ for $j\in J_x,$ are distinct. If $\frakD_m$ occurs, then we may write
\begin{equation}
    \widehat{I}_n(\calS_m) = \widehat{h}_n(\calA_m) - \sum_{x\in \supp(X)} \widehat{P}_m(x) ~ \widehat{h}_n(\calB_{m,x}),
\end{equation}
where $\calA_m := \{Y_j \}_{j\in [m]}$ and $\calB_{m,x} := \{ Y_j ~ ; ~ j\in [m], X_j=x\}.$ By the assumption of continuity of $Y,$ it holds with probability $1$ that the $Y_j,$ for $j\in \BN,$ are all distinct. In addition, we have that $P_X(x)>0$ for each $x\in \supp(X).$ Therefore, $P(\frakD_m)\to 1$ as $m\to \infty.$ Note that $\frakD_0\subset \frakD_1\subset \cdots.$ 

Let $\frakC$ be the event that $\lim_{m\to \infty} \widehat{h}_n(\calA_m)=h_n(Y)$ and, for each $x\in \supp(X),$ $\lim_{m\to \infty} \widehat{h}_n(\calB_{m,x})=h_n(Y^{(x)}).$ By Theorem~\ref{bp} and finiteness of $\supp(X),$ for each integer $m' \ge (n+1)|\supp(X)|,$ we have that $P(\frakC  \mid \frakD_{m'})=1.$ Let $\frakF$ be the event that the empirical measure $\widehat{P}_m$ converges to $P_X,$ i.e., that for each $x\in \supp(X)$ the limit $\widehat{P}_m(x)\to P_X(x)$ holds as $m\to \infty.$ By the strong law of large numbers, $P(\frakF)=1.$ Therefore,
\begin{equation}
    P\left(\lim_{m\to \infty} \widehat{I}_n(\calS_m) = I(X;Y) \right) \ge P(\frakC \cap \frakF \cap \frakD_{m'}) \ge P(\frakF) + P(\frakC\cap \frakD_{m'}) -1 =  P(\frakD_{m'}).
\end{equation}
Taking $m'\to \infty,$ we deduce that $\widehat{I}_n(\calS_m) \to I(X;Y) $ almost surely.

\section{Proofs of Subsection~\ref{wd}: Sample Complexity} \label{bz}

\subsection{Proof of Proposition~\ref{uj}: Differential Entropy} \label{gw}

Suppose $\mathrm{supp}(X) \subset [p,q] \subset (0,\infty),$ and write $\calS=\{X_j\}_{j=1}^m$; note that we may assume, without loss of generality, that $X$ is strictly positive because $h_n$ is shift-invariant. We use the same notation in Appendix~\ref{bx}. In particular, $\calX_k=\BE[X^k],$ and $\bcalX=(\calX_1,\cdots,\calX_{2n})^T.$ Let $U \sim \mathrm{Unif}(\calS).$ Let $\frakE_m$ be the event that $X_1,\cdots,X_m$ are distinct. From equations~\eqref{tp} and~\eqref{uk}, if $m>n$ and $\frakE_m$ holds, then we have that
\begin{equation} \label{gp}
    \widehat{h}_n(\calS)-h_n(X) = \frac{1}{2d_n} \log \frac{\det \bM_{U,n}}{\det \bM_{X,n}}   +  \int_0^\infty \rho_{U,n}(t) \, - \rho_{X,n}(t) \, dt. 
\end{equation}
By the assumption of continuity of $X,$ we have that $P(\frakE_m)=1$ for every $m.$ Therefore, for the purpose of proving a sample complexity bound, we may assume that $m>n$ and that $\frakE_m$ occurs.

We will consider the determinant part and the integral part in~\eqref{gp} separately, but the proof technique will be the same. Let $A_j$ and $B_\ell$ be the polynomials as defined by equation~\eqref{ul} in Appendix~\ref{bx}, so 
\begin{equation}
    \rho_{X,n}(t) = \frac{\sum_{j=1}^{d_n-2} A_j(\bcalX) ~ t^j}{c_n + \sum_{j=1}^{d_n} B_j(\bcalX) ~ t^j}
\end{equation}
where $c_n:= \prod_{k=1}^n j!.$ We split each of the polynomials $A_j$ and $B_\ell$ into a positive part and a negative part. More precisely, we collect the terms in $A_j$ that have positive coefficients into a polynomial $A_j^{(+)},$ and the terms in $A_j$ with negative coefficients into a polynomial $-A_j^{(-)}$ (so $A_j^{(-)}$ has positive coefficients, and $A_j = A_j^{(+)}-A_j^{(-)}$). Define $B_\ell^{(+)}$ and $B_\ell^{(-)}$ from $B_\ell$ similarly. By positivity of $X,$ each moment $\calX_k$ is (strictly) positive. Then, we may write
\begin{equation}
    \rho_{X,n}(t) = \frac{f_X(t) - g_X(t)}{u_X(t)-v_X(t)}
\end{equation}
with the polynomials in $t$
\begin{align}
    f_X(t) &:= \sum_{j=1}^{d_n-2} A_j^{(+)}(\bcalX)t^j \\
    g_X(t) &:= \sum_{j=1}^{d_n-2} A_j^{(-)}(\bcalX)t^j \\
    u_X(t) &:= c_n+\sum_{\ell=1}^{d_n} B_\ell^{(+)}(\bcalX)t^\ell \\
    v_X(t) &:= \sum_{\ell=1}^{d_n} B_\ell^{(-)}(\bcalX)t^\ell,
\end{align}
having all non-negative coefficients. We note that we have suppressed the dependence on $n$ in the notation used for these polynomials for readability. For $q \in \{f,g,u,v\},$ let $q_{U}$ be the random variable whose value is what is obtained via $q_X$ when the moments of $X$ are replaced with the sample moments obtained from the samples $\calS,$ e.g.,
\begin{equation}
    f_{U}(t) := \sum_{j=1}^{d_n-2} A_j^{(+)}\left( \frac{\sum_{i=1}^m X_i}{m},\cdots,\frac{\sum_{i=1}^m X_i^{2n}}{m}\right)t^j.
\end{equation}
Note that $u_{U}(t) - v_{U}(t) = \det \bM_{\sqrt{t}U+N,n}>0,$ where $N\sim \calN(0,1)$ is independent of $X,X_1,\cdots,X_m.$ Then the function
\begin{equation} \label{uu}
    \rho_{U,n}(t) = \frac{f_{U}(t) - g_{U}(t)}{u_{U}(t)-v_{U}(t)}
\end{equation}
is well-defined over $t\in [0,\infty).$ By the homogeneity properties proved in Theorem~\ref{iv}, we know that the total degree of $A_j$ is at most $2j+2,$ and the total degree of $B_\ell$ is at most $2\ell.$ Therefore, for any $\eta\in (0,1)$ and $(\xi_1,\cdots,\xi_{2n})\in \BR_{\ge 0}^{2n},$ we have the inequalities
\begin{align}
    (1-\eta)^{2j+2} A_j^{(\pm)}(\xi_1,\cdots,\xi_{2n})&\le A_j^{(\pm)}((1-\eta)\xi_1,\cdots,(1-\eta)\xi_{2n}) \label{us} \\
    A_j^{(\pm)}((1+\eta)\xi_1,\cdots,(1+\eta)\xi_{2n}) &\le (1+\eta)^{2j+2} A_j^{(\pm)}(\xi_1,\cdots,\xi_{2n})  \\
    (1-\eta)^{2\ell} B_\ell^{(\pm)}(\xi_1,\cdots,\xi_{2n})&\le B_\ell^{(\pm)}((1-\eta)\xi_1,\cdots,(1-\eta)\xi_{2n}) \\
    B_\ell^{(\pm)}((1+\eta)\xi_1,\cdots,(1+\eta)\xi_{2n}) &\le (1+\eta)^{2\ell} B_\ell^{(\pm)}(\xi_1,\cdots,\xi_{2n}) \label{ut}
\end{align}
for every $1\le j \le d_n-2$ and $1\le \ell \le d_n.$

For each $\eta \in (0,1),$ we denote the event
\begin{equation} \label{up}
    \mathfrak{A}_{n,\eta}(\calS)  := \left\{ 1-\eta \le \frac{\sum_{i=1}^m X_i^k}{m\calX_k}\le 1+\eta \text{ for every } k \in [2n] \right\},
\end{equation}
Hoeffding's inequality yields that, for any $z>0$ and $1\le k \le 2n,$
\begin{equation}
    P\left( \left| \calX_k - \frac{1}{m}\sum_{i=1}^m X_i^k \right| \ge z \right) \le 2e^{- 2mz^2/\left(q^k-p^k\right)^2}.
\end{equation}
Setting $z = \eta \calX_k \ge \eta p^k > 0$ for $\eta \in (0,1)$ yields that
\begin{equation}
P\left( (1-\eta) \calX_k < \frac{1}{m}\sum_{i=1}^m X_i^k < (1+\eta) \calX_k \right) \ge 1-2e^{- 2m \eta^2/\left( (q/p)^k-1 \right)^2}.
\end{equation}
Therefore, the union bound yields that
\begin{equation} \label{vl}
P\left( \mathfrak{A}_{n,\eta}(\calS) \right) \ge 1-4ne^{- 2m \eta^2/\left( (q/p)^{2n}-1\right)^2}. 
\end{equation}
If $\mathfrak{A}_{n,\eta}(\calS)$ occurs, we show a bound on the estimation error that is linear in $\eta$
\begin{equation}
    \widehat{h}_n(\calS)-h_n(X) = O_{X,n}(\eta),
\end{equation}
independent of the number of samples $m,$ for all small enough $\eta.$ Then, we choose $\eta$ to be linear in the error $\varepsilon$ to conclude the proof. 

We may bound $\rho_{U,n}(t)$ (see~\eqref{uu}) via the bounds in~\eqref{us}--\eqref{ut} under the assumption that $\mathfrak{A}_{n,\eta}(\calS)$ occurs. If $(1-\eta)\calX_k \le \frac{1}{m} \sum_{i=1}^m X_i^m \le (1+\eta)\calX_k$ holds for every $1\le k \le 2n,$ then by~\eqref{us}--\eqref{ut} we have that for every $t\ge 0$ and $\eta\in (0,1)$
\begin{equation} \label{uw}
    \frac{(1-\eta)^2f_X((1-\eta)^2t)-(1+\eta)^2g_X((1+\eta)^2t)}{u_X((1+\eta)^2t)-v_X((1-\eta)^2t)} \le \frac{f_{U}(t)-g_{U}(t)}{u_{U}(t)-v_{U}(t)} = \rho_{U,n}(t).
\end{equation}
For an analogous upper bound, we first verify the positivity
\begin{equation} \label{uv}
    u_X((1-\eta)^2t)-v_X((1+\eta)^2t)>0
\end{equation}
for every small enough~$\eta.$ Let
\begin{equation} \label{uo}
    \mu_X := \sup_{t\in [0,\infty)} \frac{v_X(t)}{u_X(t)}.
\end{equation}
We show that $\mu_X<1.$ We have the limit
\begin{equation}
    \xi_X := \lim_{t\to \infty} \frac{v_X(t)}{u_X(t)} = \frac{B_{d_n}^{(-)}(\bcalX)}{B_{d_n}^{(+)}(\bcalX)}.
\end{equation}
Recall that $B_{d_n}^{(+)}(\bcalX) - B_{d_n}^{(-)}(\bcalX) = B_{d_n}(\bcalX) = \det \bM_{X,n}>0$ and both $B_{d_n}^{(+)}(\bcalX)$ and $B_{d_n}^{(-)}(\bcalX)$ are non-negative, hence $B_{d_n}^{(+)}(\bcalX)>0.$ Then, $\xi_X < 1.$ Thus, there is a $t_0\ge 0$ such that $v_X(t)/u_X(t) < (1+\xi_X)/2 < 1$ whenever $t>t_0.$ Further, by the extreme value theorem, there is a $t_1\in [0,t_0]$ such that $v_X(t)/u_X(t)\le v_X(t_1)/u_X(t_1) < 1$ for every $t\in [0,t_0].$ Therefore, $\mu_X \le \max((1+\xi_X)/2,v_X(t_1)/u_X(t_1)) < 1,$ as desired. Note that if $\mu_X=0$ then $v_X \equiv 0$ identically, in which case~\eqref{uv} trivially holds by positivity of $u_X.$ So, for the purpose of showing~\eqref{uv}, it suffices to consider the case $\mu_X\in (0,1).$ Denote 
\begin{equation} \label{uq}
    \nu := \left( \frac{1+ \eta}{1-\eta} \right)^2.
\end{equation}
Now, since $v_X$ is a polynomial of degree at most $d_n,$ we have that $v_X(\alpha \tau)\le \alpha^{d_n} v_X(\tau)$ for every $\alpha \ge 1$ and $\tau\ge 0.$ Therefore, for every $1\le \nu < \mu_X^{-1/d_n}$ and $t\ge 0,$ we have that
\begin{equation}
    \frac{v_X((1+\eta)^2 t)}{u_X((1-\eta)^2 t)} \le \left( \frac{1+\eta}{1-\eta} \right)^{2d_n} \cdot \frac{ v_X((1-\eta)^2t)}{u_X((1-\eta)^2t)} \le \nu^{d_n} \mu_X < 1,
\end{equation}
i.e., inequality~\eqref{uv} holds. Therefore, for every $1\le \nu < \mu_X^{-1/d_n}$ (if $\mu_X=0,$ we allow $1\le \nu <\infty$), inequalities~\eqref{us}--\eqref{ut} imply the bound
\begin{equation} \label{ux}
\rho_{U,n}(t) = \frac{f_{U}(t)-g_{U}(t)}{u_{U}(t)-v_{U}(t)} \le \frac{(1+\eta)^2f_X((1+\eta)^2t)-(1-\eta)^2g_X((1-\eta)^2t)}{u_X((1-\eta)^2t)-v_X((1+\eta)^2t)}.
\end{equation}
Combining~\eqref{uw} and~\eqref{ux}, then integrating with respect to $t$ over $[0,\infty)$ and performing a change of variables from $t$ to $(1-\eta)^2 t,$ we obtain the bounds
\begin{equation} \label{ga}
\int_0^\infty \frac{f_X(t)-\nu g_X(\nu t)}{u_X(\nu t)-v_X(t)} \, dt \le \int_0^\infty \rho_{U,n}(t) \, dt \le \int_0^\infty \frac{\nu f_X(\nu t)-g_X(t)}{u_X(t)-v_X(\nu t)} \, dt. 
\end{equation}
Next, we further develop these bounds. For any $s\in (0,1),$ denote
\begin{equation} \label{uy}
    \nu_{X,n,s}:= \left( \frac{1-s\mu_X}{1-s} \right)^{1/d_n}.
\end{equation}
Consider the functions 
\begin{align}
    \varphi_X(t;\nu) &:= \frac{u_X(t)-v_X(t)}{u_X(t)-v_X(\nu t)},  \\
    \psi_X(t;\nu) &:= \frac{u_X(t)-v_X(t)}{u_X(\nu t)-v_X(t)}.
\end{align}
We show in Appendix~\ref{gq} that, for any $s\in (0,(1-\mu_X)/(1+\mu_X))$ and $1\le \nu \le \nu_{X,n,s},$ the uniform bounds
\begin{equation} \label{gy}
    1-s \le \psi_X(t;\nu) \le 1 \le \varphi_X(t; \nu) \le 1+s
\end{equation}
hold over $t\in [0,\infty).$ Fix $s\in (0,(1-\mu_X)/(1+\mu_X))$ and $1\le \nu \le \nu_{X,n,s}.$ 

Now, the integrand in the upper bound in~\eqref{ga} can be rewritten as
\begin{equation} \label{vc}
\frac{\nu f_X(\nu t)-g_X(t)}{u_X(t)-v_X(\nu t)} = \varphi_X(t;\nu) \left( \frac{f_X(t)-g_X(t)}{u_X(t)-v_X(t)} + \frac{\nu f_X(\nu t) -f_X(t)}{u_X(t)-v_X(t)} \right).
\end{equation}
The integrand in the lower bound in~\eqref{ga} can be rewritten as
\begin{equation} \label{vd}
    \frac{f_X(t)-\nu g_X(\nu t)}{u_X(\nu t)-v_X(t)}  = \psi_X(t;\nu) \left( \frac{f_X(t)-g_X(t)}{u_X(t)-v_X(t)} + \frac{g_X(t)-\nu g_X(\nu t)}{u_X(t)-v_X(t)} \right).
\end{equation}
By the bounds in~\eqref{gy}, we have that for every $t\ge 0$
\begin{equation} \label{uz}
    0 \le \varphi_X(t;\nu)-1 \le s.
\end{equation}
Hence, by non-negativity of $f_X$ and $g_X,$ we deduce
\begin{equation}
    \left( \varphi_X(t;\nu)-1 \right)\cdot  \frac{f_X(t)-g_X(t)}{u_X(t)-v_X(t)} \le s \cdot \frac{f_X(t)}{u_X(t)-v_X(t)},
\end{equation}
i.e.,
\begin{equation} \label{va}
    \varphi_X(t;\nu)\cdot  \frac{f_X(t)-g_X(t)}{u_X(t)-v_X(t)} \le \frac{f_X(t)-g_X(t)}{u_X(t)-v_X(t)} + s \cdot \frac{f_X(t)}{u_X(t)-v_X(t)}.
\end{equation}
In addition, since $f_X(\nu t) \le \nu^{d_n-2} f_X(t)$ over $t\in [0,\infty),$ inequality~\eqref{uz} implies that
\begin{equation} \label{vb}
    \varphi_X(t;\nu) \cdot \frac{\nu f_X(\nu t) -f_X(t)}{u_X(t)-v_X(t)} \le \frac{(1+s)(\nu^{d_n-1}-1)f_X(t)}{u_X(t)-v_X(t)}.
\end{equation}
Therefore, applying inequalities~\eqref{va} and~\eqref{vb} in formula~\eqref{vc}, we deduce in view of the upper bound in~\eqref{ga} the inequality
\begin{equation} \label{gc}
    \int_0^\infty \rho_{U,n}(t) - \rho_{X,n}(t) \, dt  \le \left( (1+s)\nu^{d_n-1}-1 \right) \int_0^\infty \frac{f_X(t)}{u_X(t)-v_X(t)} \, dt. 
\end{equation}
Similarly, we derive a lower bound on~\eqref{vd}. By~\eqref{gy}, we have that for every $t\ge 0$
\begin{equation}
    s \ge 1- \psi_X(t;\nu) \ge 0.
\end{equation}
Hence, by non-negativity of $f_X$ and $g_X,$ 
\begin{equation}
     s \cdot \frac{f_X(t)}{u_X(t)-v_X(t)}  \ge \left( 1- \psi_X(t;\nu) \right)\frac{f_X(t)-g_X(t)}{u_X(t)-v_X(t)},
\end{equation}
i.e.,
\begin{equation} \label{ve}
    \psi_X(t;\nu)\cdot\frac{f_X(t)-g_X(t)}{u_X(t)-v_X(t)}  \ge \frac{f_X(t)-g_X(t)}{u_X(t)-v_X(t)}  - s \cdot \frac{f_X(t)}{u_X(t)-v_X(t)}.
\end{equation}
In addition, from $\psi_X(t;\nu) \le 1 \le \nu$ and $g_X(\nu t)\le \nu^{d_n-2} g_X(t)$ for $t\ge 0,$ we deduce
\begin{equation} \label{vf}
    \psi_X(t;\nu)\cdot\frac{g_X(t)-\nu g_X(\nu t)}{u_X(t)-v_X(t)} \ge \psi_X(t;\nu)\cdot\frac{(1-\nu^{d_n-1})g_X(t)}{u_X(t)-v_X(t)} \ge \left( 1 - \nu^{d_n-1} \right) \frac{g_X(t)}{u_X(t)-v_X(t)}.
\end{equation}
Therefore, applying inequalities~\eqref{ve} and~\eqref{vf} in formula~\eqref{vd}, the lower bound in~\eqref{ga} yields bound
\begin{equation} \label{gd}
   \int_0^\infty \rho_{U,n}(t) - \rho_{X,n}(t) \, dt  \ge - s  \int_0^\infty \frac{f_X(t)}{u_X(t)-v_X(t)} \, dt- \left( \nu^{d_n-1}-1 \right) \int_0^\infty \frac{g_X(t)}{u_X(t)-v_X(t)} \, dt. 
\end{equation}
In particular,~\eqref{gd} implies that
\begin{equation} \label{vg}
   \int_0^\infty \rho_{U,n}(t) - \rho_{X,n}(t) \, dt  \ge - \left( \nu^{d_n-1}-(1-s) \right)  \int_0^\infty \frac{f_X(t)+g_X(t)}{u_X(t)-v_X(t)} \, dt.
\end{equation}
Now, note that $(1+s)\nu^{d_n-1}-1 \ge \nu^{d_n-1}-(1-s).$ Therefore, combining the upper bound in~\eqref{gc} and the lower bound in~\eqref{vg}, we deduce that
\begin{equation} \label{fu}
    \left| \int_0^\infty \rho_{U,n}(t) - \rho_{X,n}(t) \, dt \right|\le \left( (1+s)\nu^{d_n-1}-1 \right) \int_0^\infty \frac{f_X(t)+g_X(t)}{u_X(t)-v_X(t)} \, dt. 
\end{equation}
The upper bound in (\ref{fu}) may be made as small as needed by choosing a small~$s$ then choosing a small~$\nu.$

The second part of the proof, given in Appendix~\ref{gs}, derives the following error bound for estimating $\log \det \bM_{X,n}$ from samples. If $B_{d_n}^{(-)}(\bcalX)>0,$ we denote
\begin{equation} \label{vy}
    \tau_{X,n} := \left(  \frac{  B_{d_n}^{(+)}(\bcalX)/B_{d_n}^{(-)}(\bcalX) + 1}{2} \right)^{1/(n+1)} \in (1,\infty)
\end{equation}
and
\begin{equation} \label{we}
    \eta_{X,n}:=\min\left( \frac12, \frac{\tau_{X,n}-1}{\tau_{X,n}+1} \right) \in (0,1/2].
\end{equation}
If $B_{d_n}^{(-)}(\bcalX)=0,$ then we set $\tau_{X,n}=\infty$ and $\eta_{X,n}=1/2.$ We show that for all $\eta \in (0,\eta_{X,n}),$
if $\mathfrak{A}_{n,\eta}(\calS)$ holds, then we have the bound
\begin{equation} \label{gg}
    \left| \frac{1}{2d_n} \log \frac{\det \bM_{U,n}}{\det \bM_{X,n}} \right| \le \frac{6\eta}{n} \cdot \frac{B_{d_n}^{(+)}(\bcalX)+B_{d_n}^{(-)}(\bcalX)}{B_{d_n}^{(+)}(\bcalX)-B_{d_n}^{(-)}(\bcalX)}.
\end{equation}

To finish the proof, we choose $\eta$ so that the desired accuracy is achieved with high probability. Recall from~\eqref{vl} that
\begin{equation}
   P\left( \frakA_{n,\eta}(\calS)\right) \ge 1-4ne^{- m \eta^2 \alpha_{X,n}}
\end{equation}
where we denote the constant
\begin{equation}
    \alpha_{X,n} :=  2 \cdot \left( \left( \frac{q}{p} \right)^{2n}-1\right)^{-2}.
\end{equation}
In addition, from~\eqref{fu} and~\eqref{gg}, we know that if $s\in (0,(1-\mu_X)/(1+\mu_X)),$ $\nu \in [1,\nu_{X,n,s}],$ $\eta\in (0,\eta_{X,n}),$ and $\mathfrak{A}_{n,\eta}(\calS)$ occurs, then 
\begin{equation} \label{vr}
    \left| \widehat{h}_n(\calS) - h_n(X) \right| \le \eta \cdot \beta_{X,n}+ \left( (1+s)\nu^{d_n-1} - 1 \right)\cdot \gamma_{X,n} 
\end{equation}
where we denote the constants
\begin{align}
    \beta_{X,n} &:= \frac{6}{n} \cdot \frac{B_{d_n}^{(+)}(\bcalX)+B_{d_n}^{(-)}(\bcalX)}{B_{d_n}^{(+)}(\bcalX)-B_{d_n}^{(-)}(\bcalX)}, \\
    \gamma_{X,n} &:= \int_0^\infty \frac{f_X(t)+g_X(t)}{u_X(t)-v_X(t)} \, dt.
\end{align}
Consider the constant $\varepsilon_{X,n}\in (0, 2 \gamma_{X,n}]$ defined by
\begin{equation}
    \varepsilon_{X,n} := 2 \gamma_{X,n} \cdot \frac{1-\mu_X}{1+\mu_X}.
\end{equation}
Fix $\varepsilon \in (0, \varepsilon_{X,n}),$ set $s := \varepsilon/(6\gamma_{X,n}) \in (0,1/3],$ denote
\begin{equation} \label{vz}
    \kappa_{X,n} := \min\left( 3, \tau_{X,n} , \left( \frac{1-s\mu_X}{1-s} \right)^{1/(2d_n)},  \frac{1+\varepsilon/(2\beta_{X,n})}{1-\varepsilon/(2\beta_{X,n})}\right),
\end{equation}
and fix $\eta\in (0,(\kappa_{X,n}-1)/(\kappa_{X,n}+1)).$ Since $\kappa_{X,n}\le 3,$ we obtain $\eta < 1/2.$ In addition, $\kappa_{X,n}\le \tau_{X,n},$ hence $\eta < (\kappa_{X,n}-1)/(\kappa_{X,n}+1)$ implies that $ \eta < \eta_{X,n}.$ Note that, for $a\in (0,1)$ and $b>1,$ the inequality $a\le (b-1)/(b+1)$ is equivalent to $(1+a)/(1-a)\le b.$ By definition, 
\begin{equation}
    \kappa_{X,n} \le \left( \frac{1-s\mu_X}{1-s} \right)^{1/(2d_n)},
\end{equation}
hence we have
\begin{equation} \label{vp}
    (1+s)\nu^{d_n} = (1+s) \left( \frac{1+\eta}{1-\eta} \right)^{2d} < (1+s) \kappa_{X,n}^{2d} \le (1+s)\cdot \frac{1-s\mu_X}{1-s} \le \frac{1+s}{1-s} \le \frac{1+s+s(1-3s)}{1-s} = 1+3s.
\end{equation}
In addition, since 
\begin{equation}
    \kappa_{X,n} \le \frac{1+\varepsilon/(2\beta_{X,n})}{1-\varepsilon/(2\beta_{X,n})},
\end{equation}
and since we assume $\eta  < (\kappa_{X,n}-1)/(\kappa_{X,n}+1),$ we deduce the inequality $\eta < \varepsilon/(2\beta_{X,n}).$ Applying the two inequalities $\eta < \varepsilon/(2\beta_{X,n})$ and $(1+s)\nu^{d_n}\le 1+3s$ (see~\eqref{vp}) into inequality~\eqref{vr}, we conclude that
\begin{equation} 
    \left| \widehat{h}_n(\calS) - h_n(X) \right| \le \eta \cdot \beta_{X,n}+ \left( (1+s)\nu^{d_n-1} - 1 \right)\cdot \gamma_{X,n} \le \frac{\varepsilon}{2} +  \frac{\varepsilon}{2} = \varepsilon
\end{equation}
whenever $\mathfrak{A}_{n,\eta}(\calS)$ occurs.

Now, fix $\delta \in (0,1/(4n)).$ Set
\begin{equation} \label{wa}
    \eta := \frac12 \cdot \frac{\kappa_{X,n}-1}{\kappa_{X,n}+1}.
\end{equation}
We show that $\eta \ge \varepsilon c_{X,n},$ where we denote the constant $c_{X,n}$ by
\begin{equation}
    c_{X,n} := \min\left( \frac{1}{8 \gamma_{X,n}}, \frac{\tau_{X,n}-1}{4 \gamma_{X,n}(\tau_{X,n}+1)}, \frac{1-\mu_X}{72 \gamma_{X,n} d_{n}}, \frac{1}{4\beta_{X,n}} \right).
\end{equation}
In this definition of $c_{X,n},$ the term involving $\tau_{X,n}$ is removed if $\tau_{X,n}=\infty.$ We assume that
\begin{equation} \label{vu}
    m \ge \frac{2/(c_{X,n}^2\alpha_{X,n})}{\varepsilon^2} \log \frac{1}{\delta}.
\end{equation}
From $\eta \ge \varepsilon c_{X,n}$ and~\eqref{vu}, it follows that the probability that the event $\mathfrak{A}_{n,\eta}(\calS)$ does not occur is bounded as
\begin{equation}
    P\left( \mathfrak{A}_{n,\eta}(\calS)^c \right) \le 4n e^{-m\eta^2 \alpha_{X,n}} \le \delta.
\end{equation}
Note that this would conclude the proof, as then we would have that
\begin{equation}
    P\left( \left| \widehat{h}_n(\calS) - h_n(X) \right| \le \varepsilon \right) \ge P\left.\left( \left| \widehat{h}_n(\calS) - h_n(X) \right| \le \varepsilon ~ \right| ~ \mathfrak{A}_{n,\eta}(\calS) \right) P\left( \mathfrak{A}_{n,\eta}(\calS) \right) = P\left( \mathfrak{A}_{n,\eta}(\calS) \right) > 1-\delta.
\end{equation}
The rest of the proof is devoted to showing that $\eta \ge \varepsilon c_{X,n}$ holds. 

Let $\rho = (1-\mu_X)/(6d_n).$ We will show that
\begin{equation} \label{vv}
    \left( \frac{1-s\mu_X}{1-s} \right)^{1/(2d_n)} \ge \frac{1+s\rho}{1-s\rho}.
\end{equation}
Inequality~\eqref{vv} is equivalent to 
\begin{equation} \label{vw}
    (1-s\mu_X)(1-s\rho)^{2d_n} \ge (1+\rho s)^{2d_n}(1-s).
\end{equation}
By Bernoulli's inequality, since $0 \le s\rho \le 1,$ we have that $(1-s\rho)^{2d_n} \ge 1- 2d_n \rho s.$ In addition, the inequality $1+2az\ge e^{az}\ge (1+a)^z$ for $a,z\ge 0$ satisfying $az\le \log 2$ implies, in view of $2d_n \rho s \le 1/9 < \log 2,$ that
\begin{equation}
    1+4d_n\rho s \ge (1+\rho)^{2d_n}.
\end{equation}
Therefore, to show~\eqref{vw}, it suffices to show that
\begin{equation} \label{vx}
    (1-s\mu_X)(1-2d_n\rho s) \ge (1+4d_n \rho s)(1-s).
\end{equation}
Now, using the definition $\rho = (1-\mu_X)/(6d_n),$ inequality~\eqref{vx} follows as
\begin{align}
    (1-s\mu_X)(1-2d_n\rho s) &= (1-s \mu_X)(1-s(1-\mu_X)/3) = (1+2(1-\mu_X)s/3)(1-s) + s^2(1-\mu_X)(\mu_X+2)/3 \nonumber \\
    &\ge (1+2(1-\mu_X)s/3)(1-s) = (1+4d_n\rho s)(1-s).
\end{align}
Since~\eqref{vx} holds, we conclude that inequality~\eqref{vv} holds. 

Now, by the definition of $\kappa_{X,n}$ in~\eqref{vz} there are four possible values $\kappa_{X,n}$ can take. First, if $\kappa_{X,n}=3,$ then
\begin{equation}
    \eta = \frac14 = \varepsilon \cdot \frac{1}{4 \varepsilon} \ge \varepsilon \cdot \frac{1}{8\gamma_{X,n}} \ge \varepsilon c_{X,n}
\end{equation}
since $\varepsilon < \varepsilon_{X,n}\le 2\gamma_{X,n}.$ Now, if $\kappa_{X,n}=\tau_{X,n}$ (so $B_{d_n}^{(-)}(\bcalX)>0$), then
\begin{equation}
    \eta = \frac12 \cdot \frac{\tau_{X,n}-1}{\tau_{X,n}+1} \ge \frac{\varepsilon}{4\gamma_{X,n}} \cdot \frac{\tau_{X,n}-1}{\tau_{X,n}+1}
\end{equation}
since $\varepsilon <  2\gamma_{X,n}.$ Next, suppose that 
\begin{equation} \label{wc}
    \kappa_{X,n} = \left( \frac{1-s\mu_X}{1-s} \right)^{1/(2d_n)}.
\end{equation}
By~\eqref{vv} and~\eqref{wc}, we deduce that
\begin{equation} \label{wb}
    \kappa_{X,n} \ge   \frac{1+s\rho}{1-s\rho}.
\end{equation}
Recall that, for $0 < a < 1 < b,$ the inequalities $(1+a)/(1-a) \ge b$ and $(b-1)/(b+1) \ge a$ are equivalent.  Therefore, the definition of $\eta$ in~\eqref{wa} yields from~\eqref{wb} that $\eta \ge s\rho /2.$ Plugging in the definitions of $s$ and $\rho,$ we conclude that
\begin{equation}
    \eta \ge \varepsilon \cdot \frac{1-\mu_X}{72 \gamma_{X,n} d_{n}} \ge \varepsilon c_{X,n}.
\end{equation}
Finally, when
\begin{equation}
    \kappa_{X,n} = \frac{1+\varepsilon/(2\beta_{X,n})}{1-\varepsilon/(2\beta_{X,n})},
\end{equation}
the definition of $\eta$ implies that $\eta \ge \varepsilon/(4\beta_{X,n}) \ge \varepsilon c_{X,n}.$ Combining these four cases, we conclude that we must have $\eta \ge \varepsilon c_{X,n}$ independently of the value of $\kappa_{X,n}.$ The proof is thus complete.

\subsection{Uniform Bounds on $\varphi_X$ and $\psi_X$: Inequalities~\eqref{gy}} \label{gq}

Being polynomials of degree at most $d_n$ with non-negative coefficients, the functions $u_X$ and $v_X$ satisfy $u_X(\nu t)\le \nu^{d_n} u_X(t)$ and $v_X(\nu t)\le \nu^{d_n} v_X(t)$ for every $\nu \ge 1$ and $t\ge 0.$ Note also that both $u_X$ and $v_X$ are nondecreasing. In addition, we have $v_X(t)<u_X(t)$ for every $t\ge 0,$ because $u_X(t)-v_X(t) = \det \bM_{\sqrt{t}X+N,n} >0.$ We have also shown that $\mu_X<1,$ where $\mu_X$ is defined in~\eqref{uo} as
\begin{equation}
    \mu_X := \sup_{t\in [0,\infty)} \frac{v_X(t)}{u_X(t)}.
\end{equation}
These facts will be enough to deduce the bounds in~\eqref{gy}.

We show first the bounds on $\varphi_X$ in~\eqref{gy}. It suffices to consider the case $\mu_X>0,$ for otherwise $v_X$ vanishes identically and $\varphi_X\equiv 1$ identically. We show that for every $s>0$ and $1\le \nu \le \nu_{X,n,s}',$ where $\nu_{X,n,s}' := \left((1/s+1/\mu_X)/(1/s+1)\right)^{1/d_n} ,$ the uniform bound $1 \le \varphi_X(t;\nu) \le 1+s$ in~\eqref{gy} holds.

Consider the lower bound on $\varphi_X.$ For every $1\le \nu < \mu_X^{-1/d_n},$ we have the uniform bound
\begin{equation} \label{fh}
    \frac{v_X(\nu t)}{u_X(t)} \le \frac{\nu^{d_n} v_X(t)}{u_X(t)} \le \nu^{d_n} \mu_X < 1
\end{equation}
over $t\in [0,\infty).$ In particular,
\begin{equation} \label{um}
    u_X(t) - v_X(\nu t)>0
\end{equation}
for every $1\le \nu < \mu_X^{-1/d_n}$ and $t\ge 0.$ Since $v_X$ is nondecreasing, we conclude that $\varphi_X(t;\nu) = (u_X(t)-v_X(t))/(u_X(t)-v_X(\nu t)) \ge 1$ whenever $1\le \nu < \mu_X^{-1/d_n}.$ Note that $\nu_{X,n,s}' < \mu_X^{-1/d_n}$ for every $s>0$ since $\mu_X \in (0,1).$ 

Next, we show the upper bound on $\varphi_X.$ Fix $s>0$ and $\nu \in [1,\nu_{X,n,s}'].$ Since $v_X(t)/\mu_X \le u_X(t),$ we have for every $t\ge 0$ the bound
\begin{equation} \label{fk}
v_X(\nu t) \le \nu^{d_n} v_X(t) \le \frac{1/s+1/\mu_X}{1/s+1} \cdot v_X(t) \le \frac{v_X(t)/s+u_X(t)}{1/s+1} = v_X(t) + \frac{u_X(t)-v_X(t)}{1/s+1}. 
\end{equation}
Rearranging~\eqref{fk}, we obtain the bound
\begin{equation} \label{fl}
    \frac{-1}{1/s+1} \le \frac{v_X(t)-v_X(\nu t)}{u_X(t)-v_X(t)}.
\end{equation}
Adding $1$ to both sides of~\eqref{fl} then inverting, we obtain $\varphi_X(t;\nu) \le 1+s$; for this step, we used the fact that $u_X(t)-v_X(\nu t)>0,$ which follows by~\eqref{um} since $\nu \le \nu_{X,n,s}' < \mu_X^{-1/d_n}.$ 

Next, we prove the bounds on $\psi_X$ in~\eqref{gy}. We do not assume $\mu_X>0.$ The upper bound $\psi_X(t;\nu) \le 1$ follows for every $\nu \ge 1$ by monotonicity of $u_X.$ For the lower bound on $\psi_X,$ we show that for every $s\in (0,1)$ and $1\le \nu \le \nu_{X,n,s},$ where $\nu_{X,n,s}:=((1-s\mu_X)/(1-s))^{1/d_n},$ the uniform bound $\psi_X(t;\nu) \ge 1-s$ holds over $t\in [0,\infty).$ We have, for every $s\in (0,1)$ and $\nu \in [1,\nu_{X,n,s}],$ the bound
\begin{equation} \label{un}
    u_X(\nu t) \le \nu^{d_n} u_X(t) \le \frac{1-s \mu_X}{1-s} \cdot u_X(t) \le \frac{u_X(t) - s v_X(t)}{1-s} = \frac{u_X(t)-v_X(t)}{1-s} + v_X(t)
\end{equation}
over $t\in[0,\infty).$ Rearranging~\eqref{un}, we obtain $\psi_X(t;\nu)\ge 1-s,$ as desired.

Finally, note that $\nu_{X,n,s}\le \nu_{X,n,s}'$ is equivalent to $s\le (1-\mu_X)/(1+\mu_X).$ This concludes the proof that, for every $s\in (0,(1-\mu_X)/(1+\mu_X))$ and $\nu \in [1,\nu_{X,n,s}],$ the uniform bounds in~\eqref{gy}
\begin{equation}
    1-s \le \psi_X(t;\nu) \le 1 \le \varphi_X(t; \nu) \le 1+s
\end{equation}
hold over $t\in[0,\infty).$

\subsection{Error in Estimating $\log \det \bM_{X,n}$: Inequality~\eqref{gg} } \label{gs}

Recall that 
\begin{equation}
    \det \bM_{X,n} = B_{d_n}(\bcalX) = B^{(+)}_{d_n}(\bcalX) - B^{(-)}_{d_n}(\bcalX).
\end{equation}
We bound the error when estimating $\log \det \bM_{X,n}$ from the samples $\calS.$ Denote the random vector $\bmu := \left( \frac{\sum_{i=1}^m X_i}{m},\cdots,\frac{\sum_{i=1}^m X_i^{2n}}{m}\right),$ and note that
\begin{equation}
    \det \bM_{U,n} = B_{d_n}(\bmu) = B^{(+)}_{d_n}(\bmu) - B^{(-)}_{d_n}(\bmu).
\end{equation}
We assume that $m>n.$ Let $\eta_{X,n}$ be as defined by~\eqref{vy} and~\eqref{we}, and fix $\eta \in (0,\eta_{X,n}).$ Then we show that under $\mathfrak{A}_{n,\eta}(\calS)$
\begin{equation} \label{vh}
    \left| \frac{1}{2d_n} \log \frac{\det \bM_{U,n}}{\det \bM_{X,n}} \right| \le \frac{6\eta}{n} \cdot \frac{B_{d_n}^{(+)}(\bcalX)+B_{d_n}^{(-)}(\bcalX)}{B_{d_n}^{(+)}(\bcalX)-B_{d_n}^{(-)}(\bcalX)}.
\end{equation}
By~\eqref{mq} in Theorem~\ref{iv}, each term in the polynomials $B_{d_n}^{(\pm)}$ is a product of at most $n+1$ monomials. Thus,
\begin{equation} \label{gh}
    (1-\eta)^{n+1} B_{d_n}^{(\pm)}(\bcalX) \le B_{d_n}^{(\pm)}(\bmu) \le (1+\eta)^{n+1} B_{d_n}^{(\pm)}(\bcalX).
\end{equation}
It suffices to consider the case when $B_{d_n}^{(-)}$ is not the zero polynomial, for if $B_{d_n}^{(-)}$ is the zero polynomial then we obtain from~\eqref{gg} the bound 
\begin{equation} \label{vi}
    \left| \frac{1}{2d_n} \log \frac{\det \bM_{U,n}}{\det \bM_{X,n}} \right| = \frac{1}{2d_n} \left| \log \frac{B_{d_n}^{(+)}(\bmu)}{B_{d_n}^{(+)}(\bcalX)} \right| \le  \frac{\max\left(\log(1+\eta),-\log(1-\eta)\right)}{n} = \frac{-\log(1-\eta)}{n}< \frac{2\eta}{n}
\end{equation}
where the last inequality follow because $-\log(1-z) < 2z$ for $z\in (0,1/2),$ which can be verified by checking the derivative. Note that the bound $2\eta/n$ in~\eqref{vi} is stronger than the bound in~\eqref{vh}. Assume that $B_{d_n}^{(-)}$ does not vanish identically, so positivity of $X$ yields that $B_{d_n}^{(-)}(\bcalX)>0.$

From~\eqref{gh}, we have that
\begin{align} \label{gk}
    \log \frac{B_{d_n}^{(+)}(\bcalX)-\nu^{\frac{n+1}{2}} B_{d_n}^{(-)}(\bcalX)}{B_{d_n}^{(+)}(\bcalX)-B_{d_n}^{(-)}(\bcalX)} + (n+1)\log(1-\eta) \le \log \frac{\det \bM_{U,n}}{\det \bM_{X,n}} 
\end{align}
and
\begin{align} \label{gl}
    \log \frac{\det \bM_{U,n}}{\det \bM_{X,n}}  \le  \log \frac{B_{d_n}^{(+)}(\bcalX)-\nu^{-\frac{n+1}{2}} B_{d_n}^{(-)}(\bcalX)}{B_{d_n}^{(+)}(\bcalX)-B_{d_n}^{(-)}(\bcalX)} + (n+1)\log(1+\eta)
\end{align}
where we used our assumption that 
\begin{equation} \label{gj}
    \nu^{\frac{n+1}{2}} = \left( 1 + \frac{2}{1/\eta-1} \right)^{n+1} < \frac{1}{2} \left(\frac{B_{d_n}^{(+)}(\bcalX)}{B_{d_n}^{(-)}(\bcalX)} + 1 \right) < \frac{B_{d_n}^{(+)}(\bcalX)}{B_{d_n}^{(-)}(\bcalX)}.
\end{equation}
Now, for every $(w,z,r)\in\mathbb{R}^3$ such that $w>z>0$ and $w/z>r>1,$ rearranging $r+1/r>2$ we have that
\begin{equation}
    \frac{w-z/r}{w-z}< \frac{w-z}{w-rz}.
\end{equation}
Setting $(w,z,r)=(B_{d_n}^{(+)}(\bcalX),B_{d_n}^{(-)}(\bcalX),\nu^{(n+1)/2}),$ we obtain that
\begin{equation}
    1 < \frac{B_{d_n}^{(+)}(\bcalX)-\nu^{-\frac{n+1}{2}} B_{d_n}^{(-)}(\bcalX)}{B_{d_n}^{(+)}(\bcalX)-B_{d_n}^{(-)}(\bcalX)} < \frac{B_{d_n}^{(+)}(\bcalX)-B_{d_n}^{(-)}(\bcalX)}{B_{d_n}^{(+)}(\bcalX)-\nu^{\frac{n+1}{2}} B_{d_n}^{(-)}(\bcalX)}.
\end{equation}
Therefore,
\begin{equation} \label{vj}
    0 < \log \frac{B_{d_n}^{(+)}(\bcalX)-\nu^{-\frac{n+1}{2}} B_{d_n}^{(-)}(\bcalX)}{B_{d_n}^{(+)}(\bcalX)-B_{d_n}^{(-)}(\bcalX)} < \left| \log \frac{B_{d_n}^{(+)}(\bcalX)-\nu^{\frac{n+1}{2}} B_{d_n}^{(-)}(\bcalX)}{B_{d_n}^{(+)}(\bcalX)-B_{d_n}^{(-)}(\bcalX)} \right|.
\end{equation}
Applying~\eqref{vj} in~\eqref{gl} and combining that with~\eqref{gk}, we obtain (since $\log(1+\eta) < - \log(1-\eta)$) the bound
\begin{equation} \label{yc}
    \left| \log \frac{\det \bM_{U,n}}{\det \bM_{X,n}} \right| \le \log \frac{B_{d_n}^{(+)}(\bcalX)-B_{d_n}^{(-)}(\bcalX)}{B_{d_n}^{(+)}(\bcalX)-\nu^{\frac{n+1}{2}} B_{d_n}^{(-)}(\bcalX)} + (n+1)\log\frac{1}{1-\eta}.
\end{equation}

Now, we may write
\begin{equation}
    \frac{B_{d_n}^{(+)}(\bcalX)-B_{d_n}^{(-)}(\bcalX)}{B_{d_n}^{(+)}(\bcalX)-\nu^{\frac{n+1}{2}} B_{d_n}^{(-)}(\bcalX)}  =\left( 1 - \frac{B_{d_n}^{(-)}(\bcalX)}{B_{d_n}^{(+)}(\bcalX)-B_{d_n}^{(-)}(\bcalX)} \left( \nu^{\frac{n+1}{2}}-1 \right) \right)^{-1}. 
\end{equation}
The proof of~\eqref{vh} is completed by showing that for $(w,z,r)\in \mathbb{R}_{>0}^3$ such that $(1+z)^r<1+\frac{1}{2w}$ we have
\begin{equation} \label{gm}
    - \log\left( 1- w \left( (1+z)^r-1 \right) \right) \le (2w+1)rz.
\end{equation}
Before showing that~\eqref{gm} holds, we note how it completes the proof. Setting 
\begin{equation}
    (w,z,r) = \left(  \frac{B_{d_n}^{(-)}(\bcalX)}{B_{d_n}^{(+)}(\bcalX)-B_{d_n}^{(-)}(\bcalX)} , \frac{2\eta}{1-\eta}, n+1 \right),
\end{equation}
we obtain that
\begin{equation} \label{yd}
    \log \frac{B_{d_n}^{(+)}(\bcalX)-B_{d_n}^{(-)}(\bcalX)}{B_{d_n}^{(+)}(\bcalX)-\nu^{\frac{n+1}{2}} B_{d_n}^{(-)}(\bcalX)} \le \frac{B_{d_n}^{(+)}(\bcalX)+B_{d_n}^{(-)}(\bcalX)}{B_{d_n}^{(+)}(\bcalX)-B_{d_n}^{(-)}(\bcalX)} \cdot (n+1)\cdot \frac{2\eta}{1-\eta}
\end{equation}
since
\begin{equation} \label{gn}
    \nu^{\frac{n+1}{2}} < \frac{1}{2} \left(  \frac{B_{d_n}^{(+)}(\bcalX)}{B_{d_n}^{(-)}(\bcalX)} + 1 \right).
\end{equation}
Then $- \log (1- \eta) < 2\eta$ yields from~\eqref{yc} and~\eqref{yd} that 
\begin{equation} \label{ye}
     \frac{1}{2d_n} \left| \log \frac{\det \bM_{U,n}}{\det \bM_{X,n}} \right| \le \frac{B_{d_n}^{(+)}(\bcalX)+B_{d_n}^{(-)}(\bcalX)}{B_{d_n}^{(+)}(\bcalX)-B_{d_n}^{(-)}(\bcalX)} \cdot \frac{2\eta}{n(1-\eta)} + \frac{2\eta}{n}.
\end{equation}
Then~\eqref{yd} yields the desired inequality~\eqref{gg} as $\eta \in (0,1/2).$

Finally, to see that (\ref{gm}) holds, we consider for fixed $w,r>0$
\begin{equation}
    f(z) := (2w+1)rz + \log\left( 1 - w \left( \left(1+z\right)^r - 1 \right) \right)
\end{equation}
over $0 \le z < (1+1/(2w))^{1/r}-1.$ Inequality~\eqref{gm} is restated as $f(z)\ge 0$ for every $0 < z < (1+1/(2w))^{1/r}-1,$ which follows since $f$ is continuous, $f(0)=0,$ $f'(0^+)=(w+1)r>0,$ and 
\begin{align}
    f'(z) &= (2w+1)r - \frac{wr(1+z)^{r-1}}{1-w((1+z)^r-1)} > (2w+1)r - \frac{wr(1+z)^{r}}{1-w((1+z)^r-1)} \\
    &>(2w+1)r - \frac{wr(1+1/(2w))}{1-w((1+1/(2w))-1)}=0
\end{align}
for every $0 \le z < (1+1/(2w))^{1/r}-1.$

\subsection{Proof of Proposition~\ref{xy}: Mutual Information} \label{xz}

Let $\{(X_j,Y_j)\}_{j\in \BN}$ be i.i.d. samples drawn according to $P_{X,Y}.$ Denote $\calS_m=\{X_j\}_{j=1}^m.$ By continuity of $Y,$ we may assume that all the $Y_j,$ for $j\in \BN,$ are distinct. For each $x\in \supp(X),$ let $J_x:=\{1\le j \le m ~ ; ~ X_j=x\}.$ Let $\frakD_m$ be the event that, for every $x\in \supp(X),$ we have that $|J_x|>n.$ We use Hoeffding's inequality to obtain a lower bound on the probability
\begin{equation}
    P(\frakD_m) = P\left( \min_{x\in \supp(X)} |J_x| > n \right).
\end{equation}
Let $\widehat{P}_m$ be the empirical measure: $\widehat{P}_m(x) := m^{-1}\sum_{j=1}^m \delta_x(X_j).$ Note that $|J_x|= m \widehat{P}_m(x).$

Let $x_0\in \supp(X)$ be such that $P_X(x_0)$ is minimal, set $\zeta:= P_X(x_0)/2,$ and suppose $m \ge \zeta^{-1}n.$ Then, the union bound and $\zeta \le P_X(x)-\zeta$ for each $x\in \supp(X)$ yield that
\begin{align}
    P\left( n \ge \min_{x\in \supp(X)} |J_x| \right) &\le P\left( m \zeta \ge \min_{x\in \supp(X)} |J_x| \right) \le \sum_{x\in \supp(X)} P\left( m \zeta \ge  |J_x| \right) \\
    &\le \sum_{x\in \supp(X)} P\left( m (P_X(x)-\zeta) \ge  |J_x| \right) = \sum_{x\in \supp(X)} P\left( P_X(x)- \widehat{P}_m(x) \ge  \zeta \right).
\end{align}
Since $\BE[\widehat{P}_m(x)]=P_X(x)$ for each $x \in \supp(X),$ Hoeffding's inequality yields that $P\left( P_X(x)- \widehat{P}_m(x) \ge  \zeta \right) \le e^{-2\zeta^2 m}.$ Therefore,
\begin{equation}
    P\left( n \ge \min_{x\in \supp(X)} |J_x| \right) \le |\supp(X)| \cdot e^{-2\zeta^2 m}.
\end{equation}
In other words, for every $m\ge 2n/P_X(x_0),$ we have the bound 
\begin{equation} \label{wh}
    P(\frakD_m) \ge 1 - |\supp(X)| \cdot e^{-m P_X(x_0)^2/2}.
\end{equation}
Denote $\pi_X:= 4/P_X(x_0)^2$ and
\begin{equation}
    \delta_{X,n} := \min\left( \frac{1}{4|\supp(X)|}, e^{-P_X(x_0)n/2} \right).
\end{equation}
We conclude from~\eqref{wh} that, for every $\delta \in (0,\delta_{X,n}),$ if $m \ge \pi_X \log (1/\delta)$ then $P(\frakD_m) > 1-\delta/4.$

Consider the event $\mathfrak{P}_{m,\varepsilon}$ that the empirical measure $\widehat{P}_m$ is pointwise $\varepsilon$-close to the true measure $P_X,$ i.e.,
\begin{equation}
    \mathfrak{P}_{m,\varepsilon} := \left\{ \max_{x\in \supp(X)} \left| \widehat{P}_m(x) - P_X(x) \right| < \varepsilon \right\}.
\end{equation}
By the union bound, we have that
\begin{equation}
    P\left( \mathfrak{P}_{m,\varepsilon}^c \right) \le \sum_{x\in \supp(X)} P\left( \left| \widehat{P}_m(x)-P_X(x) \right| \ge \varepsilon \right).
\end{equation}
By Hoeffding's inequality, for each $x\in \supp(X),$ we have that
\begin{equation}
     P\left( \left| \widehat{P}_m(x)-P_X(x) \right| \ge \varepsilon \right) \le 2e^{-2m\varepsilon^2}.
\end{equation}
Therefore, we obtain the bound
\begin{equation}
    P\left( \mathfrak{P}_{m,\varepsilon} \right) > 1-2 |\supp(X)|e^{-2m\varepsilon^2}.
\end{equation}
In particular, if $\delta \in (0,1/(4|\supp(X)|)),$ then $m\ge (1/\varepsilon^2)\log(1/\delta)$ implies $P(\frakB_{m,\varepsilon}) > 1 - \delta/2.$

Recall that, if $\frakD_m$ occurs, then we may write
\begin{equation}
    \widehat{I}_n(\calS_m) = \widehat{h}_n(\calA_m) - \sum_{x\in \supp(X)} \widehat{P}_m(x) ~ \widehat{h}_n(\calB_{m,x}),
\end{equation}
where $\calA_m := \{Y_j \}_{j=1}^m$ and $\calB_{m,x} := \{ Y_j ~ ; ~ 1\le j \le m, X_j=x\}.$ Then,
\begin{align}
    \left| \widehat{I}_n(\calS_m) - I_n(X;Y) \right| \le \left| \widehat{h}_n(\calA_m) - h_n(Y) \right| &+ \sum_{x\in \supp(X)} \widehat{P}_m(x) \left| \widehat{h}_n(\calB_{m,x}) - h_n(Y^{(x)}) \right| \nonumber \\
    &+ \left( \max_{x\in \supp(X)} \left|\widehat{P}_m(x)-P_X(x)\right| \right) \sum_{x\in \supp(X)} |h_n(Y^{(x)})|.
\end{align}
Denote $H_{X,Y,n}:=\sum_{x\in \supp(X)} |h_n(Y^{(x)})|.$ Consider the events
\begin{align}
    \mathfrak{H}_{x,\varepsilon} &:= \bigcap_{x\in \supp(X)} \left\{ \left| \widehat{h}_n(\calB_{m,x}) - h_n(Y^{(x)}) \right| < \frac{\varepsilon}{3} \right\} \\
    \mathfrak{H}_{\varepsilon}' &:= \left\{ \left| \widehat{h}_n(\calA_m) - h_n(Y) \right| < \frac{\varepsilon}{3} \right\}.
\end{align}
Set $\mathfrak{H}_\varepsilon := \bigcap_{x\in \supp(X)} \mathfrak{H}_{x,\varepsilon}.$ From Theorem~\ref{bp}, we know that there is a constant $C_{X,Y,n}$ such that for every small enough $\varepsilon,\delta>0,$ if $m \ge (C_{X,Y,n}/\varepsilon^2)\log(1/\delta)$ then $P(\mathfrak{H}_{x,\varepsilon}\mid \frakD_m) \ge 1 - \delta/(8 |\supp(X)|)$ for each $x\in \supp(X)$ and $P(\mathfrak{H}_\varepsilon'\mid \frakD_m)>1-\delta/8.$ Then, $P(\mathfrak{H}_\varepsilon \cap \mathfrak{H}_\varepsilon' \mid \frakD_m) \ge 1-\delta/4.$ We conclude, possibly after increasing $C_{X,Y,n},$ that $P(\mathfrak{H}_\varepsilon \cap \mathfrak{H}_\varepsilon' \cap \frakD_m) \ge 1-\delta/2.$ Also, $P(\frakB_{m,\varepsilon/(2H_{X,Y,n})}) > 1 - \delta/2.$ Then, $P(\mathfrak{H}_\varepsilon \cap \mathfrak{H}_\varepsilon' \cap \frakD_m \cap \frakB_{m,\varepsilon/(2H_{X,Y,n})}) \ge 1-\delta.$ But under the event $\mathfrak{H}_\varepsilon \cap \mathfrak{H}_\varepsilon' \cap \frakD_m \cap \frakB_{m,\varepsilon/(2H_{X,Y,n})},$ we have the bound
\begin{equation}
     \left| \widehat{I}_n(\calS_m) - I_n(X;Y) \right| < \varepsilon,
\end{equation}
and the proof is complete.

\section{Algebraic Proof of Corollary~\ref{sl}} \label{aw}

We provide here an alternative algebraic proof of the formulas $a_X^{n,d_n-1}=\det \bM_{X,n}$ and $a_X^{n,d_n}=0$ (see equations~\eqref{mg}--\eqref{mu} for the definition of these quantities). These equations were proved in Section~\ref{jh} by invoking that $\mm(X,t)\sim 1/t$ (as $t\to \infty$) for a continuous RV $X.$ Ultimately, via Proposition~\ref{qx}, we have shown that these formulas hold identically as polynomials in the symbols $\{ \calX_j \}_{j=1}^{2n}.$ The alternative proof we present in this appendix directly derives this latter result by algebraic means without appealing to the MMSE asymptotic result. The point of including this alternative proof is that it might shed light on deriving simple expressions for the other constants $a_X^{n,j}$ for $1\le j \le d_n-2$ and $b_X^{n,\ell}$ for $2\le \ell \le d_n-1.$ See Remark~\ref{qw} for the polynomial expressions of the $a_X^{n,j}.$

We consider indeterminates $\calR_1,\cdots,\calR_{2n},$ which we think of as moments of a RV $R.$ For a permutation $\pi \in \sn$ and integers $m\in \mathbb{N}$ and $\{ i_1,\cdots,i_m\} \subseteq [n],$ it will be convenient to denote the products
\begin{equation}
    Q_R(\pi;i_1,\cdots,i_m) :=  \prod_{k\not\in \{i_1,\cdots,i_m\}} \calR_{k+\pi(k)},
\end{equation}
and $Q_R(\pi) :=  \prod_{k \in [n]} \calR_{k+\pi(k)}.$ We let $T_n^{(i,j)}\subset \calS_{[n]}$ denote the subset of permutations that send $i$ to $j,$ i.e., 
\begin{equation}
    T_n^{(i,j)} := \{ \pi \in \calS_{[n]} ~ ; ~ \pi(i) = j \}.
\end{equation}
Note that, for each fixed $i\in [n],$ we have a partition
\begin{equation} \label{cv}
\calS_{[n]} = \bigcup_{j\in [n]} T_n^{(i,j)}.
\end{equation}
We will denote for $i\in [n]$ and $\pi \in \calS_{[n]}$ the composition of permutations
\begin{equation}
    \pi_i := \pi \circ (1 \; i).
\end{equation}
If $i=1,$ then $\pi_i=\pi.$ Further, for each $i\in [n],$ as multiplication by $(1 \; i)$ is an automorphism of $\calS_{[n]},$ the mapping $\pi \mapsto \pi_i$ is a bijection of $\calS_{[n]}.$ In addition, when $i\neq 1,$
\begin{equation}
\mathrm{sgn}(\pi_i) = - \mathrm{sgn}(\pi).
\end{equation}

\subsection{Leading Coefficients} \label{cg}

We first show that the coefficient of $t^{d_n}$ in $F_{R,n}(t)$ is $\mathcal{R}_2 \det \bM_{R,n},$ i.e., that
\begin{equation} \label{cn}
 \sum_{i\in [n]} \sum_{\pi \in \calS_{[n]}} \mathrm{sgn}(\pi) \mathcal{R}_{i+1}\mathcal{R}_{\pi(i)+1}Q_{R}(\pi;i)= \mathcal{R}_2 \det \bM_{R,n}.
 \end{equation}
For each $i\in [n]\setminus \{1\},$ we have that
\begin{align}
    \sum_{\pi \in \calS_{[n]}} \mathrm{sgn}(\pi) \mathcal{R}_{\pi(i)+1}Q_{R}(\pi;i) &=\sum_{\pi \in \calS_{[n]}} \mathrm{sgn}(\pi) \mathcal{R}_{\pi(i)+1}\mathcal{R}_{\pi(1)+1}Q_{R}(\pi;1,i)  \\
    &=\sum_{\pi \in \calS_{[n]}} \mathrm{sgn}(\pi_i)\mathcal{R}_{\pi_i(i)+1}\mathcal{R}_{\pi_i(1)+1}Q_{R}(\pi_i;1,i)  \\
    &=-\sum_{\pi \in \calS_{[n]}} \mathrm{sgn}(\pi)\mathcal{R}_{\pi(1)+1}\mathcal{R}_{\pi(i)+1}Q_{R}(\pi;1,i)  \\
    &=-\sum_{\pi \in \calS_{[n]}} \mathrm{sgn}(\pi)\mathcal{R}_{\pi(i)+1}Q_{R}(\pi;i).
\end{align}
Hence,
\begin{equation}
    \sum_{\pi \in \calS_{[n]}} \mathrm{sgn}(\pi)\mathcal{R}_{\pi(i)+1}Q_{R}(\pi;i)=0.
\end{equation}
Thus, only $i=1$ could give a nonzero sum in the left hand side of~\eqref{cn}. Furthermore, when $i=1$ in~\eqref{cn}, we obtain the sum
\begin{align}
    \sum_{\pi \in \calS_{[n]}} \mathrm{sgn}(\pi) \mathcal{R}_{2} \mathcal{R}_{\pi(1)+1} Q_R(\pi;1) =\mathcal{R}_{2} \sum_{\pi \in \calS_{[n]}} \mathrm{sgn}(\pi)  Q_R(\pi) =\mathcal{R}_2 \det \bM_{R,n}.
\end{align}
Thus,~\eqref{cn} follows. In view of equation~\eqref{sd} in Lemma~\ref{kl}, equation~\eqref{cn} yields that $a_R^{n,d_n}=0.$ Next, we apply similar bijectivity tricks to show that $a_R^{n,d_n-1}= \det \bM_{R,n}.$

Via the application of Leibniz's formula in equation~\eqref{cl}, a preliminary formula for $a_R^{n,d_n-1}$ is as follows
\begin{align}
 a_R^{n,d_n-1}&=\mathcal{R}_2 \sum_{r\in [n]} \sum_{\pi \in \calS_{[n]}} \mathrm{sgn}(\pi)  \binom{r+\pi(r)}{2} \mathcal{R}_{r+\pi(r)-2} Q_{R}(\pi;r) \nonumber \\
&\hspace{5mm}- \sum_{(i,j)\in [n]^2}  \sum_{\pi \in T_n^{(i,j)}} \sum_{r\neq i} \mathrm{sgn}(\pi) \binom{r+\pi(r)}{2}  \mathcal{R}_{i+1}\mathcal{R}_{j+1} \mathcal{R}_{r+\pi(r)-2}   Q_{R}(\pi;i,r) \nonumber \\
&\hspace{5mm} -\sum_{(i,j)\in [n]^2} \sum_{\pi \in T_n^{(i,j)}} \mathrm{sgn}(\pi) Q_{R}(\pi;i)   \left( \binom{i}{2} \mathcal{R}_{i-1} \mathcal{R}_{j+1} +\binom{j}{2} \mathcal{R}_{i+1}\mathcal{R}_{j-1} \right),
\end{align}
where we set $\mathcal{R}_\ell = 0$ when $\ell<0.$ We will deal with each of the three sums in this preliminary formula separately; so, denote the three sums, in order, by $\mathfrak{S}_1,\mathfrak{S}_2,\mathfrak{S}_3$ (where, for $\mathfrak{S}_1,$ we absorb the factor $\mathcal{R}_2$ inside the sum), i.e., define
\begin{align}
  \mathfrak{S}_1 &:= \sum_{r\in [n]} \sum_{\pi \in \calS_{[n]}} \mathrm{sgn}(\pi)  \binom{r+\pi(r)}{2}  \mathcal{R}_2 \mathcal{R}_{r+\pi(r)-2} Q_{R}(\pi;r),   \\ \nonumber \\
\mathfrak{S}_2 &:= - \sum_{(i,j)\in [n]^2}  \sum_{\pi \in T_n^{(i,j)}} \sum_{r\neq i} \mathrm{sgn}(\pi) \binom{r+\pi(r)}{2}  \mathcal{R}_{i+1}\mathcal{R}_{j+1} \mathcal{R}_{r+\pi(r)-2}   Q_{R}(\pi;i,r),  \\ \nonumber  \\
\mathfrak{S}_3 &:= -\sum_{(i,j)\in [n]^2} \sum_{\pi \in T_n^{(i,j)}} \mathrm{sgn}(\pi) Q_{R}(\pi;i)   \left( \binom{i}{2} \mathcal{R}_{i-1} \mathcal{R}_{j+1} +\binom{j}{2} \mathcal{R}_{i+1}\mathcal{R}_{j-1} \right).
\end{align}
Thus, $a_R^{n,d_n-1}=\mathfrak{S}_1+\mathfrak{S}_2+\mathfrak{S}_3.$ We show that this coefficient is equal to $\det \bM_{R,n}$ by showing that
\begin{equation} \label{cq}
    \mathfrak{S}_1+\mathfrak{S}_2 = \det \bM_{R,n}
\end{equation}
and that
\begin{equation} \label{cr}
    \mathfrak{S}_3 = 0.
\end{equation}
For $i_1,\cdots,i_\ell \in [n],$ let $[n]_{i_1,\cdots,i_\ell} := [n]\setminus \{i_1,\cdots,i_\ell\}.$ 

For the first sum, $\mathfrak{S}_1,$ we partition $[n] \times \calS_{[n]}$ into three parts as
\begin{equation} \label{cp}
[n] \times \calS_{[n]} = \left( \{1\} \times T_n^{(1,1)} \right) \cup \left( \{1\} \times \bigcup_{j \in [n]_1}T_n^{(1,j)} \right) \cup \left( [n]_1 \times \calS_{[n]} \right),
\end{equation}
and we let $\mathfrak{S}_1 = \mathfrak{S}_{1,1}+\mathfrak{S}_{1,2}+\mathfrak{S}_{1,3}$ be the ensuing decomposition, which we express next. For the first part in (\ref{cp}), we obtain
\begin{equation} \label{ct}
   \mathfrak{S}_{1,1} = \sum_{\pi \in T_n^{(1,1)}} \mathrm{sgn}(\pi) Q_R(\pi),
\end{equation}
whereas the second and third parts give
\begin{equation} \label{dc}
    \mathfrak{S}_{1,2}=\sum_{j \in [n]_1} \sum_{\pi \in T_n^{(1,j)}} \mathrm{sgn}(\pi) \binom{j+1}{2} \mathcal{R}_2 \mathcal{R}_{j-1} Q_R(\pi;1)
\end{equation}
and 
\begin{equation} \label{dd}
    \mathfrak{S}_{1,3}=\hspace{-3mm}\sum_{(r,\pi)\in [n]_1 \times \calS_{[n]}} \hspace{-3mm}\mathrm{sgn}(\pi) \binom{r+\pi(r)}{2} \mathcal{R}_2 \mathcal{R}_{r+\pi(r)-2} Q_R(\pi;r),
\end{equation}
respectively. We will show that both $\mathfrak{S}_{1,2}$ and $\mathfrak{S}_{1,3}$ cancel out identically when summed with parts of the sum $\mathfrak{S}_2.$ We also note that $\mathfrak{S}_{1,1}$ provides part of the sum that will ultimately produce $\det \bM_{R,n}$; the remaining part lies in $\mathfrak{S}_2,$ which we treat next.

Let 
\begin{equation}
    U_n^{(i,j)} = \left\{ (i,j,\pi) \; ; \; \pi \in T_n^{(i,j)} \right\}.
\end{equation}
For the second sum, $\mathfrak{S}_2,$  we employ the partition
\begin{equation}
    \bigcup_{(i,j)\in [n]^2} \left( U_n^{(i,j)} \times [n]_i \right) = \mathfrak{p}_1\cup \mathfrak{p}_2 \cup \mathfrak{p}_3 \cup \mathfrak{p}_4
\end{equation}
where
\begin{align}
    \mathfrak{p}_1 &:=  \bigcup_{j\in [n]} \left( U_n^{(1,j)} \times [n]_1 \right) \\
    \mathfrak{p}_2 &:=  \bigcup_{i\in [n]_1} \left( U_n^{(i,1)} \times \{1\} \right) \\
    \mathfrak{p}_3 &:=  \bigcup_{(i,j)\in [n]_1\times [n]} \left( U_n^{(i,j)} \times [n]_{1,i} \right) \\
    \mathfrak{p}_4 &:=  \bigcup_{(i,j)\in [n]_1^2} \left( U_n^{(i,j)} \times \{1\} \right).
\end{align}
We will denote the ensuing sums by $\mathfrak{S}_{2,1}, \mathfrak{S}_{2,2}, \mathfrak{S}_{2,3}, \mathfrak{S}_{2,4},$ which we express next. We will denote a generic element $((i,j,\pi),r) \in U_n^{(i,j)} \times [n]_i$ as $(i,j,\pi,r)$ for short. The $\mathfrak{p}_1$-part yields
\begin{equation} \label{dg}
    \mathfrak{S}_{2,1} = - \sum_{(r,\pi) \in [n]_1 \times \calS_{[n]}}  \mathrm{sgn}(\pi) \binom{r+\pi(r)}{2} \mathcal{R}_2 \mathcal{R}_{r+\pi(r)-2} Q_R(\pi;r),
\end{equation}
the $\mathfrak{p}_2$-part yields
\begin{equation} \label{de}
    \mathfrak{S}_{2,2} = -  \sum_{i\in [n]_1} \sum_{\pi \in T_n^{(i,1)}}  \mathrm{sgn}(\pi) \binom{\pi(1)+1}{2} \mathcal{R}_2 \mathcal{R}_{\pi(1)-1} Q_R(\pi;1),
\end{equation}
and the $\mathfrak{p}_3$-part yields
\begin{align}
    \mathfrak{S}_{2,3} = - &\sum_{i \in [n]_1} \sum_{r\in [n]_{1,i}} \sum_{\pi \in \calS_{[n]}} \mathrm{sgn}(\pi) \binom{r+\pi(r)}{2}   \mathcal{R}_{i+1} \mathcal{R}_{\pi(i)+1}\mathcal{R}_{r+\pi(r)-2}Q_R(\pi;i,r). \label{df}
\end{align}
We further partition the $\mathfrak{p}_4$-part according to whether permutations fix $1,$ namely,
\begin{equation}
    \mathfrak{p}_4 = \mathfrak{p}_{4,1} \cup \mathfrak{p}_{4,2}
\end{equation}
where
\begin{align}
    \mathfrak{p}_{4,1} &:= \left\{ (i,j,\pi,1) \in \mathfrak{p}_4 \; ; \; \pi(1)=1 \right\} \\
    \mathfrak{p}_{4,2} &:= \left\{ (i,j,\pi,1) \in \mathfrak{p}_4 \; ; \; \pi(1)\neq 1 \right\}.
\end{align}
We denote the ensuing sums by $\mathfrak{S}_{2,4,1},\mathfrak{S}_{2,4,2}.$ The $\mathfrak{p}_{4,1}$-part gives
\begin{equation} \label{cw}
    \mathfrak{S}_{2,4,1} = -  \sum_{(i,j)\in [n]_1^2} \sum_{\pi \in T_n^{(i,j)}\cap T_n^{(1,1)}}  \mathrm{sgn}(\pi) \mathcal{R}_{i+1}\mathcal{R}_{\pi(i)+1}Q_R(\pi;i,1),
\end{equation}
whereas the $\mathfrak{p}_{4,2}$-part gives
\begin{align}
    \mathfrak{S}_{2,4,2} &= -\sum_{(i,j,k)\in [n]_1^3}\sum_{\pi \in T_n^{(i,j)}\cap T_n^{(1,k)}} \mathrm{sgn}(\pi) \binom{\pi(1)+1}{2} \mathcal{R}_{i+1}\mathcal{R}_{\pi(i)+1}\mathcal{R}_{\pi(1)-1} Q_{R}(\pi;1,i). \label{di}
\end{align}

With this decomposition of $a_R^{n,d_n-1}$ at hand, we proceed to show that equations~\eqref{cq} and~\eqref{cr} hold by showing that the following six equations hold. We will show that
\begin{align}
    \mathfrak{S}_{1,1} + \mathfrak{S}_{2,4,1} &= \det \bM_{R,n}, \label{cs} \\
    \mathfrak{S}_{1,2} + \mathfrak{S}_{2,2} &= 0, \label{cx} \\
    \mathfrak{S}_{1,3} + \mathfrak{S}_{2,1} &= 0, \label{cy} \\
    \mathfrak{S}_{2,3} &= 0, \label{cz} \\
    \mathfrak{S}_{2,4,2} &= 0, \label{da} \\
    \mathfrak{S}_3 &= 0. \label{db}
\end{align}
Summing~\eqref{cs}--\eqref{db} gives $a_R^{n,d_n-1}=\det \bM_{R,n}.$

We first show that~\eqref{cs} holds. From~\eqref{ct}, we have that
\begin{equation}
   \mathfrak{S}_{1,1} = \sum_{\pi \in T_n^{(1,1)}} \mathrm{sgn}(\pi) Q_R(\pi).
\end{equation}
We show that $\mathfrak{S}_{2,4,1}$ complements this summation to give $\det \bM_{R,n},$ i.e., that $(\ref{cs})$ holds. From the Leibniz formula for the determinant, we have that
\begin{equation}
   \det \bM_{R,n} = \sum_{\pi \in \calS_{[n]}} \mathrm{sgn}(\pi) Q_R(\pi).
\end{equation}
From the partition $\calS_{[n]} = \bigcup_{i\in [n]} T_n^{(i,1)}$ (similar to the partition in (\ref{cv})), then, it suffices to show that
\begin{equation} \label{cu}
    \mathfrak{S}_{2,4,1} = \sum_{i\in [n]_1} \sum_{\pi \in T_n^{(i,1)}} \mathrm{sgn}(\pi) Q_R(\pi).
\end{equation}
We proceed to show that (\ref{cu}) holds. We employ a similar technique to how we showed (\ref{cn}). Fix $i\in [n]_1.$ By the change of variables $\sigma = \pi_i$ (equivalently, $\pi = \sigma_i,$ since $(1 \; i)^{-1} = (1 \; i)$) we have that
\begin{align}
    \sum_{j\in [n]_1} \sum_{\pi \in T_n^{(i,j)} \cap T_n^{(1,1)}} \hspace{-5mm} \mathrm{sgn}(\pi) \mathcal{R}_{i+1}\mathcal{R}_{\pi(i)+1}Q_R(\pi;i,1) &=\sum_{j\in [n]_1} \sum_{\sigma \in T_n^{(i,1)} \cap T_n^{(1,j)}} \hspace{-3mm} \mathrm{sgn}(\sigma_i) \mathcal{R}_{i+1}\mathcal{R}_{\sigma_i(i)+1}Q_R(\sigma_i;i,1) \\
    &= -\sum_{j\in [n]_1} \sum_{\sigma \in T_n^{(i,1)} \cap T_n^{(1,j)}} \hspace{-3mm} \mathrm{sgn}(\sigma) \mathcal{R}_{i+\sigma(i)}\mathcal{R}_{\sigma(1)+1}Q_R(\sigma;i,1) \\
    &= -\sum_{j\in [n]_1} \sum_{\sigma \in T_n^{(i,1)} \cap T_n^{(1,j)}} \hspace{-3mm} \mathrm{sgn}(\sigma) Q_R(\sigma) \\
    &= -\sum_{\sigma \in T_n^{(i,1)}}\mathrm{sgn}(\sigma) Q_R(\sigma).
\end{align}
Summing over all $i\in [n]_1$ and noting the minus sign in the definition of $\mathfrak{S}_{2,4,1}$ in (\ref{cw}), we obtain that (\ref{cu}) holds. Hence, equation (\ref{cs}) holds. We now show that the other parts give a vanishing contribution, i.e., that~\eqref{cx}--\eqref{db} all hold.

Equations (\ref{cx}) and (\ref{cy}) follow from the expressions we we give in (\ref{dc}),(\ref{dd}),(\ref{dg}),(\ref{de}). For (\ref{cx}), we note that each $j$ in the summand in (\ref{dc}) may be replaced with $\pi(1)$ as $\pi \in T_n^{(1,j)}.$ Then, as
\begin{equation}
    \bigcup_{j \in [n]_1} T_n^{(1,j)} = \bigcup_{i\in [n]_1} T_n^{(i,1)}
\end{equation}
are both partitions of the same set, namely, the set of permutations that do not fix $1,$ equations~\eqref{dc} and~\eqref{de} yield that $\mathfrak{S}_{1,2} = - \mathfrak{S}_{2,2},$ i.e., (\ref{cx}) holds. For (\ref{cy}), the expressions in (\ref{dd}) and (\ref{dg}) show that $\mathfrak{S}_{1,3} = - \mathfrak{S}_{2,1},$ i.e., that (\ref{cy}) holds.

Next, we show that (\ref{cz}) holds. Fix $i\in [n]_1$ and $r\in [n]_{1,i}.$ We will show that the following sum vanishes
\begin{equation} \label{dh}
    \sum_{\pi \in \calS_{[n]}} \mathrm{sgn}(\pi) \binom{r+\pi(r)}{2} \mathcal{R}_{\pi(i)+1}\mathcal{R}_{r+\pi(r)-2}Q_R(\pi;i,r)=0.
\end{equation}
As $\mathfrak{S}_{2,3},$ according to equation (\ref{df}), is a linear combination of such sums, we would obtain that $\mathfrak{S}_{2,3}=0,$ i.e., that (\ref{cz}) holds. To show that (\ref{dh}) holds, we utilize that $\pi \mapsto \pi_i$ is an automorphism of $\calS_{[n]},$ as follows. We have that
\begin{align}
    \sum_{\pi \in \calS_{[n]}} &\mathrm{sgn}(\pi) \binom{r+\pi(r)}{2} \mathcal{R}_{\pi(i)+1}\mathcal{R}_{r+\pi(r)-2}Q_R(\pi;i,r) \nonumber \\
    &= \sum_{\pi \in \calS_{[n]}} \mathrm{sgn}(\pi) \binom{r+\pi(r)}{2}\mathcal{R}_{\pi(i)+1}\mathcal{R}_{r+\pi(r)-2}  \mathcal{R}_{\pi(1)+1}Q_R(\pi;1,i,r) \\
    &=  \sum_{\pi \in \calS_{[n]}}  \mathrm{sgn}(\pi_i) \binom{r+\pi_i(r)}{2}\mathcal{R}_{\pi_i(i)+1}\mathcal{R}_{r+\pi_i(r)-2}  \mathcal{R}_{\pi_i(1)+1}Q_R(\pi_i;1,i,r) \\
    &= -  \sum_{\pi \in \calS_{[n]}}  \mathrm{sgn}(\pi)  \binom{r+\pi(r)}{2}\mathcal{R}_{\pi(1)+1}\mathcal{R}_{r+\pi(r)-2}  \mathcal{R}_{\pi(i)+1}Q_R(\pi;1,i,r) \\
    &= -  \sum_{\pi \in \calS_{[n]}}  \mathrm{sgn}(\pi) \binom{r+\pi(r)}{2}\mathcal{R}_{r+\pi(r)-2}\mathcal{R}_{\pi(i)+1}Q_R(\pi;i,r),
\end{align}
so the vanishing in (\ref{dh}) holds. Hence, (\ref{cz}) holds.

Next, we show that (\ref{da}) holds. We rewrite (\ref{di}) as
\begin{align}
    \mathfrak{S}_{2,4,2} &= -\sum_{(i,j,k,\ell)\in [n]_1^4}\sum_{\pi \in T_n^{(i,j)}\cap T_n^{(1,k)}\cap T_n^{(\ell,1)}} \mathrm{sgn}(\pi) \binom{k+1}{2}  \mathcal{R}_{i+1}\mathcal{R}_{j+1}\mathcal{R}_{k-1}\mathcal{R}_{\ell+1} Q_{R}(\pi;1,i,\ell).  \label{dj}
\end{align}
We fix $(j,k)\in [n]_1^2$ and show the vanishing of each of the following sums
\begin{align} \label{dk}
    & \sum_{(i,\ell)\in [n]_1^2} \mathcal{R}_{i+1}\mathcal{R}_{\ell+1} \sum_{\pi \in T_n^{(i,j)}\cap T_n^{(1,k)}\cap T_n^{(\ell,1)}}  \mathrm{sgn}(\pi)   Q_{R}(\pi;1,i,\ell) = 0.
\end{align}
From (\ref{dj}), we may write $\mathfrak{S}_{2,4,2}$ as a linear combination of such sums, so we would obtain that $\mathfrak{S}_{2,4,2}=0.$ We show (\ref{dk}) next. We change variables as $\sigma = \pi \circ (i \; \; \ell)$ in the inner sum in (\ref{dk}) to obtain that 
\begin{align}
    \sum_{\pi \in T_n^{(i,j)}\cap T_n^{(1,k)}\cap T_n^{(\ell,1)}}  \mathrm{sgn}(\pi)   Q_{R}(\pi;1,i,\ell) &= \sum_{\sigma \in T_n^{(\ell,j)}\cap T_n^{(1,k)}\cap T_n^{(i,1)}}  \mathrm{sgn}(\sigma \circ (i \;\; \ell))   Q_{R}(\sigma \circ (i \;\; \ell);1,i,\ell)  \\
    &= -\sum_{\sigma \in T_n^{(\ell,j)}\cap T_n^{(1,k)}\cap T_n^{(i,1)}}   \mathrm{sgn}(\sigma )   Q_{R}(\sigma;1,i,\ell).
\end{align}
Multiplying by $\mathcal{R}_{i+1}\mathcal{R}_{\ell+1}$ then summing over $(i,\ell)\in [n]_1^2,$ we obtain that the quantity on the left hand side of (\ref{dk}) is equal to its negative. Hence (\ref{dk}) holds, and we obtain that $\mathfrak{S}_{2,4,2}=0.$

Finally, we show that $\mathfrak{S}_3=0.$ We may write
\begin{align}
    \mathfrak{S}_3 = &-\hspace{-1mm}\sum_{i\in [n]_1} \hspace{-1mm} \binom{i}{2} \mathcal{R}_{i-1} \hspace{-1mm} \sum_{\pi \in \calS_{[n]}} \hspace{-1mm}\mathrm{sgn}(\pi)   \mathcal{R}_{\pi(i)+1}  Q_{R}(\pi;i)  -\hspace{-1mm}\sum_{j\in [n]_1} \hspace{-1mm}\binom{j}{2}\mathcal{R}_{j-1}\hspace{-1mm}  \sum_{\pi \in \calS_{[n]}}\hspace{-1mm} \mathrm{sgn}(\pi)  \mathcal{R}_{\pi^{-1}(j)+1} Q_{R}(\pi;\pi^{-1}(j)). \label{dn}
\end{align}
We will show that for each $(i,j)\in [n]_1^2,$
\begin{equation} \label{dl}
    \sum_{\pi \in \calS_{[n]}} \mathrm{sgn}(\pi)   \mathcal{R}_{\pi(i)+1}  Q_{R}(\pi;i) = 0
\end{equation}
and 
\begin{equation} \label{dm}
    \sum_{\pi \in \calS_{[n]}} \mathrm{sgn}(\pi)  \mathcal{R}_{\pi^{-1}(j)+1} Q_{R}(\pi;\pi^{-1}(j))=0.
\end{equation}
Together, equations (\ref{dl}) and (\ref{dm}) imply in view of (\ref{dn}) that $\mathfrak{S}_3=0.$ To show that (\ref{dl}) holds, we apply the automorphism $\pi \mapsto \pi_i$ of $\calS_{[n]},$ from which we obtain
\begin{align}
    \sum_{\pi \in \calS_{[n]}} \mathrm{sgn}(\pi)   \mathcal{R}_{\pi(i)+1}  Q_{R}(\pi;i)  &= \sum_{\pi \in \calS_{[n]}} \mathrm{sgn}(\pi)   \mathcal{R}_{\pi(i)+1} \mathcal{R}_{\pi(1)+1} Q_{R}(\pi;1,i)  \\
    &= \sum_{\pi \in \calS_{[n]}} \mathrm{sgn}(\pi_i)   \mathcal{R}_{\pi_i(i)+1} \mathcal{R}_{\pi_i(1)+1} Q_{R}(\pi_i;1,i)  \\
    &=-\sum_{\pi \in \calS_{[n]}} \mathrm{sgn}(\pi)   \mathcal{R}_{\pi(1)+1} \mathcal{R}_{\pi(i)+1} Q_{R}(\pi;1,i)  \\
    &=-\sum_{\pi \in \calS_{[n]}} \mathrm{sgn}(\pi)    \mathcal{R}_{\pi(i)+1} Q_{R}(\pi;i)
\end{align}
and (\ref{dl}) follows. Now, we show that (\ref{dm}) reduces to (\ref{dl}) via the automorphism $\pi \mapsto \pi^{-1}$ of $\calS_{[n]}.$ First, note that for any $\pi \in \calS_{[n]}$ we have 
\begin{equation}
    Q_R(\pi;\pi^{-1}(k_1),\cdots,\pi^{-1}(k_\ell)) = Q_R(\pi^{-1};k_1,\cdots,k_\ell).
\end{equation}
Hence, the left hand side of (\ref{dm}) may be rewritten as
\begin{align}
     \sum_{\pi \in \calS_{[n]}} \mathrm{sgn}(\pi)  \mathcal{R}_{\pi^{-1}(j)+1} Q_{R}(\pi;\pi^{-1}(j))  &=    \sum_{\pi \in \calS_{[n]}}   \mathrm{sgn}(\pi)  \mathcal{R}_{\pi^{-1}(j)+1} \mathcal{R}_{\pi^{-1}(1)+1} Q_{R}(\pi;\pi^{-1}(j),\pi^{-1}(1))  \\
     &=   \sum_{\pi \in \calS_{[n]}} \mathrm{sgn}(\pi)  \mathcal{R}_{\pi^{-1}(j)+1} \mathcal{R}_{\pi^{-1}(1)+1} Q_{R}(\pi^{-1};j,1). \label{do}
\end{align}
Further, the bijection $\pi \mapsto \pi^{-1}$ yields that
\begin{align}
     \sum_{\pi \in \calS_{[n]}} \mathrm{sgn}(\pi)  \mathcal{R}_{\pi^{-1}(j)+1} \mathcal{R}_{\pi^{-1}(1)+1} Q_{R}(\pi^{-1};j,1) &=   \sum_{\pi \in \calS_{[n]}} \mathrm{sgn}(\pi^{-1})  \mathcal{R}_{\pi(j)+1} \mathcal{R}_{\pi(1)+1} Q_{R}(\pi;j,1)  \\
     &= \sum_{\pi \in \calS_{[n]}} \mathrm{sgn}(\pi)  \mathcal{R}_{\pi(j)+1} \mathcal{R}_{\pi(1)+1} Q_{R}(\pi;j,1). \label{dp}
\end{align}
Combining (\ref{do}) and (\ref{dp}), we get that 
\begin{align}
    \sum_{\pi \in \calS_{[n]}} \mathrm{sgn}(\pi)  \mathcal{R}_{\pi^{-1}(j)+1} Q_{R}(\pi;\pi^{-1}(j))  = \sum_{\pi \in \calS_{[n]}} \mathrm{sgn}(\pi)  \mathcal{R}_{\pi(j)+1} \mathcal{R}_{\pi(1)+1} Q_{R}(\pi;j,1),
\end{align}
i.e., the left hand side of (\ref{dm}) is equal to the left hand side of (\ref{dl}) with $j$ in place of $i.$ As (\ref{dl}) holds, (\ref{dm}) holds too. Therefore, $\mathfrak{S}_3=0.$ This concludes the proof that the coefficient $a_R^{n,d_n-1}$ is equal to $\det \bM_{R,n}.$

\end{appendices}

\bibliographystyle{IEEEtran}
\bibliography{Biblio}

\end{document}